\documentclass[11pt]{report}
\usepackage[utf8]{inputenc}
\usepackage[english]{babel}
\usepackage{graphicx}
\graphicspath{{images/}}
\usepackage{subfiles}
\usepackage{blindtext}
\usepackage[T1]{fontenc}
\usepackage[PetersLenny]{fncychap}
\usepackage{fancyhdr}
\usepackage{tikz}
\usetikzlibrary{quantikz}
\usepackage{lscape}
\usepackage{subfig}
\usepackage{float}
\usepackage{epigraph}
\usepackage{cite}

\usepackage{amsfonts,amssymb,amsmath,amsthm,longtable,array}
\usepackage{graphics,times,bm,bbm,bbold,amssymb,soul}
\usepackage{dsfont,color,xcolor,dcolumn}
\usepackage{epstopdf,epsfig}
\usepackage{mathtools,calc}

\DeclareMathOperator{\tr}{tr}
\DeclareMathOperator{\rk}{rk}
\DeclareMathOperator{\srk}{srk}
\DeclareMathOperator{\brk}{brk}

\DeclareMathOperator{\W}{W}

\DeclareMathOperator{\LL}{L}
\DeclareMathOperator{\M}{M}
\DeclareMathOperator{\N}{N}
\DeclareMathOperator{\Y}{Y}

\DeclareMathOperator{\Span}{Span}

\DeclareMathOperator{\sym}{sym}

\definecolor{Bittersweet}{HTML}{C04F17}

\usepackage[bookmarksnumbered,colorlinks,plainpages,linktocpage]{hyperref}
\definecolor{Darkgreen}{rgb}{0,0.4,0}
\hypersetup{%
    pdfborder = {0 0 0},
    colorlinks,
    citecolor=blue,
    filecolor=Darkgreen,
    linkcolor=Bittersweet,
    urlcolor=blue}

\usepackage{tabularx,ragged2e}

\makeatletter
\newcommand{\xdashrightarrow}[2][]{\ext@arrow 0359\rightarrowfill@@{#1}{#2}}
\def\rightarrowfill@@{\arrowfill@@\relax\relbar\rightarrow}
\def\arrowfill@@#1#2#3#4{%
  $\m@th\thickmuskip0mu\medmuskip\thickmuskip\thinmuskip\thickmuskip
   \relax#4#1
   \xleaders\hbox{$#4#2$}\hfill
   #3$%
}
\makeatother

\newcommand\xxoverset[3]{%
  \resizebox{#1+\widthof{\scriptsize #2}}{\height}{$#3$}}
\newcommand\extoverset[3][0pt]{%
  \mathrel{\overset{\textup{#2}}{\xxoverset{#1}{#2}{#3}}}}


\newtheorem{lemma}{Lemma}[chapter]

\newtheorem{proposition}{Proposition}[chapter]
\newtheorem{theorem}{Theorem}[chapter]
\newtheorem{corollary}{Corollary}[chapter]
\newtheorem{conjecture}{Conjecture}[chapter]

\theoremstyle{definition}
\newtheorem{definition}{Definition}[chapter]
\newtheorem{example}{Example}[chapter]
\newtheorem{remark}{Remark}[chapter]

\PassOptionsToPackage{breaklinks}{hyperref}

\usepackage{cleveref}
\crefname{proposition}{Proposition}{Propositions}
\crefname{definition}{Definition}{Definitions}
\crefname{lemma}{Lemma}{Lemmas}
\crefname{figure}{Figure}{Figures}
\crefname{theorem}{Theorem}{Theorems}
\crefname{corollary}{Corollary}{Corollary}
\crefname{conjecture}{Conjecture}{Conjectures}
\crefname{section}{Section}{Sections}
\crefname{appendix}{Appendix}{Appendices}
\crefname{observation}{Observation}{Observation}
\crefname{question}{Question}{Questions}
\crefname{remark}{Remark}{Remark}
\crefname{example}{Example}{Examples}
\crefname{equation}{Eq.}{Eqs.}
\crefname{table}{Table}{Tables}
\crefname{chapter}{Chapter}{Chapters}
\makeatother

\usepackage{pdfpages} 

\begin{document}
 

\begin{titlepage}
	\centering

\begin{figure}[!htb]
\minipage{0.41\textwidth}
  \includegraphics[scale=0.1]{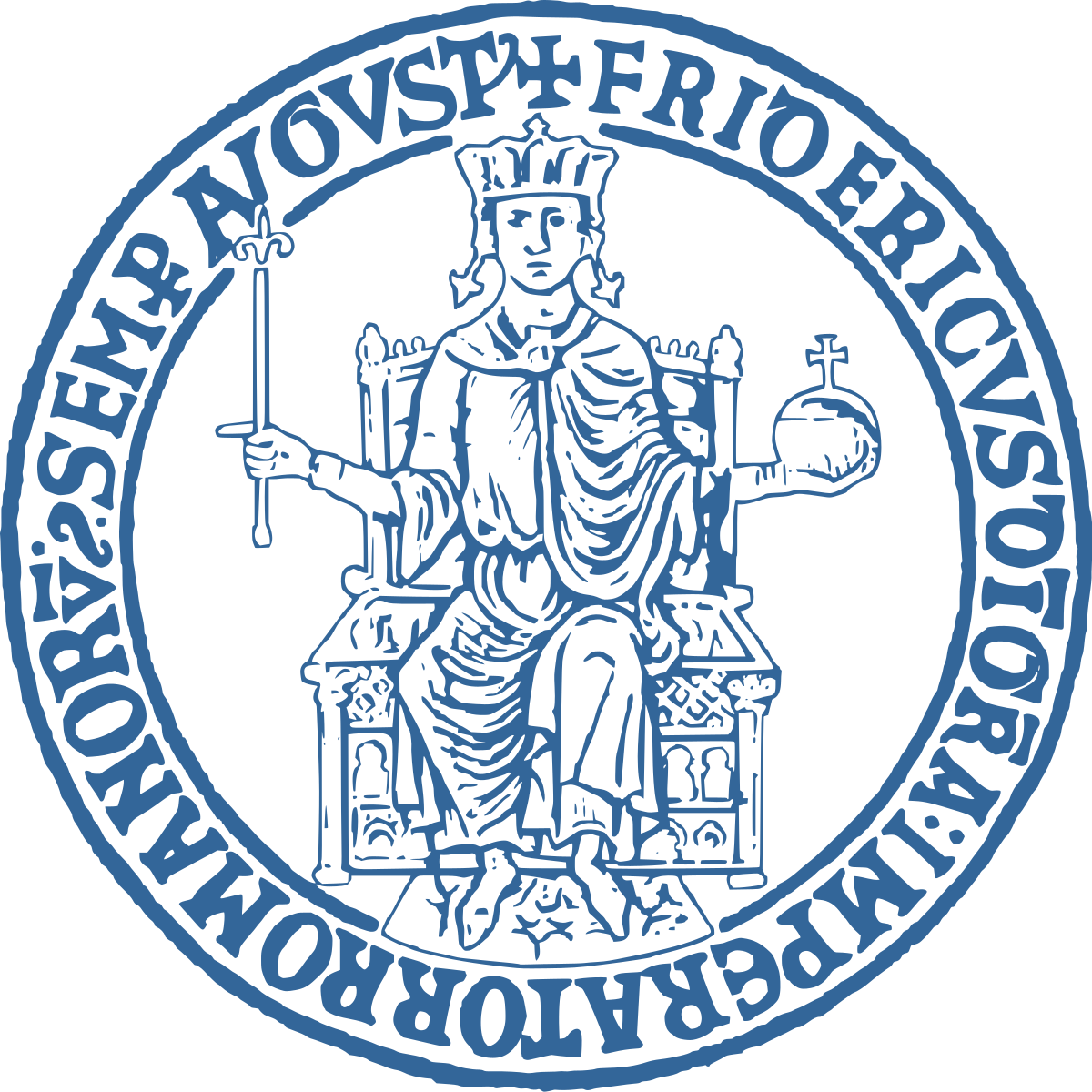}
\endminipage\hfill
\minipage{0.33\textwidth}
  \includegraphics[scale=0.15]{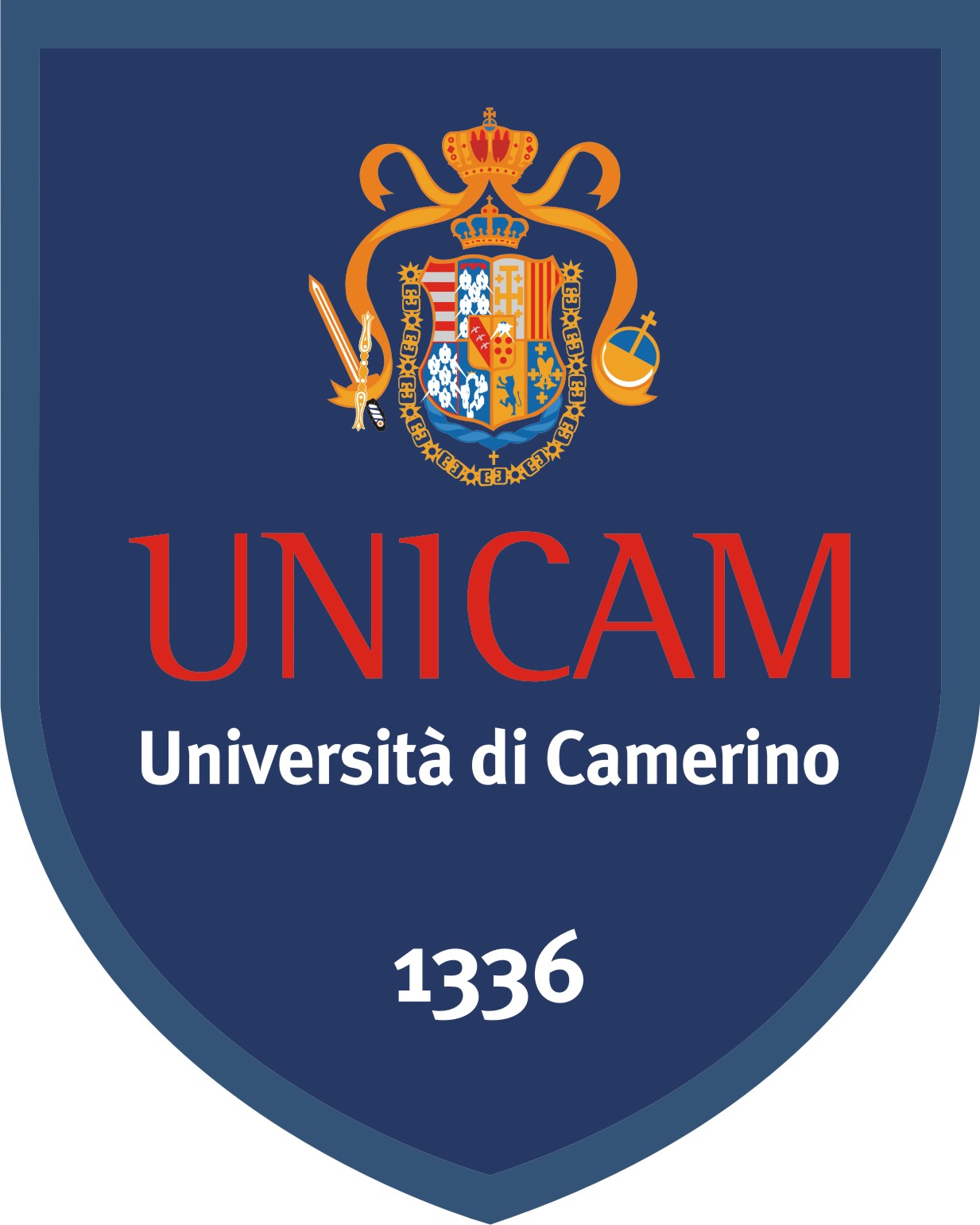}
\endminipage\hfill
\minipage{0.24\textwidth}%
  \includegraphics[scale=2.2]{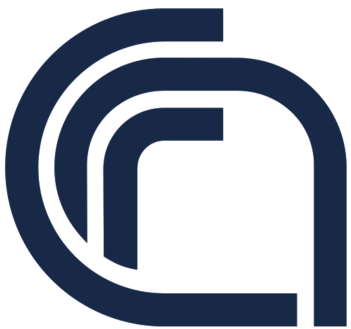}
\endminipage
\end{figure}

	\par\vspace{0.7cm}
	
	{\large\scshape\bfseries Dottorato di Ricerca in Quantum Technologies \par}
	\vspace{0.25cm}
	{\large\scshape\bfseries XXXV Cycle \par}
	\vspace{1.5cm}
	
	{\LARGE\bfseries \mbox{Classifying~Entanglement~by~Algebraic~Geometry}\par}
	\vspace{1.5cm}
	
    {\bfseries PhD Dissertation by \par}
	\vspace{0.25cm}	
	{\huge Masoud Gharahi\par}
	\vspace{1.5cm}

	{\bfseries Settore Scientifico Disciplinare FIS/02 \par}
	
	\vspace{1.5cm}
	
	{\bfseries Tutor \par}
	\vspace{0.25cm}
	{\Large Prof. Stefano Mancini}\par
	\vspace{2.1cm}
	
	{\large Camerino, Italy, a.y. 2019/2022}
\end{titlepage}


\newpage
\begin{titlepage}
	\begin{flushright}
	\mbox{}
    \vfill
	{\Large \textit{``Pure mathematics is the world's best game. It is more absorbing than chess, more of a gamble than poker, and lasts longer than Monopoly. It's free. It can be played anywhere - Archimedes did it in a bathtub.''}\par}
	\vspace{1cm}
	{\Large — Richard J. Trudea\par}
	\end{flushright}
\end{titlepage}
	

\newpage
\begin{titlepage}
	\begin{flushright}
	\mbox{}
     \vfill
     {\Large Dedicated\par}
    \vspace{1cm}
    {\Large \textit{to my wonderful mother}\par}
	\vspace{0.5cm}
	{\Large \textit{for her endless love}\par}
	\vspace{1cm}
	{\Large \textit{and to the memory of my father}\par}
	\vspace{0.5cm}
	{\Large \textit{the first and the best math tutor of my life}\par}
	\end{flushright}
\end{titlepage}

\newpage
\thispagestyle{empty}

\begin{flushleft}
    {\Large \textbf{Masoud Gharahi}\par}
    \vspace{0.1cm}
    {\Large School of Science and Technology\par}
    \vspace{0.1cm}
    {\Large University of Camerino, Italy \par}
    \vspace{0.1cm}
    {\Large \href{mailto:masoud.gharahi@gmail.com}{masoud.gharahi@gmail.com}\par}

\vfill

\[
\large
\begin{array}{lcccl}
\text{Date of Submission}  & && & \text{December 01, 2022} \\ [2ex]
\text{Date of Defense}     & && & \text{February 08, 2023} \\ [2ex]
\text{Advisor}          & && & \textbf{Prof. Stefano Mancini} \\
                           & && & \text{University of Camerino, Camerino, Italy} \\ [2ex]
\text{Assessment Commitee} & && & \textbf{Prof. Karol \.{Z}yczkowski} \\
                           & && & \text{Jagiellonian University, Krak\'{o}w, Poland} \\[2ex]
                           & && & \textbf{Prof. Fabio Benatti} \\
                           & && & \text{University of Trieste, Trieste, Italy} \\ [2ex]
                           & && & \textbf{Prof. Paolo Facchi} \\
                           & && & \text{University of Bari, Bari, Italy}
\end{array}
\]

\end{flushleft}


\newpage
\thispagestyle{empty}
\addcontentsline{toc}{chapter}{\numberline{}Abstract}
\begin{center}
  \textbf{\Large Classifying Entanglement by Algebraic Geometry}

  \vspace*{1cm}
  \textbf{\large Masoud Gharahi}

  \vspace*{0.5cm}
  {\large Submitted for the degree of Doctor of Philosophy\\
  December 2022}

  \vspace*{1cm}
  \textbf{\large Abstract}
\end{center}

Quantum Entanglement is one of the key manifestations of quantum mechanics that separate the quantum realm from the classical one. Characterization of entanglement as a physical resource for quantum technology became of uppermost importance. While the entanglement of bipartite systems is already well understood, the ultimate goal to cope with the properties of  entanglement of multipartite systems is still far from being realized. This dissertation covers characterization of multipartite entanglement using algebraic-geometric tools. Firstly, we establish an algorithm to classify multipartite entanglement by $k$-secant varieties of the Segre variety and $\ell$-multilinear ranks that are invariant under Stochastic Local Operations with Classical Communication (SLOCC). We present a fine-structure classification of multiqubit and tripartite entanglement based on this algorithm.
Another fundamental problem in quantum information theory is entanglement transformation that is quite challenging regarding to multipartite systems. It is captivating that the proposed entanglement classification by algebraic geometry can be considered as a reference to study SLOCC and asymptotic SLOCC interconversions among different resources based on tensor rank and border rank, respectively. In this regard, we also introduce a new class of tensors that we call \emph{persistent tensors} and construct a lower bound for their tensor rank. We further cover SLOCC convertibility of multipartite systems considering several families of persistent tensors.

\setcounter{page}{1}
\pagenumbering{roman}

\chapter*{Declaration}
\addcontentsline{toc}{chapter}{\numberline{}Declaration}

The work in this dissertation is based on research carried out at the
University of Camerino, School of Science and Technology, Italy. No part of this dissertation has been submitted elsewhere for any other degree or qualification and it is all my own work unless referenced to the contrary in the text.

\vspace{2in}
\noindent \textbf{Copyright \copyright\; MASOUD GHARAHI, 2022}.\\
``The copyright of this dissertation rests with the author.  No quotations
from it should be published without the author's prior written consent and information derived from it should be acknowledged''.


\chapter*{Acknowledgments}
\addcontentsline{toc}{chapter}{\numberline{}Acknowledgments}

Here, I would like to take this opportunity to thank everyone for their help and support during the last  few years that I have been living outside of my homeland.

First and foremost I would like to thank my supervisor Stefano Mancini for giving me the opportunity to do my doctoral studies within his research group. Not only as a professor, he guided me in the academic ﬁeld, but also as a friend in life. Grazie infinite Stefano per il tuo continuo supporto.

During my Ph.D., I had the honor to work with bright mathematicians and physicists. I remember the day I contacted Giorgio Ottaviani to ask some \mbox{questions} which started my journey to algebraic geometry. I continued this connection and learned a lot from him. Grazie mille Giorgio for your willingness to share your knowledge. Subsequently, I would like to thank Alessandra Bernardi, Jarosław Buczyński, Fulvio Gesmundo, Frédéric Holweck, Joachim Jelisiejew, Luke \mbox{Oeding} and Adam Sawicki for all delightful and fruitful discussions. I would like to \mbox{acknowledge} the hospitality of the QMATH Center at the University of \mbox{Copenhagen} where I had a long-term visit. There, I had an enjoyable collaboration with Vladimir Lysikov on the Persistent Tensors project. Thank you Vladimir for your guidance and persistence until a solution was found.

\sloppy Here, I also would like to thank my former supervisor Seyed Javad Akhtarshenas who introduced me the field of quantum information theory.

I would like to thank a true old friend. Thank you Pedram for all the interesting conversations and discussions whether it was mathematics or any other subjects.

Last but not least, I would like to thank my lovely family for always being there. All of this would not have been possible without the love coming from my family.

\pagestyle{fancy} \fancyhead[LO,RE]{}
\fancyhead[LE]{\nouppercase{\bfseries \leftmark}}
\fancyhead[RO]{\nouppercase{\bfseries \rightmark}}
\setcounter{tocdepth}{3}
\tableofcontents
\listoffigures
\listoftables
\pagestyle{plain}
\newpage
\chapter*{List of Publications}
\addcontentsline{toc}{chapter}{\numberline{}Publications}
{\bf{This PhD Dissertation is based on the following publications available online}}
\vspace{1cm}
\begin{enumerate}
\item[{[A]}]
\textbf{M. Gharahi} and S. Mancini,
\textit{Comment on ``Inductive entanglement classification of four qubits under stochastic local operations and classical communication''}, 
\href{https://doi.org/10.1103/PhysRevA.98.066301}{Phys. Rev. A \textbf{98}, 066301 (2018)}.

\hfill (contains results presented in \Cref{chap2})
\item[{[B]}]
\textbf{M. Gharahi}, S. Mancini, and G. Ottaviani,
\textit{Fine-structure classification of multiqubit entanglement by algebraic geometry}, 
\href{https://doi.org/10.1103/PhysRevResearch.2.043003}{Phys. Rev. Research \textbf{2}, 043003 (2020)}.

\hfill (contains results presented in Chapters \ref{chap3} and \ref{chap4})
\item[{[C]}]
\textbf{M. Gharahi} and S. Mancini,
\textit{Algebraic-geometric characterization of tripartite entanglement}, 
\href{https://doi.org/10.1103/PhysRevA.104.042402}{Phys. Rev. A \textbf{104}, 042402 (2021)}.

\hfill (contains results presented in Chapters \ref{chap3} and \ref{chap5})
\item[{[D]}]
\textbf{M. Gharahi} and V. Lysikov, 
\textit{Persistent Tensors and Multiqudit Entanglement Transformation},
\href{https://doi.org/10.22331/q-2024-01-31-1238}{Quantum \textbf{8}, 1238 (2024)}.

\hfill (contains results presented in Chapter \ref{chap6})
\end{enumerate}

\newpage
\chapter*{List of Acronyms}
\addcontentsline{toc}{chapter}{\numberline{}Acronyms}

\begin{table}[h!]
\begin{tabular}{ll}
CPTP & Completely Positive Trace-Preserving \\
EPR & Einstein–Podolsky–Rosen \\
$\rm{GHZ}$ & Greenberger-Horne-Zeilinger \\
$\rm{GL}$ & General Linear \\
iff & if and only if \\
LOCC &  Local Operations with Classical Communication \\
LU &  Local Unitary \\
MES & Maximally Entangled State \\
POVM & Positive Operator-Valued Measure \\
$\rm{PGL}$ & Projective General Linear \\
$\rm{SL}$ & Special Linear \\
$\rm{SO}$ & Special Orthogonal \\
$\rm{SU}$ & Special Unitary \\
SLOCC &  Stochastic Local Operations with Classical Communication \\
\end{tabular}
\end{table}

\pagestyle{fancy} \fancyhead[LO,RE]{}
\fancyhead[LE]{\nouppercase{\bfseries \leftmark}}
\fancyhead[RO]{\nouppercase{\bfseries \rightmark}}
\chapter{Introduction}
\setcounter{page}{1}
\pagenumbering{arabic}

\epigraph{``When I consider what people generally want in calculating, I found that it always is a number.''}{\textit{Mohammad ibn Musa Kharazmi}}

\section{Overview}
Quantum entanglement is one of the most strinking features of quantum mechanics. Actually it is synonimous of nonclassical correlations between systems that can be also spatially separated. This peculiar phenomenon first described by Albert Einstein with his two postdoctoral research associates Boris Podolsky and Nathan Rosen in their seminal paper entitled ``Can Quantum-Mechanical Description of Physical Reality Be Considered Complete?'' in 1935 \cite{EPR35}. In this paper, they formulated an apparent paradox of quantum theory which is known as EPR paradox. Shortly after, Erwin Schr\"odinger coined the term ``entanglement''\footnote{Verschr\"ankung in German which means ``folding the arms''.} in his seminal paper \cite{Schrodinger35} to describe this peculiar correlation between quantum systems.

Although quantum entanglement was discovered many decades ago it has recently attracted much attention. Not only does entanglement play a central role in quantum information theory, but also it is the key ingredient in many quantum information processing and communication. In fact, quantum entanglement is a physical resource, like energy. In more recent years, it has gained interest as a fundamental resource in many quantum information protocols such as quantum cryptography and quantum key distribution \cite{BB84,Ekert91,GRTZ02}, teleportation of quantum states \cite{BBCJPW93,BPMEWZ97}, superdense coding \cite{BW92}, and quantum error-correction \cite{Shor95,CS96}. Moreover, entanglement is considered as a crucial component for speed-up in quantum computation \cite{Shor97,Grover97,BCJLPS99}. In Ref. \cite{GHSZ90}, the authors have shown that multipartite entanglement violates classical principles even stronger than bipartite entanglement as proposed by John Bell in Ref. \cite{Bell64}. Moreover, it has been established that most of the multipartite quantum states are even too entangled to be useful for certain quantum information processing tasks \cite{GFE09}. Furthermore, theoretical and experimental results based on entangled states as resource states have contributed to new technology involving quantum information\footnote{The Nobel Prize in Physics 2022 was awarded to Alain Aspect, John F. Clauser, and Anton Zeilinger ``for experiments with entangled photons, establishing the violation of Bell inequalities and pioneering quantum information science'' \cite{Nobel2022}.}.

Taking all these reasons into account, the characterization, quantification, and classification of (multipartite) entanglement are crucial milestones in quantum information theory.  For bipartite systems this was done by developing entanglement monotones \cite{HHHH09}. However extension of these to multipartite systems soon appeared quite challenging. That is why a classification of entangled states in multipartite systems was pursued on the basis of one out of the many properties satisfied by entanglement monotones, namely the invariance under Local Operation and Classical Communication (LOCC). Actually, this property is reinforced by requiring stochasticity of local operation and classical communication (SLOCC). Such an invariance property is relevant to single out states that perform (probabilistically) equally well quantum information tasks. On the grounds of group theory, SLOCC equivalence classes are orbits under the action of special linear group ${\rm{SL}}(d_1,\mathbbm{C})\times\cdots\times{\rm{SL}}(d_n,\mathbbm{C})$ on the set of $n$-partite quantum states in the tensor product Hilbert space $\mathcal{H}=\otimes_{i=1}^n\mathbbm{C}^{d_i}$. For unnormalised $n$-partite states, the number of parameters needed to describe inequivalent states is thus given by the dimension of the space of orbits, i.e., the quotient space
\begin{equation*}
\frac{\otimes_{i=1}^n\mathbbm{C}^{d_i}}{{\rm{SL}}(d_1,\mathbbm{C})\times\cdots\times{\rm{SL}}(d_n,\mathbbm{C})}\,.
\end{equation*}
\sloppy Therefore the set of equivalence classes under SLOCC depends at least on $2(\prod_{i=1}^n d_i-1-\sum_{i=1}^n(d_i^2-1))$ parameters. Finding such equivalence classes, that will provide an entanglement classification based on a finite number of entanglement families, was a long-standing open problem in quantum information theory \cite{HHHH09}.

While SLOCC classification works well for small systems like two- and three-qubit systems which feature two and six orbits, respectively, there are inﬁnitely many (actually uncountable) SLOCC classes for four (or more) qubits \cite{DVC00}. This will be even  worse for higher dimensional systems. Along finding methods that group the infinite number of SLOCC classes in a finite number of families, there have been several attempts to classify four-qubit entanglement, the case which attracted most attention \cite{VDDV02, CD07, CW07, BDDMR10, BK12, CDGZ13, GA16}, and also for $n$-qubit symmetric states \cite{BKMGLS09, RM11,Aulbach12}. Although the general case of $n$-qubit entanglement has been addressed, its classification suffers from family overlapping \cite{LL12,GM18}, or still shows an infinite number of classes \cite{GW13}. Thus, it necessitates new methods to establish a finite classification.

Formally, (pure) quantum states are rays in a Hilbert space. As a consequence, the space of states is more appropriately described by projective Hilbert space $\mathbbm{P}(\mathcal{H}_{n})$. Thus, a natural way to study entanglement of pure states is with algebraic geometry, which is the ``language'' of projective spaces. This avenue was put forward in Refs. \cite{BH01,Miyake03,BGH07,ST13}, where the authors investigated the geometry of entanglement and considered small systems (up to $\mathbbm{C}^3\otimes\mathbbm{C}^2\otimes\mathbbm{C}^2$) to lighten it. Following this, it has been recently realized the existence, for four qubit systems, of families, each including an infinite number of SLOCC classes with common properties \cite{HLT14,HLT17,SBSE17,SMKKKO18}. The framework of algebraic geometry also helped to visualize entanglement families with polytopes \cite{WDGC13, SOK14}, which would be of practical use if a finite classification existed.

It is worth noting that in the algebraic-geometric approach, the Segre, Veronese, and Pl\"{u}cker embedding maps together with their secant varieties are the key tools (to be used with distinguishable particles, bosons and fermions, respectively) \cite{HLT14, HLT17,HLT12}.

In this dissertation, by referring to Ref. \cite{GMO20}, we introduce an entanglement classification of ``generic'' $n$-qubit pure states under SLOCC that is based on a finite number of \emph{families} and \emph{subfamilies} (i.e., a fine-structure classification). We do this by employing tools of algebraic geometry that are SLOCC invariants. In particular, the families and subfamilies are identified by $k$-secants and $\ell$-multilinear ranks (hereafter $\ell$-multiranks), respectively. A $k$-secant of a variety $\mathfrak{X}\subset\mathbbm{P}(\mathcal{H}_{n})$ is the projective span of $k$ points of $\mathfrak{X}$. Geometrically, the $k$-secant variety is the zero locus of a set of polynomial equations. Physically, as the $k$-secant of a variety joins its $k$ points, it can liaise to the concept of quantum superposition. On the other hand, $\ell$-multiranks are a collection of integers which are just ranks of different matricizations of a given $n$-qubit state as an order-$n$ tensor in $\otimes^n\mathbbm{C}^2$. Actually, the $\ell$-multiranks tell us about the separability of such a state; when all of them are equal to one we are dealing with a fully separable state. Furthermore, each $k$-secant is a counterpart of the generalized Schmidt rank \cite{CDS08, CCDJW10} which is an entanglement measure. These connections made our classification also operationally meaningful.

Going beyond qubit, more information can be encoded in qudits and more robustness against noise can be achieved \cite{CMM12}. Also, in quantum cryptography, entangled qudits guarantee more security against coherent attacks than using entangled qubits \cite{CBKG02}. These facts motivate the classification of entangled states for higher than two-dimensional systems. Ref \cite{YWWF08}, has investigated the SLOCC classification of two- and three-qutrit entanglement based on the inductive method. However, this method suffers from a flaw already at the qubits level \cite{GM18}, which propagates to higher dimensional systems \cite{example}. Therefore, there is an overlap between some families of three-qutrit entanglement, and thus this approach cannot be
used to exactly identify to which family a given three-qutrit
state belongs to. In Ref. \cite{BLTV04}, the invariants of three-qutrit entanglement has been studied, while Ref. \cite{HJ16} used singularity theory to study the entanglement of pure three-qutrit states. More specifically, Refs. \cite{Nurmiev1,Nurmiev2} provided an implicit description of all three  fundamental invariants of ${\rm{SL}}(3,\mathbbm{C})^{\times{3}}$, and classified the normal forms in five families, which can also be derived as a special case of entanglement classification of three-fermions with nine single-particle states \cite{SL14}.

Following Ref. \cite{GM21} we then pursue the extendibility of the algebraic-geometric approach of Ref. \cite{GMO20} to multiqudit states using as a benchmark tripartite $\mathbbm{C}^d\otimes\mathbbm{C}^d\otimes\mathbbm{C}^d$ systems and achieving in particular a fine-structure classification of three-qutrit entanglement.

Regarding entanglement classification, there is (uncountably) infinite number of inequivalent entangled states. Therefore, a natural question is whether two given (entangled) states can be probabilistically converted to each other via SLOCC. Since quantum mechanics is inherently probabilistic, the more natural question is with what probability $p$ one can obtained a target quantum state $|\varphi\rangle$ from a source quantum state $|\psi\rangle$ via LOCC. For $p=1$, the transformation protocol is called deterministic (LOCC), and for $0<p<1$, the protocol is called probabilistic (SLOCC).

The interconversion between different resources by the SLOCC and asymptotic SLOCC is another fundamental problem in quantum information theory that we study in this dissertation by referring to Ref. \cite{GL22}. This problem encodes some of the most challenging open problems in mathematics and computer science.

Using the Schmidt rank, one can characterize the SLOCC convertibility of bipartite systems. In fact, a bipartite quantum state is SLOCC convertible to another bipartite quantum state iff the Schmidt rank of the initial state is not smaller than that of the latter (notation: $|\psi\rangle\xrightarrow[]{\text{SLOCC}}|\varphi\rangle\Leftrightarrow\rk_S(\psi)\geq\rk_S(\varphi)$). A generalization of Schmidt rank in multipartite systems is the tensor rank. Using the Dirac bra-ket notation, the tensor rank can be considered as the length of the shortest bra-ket representation of a quantum state. Another tool relevant to the SLOCC entanglement transformation is the tensor border rank (border rank, for short). The border rank of a tensor $\mathcal{T}$ is defined as the smallest $r$ such that $\mathcal{T}$ is a limit of tensors of rank $r$. Both the tensor rank and the border rank have been extensively studied in algebraic complexity theory \cite{BCS97} and algebraic geometry \cite{Landsberg}. Recently, connections have been discovered between algebraic complexity theory, algebraic geometry, asymptotic SLOCC transformations, and SLOCC equivalence. \cite{CDS08,CCDJW10,GMO20,YCGD10,CDS10,YGD14,VC15,VC17}. An important property of the tensor rank is that it is an SLOCC monotone, that is, if a source quantum state $|\psi\rangle$ can be transformed into a target quantum state $|\varphi\rangle$ via SLOCC, then the tensor rank of the source is not smaller than that of the target (notation: $|\psi\rangle\xrightarrow[]{\text{SLOCC}}|\varphi\rangle\Rightarrow\rk(\psi)\geq\rk(\varphi)$). Although in general the inverse is not necessarily true, as in Ref. \cite{CDS08} it has been shown that a $\rm{GHZ}$-equivalent state (a state in the $\rm{GHZ}$ orbit) $|\psi_{\rm{GHZ}}\rangle$ can be transformed into another state $|\varphi\rangle$ iff the tensor rank of the $\rm{GHZ}$-equivalent state is not smaller than that of the latter, i.e., $|\psi_{\rm{GHZ}}\rangle\xrightarrow[]{\text{SLOCC}}|\varphi\rangle\Leftrightarrow\rk(\psi_{\rm{GHZ}})\geq\rk(\varphi)$. On the other hand, it is well known that a $\rm{GHZ}$ state cannot be transformed into a $\rm{W}$ state by SLOCC \cite{DVC00}, as they belong to distinct entanglement classes of multiqubit states, but one can asymptotically produce a $\rm{W}$-equivalent state from a $\rm{GHZ}$-equivalent state with rate arbitrarily close to one (see Refs. \cite{GMO20,GM21,VC15} for theory and Ref. \cite{WRZ05} for an experimental way). Actually, the reason is that the tensor rank of the multiqubit $\rm{GHZ}$ state is less than the tensor rank of the multiqubit $\rm{W}$ state, but the border rank of both of them is the same (geometrically, the $\rm{W}$ state is in the orbit closure of the $\rm{GHZ}$ state; see Ref. \cite{GMO20}). This phenomenon is known as degeneration in algebraic complexity theory \cite{BCS97} and algebraic geometry \cite{Landsberg} and is related to the asymptotic SLOCC transformation in entanglement theory \cite{GMO20,GM21,VC15,VC17}. Indeed, the border rank also has the same property as the tensor rank that is SLOCC monotone, i.e., a target quantum state $|\varphi\rangle$ can be approximated from a source quantum state $|\psi\rangle$ via SLOCC, then the border rank of the source is not smaller than that of the target (notation: $|\psi\rangle\xdashrightarrow[]{\text{SLOCC}}|\varphi\rangle\Rightarrow\brk(\psi)\geq\brk(\varphi)$). Interestingly, if a target quantum state $|\varphi\rangle$ can be obtained approximated from a $\rm{GHZ}$-equivalent state $|\psi_{\rm{GHZ}}\rangle$, then the border rank of the $\rm{GHZ}$-equivalent state is not smaller than that of the target, i.e., $|\psi_{\rm{GHZ}}\rangle\xdashrightarrow[]{\text{SLOCC}}|\varphi\rangle\Leftrightarrow\brk(\psi_{\rm{GHZ}})\geq\brk(\varphi)$.

The proposed entanglement classification by algebraic geometry \cite{GMO20,GM21} can be considered as a reference to study SLOCC and asymptotic SLOCC interconversions among different resources based on tensor rank and border rank, respectively. In this regard, we introduce a new class of tensors that we call \emph{persistent tensors} and have constructed a lower bound for their tensor rank. We present three specific families of minimum-rank persistent tensors that are different generalizations of multiqubit $\rm{W}$ state within multiqudit systems. We further cover SLOCC convertibility of multipartite systems considering several families of persistent tensors. Furthermore, we show that the obtained tensor rank lower bound can be extended to direct sums with persistent summands and to even more general combinations of tensors, which we call \emph{block pyramidal tensors} \cite{GL22}.

\section{Outline}
In this dissertation, we address two central problems in quantum information theory. The first one is the entanglement classification of pure multipartite states under SLOCC. The second problem is the convertibility between different multipartite entangled states by using SLOCC and asymptotic SLOCC.

The main goals and results of the dissertation are presented in the five next chapters. Hereafter we shortly describe the content of each of them.

\subsubsection*{Chapter 2: Quantum Entanglement}
This chapter contains the mathematical framework quantum mechanics as well as basic definitions of some concepts in entanglement theory. Although our purpose is entanglement classification of pure multipartite states, we have also extended the definitions to mixed states. Here, we have started by the definition of entanglement in pure and mixed states in biparite systems and have extended it to multipartite one. Then, we have introduced most famous quantum local operations, namely, Local Unitay (LU), Local Operation and Classical Communication (LOCC), and Stochastic LOCC. This followed by the definition of entanglement monotone and introducing some important entanglement measures. Finally, we have finalized the chapter by presenting the SLOCC classification of bipartite syatems and pure and mixed three-qubit states.

In this chapter, we have used some materials from Refs. \cite{HHHH09,Nielsen-Chuang,Mancini-Winter,Peres}.

\subsubsection*{Chapter 3: Algebraic Geometry}
In order to present algebro-geometric tools to solve our problems in quantum information theory, we have presented some basic definitions in this chapter. We introduced the most important concepts in affine and projective geometries. Then, using the definition of morphism between algebraic varieties we have introduced two important morphisms in algebraic geometry, namely, Veronese and Segre embeddings. Finally, we have introduced algebro-geometric tools that are SLOCC invariants, namely, tensor rank, border rank, $\ell$-multirank, and $k$-secant varieties of Segre variety.

For the definitions, lemmas, and theorems presented in this chapter we have used materials from Refs. \cite{Landsberg,Harris,Hartshorne,CLO15,Shafarevich,Holme}.

\subsubsection*{Chapters 4 \& 5: Fine-Structure Classification of Multiqubit \& Tripartite Entanglement}
In these chapters, we focus on the first problem that is multipartite entanglement classification. To this end, we use $k$-secant varieties and $\ell$-multiranks as the SLOCC invariants and we present the entanglement classification algorithm based on them. Regarding this algorithm, one is able to group orbits (classes) into ﬁnite number of families and subfamilies. Then, we study in details two- to five-qubit entanglement (in Chapter \ref{chap4}) and three-qutrit entanglement (in Chapter \ref{chap5}) achieving a fine-structure classification as relevant examples. Several issues of these cases will be generalized to multiqudit and multipartite systems.

\subsubsection*{Chapter 6: Persistent Tensors}
In this chapter, we introduce a new class of tensors in $\otimes_{i=1}^n\mathbbm{C}^{d_i}$ that we call \emph{persistent tensors} and construct a lower bound for their tensor rank. We present three specific families of tensors in this class which the lower bound is tight, and therefore, their corresponding $n$-qudit states can be seen as different generalizations of the multiqubit $\rm{W}$ state within multiqudit systems. Then, we provide a chain of degenerations among the families that can be used to study the entanglement transformation between them. In addition, we prove that the border rank of these three families of minimal-rank persistent tensors is equal to $d$.
This will reveal the fact that different generalizations of multiqubit $\rm{W}$ state within multiqudit systems are geometrically in the orbit closure of multiqudit $\rm{GHZ}$ states. Consequently, we show that a multiqudit $\rm{GHZ}$-equivalent state can can be transformed into each one of the generalizations of $\rm{W}$ state via asymptotic SLOCC with rate one. Furthermore, we show that the obtained tensor rank lower bound can be extended to direct sums with persistent summands and to even more general combinations of tensors, which we call \emph{block pyramidal tensors}. Accordingly, we show that the tensor rank is multiplicative under the Kronecker and the tensor product of minimal-rank persistent tensors with the $\rm{GHZ}$ tensor, and hence we answer an open question posed in Ref. \cite{CF18}.

\subsubsection*{Chapter 7: Summary and Outlook}
In this chapter, we summarize results obtained in the dissertation and outlines the potential applications and open problems for further research.

\chapter{Quantum Entanglement}\label{chap2}

\epigraph{``I never wish to be easily defined.''}{\textit{Franz Kafka}}

In this chapter, the mathematical formulation and the fundamental concepts of quantum information theory are introduced. We start with the mathematical description of single quantum systems residing in Hilbert spaces. Then, it will be shifted to the mathematical description of composite quantum systems where entanglement can be manifested. The main focus of this chapter lies on the description of entangled systems and the mathematical tools needed to characterize and classify entanglement in bipartite and multipartite quantum systems. In this chapter, we have used some introductory materials and notions from Refs. \cite{HHHH09,Nielsen-Chuang,Mancini-Winter,Peres}. The last part of this chapter is based on the Ref. \cite{GM18}.

\section{Hilbert space: space of quantum states}\label{sec.2.1}

The mathematical formulation of quantum mechanics that permit a rigorous description of quantum mechanics is due to John von Neumann \cite{vonNeumann}. In this formalism, the possible states\footnote{By state we mean something that determines the values of observables.} that describe a quantum mechanical system are represented by vectors, called state vectors, residing in a complex separable Hilbert space\footnote{A Hilbert space is separable iff it admits a countable orthonormal basis.}. While any quantum system is identified with a finite or infinite dimensional Hilbert space, in this dissertation, we will be only concerned with the quantum systems having finite dimensional Hilbert space. The Hilbert space representing the possible states of a $d$-dimensional quantum state is then $\mathcal{H}=\mathbbm{C}^d$ and $\dim\mathcal{H}=d$. In the following we introduce the notations starting from the definition of the Hilbert space.

\begin{definition}[Hilbert space]
A Hilbert space, denoted by $\mathcal{H}$, is a complete inner product space.
\end{definition}

In quantum mechanics, we represent elements of the Hilbert space by Dirac bra-ket notation, firstly introduced by Paul Dirac \cite{Dirac}, over the field of complex numbers $\mathbbm{C}$. A ket and a bra are of the form $|\psi\rangle$ and $\langle\varphi|$ that denote a vector in the Hilbert space $\mathcal{H}$ and a linear form (also known as a one-form, a covector, or linear functional) $\langle\varphi|\colon\mathcal{H}\to\mathbbm{C}$, respectively. Actually, the linear form $\langle\varphi|$ is a covector to $|\varphi\rangle$, and the set of all covectors forms the dual Hilbert space $\mathcal{H}^{\vee}$. Hence, the inner product (or scalar product) is defined as a sesquilinear form:
\begin{equation}
\langle\cdot|\cdot\rangle\colon\mathcal{H}\times\mathcal{H}\to\mathbbm{C}\,,
\end{equation}
that is antilinear (conjugate-linear) in the first argument and linear in their second, i.e.,
\begin{align}\nonumber
& \langle \alpha\varphi_1+\beta\varphi_2|\psi\rangle=\alpha^{*}\langle\varphi_1|\psi\rangle+\beta^{*}\langle\varphi_2|\psi\rangle\,, \\
& \langle\varphi|\alpha\psi_1+\beta\psi_2\rangle=\alpha\langle\varphi|\psi_1\rangle+\beta\langle\varphi|\psi_2\rangle\,,
\end{align}
where $\alpha,\beta\in\mathbbm{C}$ and $*$ denotes the complex conjugation. More generally, the inner product on any complex Hilbert space is a Hermitian form, that is $\langle\varphi|\psi\rangle=\langle\psi|\varphi\rangle^*$. Subsequently, the inner product of a vector state $|\psi\rangle$ and its dual complement $\langle\psi|=(|\psi\rangle)^{\dagger}$, where $\dagger$ indicates the conjugate transpose (also known as Hermitian transposition), is positive semidefinite, i.e., $\langle\psi|\psi\rangle\geq0$, and the equality holds iff the vector state is the zero vector. Concerning this property, a definition of the norm directly emerges as the square root of the inner product
\begin{equation}
\parallel|\psi\rangle\parallel\,=\sqrt{\langle\psi|\psi\rangle}\,.
\end{equation}

Finally, completeness is satisfied if every Cauchy sequence of vectors in $\mathcal{H}$ has a limit vector in $\mathcal{H}$. In other words, if $\{|\psi_i\rangle\}_{i=1}^{\infty}$ is a Cauchy sequence, then there exists a vector $|\psi\rangle\in\mathcal{H}$ such that
\begin{equation}
\lim_{i\to\infty}\parallel|\psi_i\rangle-|\psi\rangle\parallel\,=0\,.
\end{equation}

It is worth noting that the inner product allows to define geometric measures as the distance between elements of Hilbert space.

For the finite-dimentional Hilbert spaces we have the following proposition.
\begin{proposition}
Every ﬁnite-dimensional complex (or real) Hilbert space is complete with respect to the norm induced by its inner product.
\end{proposition}

Therefore, a Hilbert space can be defined as a complete metric space (a Banach space) with respect to the distance function induced by the inner product.

\subsubsection{Composite System}\label{subsub2.1_1}

In quantum information theory, we frequently deal with quantum systems consisting of several subsystems, called composite systems (or multipartite systems). The Hilbert space of an $n$-partite quantum system is the tensor product of the Hilbert spaces of each individual subsystem, that is
\begin{equation}
\mathcal{H}=\bigotimes_{i=1}^n\mathcal{H}_i\,.
\end{equation}

\begin{definition}[Tensor product Hilbert space]
The tensor product Hilbert space $\mathcal{H}_1\otimes\mathcal{H}_2$ of two Hilbert spaces $\mathcal{H}_1$ and $\mathcal{H}_2$ is a Hilbert space to which is associated a bilinear map $\mathcal{H}_1\times\mathcal{H}_2\to\mathcal{H}_1\otimes\mathcal{H}_2$ that maps a pair $(|\psi_1\rangle,|\psi_2\rangle)$,  where $|\psi_1\rangle\in\mathcal{H}_1$ and $|\psi_2\rangle\in\mathcal{H}_2$, to an element of $\mathcal{H}_1\otimes\mathcal{H}_2$ which is denoted by $|\psi_1\rangle\otimes|\psi_2\rangle\equiv|\psi_1\otimes\psi_2\rangle$. The inner product of the tensor product Hilbert space is defined by
\begin{equation}
\langle\varphi_1\otimes\varphi_2|\psi_1\otimes\psi_2\rangle:=\langle\varphi_1|\psi_1\rangle\langle\varphi_2|\psi_2\rangle\,,
\end{equation}
for all $|\varphi_1\rangle,|\psi_1\rangle\in\mathcal{H}_1$ and $|\varphi_2\rangle,|\psi_2\rangle\in\mathcal{H}_2$.
\end{definition}

Let $\mathcal{H}_1$ and $\mathcal{H}_2$ be two Hilbert spaces of dimensions $d_1$ and $d_2$, respectively. The tensor product Hilbert space $\mathcal{H}_1\otimes\mathcal{H}_2$ is a Hilbert space which has as a basis the set of all $|e_i\rangle\otimes|f_j\rangle$ where $\{|e_i\rangle\mid i\in\mathbbm{Z}_{d_1}\}$ and $\{|f_j\rangle\mid j\in\mathbbm{Z}_{d_2}\}$ are bases of $\mathcal{H}_1$ and $\mathcal{H}_2$, respectively. Therefore, the dimension of the tensor product Hilbert space of an $n$-partite quantum system, $\mathcal{H}=\otimes_{i=1}^n\mathcal{H}_i$, is equal to the product of dimensions of individual Hilbert spaces, i.e., $\dim\mathcal{H}= \prod_{i=1}^n\dim\mathcal{H}_i$.
 
\subsection{Pure quantum states}\label{subsec.2.1.1}
In classical mechanics, the pure state of the system is represented by a point in a real vector space that is called phase space. The dimension of phase space is defined by the number of the degrees of freedom of the system.

The case of quantum mechanics is more subtle. Instead of the real finite dimensional phase space we have a finite dimensional complex separable Hilbert space. Mathematically, a $d$-dimensional pure quantum state, which describes an isolated quantum system, is represented by a norm-one vector belonging to a Hilbert space $\mathcal{H}=\mathbbm{C}^d$. A pure quantum state $|\psi\rangle$ can be expressed as a linear combination of some orthonormal basis $\{|e_i\rangle\mid i\in\mathbbm{Z}_d\}$ of the Hilbert space $\mathbbm{C}^d$:
\begin{equation}\label{pure-state}
|\psi\rangle=\sum_{i=0}^{d-1}\mathbf{c}_i|e_i\rangle\,,
\end{equation}
where due to the normalization, the complex-valued coefficients $\mathbf{c}_i$ satisfy the following constraint, $\sum_{i=0}^{d-1}|\mathbf{c}_i|^2=1$. Thus the set of all pure states corresponds to the unit sphere in the Hilbert space. We usually use the standard basis (or canonical basis) in linear algebra, i.e., $\{|i\rangle\mid i\in\mathbbm{Z}_d\}$, which is called computational basis in quantum information and computation.

Pure states are also known as state vectors or wave functions, the latter term applying particularly when they are represented as functions of position or momentum.

In case of composite systems containing $n$ quantum subsystems, each with a respective Hilbert space $\mathcal{H}_j=\mathbbm{C}^{d_j}$, the associated state vector describing a pure quantum state is an element in the composite Hilbert space, denoted by the tensor product $\mathcal{H}_1\otimes\cdots\otimes\mathcal{H}_n$, i.e.,
\begin{equation}\label{pure-multi}
|\psi\rangle=\sum_{j=1}^{n}\sum_{i_j\in\mathbbm{Z}_{d_j}}\mathbf{c}_{i_1\cdots i_n}|i_1\rangle\otimes\cdots\otimes|i_n\rangle\,.
\end{equation}
It is worth remarking that often for the sake of simplicity, we will omit the tensor product symbol and merge ket vectors into one, i.e.,
\begin{equation}
|i_1\rangle\otimes\cdots\otimes|i_n\rangle\equiv|i_1\rangle\cdots|i_n\rangle\equiv|i_1\cdots i_n\rangle\,.
\end{equation}

\subsection{Mixed quantum states}\label{subsec.2.1.2}

The description of a quantum system by its state vector is possible only if preparation of the quantum system is fully known. In practice, such perfect information is mostly not available to specify the state vector of a quantum system but rather we know that with some probabilities $p_i$ the quantum system is in a normalized state vector $|\psi_i\rangle$ (see Fig. \ref{fig:MixedState}). These quantum states that incorporate the incomplete knowledge about the quantum system  are called mixed quantum states.

\begin{figure}[t]
\center{\includegraphics[width=5cm]{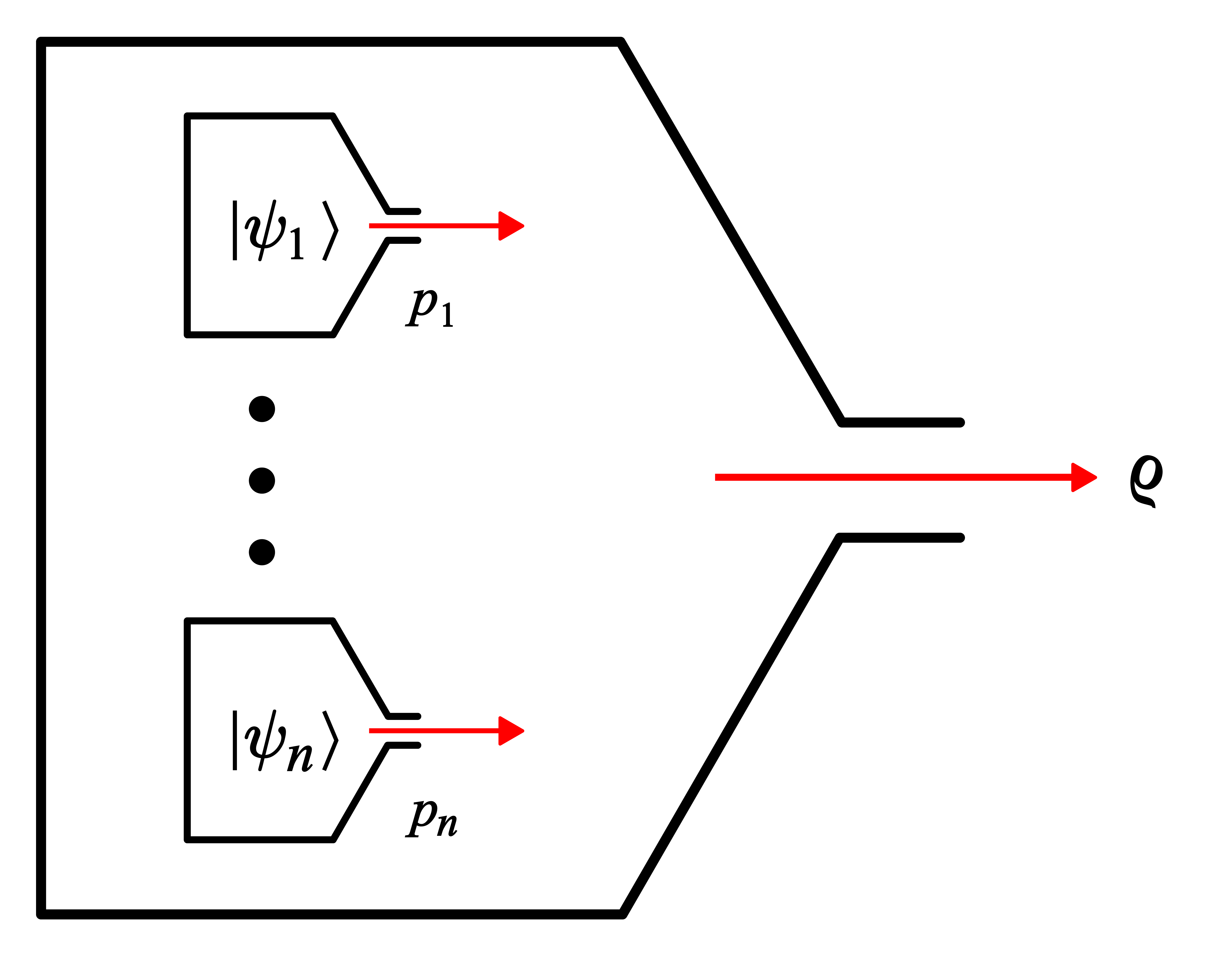}}
\caption[Mixed quantum state]{\label{fig:MixedState} Representation of a mixed quantum state as a statistical mixture with ensemble $\{p_i,|\psi_i\rangle\}$.}
\end{figure}

Mathematically, a mixed quantum state, consisting of several possible pure states $|\psi_i\rangle\in\mathcal{H}$, each with probability $p_i$ of being occupied, is described by a density matrix of the form:
\begin{equation}\label{MixedState}
\varrho=\sum_{i}p_i|\psi_i\rangle\langle\psi_i|\,,
\end{equation}
which is an element of $\mathcal{B}(\mathcal{H})=\mathcal{H}\otimes\mathcal{H}^{\vee}\cong{\rm{End}}(\mathcal{H})$, the space of endomorphisms of the Hilbert space $\mathcal{H}$. In other words, the density matrix in Eq. \eqref{MixedState} is a square matrix of size equal to the dimension of the Hilbert space $\mathcal{H}$. Since $p_i$'s are probabilities, they are non-negative real numbers that sum to one, i.e., $p_i\in\mathbbm{R}^+$ and $\sum_ip_i=1$. It follows that the density matrix $\varrho$ is normalized, i.e., $\tr(\varrho)=1$ where $\tr$ denotes the trace. Moreover, since every matrix $|\psi_i\rangle\langle\psi_i|$ is Hermitian positive semi-definitive so is the density matrix $\varrho$, i.e., $\varrho^\dagger=\varrho$ and $\langle\phi|\varrho|\phi\rangle\geq0$ for all $|\phi\rangle\in\mathcal{H}$.

It is worth noting that the decomposition of a density matrix in Eq. \eqref{MixedState} in a statistical ensemble of pure states is not unique since the vector states $|\psi_i\rangle$ need be neither orthogonal nor linearly independent.

Regarding Eq. \eqref{MixedState}, the density matrix can be seen as a weighted sum of projectors on all pure states within the statistical ensemble $\{p_i,|\psi_i\rangle\}$. So it is easy to see that this general definition of the density matrix also holds for pure states, for which we will only have one vector state $|\psi_i\rangle$ with $p_i=1$, i.e., the density matrix has rank one and thus $\varrho=|\psi_i\rangle\langle\psi_i|$ and $\varrho^2=\varrho$.

Mathematically, pure and mixed state can be distinguished by computing $\tr(\varrho^2)$ where $\varrho$ is the density matrix of the given state:
\begin{align}\nonumber
& \text{if}~\tr(\varrho^2)=1\quad\Leftrightarrow\quad\varrho\,~\text{is pure,} \\ \label{mixedness}
& \text{if}~\tr(\varrho^2)<1\quad\Leftrightarrow\quad\varrho\,~\text{is mixed.}
\end{align}
The minimum value of $\tr(\varrho^2)$ is attained when the density matrix is proportional to the identity matrix corresponding a quantum state called maximally mixed.

In summary, any trace-one Hermitian positive semidefinite matrix is a density matrix which describes a (pure or mixed) state of a quantum system. Therefore, the set of all quantum states is a closed convex set with pure states on its boundary.

\subsubsection{Reduced density matrix}\label{2.1.2-1}
As we mentioned before the space state of composite systems is obtained by the tensor product of its subsystems Hilbert spaces. Suppose now that, given a multipartite state, we are interested to have information about only one or some parts of the system. Thus, we need an operation that is somehow contrary to the tensor product. This operation is called the partial trace. In other word, the information we are interested can be found by taking the trace over the subspaces of the Hilbert space that represent subsystems we are not interested in. While the trace is a scalar-valued function on operators, the partial trace is an operator-valued function. For example, if we have a bipartite quantum system $\varrho_{AB}$, consisting of the subsystems $A$ and $B$, we can obtain a state that has information about subsytem $B$ by taking the partial trace over subsystem $A$. It is called reduced density matrix:
\begin{equation}\label{partial-trace}
\varrho_{B}=\tr_{A}(\varrho_{AB})=\sum_{i=0}^{d_A-1}\langle i|\varrho_{AB}|i\rangle\,.
\end{equation}

\subsection{Observables}\label{subsec.Observables}
In general, an operator $\hat{O}$ is a linear map that takes a quantum state $|\psi\rangle$ and produces another quantum state $\hat{O}|\psi\rangle$ which, possibly, will not be normalized\footnote{In the rest of the thesis, where there is no ambiguity we omit the hat notation of the operator.}. On the other hand, an observable is a measurable physical quantity. Hence, an observable is a self-adjoint operator (a Hermitian operator in the finite-dimensional case) since it has real spectrum.




\begin{definition}[Projective measurement]
A projective measurement is described by an observable $\hat{M}$ on the Hilbert space of the quantum system being measured. The observable $\hat{M}$ has a spectral decomposition
\begin{equation}
\hat{M}=\sum_mm\hat{P}_m\,,
\end{equation}
where $P_m$ is the projector onto the eigenspace of $\hat{M}$ corresponding to the eigenvalue $m$. The possible outcomes of the measurement correspond to the eigenvalues of the observable. Upon measuring the state $|\psi\rangle$, the probability of getting result $m$ is
\begin{equation}
p(m)=\langle\psi|\hat{P}_m|\psi\rangle\,.
\end{equation}
Given that outcome $m$ occurred, the state of the quantum system immediately after the measurement is
\begin{equation}
\frac{\hat{P}_m|\psi\rangle}{\sqrt{p(m)}}\,.
\end{equation}
\end{definition}
For instance, the square of the absolute value of each coefficient $\mathbf{c}_i$ in Eq. \eqref{pure-state}, corresponds to the probability of obtaining an outcome $e_i$, once the system is measured in the basis $\{|e_i\rangle\}_{i=0}^{d-1}$ which are the eigenvectors of the measurement operator. The quantum state after the measurement will be $|\psi'\rangle=|e_i\rangle$.

In quantum mechanics, an experimental setup is described by the observable $\hat{O}$ to be measured, and the state of the system is given by the density matrix $\varrho$. The probabilistic expected value of the result (measurement) of an experiment is called expected value and is defined as follows
\begin{equation}
\langle\hat{O}\rangle_{\varrho}=\tr(\varrho\hat{O})\,.
\end{equation}

\begin{corollary}\label{U1-gauge}
The global phase of a quantum state has no physical consequences, that is,
\begin{equation}
|\psi\rangle\sim{\rm{e}}^{\mathbf{i}\delta}|\psi\rangle=|{\psi}'\rangle\,.
\end{equation}
This is because the global phase $\delta$ does not affect the result of any measurement
\begin{equation}
\langle{\psi}'|\hat{O}|{\psi}'\rangle=\langle\psi|{\rm{e}}^{-\mathbf{i}\delta}\hat{O}{\rm{e}}^{\mathbf{i}\delta}|\psi\rangle=\langle\psi|\hat{O}|\psi\rangle\,.
\end{equation}
\end{corollary}

\subsection{Qubits}\label{subsec.2.1.4}
In classical information theory, bit, contracted from binary digit, is the most basic unit of classical information. By its name, a bit is commonly represented as either $0$ or $1$. It describes a logical state with one of these two possible values. A bit can be described by a classical system with two independent physical states. These two states are connected since the classical system can hold a maximum of one bit of information. In physics, an observable is a physical quantity that can be measured.

Analogously, in quantum information theory, qubit is the fundamental unit of quantum information. The name of qubit, coined by Benjamin Schumacher \cite{Schumacher95}, is a portmanteau of quantum bit. A qubit is a two-level quantum-mechanical system which is the simplest non-trivial quantum system. While a bit is always in precisely one of two states, i.e., it can be either $0$ or $1$, the general state of a qubit can be in a superposition of both states simultaneously, that is a fundamental property of quantum mechanics. Any two-level quantum system can be used to provide a physical implementation of a qubit. Following are the important physical implementations of qubit systems:
\begin{enumerate}
\item[-] The orientation of a spin-half particle (spin up $|\!\uparrow\rangle$ and spin down $|\!\downarrow\rangle$). \\ \vspace{-4mm}
\item[-] The polarization of a photon (horizontal polarization $|\!\leftrightarrow\rangle$ and vertical polarization $|\!\updownarrow\rangle$). \\ \vspace{-4mm}
\item[-] A pair of electronic energy levels in an atom, ion, or quantum dot (ground state $|\texttt{g}\rangle$ and excited state $|\texttt{e}\rangle$).
\end{enumerate}

\subsubsection*{Qubit pure states}\label{2.1.3-1}
\begin{definition}[Qubit]
Any two-level quantum state which is an element of two-dimensional Hilbert space $\mathcal{H}=\mathbbm{C}^2$ is a qubit state. Therefore, a qubit pure state is represented by
\begin{equation}\label{pure-qubit}
|\psi\rangle =\alpha|0\rangle+\beta|1\rangle\,,
\end{equation}
where $\alpha,\beta\in\mathbbm{C}$, $|\alpha|^2+|\beta|^2=1$, and
\begin{equation}
|0\rangle=\begin{pmatrix}
1\\0 \end{pmatrix}\,, \qquad \text{and} \qquad |1\rangle=\begin{pmatrix}
0\\1 \end{pmatrix}\,,
\end{equation}
are two orthonormal basis vectors, which form what is known as the computational basis.
\end{definition}
Since $|\alpha|^2+|\beta|^2=1$, and because the state does not care about a global phase change, we can use the following useful parameterization for the amplitudes of the qubit sate: 
\[
\alpha=\cos\frac{\theta}{2} \qquad \text{and} \qquad \beta={\rm{e}}^{\mathbf{i}\varphi}\sin\frac{\theta}{2}\,,
\]
where $\theta\in[0,\pi]$ and $\varphi\in[0,2\pi)$ that yield\footnote{In this parameterization, the global phase is omitted (see \cref{U1-gauge}).}
\begin{equation}\label{qubit-polar}
|\psi\rangle=\cos\frac{\theta}{2}|0\rangle+{\rm{e}}^{\mathbf{i}\varphi}\sin\frac{\theta}{2}|1\rangle\,.
\end{equation}

An $n$-qubit pure state $|\psi\rangle\in\otimes^n\mathbbm{C}^2$ can be represented as follows
\begin{equation}\label{n-qubit-pure}
|\psi\rangle=\sum_{j=1}^{n}\sum_{i_j\in\mathbbm{Z}_2}\mathbf{c}_{i_1\cdots i_n}|i_1\cdots i_n\rangle\equiv\sum_{i\in\{0,1\}^n}\mathfrak{c}_i|i\rangle\,.
\end{equation}

We can generalize the concept of qubit to the higher dimensional Hilbert spaces.
\begin{definition}[Qudit]
Any element of $d$-dimensional Hilbert space $\mathcal{H}=\mathbbm{C}^d$ is called a qudit ($d$-level quantum state) which is represented by
\begin{equation}\label{pure-qudit}
|\psi\rangle =\alpha_0|0\rangle+\alpha_1|1\rangle+\cdots+\alpha_{d-1}|d-1\rangle\,,
\end{equation}
where $\alpha_i\in\mathbbm{C}$, $\sum_{i=0}^{d-1}|\alpha_i|^2=1$, and $\{|i\rangle\mid i\in\mathbbm{Z}_d\}$ is the computational basis.
\end{definition}

\subsubsection*{Qubit mixed states}\label{2.1.3-2}
In order to describe a qubit, it is convenient to treat a qubit as a spin-half particle and to introduce the Pauli operators. Pauli matrices are defined as follows
\begin{equation}\label{Pauli-Mat}
\sigma_1=\begin{pmatrix}
0 & 1 \\
1 & 0
\end{pmatrix}\,, \quad
\sigma_2=\begin{pmatrix}
0 & -\mathbf{i} \\
\mathbf{i} & 0
\end{pmatrix}\,, \quad
\sigma_3=\begin{pmatrix}
1 & 0 \\
0 & -1
\end{pmatrix}\,.
\end{equation}
These matrices together with the identity matrix $\mathbb{1}_2$ (sometimes considered as the zeroth Pauli matrix $\sigma_0$) form a basis for the $4$-dimensional real vector space of $2\times2$  Hermitian matrices. This means that any $2\times2$ Hermitian matrix can be written in a unique way as a linear combination of Pauli matrices, with all coefficients being real numbers.

\begin{definition}
Any trace-one Hermitian positive semidefinite matrix of size $2\times2$ is a qubit state. Therefore, a qubit mixed state can be written as follows
\begin{equation}\label{qubit-mixed}
\varrho=\frac{1}{2}(\mathbb{1}_2+\vec{r}\cdot\vec{\sigma})\,,
\end{equation}
where $\vec{r}=(r_1,r_2,r_3)\in\mathbbm{R}^3$, and $\vec{\sigma}=(\sigma_1,\sigma_2,\sigma_3)^{\rm{T}}$ ($\rm{T}$ denotes transposition).
\end{definition}

If we calculate $\tr(\varrho^2)$ for the Eq. \eqref{qubit-mixed}, we will find the following relations
\begin{equation}\label{Bloch-vec}
\begin{cases}
|\vec{r}|=1 \qquad \varrho\,~\text{is pure,}\\
|\vec{r}|<1 \qquad \varrho\,~\text{is mixed.}
\end{cases}
\end{equation} 

\begin{remark}
Using Pauli matrices in Eq. \eqref{Pauli-Mat}, the matrices $\{\mathbf{i}\sigma_i\}_{i=1}^3$ are the generators of the group ${\rm{SU}}(2)$.
\end{remark}

\begin{definition}
Any trace-one Hermitian positive semidefinite matrix of size $d\times d$ is a qudit state which can be expressed as follows \cite{HE81}
\begin{equation}\label{qudit-mixed}
\varrho=\frac{1}{d}(\mathbb{1}_{d}+\vec{r}\cdot\vec{\mathfrak{s}})\,,
\end{equation}
where $\vec{r}=(r_1,\ldots,r_{d^2-1})\in\mathbbm{R}^{d^2-1}$, with $r_i=\langle\mathfrak{s}_i\rangle=\tr(\varrho\,\mathfrak{s}_i)$, and $\vec{\mathfrak{s}}=(\mathfrak{s}_1,\ldots,\mathfrak{s}_{d^2-1})^{\rm{T}}$, with $\{\mathbf{i}\mathfrak{s}_i\}_{i=1}^{d^2-1}$ as the generators of the group ${\rm{SU}}(d)$.
\end{definition}

By calculating $\tr(\varrho^2)$ for the Eq. \eqref{qudit-mixed} we will have \cite{BK03,Kimura03}
\begin{equation}
\begin{cases}
|\vec{r}|=\sqrt{d(d-1)/2} \qquad \varrho\,~\text{is pure,}\\
|\vec{r}|<\sqrt{d(d-1)/2} \qquad \varrho\,~\text{is mixed.}
\end{cases}
\end{equation}

For $n$-qubit mixed states one can consider the following basis for the state space $\otimes^n(\mathbbm{C}^2\otimes{\mathbbm{C}^2}^{\vee})$
\begin{equation}\label{Pauli-n-qubit}
\big\{\mathfrak{s}_k\big\}_{k=0}^{4^n-1}=\bigotimes_{j=1}^n\sigma_{i_j}\,,
\end{equation}
where $\mathfrak{s}_0=\mathbb{1}_{d}$ and for any $j$, $i_j\in\mathbbm{Z}_4$.

As an example, any two-qubit mixed state can be written in Fano form \cite{Fano83}
\begin{equation}
\varrho=\frac{1}{4}\big(\mathbb{1}_4+\sum_{i=1}^3\alpha_i\sigma_i\otimes\mathbb{1}_2+\sum_{i=1}^3\beta_i\mathbb{1}_2\otimes\sigma_i+\sum_{i,j=1}^3\gamma_{ij}\sigma_i\otimes\sigma_j\big)\,.
\end{equation}
Here $\vec{\alpha}$ and $\vec{\beta}$ are Bloch vectors of the partially reduced states and $\gamma_{ij}$ is a real $3\times3$ matrix describing the correlations between the two subsystems. If $\gamma=0$, then the state is separable, but the reverse is not always true.

\section{Entangled states}\label{sec.2.2}
Entanglement is a non-classical correlation between two or more quantum subsystems with the property that the state of each individual subsystem cannot be described independently of the states of the other subsystems. Therefore, an entangled state provides complete information about the system as a whole but not about the subsystems. If we have no information about the subsystems, the entangled state is maximally entangled. It is also a central element in quantum information theory. However, it is not easy to give a precise and comprehensive definition of entanglement other than that it is a property of entangled states. Actually, it is simpler to define an entangled state by what it is not. Mathematically, an entangled state is described by a single state vector for pure states, or a single density matrix for mixed states, combining two or more subsystems, that does not factorize as a product of states of its local constituents. Factorizing as a product of states means that a local measurement acting on one subsystem is independent from the local measurement acting on other subsystems. That is, from measurement of one subsystem we can derive nothing about the measurement results of the other subsystems.

\subsection{Entangled pure states}\label{subsec.2.2.1}
\subsubsection{Bipartite entanglement}\label{2.2.1_1}
Consider the state of a composite quantum system consisting of two subsystems that are denoted by $A$ and $B$ and have associated in Hilbert spaces $\mathcal{H}_A$ and $\mathcal{H}_B$, respectively, as follows
\begin{equation}\label{sep-state}
|\psi\rangle_{AB}=|\varphi_A\rangle\otimes|\varphi_B\rangle\,,
\end{equation}
where $\otimes$ denotes tensor product.

Now, consider that the eigenvalue equations of two observables $\hat{A}$ and $\hat{B}$ that are defined in subsystems $A$ and $B$, respectively, are given as follows
\begin{align}\nonumber
&\hat{A}|a_i\rangle=a_i|a_i\rangle\,, \\
&\hat{B}|b_j\rangle=b_j|b_j\rangle\,.
\end{align}
The probability of obtaining the eigenvalue $a_i$ ($b_j$) as the outcome of the measurement of the observable $\hat{A}$ ($\hat{B}$) alone is given by
\begin{align}\nonumber
&\text{Pr}(a_i)=|\langle a_i|\varphi_A\rangle|^2\,,\\
&\text{Pr}(b_j)=|\langle b_j|\varphi_B\rangle|^2\,.
\end{align}
If one measures both observables $\hat{A}$ and $\hat{B}$ simultaneously, the probability of obtaining the eigenvalues $a_i$ and $b_j$, respectively, is given by
\begin{align}\nonumber
\text{Pr}(a_i,b_j)&=|\langle a_i,b_j|\psi\rangle_{AB}|^2=|\langle a_i\otimes b_j|\varphi_A\otimes\varphi_B\rangle|^2 \\
&=\text{Pr}(a_i)\cdot\text{Pr}(b_j)\,.
\end{align}
So the probability of the simultaneous measurement of subsystems is equal to the product of the probabilities of measurements of subsystems separately. Indeed, in a product state, any measurement on one of the subsystems does not affect the state of another subsystem. One question then arises: can we always write the state vector of a composite system as in Eq. \eqref{sep-state}? The answer is no and the reason is the fact that in the Hilbert space we have the possibility of linear combination of vectors (this is often referred to the quantum superposition principle in quantum mechanics). For instance, consider the following state
\begin{align}\nonumber
|\Psi\rangle_{AB}&=\alpha|\psi\rangle_{AB}+\beta|\psi'\rangle_{AB} \\ \label{superposition}
&=\alpha(|\varphi_A\rangle\otimes|\varphi_B\rangle)+\beta(|\varphi'_A\rangle\otimes|\varphi'_B\rangle)\,.
\end{align}
Now, if we measure both observables $\hat{A}$ and $\hat{B}$ simultaneously, the probability of obtaining the eigenvalues $a_i$ and $b_j$, respectively, is given by
\begin{align}\nonumber
\text{Pr}(a_i,b_j)&=|\langle a_i,b_j|\Psi\rangle_{AB}|^2 \\
&=|\alpha\langle a_i|\varphi_A\rangle\langle b_j|\varphi_B\rangle+\beta\langle a_i|\varphi'_A\rangle\langle b_j|\varphi'_B\rangle|^2\,.
\end{align}
From above equation we can conclude that there is a correlation between subsystems of the composite quantum system defined in Eq. \eqref{superposition}.

We can conclude that entanglement is a direct result from superposition. In Ref. \cite{AAPP22}, it has proved that entanglement can exist between different systems iff superposition can exist in each of them.

\begin{definition}[Bipartite separable state]
A pure bipartite quantum state $|\psi\rangle_{AB}\in\mathcal{H}_A\otimes\mathcal{H}_B$ is called separable, if it can be written as the tensor product of the quantum states of subsystems $|\varphi_A\rangle\in\mathcal{H}_A$ and $|\varphi_B\rangle\in\mathcal{H}_B$, i.e.,
\begin{equation}
|\psi\rangle_{AB}=|\varphi_A\rangle\otimes|\varphi_B\rangle\,.
\end{equation}
\end{definition}

\begin{definition}[Bipartite entangled state]
A pure bipartite quantum state $|\psi\rangle_{AB}\in\mathcal{H}_A\otimes\mathcal{H}_B$ is called entangled (or inseparable), if it is not separable. So, it cannot be written as the tensor product of the quantum states of subsystems, i.e.,
\begin{equation}
|\psi\rangle_{AB}\neq|\varphi_A\rangle\otimes|\varphi_B\rangle\,,
\end{equation}
or in other words,
\begin{equation}
\nexists~|\varphi_A\rangle\in\mathcal{H}_A \quad \text{and} \quad |\varphi_A\rangle\in\mathcal{H}_B \quad \text{s.t.} \quad |\psi\rangle_{AB}=|\varphi_A\rangle\otimes|\varphi_B\rangle\,.
\end{equation}
\end{definition}

\begin{example}
The most famous example of an entangled pure state is the two-qubit $\rm{EPR}$ state:
\begin{equation}\label{EPR}
|{\rm{EPR}}\rangle=\frac{1}{\sqrt{2}}(|00\rangle+|11\rangle)\,.
\end{equation}
This can be understood by trying to write the above mentioned $\rm{EPR}$ state as the tensor product of two single qubits, that is
\begin{align}\nonumber
|{\rm{EPR}}\rangle&\overset{?}{=}(\alpha|0\rangle+\beta|1\rangle)\otimes(\gamma|0\rangle+\delta|1\rangle) \\ \label{EPR-2-qubit}
&=\alpha\gamma|00\rangle+\alpha\delta|01\rangle+\beta\gamma|10\rangle+\beta\delta|11\rangle\,.
\end{align}
Comparing Eqs. \eqref{EPR} and \eqref{EPR-2-qubit}, one concludes that $\alpha\gamma=\beta\delta=\frac{1}{\sqrt{2}}$ and $\alpha\delta=\beta\gamma=0$, where there is no common solution for this system of equations.
\end{example}

Since a separable pure state can be written as the tensor product of the states of subsystems it is concluded that the reduced density matrix of each single subsystem is a pure state. For example, let $|\psi\rangle_{AB}=|\varphi_A\rangle\otimes|\varphi_B\rangle$ represent a separable pure state containing two parties, then we have
\begin{align}\nonumber
\varrho_A=\tr_B(|\psi\rangle_{AB}\langle\psi|)=|\varphi_A\rangle\langle\varphi_A|\,, \\
\varrho_B=\tr_A(|\psi\rangle_{AB}\langle\psi|)=|\varphi_B\rangle\langle\varphi_B|\,.
\end{align}
Therefore, concerning Eq. \eqref{mixedness}, entanglement can be related to the purity of the reduced density matrices as follows
\begin{align}\nonumber
& \text{if}~\tr(\varrho_A^2)=1\quad\Leftrightarrow\quad|\psi\rangle_{AB}~\text{is separable,} \\
& \text{if}~\tr(\varrho_A^2)<1\quad\Leftrightarrow\quad|\psi\rangle_{AB}~\text{is entangled.}
\end{align}

\subsubsection{Schmidt decomposition}\label{subsubsec.2.2.1_2}
There is a very powerful tool to characterize entanglement in bipartite systems, called Schmidt decomposition. Actually, for the special case of pure bipartite states, any state of the form Eq. \eqref{pure-multi} can be written in the Schmidt decomposition.

\begin{theorem}[Schmidt decomposition \cite{Schmidt}]
Let
\begin{equation}\label{pure-bipartite}
|\psi\rangle=\sum_{i=0}^{d_A-1}\sum_{j=0}^{d_B-1}\mathbf{c}_{ij}|ij\rangle\,,
\end{equation}
be a normalized bipartite state in $\mathcal{H}_A\otimes\mathcal{H}_B$. Then there exists orthonormal bases $\{e_i\}_{i=0}^{d_A-1}$ and $\{f_i\}_{i=0}^{d_B-1}$ in $\mathcal{H}_A$ and $\mathcal{H}_B$, respectively, such that
\begin{equation}\label{Schmidt-form}
|\psi\rangle=\sum_{i=0}^{d-1}\sqrt{\lambda_i}|e_if_i\rangle\,,
\end{equation}
where $d=\min(d_A,d_B)$ and $\lambda_i\in\mathbbm{R}^+$ are called Schmidt coefficients.
\end{theorem}

Schmidt coefficients and Schmidt bases of an arbitrary pure bipartite quantum state in Eq. \eqref{pure-bipartite} can be found from the eigenvalues and eigenvectors of its reduced density matrices:
\begin{align}\nonumber
&\varrho_A=\tr_B|\psi\rangle\langle\psi|=\sum_i\lambda_i|e_i\rangle\langle e_i|\,, \\
&\varrho_B=\tr_A|\psi\rangle\langle\psi|=\sum_i\lambda_i|f_i\rangle\langle f_i|\,.
\end{align}

Regarding characterization of entanglement in pure bipartite quantum states, the Schmidt decomposition is a strong and important tool in quantum information theory .

\begin{definition}[Schmidt rank]
The number of non-zero Schmidt coefficients is called Schmidt rank.
\end{definition}

\begin{lemma}
A pure bipartite quantum state is separable iff its Schmidt rank is one.
\end{lemma}

Not only do the Schmidt coefficients determine separability of a pure bipartite quantum state but they also correlate the amount of entanglement to the mixedness of the reduced density matrix of a pure bipartite quantum state. A separable state indicates by a vector with components corresponding to Schmidt coefficients (Schmidt vector) that has only one non-zero component, i.e., 
\[
\vec{\lambda}_{\text{Sep.}}=\begin{pmatrix}
1 & 0 & \cdots & 0
\end{pmatrix},
\]
and therefore the corresponding reduced density matrix is a pure state. The Schmidt vector of an entangled state has at least two non-zero components. The Schmidt vector
\[
\vec{\lambda}_{\rm{MES}}=\frac{1}{d}\begin{pmatrix}
1 & 1 & \cdots & 1
\end{pmatrix},
\]
corresponds to a reduced density matrix that is proportional to identity matrix, i.e., the reduced density matrix is a maximally mixed state, indicates an entangled state which is called a Maximally Entangled State (MES).

The generalized Bell states can be considered as MESs:
\begin{equation}
|\psi_{\rm{MES}}\rangle=\frac{1}{\sqrt{d}}\sum_{i=0}^{d-1}|e_if_i\rangle\,,
\end{equation}
where $d=\min(d_A,d_B)$. For $d_A=d_B=2$ the original four orthogonal Bell states are
\begin{equation}
|\Psi^{\pm}\rangle=\frac{1}{\sqrt{d}}(|01\rangle\pm|10\rangle) \qquad~ |\Phi^{\pm}\rangle=\frac{1}{\sqrt{d}}(|00\rangle\pm|11\rangle)\,.
\end{equation}

\subsubsection{Multipartite entanglement}\label{2.2.1_3}
Composite systems get more complicated with the increasing number of parties, as there exist separability with respect to different partitioning. Here, the situation could be that in an $n$-partite system, an arbitrary number of subsystems are entangled but there is no entanglement between other parts. So there exist different notions of entanglement in multipartite systems. These considerations lead to the definition of $k$-separability \cite{HHHH09}.

In what follows we will be concerned with $n$-partite pure quantum states
\begin{equation}\label{n-partite-pure}
|\psi\rangle=\sum_{j=1}^{n}\sum_{i_j\in\mathbbm{Z}_{d_j}}\mathbf{c}_{i_1\cdots{i_n}}|i_1\cdots i_n\rangle\,,
\end{equation}
which are elements of the Hilbert space $\mathcal{H}=\otimes_{j=1}^n\mathcal{H}_j=\otimes_{j=1}^{n}\mathbbm{C}^{d_j}$, with $d_j$ standing for the dimension of the local Hilbert space corresponding to the subsystem $A_j$. Now, let $\alpha_k=\{S_1,\ldots,S_k\}$ denote a partition of $[n]:=\{1,\ldots,n\}$ into $k$ disjoint nonempty subsets ($k\leq n$). Such a partition corresponds to a division of the system into $k$ distinct subsystems, also called a $k$-partite split \cite{DC00}. For instance, the set $\{1,2,3\}$ has these five different partitions: $\{\{1\},\{2\},\{3\}\}$, $\{\{1,2\},\{3\}\}$, $\{\{1,3\},\{2\}\}$, $\{\{1\},\{2,3\}\}$, and $\{\{1,2,3\}\}$.

\begin{definition}[$k$-separable pure state]\label{def:k-separable-pure}
A pure multipartite quantum state in Eq. \eqref{n-partite-pure} is called $k$-separable state with respect to a specific $k$-partite split, iff it can be written as tensor product of $k$ factors of subsystem states, i.e.,
\begin{equation}\label{k-separable-pure}
|\psi_{\rm{ks}}\rangle=\bigotimes_{i=1}^{k}|\phi_i\rangle\,, \quad k\in[n]\,,
\end{equation}
where the state vector $|\phi_i\rangle$ may consist of more than one subsystem with the maximum number of $(n-k+1)$ subsystems which corresponds the situation of $(k-1)$ single subsystems and one $(n-k+1)$-partite subsystem.
\end{definition}

Concerning \Cref{k-separable-pure}, it follows that two special cases emerge for $k=n$ and $k=1$.

\begin{definition}[Fully separable pure state]
A pure multipartite quantum state in Eq. \eqref{k-separable-pure} with $k=n$ is called fully separable. That is iff it can be written as a tensor product of the quantum states of subsystems $|\varphi_i\rangle\in\mathcal{H}_i$ for all $i\in[n]$, i.e.,
\begin{equation}\label{full-sep-pure}
|\psi_{\rm{fs}}\rangle=\bigotimes_{i=1}^{n}|\varphi_i\rangle\,.
\end{equation}
\end{definition}

\begin{definition}[Genuine entangled pure state]
A pure multipartite quantum state in Eq. \eqref{k-separable-pure} with $k=1$ is called genuine entangled state. That is iff
\begin{equation}\label{genuine-entangled}
|\psi_{e}\rangle\neq|\phi\rangle_S\otimes|\phi'\rangle_{\bar{S}}\,,
\end{equation}
for any bipartition $S|\bar{S}$, where $S$ is a subset of $[n]$ which denotes the indices of subsystems $A_j$'s and $\bar{S}:=[n]\setminus S$ denotes the rest of them.
\end{definition}

\begin{example}
Let $|\psi\rangle\in\otimes^3\mathbbm{C}^2$ represent the vector state of a three-qubit system. We have the following possibilities for different notions of entanglement:
\begin{enumerate}
\item[1.] Fully separable state: $|000\rangle$.
\item[2.] Biseparable states: $|0\rangle_i\otimes|{\rm{EPR}}\rangle_{jk}$ where the indices $\{i,j,k\}=\{1,2,3\}$ indicate the corresponding subsystems.
\item[3.] Genuine entangled states: $|{\rm{GHZ}}\rangle=|000\rangle+|111\rangle$ which is known as Greenberger–Horne–Zeilinger ($\rm{GHZ}$) state \cite{GHZ89}. $|{\rm{W}}\rangle=|001\rangle+|010\rangle+|100\rangle$ which is known as $\rm{W}$ state \cite{DVC00}.
\end{enumerate}
\end{example}

In general, characterization of entanglement in pure multipartite state is quite challenging, especially for what concerns the characterization of different types of genuine multipartite entangled states. 

\subsection{Entangled mixed states}\label{subsec.2.2.2}
Since bipartite system is a special case of multipartite systems, here we just give the definition for the general case.

In what follows we will be concerned with $n$-partite mixed quantum states
\begin{equation}\label{n-partite-mixed}
\varrho\in\bigotimes_{i=1}^{n}(\mathcal{H}_i\otimes\mathcal{H}_i^{\vee})\,.
\end{equation}

\begin{definition}[$k$-separable mixed state]\label{def:k-separable-mixed}
A mixed multipartite quantum state in Eq. \eqref{n-partite-mixed} is called $k$-separable state, iff it can be written as a convex combination of $k$-separable pure states, i.e.,
\begin{equation}\label{k-separable-mixed}
\varrho_{\rm{ks}}=\sum_jp_j|\psi_{\rm{ks}}^{(j)}\rangle\langle\psi_{\rm{ks}}^{(j)}|=\sum_jp_j\big(\otimes_{i=1}^k\rho_i^{(j)}\big)\,,
\end{equation}
where $p_j\in\mathbbm{R}^+$, $\sum_jp_j=1$, and $k$-separable states $|\psi_{\rm{ks}}^{(j)}\rangle$ might be $k$-separable with respect to different $k$-partite splits.
\end{definition}

It is worth noting that if pure states $|\psi_{\rm{ks}}^{(j)}\rangle$ in Eq. \eqref{k-separable-mixed} are $k$-separable with respect to a specific $k$-partite split, then the mixed state is called $\alpha_k$-separable \cite{SU08}.

\begin{definition}[Fully separable mixed state]
A mixed multipartite quantum state in Eq. \eqref{k-separable-mixed} with $k=n$ is called fully separable. That is iff it can be written as a convex combination of fully separable pure states, i.e.,
\begin{equation}\label{full-sep-mix}
\varrho_{\rm{fs}}=\sum_jp_j|\psi_{\rm{fs}}^{(j)}\rangle\langle\psi_{\rm{fs}}^{(j)}|=\sum_jp_j\big(\otimes_{i=1}^{n}\varrho_i^{(j)}\big)\,,
\end{equation}
where $p_j\in\mathbbm{R}^+$, $\sum_jp_j=1$ and fully-separable states $|\psi_{\rm{fs}}^{(j)}\rangle$ are defined in Eq. \eqref{full-sep-pure}.
\end{definition}

\begin{figure}[t]
\center{\includegraphics[width=8cm]{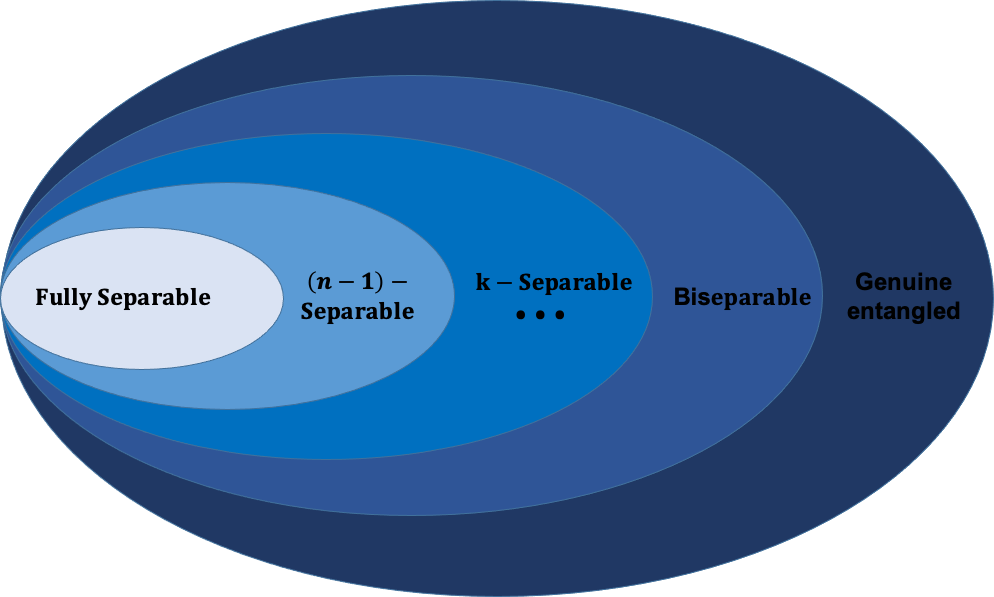}}
\caption[Convex set of $k$-separable states]{\label{fig:k-separable} Convex set of $k$-separable states.}
\end{figure}

Note that whenever a state is $k$-separable, it is automatically
also $l$-separable for all $1\leq l\leq k$. If we denote the set of all
$k$-separable states by $D_k$, then each set $D_k$ is a convex set and embedded within the next set, i.e., $D_n\subset D_{n-1}\subset\cdots\subset D_1$. The complement $D_1\setminus D_k$ of $D_k$ in $D_1$ is the set of all $k$-nonseparable states. In particular, the complement $D_1\setminus D_2$ is the set of all $2$-nonseparable states which are called genuine $n$-partite entangled states. Therefore, the cone of fully separable states lies in the middle and the cone of genuine multipartite entangled states lies at the outermost one (see Fig. \ref{fig:k-separable}).

\section{Quantum operations}\label{sec.2.3}
In quantum mechanics, the term quantum operation (also known as quantum dynamical map, quantum process or quantum superoperator) defines the class of transformations that a quantum system can undergo. Mathematically, a quantum operation is a linear map $\Lambda\colon\mathcal{B}(\mathcal{H})\to\mathcal{B}(\mathcal{H})$ that evolves a density matrix $\varrho$ to another density matrix $\Lambda(\varrho)=\varrho'$. The map $\Lambda$ is necessarily characterized by the following properties:
\begin{enumerate}
\item The probability that a physical process represented by map $\Lambda$ occurs is given by $\tr\big(\Lambda(\varrho)\big)$, where $\varrho$ is the initial state. Thus
\begin{equation}
0\leq\tr\big(\Lambda(\varrho)\big)\leq1\,.
\end{equation}
\item The map $\Lambda$ has to be a convex-linear map on the set of density matrices, that is, for probabilities $\{p_i\}$,
\begin{equation}
\Lambda\big(\sum_ip_i\varrho_i\big)=\sum_ip_i\Lambda(\varrho_i)\,.
\end{equation}
\item Since in quantum information we mostly deal with composite systems, the map $\Lambda$ has to be completely positive. A positive map $\Lambda$ is called completely positive if any tensor extension to a larger Hilbert space, i.e., $\mathbb{1}_d\otimes\Lambda$, is a positive map. Here, $d$ denotes the
dimension of the extension and is arbitrary. So we have
\begin{equation}
\mathbb{1}_{d_A}\otimes\Lambda_B(\varrho_{AB})\geq0 \qquad \forall~\varrho_{AB}\in\mathcal{B}(\mathcal{H}_A)\otimes\mathcal{B}(\mathcal{H}_B)\,.
\end{equation}
\end{enumerate}

In summary, any quantum operation that describes a physical process a quantum system can undergo, is described by a Completely Positive Trace-Preserving linear map (CPTP map).

Note that in the context of quantum information and computation, a quantum operation is called a quantum channel.

Unitary evolution is the simple example of quantum operation, for which
\begin{equation}
\Lambda_U(\varrho)=U\varrho\, U^{\dagger}\,.
\end{equation}

\begin{example}
A two-qubit controlled NOT (CNOT) gate that is a fundamental quantum logic gate (or simply quantum gate) is a quantum operation. Actually, it is a unitary operation which is defined as follows
\begin{equation}
U_{\rm{CNOT}}=|0\rangle\langle0|\otimes\mathbb{1}+|1\rangle\langle1|\otimes\sigma_1=\begin{pmatrix}
1 & 0 & 0 & 0\\
0 & 1 & 0 & 0\\
0 & 0 & 0 & 1\\
0 & 0 & 1 & 0
\end{pmatrix}\,,
\end{equation}
\sloppy where the matrix representation is with respect to the canonical basis $\{|00\rangle, |01\rangle,|10\rangle, |11\rangle\}$. The quantum circuit illustration of this gate is represented in Fig. \ref{fig:CNOT}.
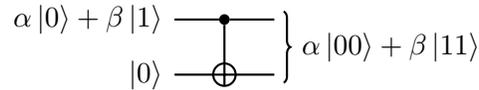
\begin{figure}[th]
\center{\begin{quantikz}
\lstick{$\alpha\ket{0}+\beta\ket{1}$} & \ctrl{1} & \qw\rstick[wires=2]{$\alpha\ket{00}+\beta\ket{11}$} \\
\lstick{$\ket{0}$}  & \targ{} & \qw
\end{quantikz}
\caption[Quantum circuit of two-qubit CNOT gate]{\label{fig:CNOT} Quantum circuit of two-qubit CNOT gate.}}
\end{figure}
\end{example}

As another example of quantum operation, we can consider quantum measurement with outcomes labeled by $m$ which is described by Positive Operator-Valued Measure (POVM)\footnote{Projective measurement (discussed in Section \ref{subsec.Observables}) is the simplest case of a POVM which is a set of orthogonal projectors.}. POVM is a set of $\{M_m\}$ such that $\sum_m M_m^{\dagger}M_m=\mathbb{1}$. So quantum measurement can be defined as the following map
\begin{equation}
\Lambda_m(\varrho)=M_m\varrho M_m^{\dagger}\,.
\end{equation}
The state of the quantum system immediately after the measurement is
\begin{equation}\label{post-measure}
\varrho_m=\frac{\Lambda_m(\varrho)}{\tr(\Lambda_m(\varrho))}\,,
\end{equation}
and the probability of obtaining this measurement result is
\begin{equation}\label{post-measure-prob}
p(m)=\tr(\Lambda_m(\varrho))\,.
\end{equation}

In general, quantum operations (CPTP-maps) can be represented in an elegant form known as the Kraus representation \cite{Kraus71}.

\begin{theorem}[Kraus representation of CPTP-maps]
Any quantum operation $\Lambda$ acting on a quantum system with associated a $d$-dimensional Hilbert space $\mathcal{H}$ can be represented as 
\begin{equation}
\Lambda(\varrho)=\sum_{i=1}^{N}K_i\varrho K_i^{\dagger}\,, \qquad \sum_{i=1}^{N}K_i^{\dagger}K_i=\mathbb{1}\,,
\end{equation}
where the operators $K_i\colon\mathcal{H}\to\mathcal{H}$ are called Kraus operators (also known as effects) and $1\leq N\leq d^2$.
\end{theorem}

It is worth noting that when the map $\Lambda$ is describing a quantum measurement, the trace does not need to be preserved. It is because of this fact that the trace of the post measurement state, i.e., $\tr(\Lambda(\varrho))$ gives the probability with which the measurement outcome does occur (see Eqs. \eqref{post-measure} and \eqref{post-measure-prob}). Therefore, in the case where the process is deterministic, that is, no measurement is taking place, this reduces to the following requirement
\begin{equation}
\tr\big(\Lambda(\varrho)\big)=\tr(\varrho)=1 \qquad \forall~\varrho\in\mathcal{B}(\mathcal{H})\,.
\end{equation}

Overall, all valid quantum operations can be written in Kraus representation. There are two main classes of quantum operations:
\begin{enumerate}
\item Trace-preserving quantum operations, that is $\sum_i K_i^{\dagger}K_i=\mathbb{1}$.
\item Trace decreasing quantum operations, that is $\sum_i K_i^{\dagger}K_i<\mathbb{1}$.
\end{enumerate}

It is worth noting that the Kraus representation is not unique.

\subsubsection*{Local quantum operations}
Local quantum operations can be written as the tensor product of CPTP maps $\Lambda_i\colon\mathcal{B}(\mathcal{H}_i)\to\mathcal{B}(\mathcal{H}_i)$ acting on all subsystems respectively, that is
\begin{equation}
\Lambda_i=\bigotimes_{i=1}^n\Lambda_i\,.
\end{equation}
Regarding Kraus representation, any local quantum operation can be represented in terms of the Kraus operators $\{K_{j_i}\}$ as follows
\begin{equation}\label{Kraus-local}
\Lambda(\varrho)=\sum_{j_1,\ldots,j_n}\big(\otimes_{i=1}^n K_{j_i}\big)\varrho\big(\otimes_{i=1}^n K_{j_i}^{\dagger}\big)\,,
\end{equation}
where $\sum_{j_i}(\otimes_{i=1}^n K_{j_i}^{\dagger}K_{j_i})=\mathbb{1}$ and $\sum_{j_i}(\otimes_{i=1}^n K_{j_i}^{\dagger}K_{j_i})<\mathbb{1}$ hold for trace-preserving quantum operations and for trace decreasing quantum operations, respectively.

Below we give a brief review of three most studied operations: Local Unitary (LU), Local Operations and Classical Communication (LOCC), and Stochastic LOCC (SLOCC).

\subsection{LU operations}\label{subsec.2.3.1}
Unitary operations belong to the class of deterministic CPTP maps. A unitary transformation is a transformation corresponding to a change of a basis so it preserves the inner product . A LU transformation corresponds to a change of a basis in each of the subsystems. These transformations simply reflect the choice of our point of view rather than any specific manipulation of the physical system. It reflects the fact that two equivalent states under local unitary operations have the same matrix form, only the choice basis of subsystems is different.

\begin{definition}[LU equivalence]
Given two $n$-partite quantum states $\varrho,\sigma\in\mathcal{B}(\mathcal{H})$
\begin{equation}
\varrho\extoverset[6pt]{LU}{\sim}\sigma \quad \text{iff} \quad \exists~{\rm{U}}={\rm{U}}(d_1)\otimes\cdots\otimes{\rm{U}}(d_n) \quad \text{s.t.} \quad \varrho={\rm{U}}\sigma{\rm{U}}^{\dagger}\,.
\end{equation}
\end{definition}

For two pure $n$-partite quantum states $|\psi\rangle,|\phi\rangle\in\mathcal{H}=\otimes_{i=1}^{n}\mathbbm{C}^{d_i}$, the LU equivalence between them implies:
\begin{equation}
|\psi\rangle\extoverset[6pt]{LU}{\sim}|\varphi\rangle \quad \text{iff} \quad \exists\, {\rm{U}}(d_i)~\forall\, i\in[n]~~\text{s.t.}~~|\psi\rangle={\rm{U}}(d_1)\otimes\cdots\otimes{\rm{U}}(d_n)|\varphi\rangle\,.
\end{equation}

\subsection{LOCC operations}\label{subsec.2.3.2}
LOCC is one of the most important classes of quantum operations in quantum information theory. Generally, LOCC is a method where a local operation is performed on a part of the system, and the result of that operation is communicated classically to another part where usually another local operation is performed conditioned on the information received. All CPTP maps in Eq. \eqref{Kraus-local} can be considered as LOCC operations. More precisely, LOCC can be considered as local unitary operations, local measurements and coupling to ancillary systems followed by their removal. These operations are able to answer the question whether or not a multipartite quantum state can be transformed deterministically into another multipartite quantum state in case each party of the multipartite system has access exclusively to its own subsystem. 

LOCC is also known as the ``free operation'' in the resource theories of entanglement since entanglement cannot be produced from separable states via LOCC. This can be understood by considering a general local operation on a bipartite quantum system as follows
\begin{equation}
\Lambda(\varrho)=\sum_{i}(K_i\otimes L_i)\varrho(K_i^{\dagger}\otimes L_i^{\dagger})\,,
\end{equation}
where $K_i$ and $L_i$ are Kraus operators and $\sum_{i}K_{i}^{\dagger}K_{i}\otimes L_{i}^{\dagger}L_{i}=\mathbb{1}_{\mathcal{H}_A\otimes\mathcal{H}_B}$. For a bipartite product state $\varrho=\varrho_A\otimes\varrho_B$ we have
\begin{equation}\label{free-operation}
\Lambda(\varrho)=\sum_i(K_i\varrho_A K_i^{\dagger})\otimes(L_i\varrho_B L_i^{\dagger})=\sum_ip_i\tilde{\varrho}_A^i\otimes\tilde{\varrho}_B^i\,,
\end{equation}
where
\begin{equation}
\tilde{\varrho}_A^i=\frac{K_i\varrho_A K_i^{\dagger}}{\tr(K_i\varrho_A K_i^{\dagger})}\,, \quad \tilde{\varrho}_B^i=\frac{L_i\varrho_B L_i^{\dagger}}{\tr(L_i\varrho_B L_i^{\dagger})}\,, \quad p_i=\tr(K_i\varrho_A K_i^{\dagger})\tr(L_i\varrho_B L_i^{\dagger})\,.
\end{equation}
Regarding Eq. \eqref{free-operation}, one can see that by a general local operation on a bipartite product state we could produce classical correlation but the produced state is still separable.

For pure quantum states it has been shown that the LOCC equivalence coincides with the LU equivalence \cite{BPRST00}.

\subsection{SLOCC operations}\label{subsec.2.3.3}

For pure quantum states we have the following definition \cite{DVC00}.

\begin{definition}[SLOCC equivalence]
Given two $n$-partite quantum states $|\psi\rangle$ and $|\varphi\rangle$ in $\mathcal{H}=\otimes_{i=1}^n\mathbbm{C}^{d_i}$
\begin{equation}\label{SLOCC-equiv}
|\psi\rangle\extoverset[6pt]{SLOCC}{\sim}|\varphi\rangle~~\text{iff} ~~\exists\, A_i\in{\rm{SL}}(d_i,\mathbbm{C})~\forall\, i\in[n]~~\text{s.t.}~~|\psi\rangle=A_1\otimes\cdots\otimes A_n|\varphi\rangle\,.
\end{equation}
\end{definition}

\subsubsection{Entanglement transformation}
Entanglement transformation is a fundamental problem in quantum information theory. For bipartite systems, the problem is completely solved. Nielsen has provided necessary and sufficient conditions for the LOCC convertibility of pure bipartite states \cite{Nielsen99}.

\begin{definition}[Majorization]
For a vector $\mathbf{x}\in\mathbbm{R}^d$ let $\mathbf{x} ^{\downarrow}\in\mathbbm{R}^d$ be a vector with the same components of the vector $\mathbf{x}$, but sorted in descending order. Given two $d$-dimensional vectors $\mathbf{x}=(x_1,\ldots,x_d)$ and $\mathbf{y}=(y_1,\ldots,y_d)$ in $\mathbbm{R}^d$ we say that $\mathbf{x}$ is majorized by $\mathbf{y}$ (equivalently $\mathbf{y}$ majorizes $\mathbf{x}$), denoted by $\mathbf{x}\prec\mathbf{y}$, iff
\begin{equation}
\sum_{i=1}^{k}x_i^{\downarrow}\leq\sum_{i=1}^{k}y_i^{\downarrow} \quad \forall~k\in\{1,\ldots,d\}\,,
\end{equation}
with equality holding when $k=d$.
\end{definition}

\begin{theorem}[Ref. \cite{Nielsen99}]
A pure bipartite quantum state $|\psi\rangle$ can be transformed deterministically to another pure bipartite quantum state $|\varphi\rangle$ via LOCC operations iff the Schmidt Coefficients of the first state is majorized by those of the second one, i.e.,
\begin{equation}
|\psi\rangle\xrightarrow[]{\text{LOCC}}|\varphi\rangle\qquad\text{iff}\qquad\lambda_{\psi}\prec\lambda_{\varphi}\,.
\end{equation}
\end{theorem}

Regarding SLOCC operations, any bipartite state $|\psi\rangle$ can be probabilistically transformed into $|\varphi\rangle$ iff the Schmidt rank of the reduced density operator of $|\psi\rangle$ is greater than that one of $|\varphi\rangle$, i.e.,
\begin{equation}
|\psi\rangle\xrightarrow[]{\text{SLOCC}}|\varphi\rangle\qquad\text{iff}\qquad\rk_S(\psi)\geq\rk_S(\varphi)\,.
\end{equation}

\section{Entanglement quantification}\label{sec.2.4}
Regarding entanglement as a resource, one would like to know how much of that resource is available in any given situation. An entanglement measure is a function that quantifies the amount of entanglement present in a quantum state. Actually, an entanglement measure is an entanglement monotone that vanishes on all separable states. An entanglement monotone is a linear form that maps a quantum state to a positive real numbers. More precisely, we have the following definition of entanglement monotone \cite{Vidal00}.

\begin{definition}[Entanglement monotone]
An entanglement monotone is a function $\mu\colon\mathcal{B}(\mathcal{H})\to\mathbbm{R}^+$ which maps density operators to positive real numbers. This function is invariant under unitary operations and non–increasing, on average, under LOCC. That is,
\begin{equation}
\mu(\Lambda_{\text{LOCC}}\varrho)\leq\mu(\varrho) \qquad \forall~\varrho\in\mathcal{B}(\mathcal{H})~\&~\Lambda_{\text{LOCC}}\,.
\end{equation}
\end{definition}

The phrase ``on average'' refer to the general case, in which a pure state $\varrho$ is transformed by a probabilistic local operation into a mixture,
\begin{equation}
\varrho\to\sum_ip_i\varrho_i \qquad \Rightarrow \qquad \mu(\varrho)\geq\sum_ip_i\mu(\varrho_i)\,.
\end{equation}

In the following, we give the requirements for a good entanglement measure. See Ref. \cite{PV07} for a review on the subject.

\subsection{Entanglement measure}\label{subsec.2.4.1}
 A good entanglement measure $\mathcal{E}$ has to fulfill several requirements. However, it is still an open question whether all of these conditions are indeed necessary.
\begin{itemize}
\item[1)] {\em Discriminance:} $\mathcal{E}(\varrho)=0$ iff $\varrho$ is separable.
\item[2)] {\em Monotonicity under LOCC:} applying local operations to $\varrho$ and classically communicating cannot increase the entanglement of $\varrho$, i.e.,
\begin{equation}
\mathcal{E}(\Lambda_{LOCC}(\varrho))\leq\mathcal{E}(\varrho)\,.
\end{equation}
\item[3)] {\em Convexity:} The entanglement measure should be a convex function, i.e.
\begin{equation}
\mathcal{E}(p\varrho+(1-p)\sigma)\leq p\,\mathcal{E}(\varrho)+(1-p)\mathcal{E}(\sigma)\,,
 \end{equation}
for $0\leq p\leq1$.
\item[4)]\ {\em Continuity:} In the limit of vanishing distance between two density matrices the difference between  their  entanglement should tend to zero, i.e.,
\begin{equation}
\mathcal{E}(\varrho)-\mathcal{E}(\sigma)\rightarrow0 \quad \text{for} \quad \parallel\varrho-\sigma\parallel_1\rightarrow0\,.\footnote{The trace norm $\parallel O \parallel_1$ of an operator $O$ is the sum of its singular values. That is \\ $\parallel O \parallel_1=\tr\sqrt{OO^{\dagger}}$.}
  \end{equation}
\item[5)]\ {\em Additivity:} A certain number $n$ of 
identical copies of the state $\varrho$ should contain $n$ times the entanglement of one copy,
\begin{equation}
\mathcal{E}(\varrho^{\otimes n})=n\,\mathcal{E}(\varrho)\,.
\end{equation}
\item[6)]\ {\em  Subadditivity:} The entanglement of 
 the tensor product of two states
 $\varrho$ and $\sigma$ should not be larger than the sum of 
 the entanglement of each of the states,
\begin{equation}
\mathcal{E}(\varrho\otimes\sigma )\leq\mathcal{E}(\varrho )+\mathcal{E}(\sigma)\,,
\end{equation}
\item[7)]\ {\em Normalization:}  The entanglement of a maximally entangled state of two $d$-dimensional systems is given by
\begin{equation}
\mathcal{E}(|\psi_{\rm{MES}}\rangle\langle\psi_{\rm{MES}}|)=\log d\,.
\end{equation}
\end{itemize}
 
\subsubsection{Some important entanglement measures}\label{subsec.2.4.2.1}
Following we introduce some important entanglement measures, without making the attempt to discuss all existing entanglement measures.

\begin{itemize}
\item[-]{\em Entropy of entanglement \cite{BBPS96}:}
For pure bipartite states, the entropy of entanglement is defined as the von Neumann entropy of the reduced density matrix, i.e.,
\begin{equation}
\mathcal{E}_{\textit{E}}(|\psi\rangle_{AB})=S(\varrho_A)=-\tr(\varrho_A\log\varrho_A)\,.
\end{equation}
\item[-] {\em Distillable entanglement \cite{BDSW96,BBPSSW97,Rains99}:}
The distillable entanglement tells us how much maximally entangled state can be extracted from a given entangled state $\varrho$ by LOCC. Mathematically, it is the ratio of the number of maximally entangled state $|\psi_{\rm{MES}}\rangle$ as output states over the needed input states $\varrho$, maximized over all LOCC operations, in an asymptotic setting, i.e.,
\begin{equation}
\mathcal{E}_{\textit{D}}(\varrho)=\sup_{\{\Lambda_{LOCC}\}}\lim_{n_\varrho\rightarrow\infty}\frac{n^{\rm{out}}_{|\psi_{\rm{MES}}\rangle}}{n^{\rm{in}}_\varrho}\,.
\end{equation}
\item[-] {\em Entanglement of formation \cite{BDSW96,BBPSSW97}:}
The entanglement of formation is the averaged von Neumann entropy of the reduced density matrices of the pure states $|\psi_i\rangle$, minimized over all possible decompositions $\varrho =\sum_ip_i\proj{\psi_i}$, i.e.,
\begin{equation}
\mathcal{E}_{\textit{F}}(\varrho)=\inf_{\{p_i,|\psi_i\rangle\}}\sum_ip_i\mathcal{E}_{\textit{E}}(|\psi_i\rangle)\,.
\end{equation}
Actually, entanglement of formation $\mathcal{E}_{\rm F}$ is an extension of entropy of entanglement for mixed  via the convex roof extension.
\item[-]{\em Entanglement cost \cite{BDSW96}:}
The entanglement cost tells us how expensive it is to create an entangled state $\varrho$. Mathematically, it is the ratio of the number of maximally entangled states $|\psi_{\rm{MES}}\rangle$ as input states over the produced output states $\varrho$, minimized over all LOCC operations, in an asymptotic setting, i.e.,
\begin{equation}
\mathcal{E}_{\textit{C}}(\varrho)=\inf_{\{\Lambda_{LOCC}\}}\lim_{n_\varrho\rightarrow\infty}\frac{n^{\rm{in}}_{|\psi_{\rm{MES}}\rangle}}{n^{\rm{out}}_\varrho}\ .
\end{equation}
Entanglement cost can be related to entanglement of formation as follows (see Ref. \cite{HHT01})
\begin{equation}
\mathcal{E}_{\textit{C}}(\varrho)=\lim_{n\to\infty}\frac{\mathcal{E}_{\textit{F}}(\varrho^{\otimes n})}{n}\,.
\end{equation}
\item[-] {\em Relative entropy of entanglement \cite{VPRK97}:}
The relative entropy can be seen intuitively as the ``distance'' of the entangled $\varrho$ to the closest separable state $\sigma$, although it is not a distance in the mathematical sense,
\begin{equation}
\mathcal{E}_{\textit{R}}(\varrho)=S(\varrho\parallel\sigma)=\inf_{\sigma\in{\rm{D}}}\tr[\varrho(\log\varrho-\log\sigma)]\,.
\end{equation}
where ${\rm{D}}$ is the set of separable states.
\end{itemize}

It is worth noting that the extension of any entanglement measure for pure states to an entanglement measure for mixed states can be done by the convex roof construction \cite{Ulmann98}.

In Ref. \cite{HW97, Wootters98} Hill and Wootters introduced the concurrence as a measure for two-qubit systems. It is defined as the overlap of a given state $|\psi\rangle$ with its spin-flipped state, i.e., $|\tilde{\psi}\rangle=\sigma_2|\psi^*\rangle$. That is
\begin{equation}
\mathcal{C}(|\psi\rangle)=\langle\psi|\sigma_2\otimes\sigma_2|\psi^*\rangle\,.
\end{equation}
It can also be extended to mixed states by the convex roof construction. That is
\begin{equation}
\mathcal{C}(\varrho)=\inf_{\{p_i,|\psi_i\rangle\}}\sum_ip_i\mathcal{E}_{\textit{E}}(|\psi_i\rangle)\,.
\end{equation}
An explicit formula for this is given as follow
\begin{equation}
\mathcal{C}(\varrho)=\max\{\sqrt{\lambda_1}-\sqrt{\lambda_2}-\sqrt{\lambda_3}-\sqrt{\lambda_4}\}\,,
\end{equation}
where $\lambda_i$'s are eigenvalues of the Hermitian operator
\begin{equation}
\varrho\tilde{\varrho} \qquad \text{with} \qquad \tilde{\varrho}=\sigma_2\otimes\sigma_2\varrho^*\sigma_2\otimes\sigma_2\,
\end{equation}
which are sorted in descending order. The concurrence has been generalized to higher dimensional systems \cite{MKB04}
\begin{equation}
\mathcal{C}(\varrho)=\sqrt{2(1-\tr(\varrho_r))}\,,
\end{equation}
where $\varrho_r$ is the reduced density matrix.

Moreover, based on the concurrence, an entanglement measure for three-qubit states is deﬁned in Ref. \cite{CKW00}. That is the tangle:
\begin{equation}
\tau=\mathcal{C}_{A(BC)}^2-\mathcal{C}_{AB}^2-\mathcal{C}_{AC}^2\,.
\end{equation}

\section{Entanglement classification}\label{sec.2.5}
Multipartite entanglement is a physical resource, like energy, associated with the complex nonclassical correlations that are a basis for quantum-enhanced applications. Hence, it is important to characterize entanglement in order to learn which states are more powerful. Thanks to Schmidt decomposition, this was done by developing entanglement monotones for bipartite systems (see Section \ref{sec.2.4}). However, extension of these to multipartite systems is quite complicated. That is why a classification of multipartite entanglement is required to single out states that perform (probabilistically) equally well quantum information tasks. This is done by considering the following questions: Let us have two quantum states. How can we transfer entanglement from one state to another one? In other words, how can we testify entanglement equivalency of two quantum states using only local operations?

Because of nonlocality of quantum entanglement, a proper equivalent relation should be local operation. The most studied local operations are LU, LOCC, and SLOCC for mixed states and LU and SLOCC for pure states.

In this dissertation, we focus on pure multipartite quantum states.

While LU equivalence can provide a useful division into different equivalence classes, already for bipartite systems, a mathematical analysis is only possible in some lower dimensional systems \cite{Kraus10-1,Kraus10-2}. In contrast, SLOCC equivalence provides a neat characterization of different classes of entanglement.

\subsection{Entanglement classification of bipartite systems}\label{subsec.2.5.1}
Thanks to the existence of the Schmidt decomposition, bipartite entanglement classification can be fully determined under SLOCC.

\begin{theorem}
Schmidt rank is an SLOCC invariant that provides us with a complete entanglement classification of all bipartite states.
\end{theorem}

In the next chapters, we will see that the rank of the coefficient matrix of a quantum state is equal to its Schmidt rank.

\subsection{Classification of three-qubit entanglement}\label{subsec.2.5.2}
In this section, we recall the complete entanglement classification for both pure and mixed three-qubit states under SLOCC \cite{DVC00,ABLS01}.

\subsubsection{Pure three-qubit states}
Although there is no Schmidt decomposition for arbitrary states of more than two parties \cite{Peres95,Pati00} there is a generalization of Schmidt decomposition for pure multipartite quantum states \cite{CHS00,AACJLT00,AAJT01}.

For instance, In Ref. \cite{AACJLT00} the authors gave a generalized Schmidt decomposition for three-qubit pure states, in the sense that the coefficients of this decomposition carry all the information
about the non-local properties of the state.

\begin{theorem}[Ref. \cite{AACJLT00}]
Up to an LU operation, any three-qubit pure state can be written as follows
\begin{equation}\label{3-qubit-GSD}
|\psi\rangle=\lambda_0|000\rangle+\lambda_1{\rm{e}}^{\mathbf{i}\theta}|100\rangle+\lambda_2|101\rangle+\lambda_3|110\rangle+\lambda_4|111\rangle\,,
\end{equation}
where $\lambda_i\geq0$, $0\leq\theta\leq\pi$, and $\sum_i\lambda_i^2=1$.
\end{theorem}

The concurrence and the tangle of the state in Eq. \eqref{3-qubit-GSD} reads
\begin{align}\nonumber
&\mathcal{C}_{AB}=2\lambda_0\lambda_3 \qquad \mathcal{C}_{BC}=2\sqrt{(\lambda_1\lambda_4)^2+(\lambda_1\lambda_3)^2 -2\lambda_1\lambda_2\lambda_3\lambda_4\cos\theta} \\ \label{3-qubit-C-T}
& \mathcal{C}_{AC}=2\lambda_0\lambda_2 \qquad \tau=4\lambda_0^2\lambda_4^2\,.
\end{align}

\begin{table}[t]
\centering
\caption[SLOCC classification of three-qubit entanglement]{\label{table:2.1} SLOCC classification of three-qubit entanglement.}
\begin{tabular*}{\linewidth}{c@{\extracolsep{\fill}}cccccc}
\hline\hline
&&&&& \\ [-2ex]
Class & Representative & $\mathcal{C}_{AB}$ & $\mathcal{C}_{AB}$ & $\mathcal{C}_{AB}$ & $\tau$ \\ [0.5ex]
\hline
&&&&& \\ [-2ex]
$|\rm{Sep}\rangle$ & $|000\rangle$ & $=0$ & $=0$ & $=0$ & $=0$ \\ [0.5ex]
$|{\rm B}_1\rangle$ & $|000\rangle+|011\rangle$ & $\neq0$ & $=0$ & $=0$ & $=0$ \\ [0.5ex]
$|{\rm B}_2\rangle$ & $|000\rangle+|101\rangle$ & $=0$ & $\neq0$ & $=0$ & $=0$ \\ [0.5ex]
$|{\rm B}_3\rangle$ & $|000\rangle+|110\rangle$ & $=0$ & $=0$ & $\neq0$ & $=0$ \\ [0.5ex]
$|{\rm{W}}\rangle$ & $|001\rangle+|010\rangle+|100\rangle$ & $\neq0$ & $\neq0$ & $\neq0$ & $=0$  \\ [0.5ex]
$|\rm{GHZ}\rangle$ & $|000\rangle+|111\rangle$ & $\neq0$ & $\neq0$ & $\neq0$ & $\neq0$  \\
\hline\hline
\end{tabular*}
\end{table}

\begin{figure}[t]
\centering
\subfloat[\centering]{\includegraphics[width=5.5cm]{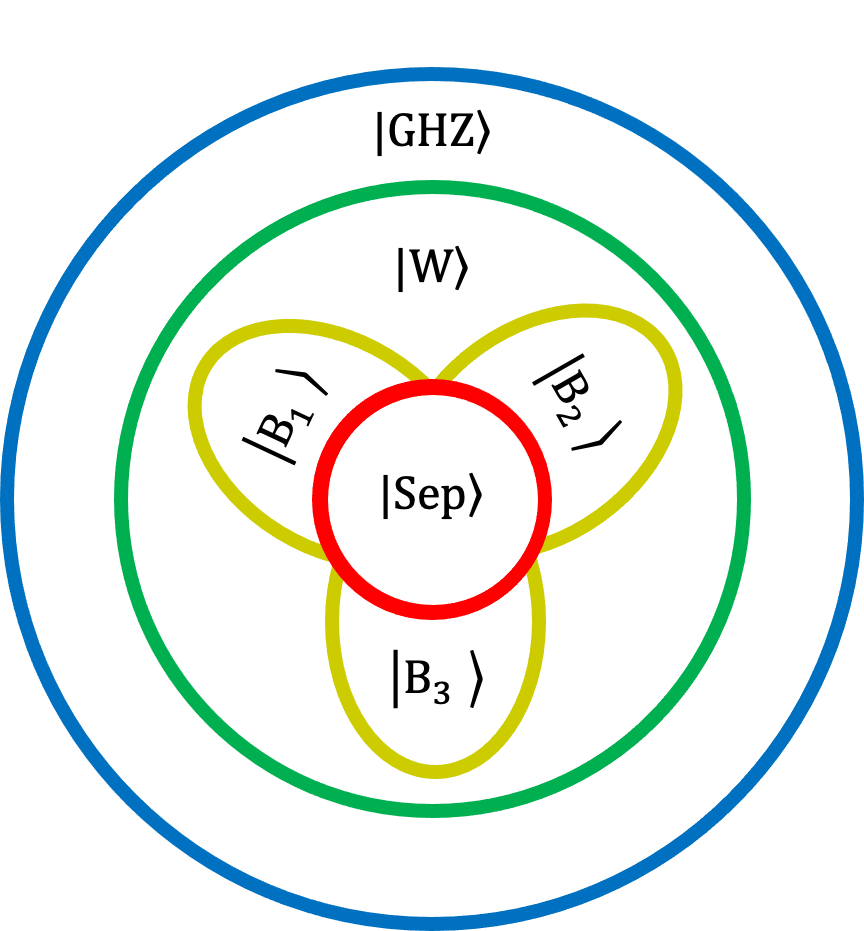}} \label{fig:2.4.1}
  \qquad
\subfloat[\centering]{\includegraphics[width=5.5cm]{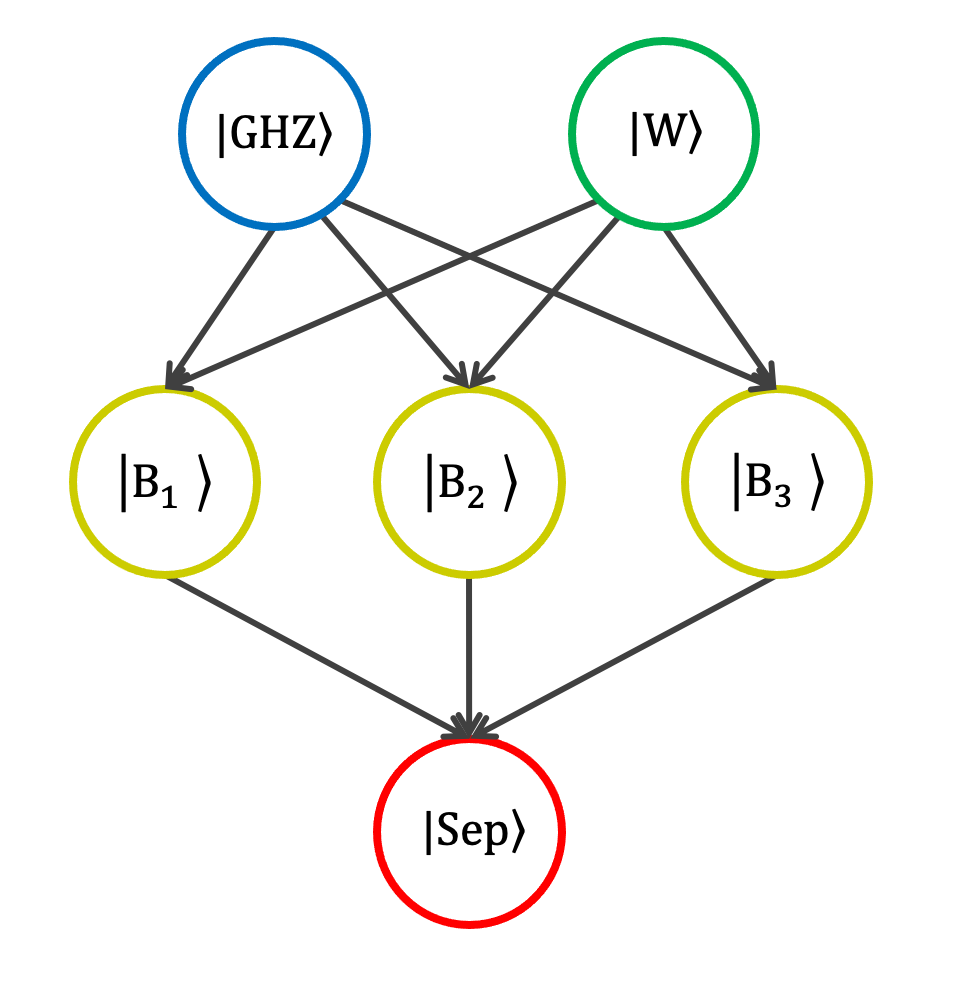}} \label{fig:2.4.2}
\caption[Entanglement classification of pure three-qubit states]{(a) SLOCC orbits of three-qubit entanglement. (b) Hasse diagram of the SLOCC classification of three-qubit entanglement. The direction of the arrows indicates noninvertible SLOCC transformations between classes that generate the entanglement hierarchy.}
\label{fig:3-qubit}%
\end{figure}

Since these entanglement measures are SLOCC invariants, we have the following complete entanglement classification for pure three-qubit states under SLOCC:
\begin{enumerate}
\item {\bf Fully separable:} This class contains fully separable states denoted by $\rm{Sep}$. This class represents all fully separable states that are SLOCC equivalent to $|000\rangle$, i.e., all pure three-qubit states with zero concurrences and tangle in Eq. \eqref{3-qubit-C-T}.

\item {\bf Biseprable $\mathbf{B}_1$:} This class contains biseparable states such that parties $B$ and $C$ are entangled and they are separable from party $A$. So this class contains all biseparable states that are SLOCC equivalent to $|{\rm B}_1\rangle=|0_A\rangle|{\rm{EPR}}_{BC}\rangle$, i.e., all pure three-qubit states with $\mathcal{C}_{AB}=\mathcal{C}_{AC}=\tau=0$ and $\mathcal{C}_{BC}\neq0$.

\item {\bf Biseprable $\mathbf{B}_2$:} This class contains biseparable states such that parties $A$ and $C$ are entangled and they are separable from party $B$. So this class contains all biseparable states that are SLOCC equivalent to $|{\rm B}_2\rangle=|0_B\rangle|{\rm{EPR}}_{AC}\rangle$, i.e., all pure three-qubit states with $\mathcal{C}_{AB}=\mathcal{C}_{BC}=\tau=0$ and $\mathcal{C}_{AC}\neq0$.

\item {\bf Biseprable $\mathbf{B}_3$:} This class contains biseparable states such that parties $A$ and $B$ are entangled and they are separable from party $C$. So this class contains all biseparable states that are SLOCC equivalent to $|{\rm B}_3\rangle=|0_C\rangle|{\rm{EPR}}_{AB}\rangle$, i.e., all pure three-qubit states with $\mathcal{C}_{AC}=\mathcal{C}_{BC}=\tau=0$ and $\mathcal{C}_{AB}\neq0$.
 
\item {\bf W:} This class represents all genuine entangled states such that are SLOCC equivalent to $|{\rm{W}}\rangle=|001\rangle+|010\rangle+|100\rangle$, i.e., all pure three-qubit states with non-zero concurrences and zero tangle in Eq. \eqref{3-qubit-C-T}.

\item {\bf GHZ:} This class which is denoted by $|{\rm{GHZ}}\rangle=|000\rangle+|111\rangle$ represent all genuine entangled states that are SLOCC equivalent to $|{\rm{GHZ}}\rangle$, i.e., all pure three-qubit states with non-zero concurrences and tangle in Eq. \eqref{3-qubit-C-T}.
\end{enumerate}

In summary, regarding entanglement classification of pure three-qubit states under the SLOCC, we have six different classes, namely, fully separable states, three different classes of biseparable states with respect to three different bipartitions, and two inequivalent genuine entangled states $\rm{W}$ and $\rm{GHZ}$. These classes and the their representative states from each class are summarized in Table \ref{table:2.1}. A visual representation of these SLOCC orbits is illustrated in Fig. \ref{fig:3-qubit} (a).

According to Eq. \eqref{SLOCC-equiv}, a noninvertible local operator transforms $|\psi\rangle$ into $|\varphi\rangle$ where at least one of the local operators is not full rank. So it is possible to transform each of $\rm{GHZ}$ and $\rm{W}$ states to one of the biseparable states or even to fully separable states. Note that the inverse transformations, e.g., from the class of fully separable states to one of the biseparable or genuine entangled states, are impossible as they would imply an increase of the rank of at least
one of the reduced density operators. These results are summarized in Fig. \ref{fig:3-qubit} (b).

\subsubsection{Mixed three-qubit states}
Now, we present a complete entanglement classification for mixed three-qubit states under SLOCC \cite{ABLS01}. A mixed three-qubit state $\varrho$ can be written as a convex combination of pure three-qubit states (See Eq. \ref{MixedState} in Subsection \ref{subsec.2.1.2}). Regarding the entanglement classification of pure three-qubit states and 
using \Cref{def:k-separable-mixed} in Subsection \ref{subsec.2.2.2} we have following classes:
\begin{enumerate}
\item {\bf Fully separable:} If $\varrho$ can be decomposed as a convex combination of projectors onto only pure separable states, then it belongs to the convex compact set of fully separable states, denoted by $\rm{Sep}$.

\item {\bf Biseparable (B):} If in the decomposition of $\varrho$ at least one pure biseparable state of any kind (and no genuine entangled state) is needed, then it belongs to the convex compact set of biseparable states. More precisely $\varrho$ belongs to $\rm{B}\setminus\rm{Sep}$.

\item {\bf W:} If in the decomposition of $\varrho$ at least one pure $\rm{W}$ state (and no pure $\rm{GHZ}$ state) is needed, then it belongs to the convex compact set of $\rm{W}$ states. More precisely $\varrho$ belongs to $\rm{W}\setminus\rm{B}$.

\item {\bf GHZ:} If in the decomposition of $\varrho$ at least one pure $\rm{GHZ}$ state is needed, then it belongs to the convex compact set of $\rm{GHZ}$ states. More precisely $\varrho$ belongs to $\rm{GHZ}\setminus\rm{W}$.
\end{enumerate}

\begin{figure}[t]
\centering
{\includegraphics[width=5.5cm]{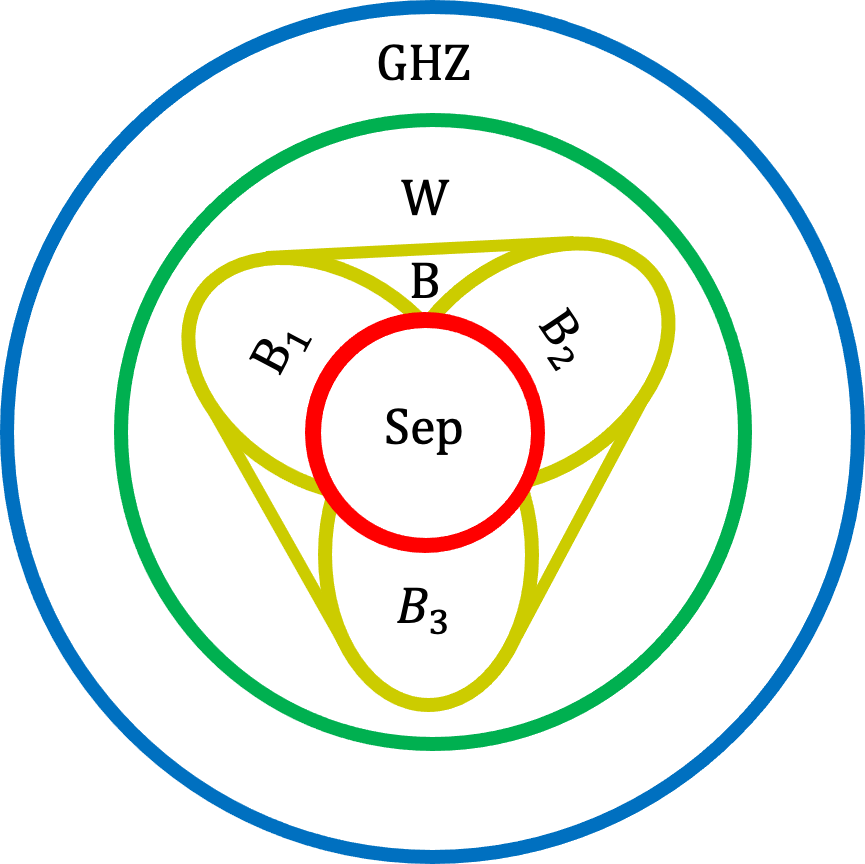}}
\caption[Entanglement classification of mixed three-qubit states]{ Entanglement classification of mixed three-qubit states under SLOCC.
$\rm{Sep}$: fully separable class; B: biseparable class (convex hull of biseparable states with respect to any bipartition), W class, and GHZ
class.}
\label{fig:3-qubit-mixed}%
\end{figure}

It is worth noting that these classes are embedded into each other (compare with Fig. \ref{fig:k-separable}), i.e.,
\[
\rm{Sep}\subset\rm{B}\subset\rm{W}\subset\rm{GHZ}\,.
\]
Also, it is important to note that $\rm{GHZ}\subset\rm{W}$ is not correct since otherwise the class $\rm{GHZ}$ would not be compact, as can be seen by studying the most general form of a $\rm{W}$-type state versus a $\rm{GHZ}$-type state, as given in \cite{AACJLT00}. In the next chapter, we also will see that we cannot produce a $\rm{GHZ}$-type state from a $\rm{W}$-type state but the reverse is approximately possible.

In summary, regarding entanglement classification of mixed three-qubit states under the SLOCC, we have four different classes, namely, fully separable class, biseparable class that is a convex hull of biseparable states with respect to any bipartition, $\rm{W}$ class, and $\rm{GHZ}$ class. A visual representation of these SLOCC classes is illustrated in Fig. \ref{fig:3-qubit-mixed}.

\subsection{Inductive entanglement classification}\label{subsec.2.5.3}
It is known that for four or more qubits there are an (uncountable) infinite number of SLOCC classes \cite{DVC00}. Hence, it is desirable to bunch the infinite number of SLOCC classes in a finite number of families with the common physical and/or
mathematical properties.

In Ref. \cite{VDDV02}, Verstraete et al., have introduced an entanglement classification of four-qubit systems. They have benefited from the the following elegant isomorphism
\begin{equation}
{\rm{SL}}(2,\mathbbm{C})\times{\rm{SL}}(2,\mathbbm{C})\cong{\rm{SO}}(4,\mathbbm{C})\,,
\end{equation}
and grouped the infinite number of SLOCC classes into nine families. However, this method cannot be extended to more than four qubits since there is no known isomorphism for this purpose.

Lamata et al. \cite{LLSS06} also introduced an inductive method to partition the infinite number of SLOCC classes into a finite number of families. Based on this method, they have bunched four qubits SLOCC classes into eight entanglement families (up to qubit permutation) \cite{LLSS07}. Some years later Backens has shown that the inductive approach yields ten entanglement families for four-qubit entangled states \cite{Backens17}. The main idea of this approach is to investigate the possible entanglement families of the right singular subspace of the coefficient matrix of each pure state $|\psi\rangle$ expressed in the canonical basis. In this order, one can write an $n$-qubit pure state as
\begin{equation}
|\psi\rangle=|e_{0}\rangle|\varphi_{0}\rangle+|e_{1}\rangle|\varphi_{1}\rangle\,,
\end{equation}
where $|e_{0}\rangle$ and $|e_{1}\rangle$ are two linearly independent states of the $i$-th ($i=1,\cdots,n$) qubit, and $|\varphi_{0}\rangle$ and $|\varphi_{1}\rangle$ are the states of the rest $n-1$ qubits. In general, normalization is not needed as SLOCC operations can change the norm of the states. The entanglement families are determined by considering all combinations of entanglement types of the rest $n-1$ qubits in different spanning sets for span$\{|\varphi_{0}\rangle,|\varphi_{1}\rangle\}$. In Refs. \cite{LLSS06,LLSS07}, the authors label the entanglement families according to the types of entangled vectors in the spanning set where $|\varphi_{0}\rangle$ and $|\varphi_{1}\rangle$ can take the values $000$, $0_{k}\Psi$, GHZ, or W. Here, $000$ denotes a fully separable state, while $0_{k}\Psi$ denotes a biseparable state where $0_{k}$ is the $k$-th single qubit in a product with an entangled state of the remaining qubits.

Here, we present two examples of four qubits states that invalidates this approach when we consider different partitioning of subsystems.

In Refs. \cite{LLSS06, LLSS07, Backens17}, the authors have considered the partition of the first qubit from the rest ($1|23$ and $1|234$) but there should be no loss of generality in choosing other partitions as it was also mentioned in Ref. \cite{LLSS06}. For the first example, let us consider two states of the family span$\{000,0_{k}\Psi\}$ for four qubits where the partition is $1|234$
\begin{align*}
&|0000\rangle+|1100\rangle+|1111\rangle=|0\rangle|000\rangle+|1\rangle|1\rangle(|00\rangle+|11\rangle)\,, \\
&|0000\rangle+|1101\rangle+|1110\rangle=|0\rangle|000\rangle+|1\rangle|1\rangle(|01\rangle+|10\rangle)\,.
\end{align*}
Now, let to consider the partition $123|4$ for the above states. We will have
\begin{align*}
&|0000\rangle+|1100\rangle+|1111\rangle=(|00\rangle+|11\rangle)|0\rangle|0\rangle+|111\rangle|1\rangle\,, \\
&|0000\rangle+|1101\rangle+|1110\rangle=(|000\rangle+|111\rangle)|0\rangle+|110\rangle|1\rangle\,.
\end{align*}
As we see by changing the partition (changing the algorithm) the first state remains in the family span$\{000,0_{k}\Psi\}$ but the second one goes to the family span$\{000,\text{GHZ}\}$. We can also see this problem with two other states of the family span$\{0_{k}\Psi,0_{k}\Psi\}$ by considering the partition as $1|234$
\begin{align*}
&|0000\rangle+|1100\rangle+\lambda_{1}|0011\rangle+\lambda_{2}|1111\rangle = \\
&|0\rangle|0\rangle(|00\rangle+\lambda_{1}|11\rangle)+|1\rangle|1\rangle(|00\rangle+\lambda_{2}|11\rangle)\,,
\end{align*}
and
\begin{align*}
&|0000\rangle+|1100\rangle+\lambda_{1}|0001\rangle+\lambda_{1}|0010\rangle+\lambda_{2}|1101\rangle+\lambda_{2}|1110\rangle = \\
&|0\rangle|0\rangle(|00\rangle+\lambda_{1}|01\rangle+\lambda_{1}|10\rangle)+|1\rangle|1\rangle(|00\rangle+\lambda_{2}|01\rangle+\lambda_{2}|10\rangle)\,.
\end{align*}
But when we consider the partition as $123|4$ we have
\begin{align*}
&|0000\rangle+|1100\rangle+\lambda_{1}|0011\rangle+\lambda_{2}|1111\rangle = \\
&(|00\rangle+|11\rangle)|0\rangle|0\rangle+(\lambda_{1}|00\rangle+\lambda_{2}|11\rangle)|1\rangle|1\rangle\,,
\end{align*}
and
\begin{align*}
&|0000\rangle+|1100\rangle+\lambda_{1}|0001\rangle+\lambda_{1}|0010\rangle+\lambda_{2}|1101\rangle+\lambda_{2}|1110\rangle = \\
&(|000\rangle+|110\rangle+\lambda_{1}|001\rangle+\lambda_{2}|111\rangle)|0\rangle+(\lambda_{1}|00\rangle+\lambda_{2}|11\rangle)|0\rangle|1\rangle\,,
\end{align*}
where the first state remains in the family span$\{0_{k}\Psi,0_{k}\Psi\}$ and the second one goes to the family span$\{0_{k}\Psi,\text{GHZ}\}$. It is worth noting that since genuine entanglement families for three qubits are merely W and GHZ with the symmetric canonical forms, there is no such a problem. Clearly this is due to the fact that these two inseparable classes are invariant under exchanging the parties and hence under changing the partition. The situation is different in the systems of four or more qubits. Indeed, from the above examples one can easily find out that there are some genuine entanglement families for four qubits for which the canonical forms are  not symmetric and these are where violation to the classification of Ref. \cite{LLSS07} takes place. Evidently, symmetric genuine entanglement families for four or more qubits are invariant under changing the partition, e.g. the family span$\{000,000\}$ and the family span$\{000,\text{W}\}$ in four qubits systems. Generally, entanglement families W and GHZ are symmetric in $n\geq3$ qubits systems.

In summary, if by changing the partition all states from one family are mapped to another family there would be no problem\footnote{It is worth to mention that for four or more qubits, the number of SLOCC classes is infinite, indeed an uncountable infinite. Therefore, Cantor's theorem ensures that there are infinite ways to allocate them into a finite number of families. It means that there are infinite number of ways to ``classify entanglement" into families. However, most of them are meaningless classifications from the point of view of physics. Hence, even if all states from one family are mapped to another family we have a meaningless entanglement classification.}, but we have seen that with the approach of Refs. \cite{LLSS06,LLSS07} it may happen that only some states go to another family. Hence, there is an overlap between some families of four-qubit entangled states and thus this approach cannot be used to exactly identify to which family a given four-qubit state belongs to. Furthermore, being the approach inductive, the entanglement classification turns out to be flawed also for more than four qubits systems.

\chapter{Algebraic Geometry}\label{chap3}

\epigraph{Where are you?\\
In the boundless expanse of this world,\\
Where are you?\\
I am standing at the farthest end of the world.\\
By your side.}{\textit{Ahmad Shamloo, Meeting point}}

Projective space plays a central role in algebraic geometry. The aim of this chapter is to define the notion of multpartite entanglement in terms of abstract algebraic geometry.  So, within this chapter, we introduce some algebro-geometric tools that are needed to characterize multipartite entanglement. We start by the definition of affine space and affine variety and then we define projective space and projective variety. As examples, we introduce Veronese and Segre varieties. Then, we investigate their $k$-secant and $k$-tangent varieties. On the one hand any multipartite quantum state is indeed a tensor on the other hand the notion of $k$-secant variety is deeply connected to the notion of rank and border rank of tensors. Therefore, we also introduce tensor flattening and multilinear rank (multirank, for short). We will see that $k$-secant variety and $\ell$-multirank can be used as SLOCC invariants for the purpose of multipartite entanglement classification. In this chapter, we have used some materials and notions from Refs. \cite{Landsberg,Harris,Hartshorne,CLO15,Shafarevich,Holme}. The last section of this chapter is based on the Refs. \cite{GMO20,GM21}.

\section{Motivation}\label{sec.3.1}
One of the fundamental principles of quantum mechanics is that a quantum state that describes a quantum system corresponds to a vector in a Hilbert space $\mathcal{H}$, and that the Born rule gives the probability for a system in state $|\psi\rangle$ to be in state $|\varphi\rangle$ by
\begin{equation}
\text{Pr}(|\psi\rangle,|\varphi\rangle)=\frac{|\langle\psi|\varphi\rangle|^2}{|\langle\psi|\psi\rangle|\,|\langle\varphi|\varphi\rangle|}\,.
\end{equation}
Now, for any $\lambda\in\mathbbm{C}\setminus\{0\}$, one can see that $\text{Pr}(|\psi\rangle,|\varphi\rangle)=\text{Pr}(\lambda|\psi\rangle,|\varphi\rangle)=\text{Pr}(|\psi\rangle,\lambda|\varphi\rangle)=\text{Pr}(\lambda|\psi\rangle,\mu|\varphi\rangle)$. Therefore, $|\psi\rangle$ is said to be equivalent to the state $\lambda|\psi\rangle$, i.e.,
\begin{equation}\label{equivalence}
|\psi\rangle\sim\lambda|\psi\rangle \qquad \forall~\lambda\in\mathbbm{C}\setminus\{0\}\,,
\end{equation}
since these states will yield the same results when we use them in the Born rule. More precisely, the equivalence in Eq. \eqref{equivalence} comes from ${\rm{U}}(1)$ gauge symmetry in quantum theory. That is, if we consider normalized states $\langle\psi|\psi\rangle=1$, then the quantum states $|\psi\rangle$ and $\lambda|\psi\rangle$ represent the same physical state, for any $\lambda={\rm{e}}^{\mathbf{i}\delta}$ with $0\leq\delta<2\pi$ called the global phase (see \cref{U1-gauge}). It is known that one can not measure the global phase of the quantum states. So, one postulates that any calculation of a measurable quantity must be invariant under any change of phases, and therefore, the theory must be symmetric\footnote{The importance of such symmetries comes from Noether's theorem which states that such gauge symmetries lead to the conservation of a related quantity \cite{Noether}.} under such phase shifts.

Therefore, from a physics point of view and the fact that quantum mechanics is a ${\rm{U}}(1)$ gauge invariant theory, equivalent vectors in the Hilbert space represent the same pure quantum states. As a consequence, the proper state space of a quantum system is not the original Hilbert space $\mathcal{H}$ but rather the projective Hilbert space $\mathbbm{P}(\mathcal{H})$ where sets of equivalent states are its points. Thus, a natural way to study entanglement of pure states is with algebraic geometry, which is the ``language'' of projective spaces.

\section{Affine geometry}
In 1748, Leonhard Euler introduced the term ``affine'' in his book, ``Introductio in analysin infinitorum'' \cite{Euler}. After Felix Klein's Erlangen program, affine geometry was recognized as a generalization of Euclidean geometry \cite{Coxeter}. Actually, affine geometry is an incidence geometry that generalizes the Euclidean geometry when one forgets the notions of distance and angle. The only properties that are preserved are those related to parallelism and ratio of lengths for parallel line segments.

An affine space is nothing more than a vector space whose origin we try to forget about, by adding translations to the linear maps \cite{Berger}. In fact, the origin plays no special role in the affine space. Hence, any vector space is an affine space but the inverse is not true since in an affine space there is no distinctive point that serves as the origin. Therefore, we cannot associate a unique vector to a point in affine space. Instead, there are translation vectors between two points of the affine space. Hence, it makes sense to subtract two points of the space, giving a translation vector, but it does not make sense to add two points of the space. Subtraction of two points of the affine space can be seen in another way: adding a translation vector to a point of an affine space results in a new point translated from the starting point by that vector. See, for instance, Fig. \ref{fig:Affine-R3}.

\begin{figure}[t]
\center{\includegraphics[width=6cm]{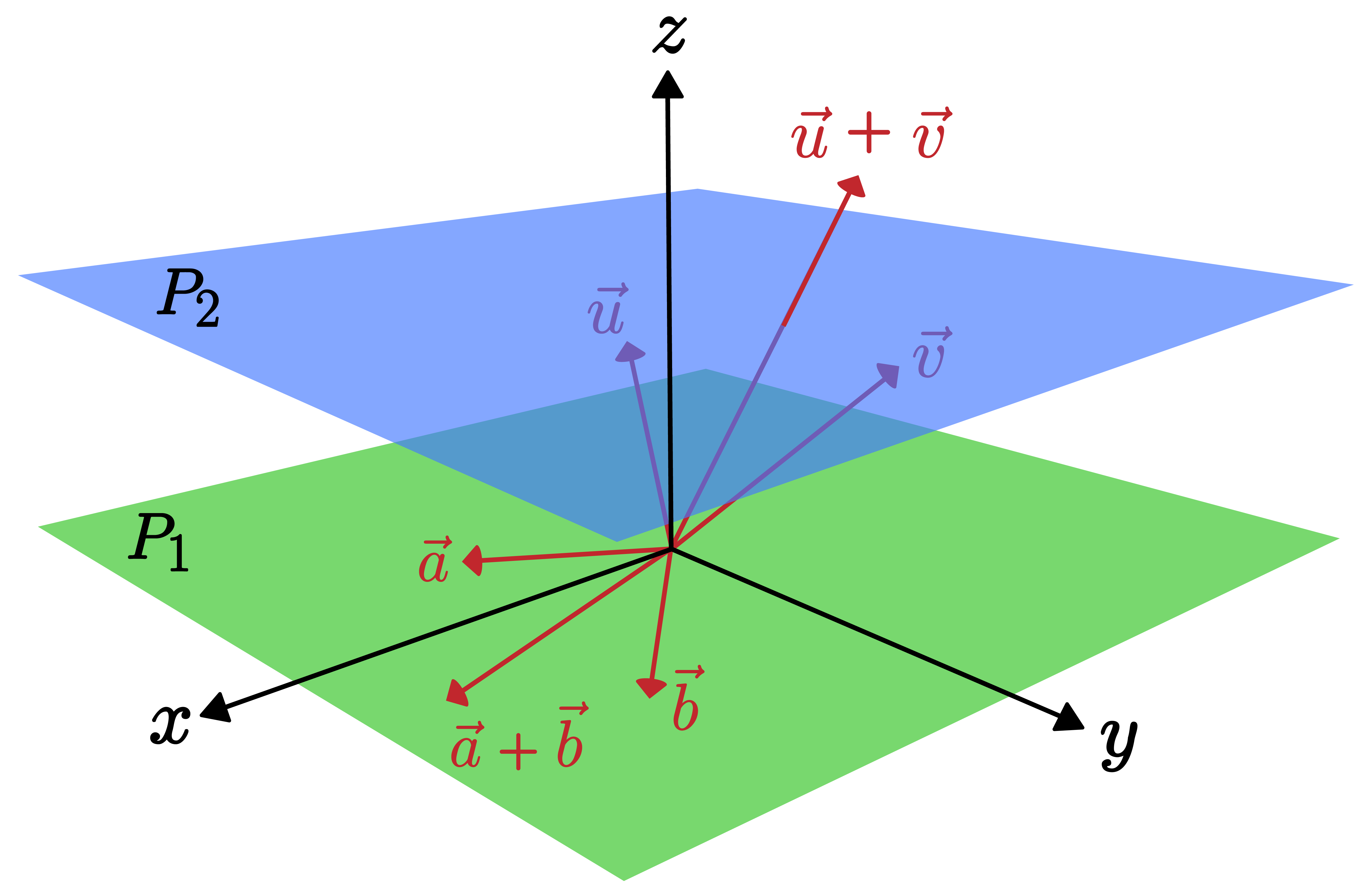}}
\caption[Affine space]{\label{fig:Affine-R3} Inside the vector space $\mathbbm{R}^3$ we can choose two planes $P_1$ (the green one) and $P_2$ (the blue one). While the plane $P_1$ that passes through the origin is a vector subspace ($P_1\cong\mathbbm{R}^2\subset\mathbbm{R}^3$), the plane $P_2$ is not a vector subspace. This is because $0\in P_1$ and $\vec{a}+\vec{b}\in P_1$ for any vectors $\vec{a},\vec{b}\in P_1$ but $0\notin P_2$ and $\vec{u}+\vec{v}\notin P_1$. The plane $P_2$ is an affine subspace.}
\end{figure}

\begin{definition}[Affine space]
Given a field $K$ and a positive integer $d$, we define the $d$-dimensional affine space over $K$, denoted by $\mathbbm{A}_K^d$, to be the set $K^d$ of all $d$-tuples of elements from $K$, i.e.,
\begin{equation}
\mathbbm{A}_K^d=\{(a_0,\ldots,a_{d-1})\mid a_j\in K~\forall~j\in\mathbbm{Z}_d\}\,.
\end{equation}
The elements of an affine space are called points.
\end{definition}

It worth remarking that within this chapter, $K$ is a commutative field.

\begin{remark}
The dimension of an affine space is defined as the dimension of its associated vector space that contains translations.
\end{remark}

Specifically, $\mathbbm{A}_K^1$ and $\mathbbm{A}_K^2$ are called affine line and affine plane, respectively. A $(d-1)$-dimensional affine subspace $\mathbbm{A}_K^{d-1}$ in an affine space (or a vector space) of dimension $d$ is called an affine hyperplane.

\subsection{Affine variety}
To link algebra and geometry, we need to relate polynomials over a field $K$ to an affine space over the same field.

\begin{definition}[Monomial]
A $d$-variate monomial is a product of the form
\begin{equation}
\mathbf{x}^{\alpha}=\prod_{k\in\mathbbm{Z}_d}x_k^{\alpha_k}\,,
\end{equation}
where $\alpha=(\alpha_0,\ldots,\alpha_{d-1})$ is a $d$-tuple of nonnegative integers. The total degree of this monomial is the sum $\alpha_0+\cdots+\alpha_{d-1}$.
\end{definition}

\begin{definition}[Polynomial]
A $d$-variate polynomial $f$ is a ﬁnite linear combination of $d$-variate monomials with coefficients in a field $K$. That is
\begin{equation}\label{Poly}
f(x_0,\ldots,x_{d-1})=\sum_{\alpha}a_{\alpha}\mathbf{x}^{\alpha}\,,
\end{equation}
where $a_{\alpha}\in K$ and the sum is over a finite number of $d$-tuples $\alpha=(\alpha_0,\ldots,\alpha_{d-1})$.
\end{definition}

\begin{definition}[Polynomial ring]
The set of all $d$-variate polynomials $f$ with coefﬁcients in a field $K$ is denoted by $K[x_{0},\ldots,x_{d-1}]$ which refers to a polynomial ring.
\end{definition}

Actually, under addition and multiplication, $K[x_{0},\ldots,x_{d-1}]$ satisfies all of the field axioms except for the existence of multiplicative inverses because, for example, $x^{-1}$ is not a polynomial. Such a mathematical structure is called a commutative ring.

The key idea of seeing how polynomials relate to an affine space is that a polynomial $f\in K[x_0,\ldots,x_{d-1}]$ corresponds to a multilinear form
\begin{equation}
f\colon\mathbbm{A}_K^d\to K\,.
\end{equation}
It is obvious that if we are given by a $(a_0,\ldots,a_{d-1})\in\mathbbm{A}_K^d$, then we can replace every $x_i$ in the expression of $f$ by $a_i$. Since the coefficients of $f$ are also in the field $K$, this function gives an element of the field $K$, i.e., $f(a_0,\ldots,a_{d-1})\in K$.

\begin{definition}[Algebraically closed field]
A field $K$ is algebraically closed if every nonconstant polynomial in $K[x_{0}, \ldots,x_{d-1}]$ has a root in $K$.
\end{definition}
Therefore, $\mathbbm{R}$ is not algebraically closed, whereas $\mathbbm{C}$ is\footnote{It is proved as the fundamental theorem of algebra (see Ref. \cite{Sheldon}).}.

Roughly speaking, every variety can be defined as the locus of vanishing a finite number of polynomials.

\begin{definition}[Affine variety]
Let $F=\{f_1,\ldots,f_s\}\subseteq K[x_0,\ldots,x_{d-1}]$ be a set of polynomials. The zero locus of the set of polynomials $F$, i.e.,
\begin{equation}
\mathfrak{V}(F)=\{(a_0,\ldots,a_{d-1})\in\mathbbm{A}_K^d\mid f_i(a_0,\ldots,a_{d-1})=0~\forall~f_i\in F\}\,.
\end{equation}
is called affine variety.
\end{definition}

In other words, an affine variety $\mathfrak{V}(F)\subseteq\mathbbm{A}_K^d$ is the the set of all solutions of the system of equations $f_1(x_0,\ldots,x_{d-1})=\cdots=f_s(x_0,\ldots,x_{d-1})=0$.

\begin{example}
The affine variety $\mathfrak{V}(f)$ with
\[
f(x,y,z)=3 (x^2+y^2+z^2)-2(x^2y^2+x^2z^2+y^2 z^2)-3\,,
\]
is shown in Fig. \ref{fig:3.2}.
\begin{figure}[t]
\center{\includegraphics[width=4cm]{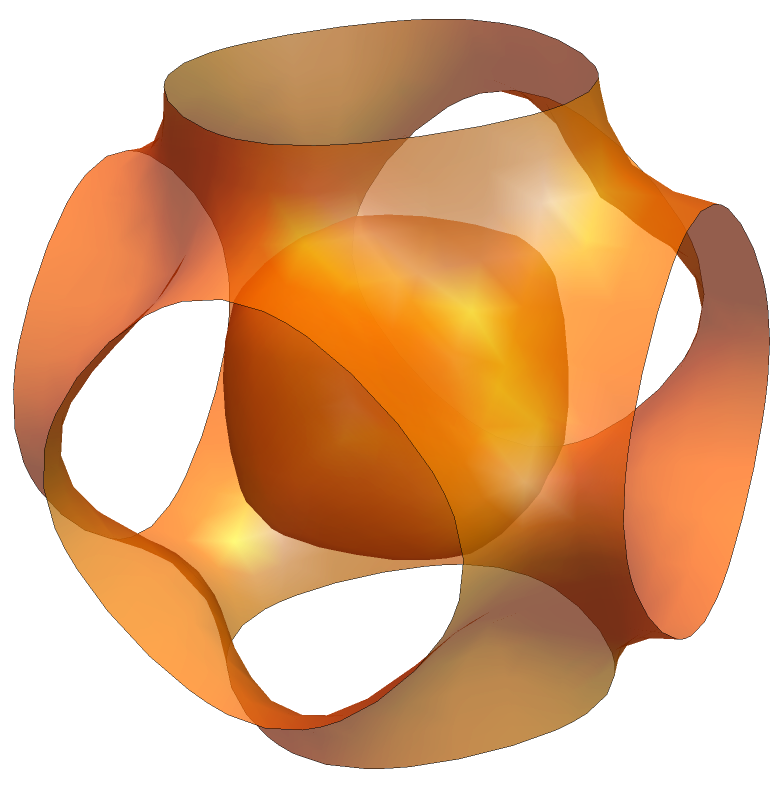}}
\caption[Affine variety]{\label{fig:3.2} 
The affine variety $\mathfrak{V}(3 (x^2+y^2+z^2)-2(x^2y^2+x^2z^2+y^2 z^2)-3)$.}
\end{figure}
\end{example}

\begin{definition}[Orbit]
Let $G\subseteq{\rm{GL}}(d,K)$ be a group and $\mathfrak{X}\subseteq\mathbbm{A}_K^d$ be an affine variety. The $G$-orbit of a point $p=(a_0,\ldots,a_{d-1})\in\mathfrak{X}$ is the set
\begin{equation}
[p]_G:=G\cdot p=\{g\cdot p\mid g\in G\}\,.
\end{equation}
That is the full set of points that $p$ is sent to under the action of  group $G$.
\end{definition}
\begin{definition}[Orbit space]
The orbit space, is the set of all orbits of $\mathfrak{X}\subseteq\mathbbm{A}_K^d$ under the action of a group $G$ and is denoted by
\begin{equation}
\frac{\mathfrak{X}}{G}\,.
\end{equation}
It is also called the quotient of the action.
\end{definition}

Since orbits are equivalence classes, therefore the orbit space $\mathbbm{A}_K^d/G$ is the set of equivalence classes of $\sim_G$. Moreover, the orbits are disjoint since $\sim_G$ is an equivalence relation.

\begin{remark}
Given an action $(G,\mathfrak{X})$, the problem of finding the parameterization of $\mathfrak{X}/G$ and the orbits is known as the classification problem.
\end{remark}

\begin{theorem}
The orbit space $\mathfrak{X}/G$ has the structure of an afﬁne variety.
\end{theorem}

\begin{lemma}\label{variety-cup-cap}
If $\mathfrak{X},\mathfrak{Y}\subseteq\mathbbm{A}_K^d$ are two affine varieties, then $\mathfrak{X}\cup\mathfrak{Y}$ and $\mathfrak{X}\cap\mathfrak{Y}$ are also affine varieties.
\end{lemma}

In \Cref{variety-cup-cap}, suppose that $\mathfrak{X}=\mathfrak{X}(f_1,\ldots,f_s)$ and $\mathfrak{Y}=\mathfrak{Y}(g_1,\ldots,g_t)$. Then,
\begin{equation}
\mathfrak{X}\cup\mathfrak{Y}=\mathfrak{V}(f_ig_j\mid 1\leq i\leq s, 1\leq j\leq t)\,,
\end{equation}
and
\begin{equation}
\mathfrak{X}\cap\mathfrak{Y}=\mathfrak{V}(f_1,\ldots,f_s,g_1,\ldots,g_t)\,,
\end{equation}
are affine varieties.

\begin{definition}[Irreducible variety]
An affine variety $\mathfrak{V}\subset\mathbbm{A}_K^d$ is called reducible if it can be written as a non-trivial union of two subvarieties $\mathfrak{V}=\mathfrak{V}_1\cup\mathfrak{V}_2$ where $\mathfrak{V}_1\neq\emptyset$ and $\mathfrak{V}_2\neq\emptyset$. Otherwise it is called irreducible.
\end{definition}

\begin{definition}[Ideal]
A subset $\mathbf{I}\subseteq K[x_{0},\ldots,x_{d-1}]$ is called an ideal if it satisfies the following:
\begin{enumerate}
\item[(i)] $0\in\mathbf{I}$.
\item[(ii)] If $f_1,f_2\in\mathbf{I}$, then $f_1+f_2\in\mathbf{I}$.
\item[(iii)] If $f\in\mathbf{I}$ and $g\in K[x_{0},\ldots,x_{d-1}]$, then $fg\in\mathbf{I}$.
\end{enumerate}
\end{definition}

\begin{definition}[Ideal of a variety]
Let $\mathfrak{V}\subseteq\mathbbm{A}_K^d$ be an afﬁne variety. We define the ideal of this variety as follows
\begin{equation}
\mathbf{I}(\mathfrak{V})=\{f\in K[x_{0},\ldots,x_{d-1}]\mid f(a_0,\ldots,a_{d-1})=0~\forall~(a_0,\ldots,a_{d-1})\in\mathfrak{V}\}\,.
\end{equation}
\end{definition}

\begin{remark}
The ideal of an affine variety $\mathfrak{V}\subseteq\mathbbm{A}_K^d$, $\mathbf{I}(\mathfrak{V})$, is an ideal.
\end{remark}

\begin{definition}[Subset closure]
The algebro-geometric closure of a subset of an affine space $S\subset\mathbbm{A}_K^d$ is defined as follows
\begin{equation}
\bar{S}=\{(a_0,\ldots,a_{d-1})\in\mathbbm{A}_K^d\mid f(a_0,\ldots,a_{d-1})=0~\forall~f\in\mathbf{I}(S)\}=\mathfrak{V}(\mathbf{I}(S))\,.
\end{equation}
A subset $S\subset\mathbbm{A}_K^d$ is closed if $S=\bar{S}$; $T \subset\mathbbm{A}_K^d$ is open if its complement $\mathbbm{A}_K^d\setminus T$ is closed in $\mathbbm{A}_K^d$.
\end{definition}

\begin{lemma}\label{lem:Zariski}
(i) A finite union of affine varieties $\cup_{i=1}^{N}\mathfrak{V}_i$ is an affine variety. \\ (ii) An arbitrary intersection of affine varieties $\cap_j\mathfrak{V}_j$ is an affine variety.
\end{lemma}

The \Cref{lem:Zariski} shows that closed subsets of an afﬁne space (in the sense of algebraic geometry) satisfy the axioms of a topological space.

\begin{definition}[Zariski topology on an affine space]
The Zariski topology on an affine space $\mathbbm{A}_K^d$ is the topology whose closed subsets are the affine varieties $\mathfrak{V}(\mathbf{I})$ for $\mathbf{I}\subseteq K[x_0,\ldots,x_{d-1}]$.
\end{definition}

\begin{definition}[Zariski topology on an affine variety]
The Zariski topology on an affine variety $\mathfrak{X}\subseteq\mathbbm{A}_K^d$ is the restricted topology, that is, the closed subsets of $\mathfrak{X}$ are subvarieties of $\mathfrak{X}$, i.e., $\{\mathfrak{X}\cap\mathfrak{V}(\mathbf{I})\mid\mathbf{I}\subseteq K[x_0,\ldots,x_{d-1}]\}$.
\end{definition}

\section{Projective geometry}
Projective geometry has its origins in the early Italian Renaissance, particularly in the architectural drawings of Filippo Brunelleschi (1377–1446) and Leon Battista Alberti (1404–72), who invented the method of perspective drawing \cite{Artmann}. 

Projective geometry is less restrictive than affine geometry. A projective space can be seen as an extension of a Euclidean space, that is an affine space with points at infinity, in such a way that there is one point at infinity of each direction of parallel hyperplanes, and so two hyperplanes always intersect. Actually, in projective space there is a hyperplane at infinity which is the locus of all points of intersections of parallel hyperplanes.

Consider the real projective plane $\mathbbm{R}\mathbbm{P}^2$ which is an example of a compact non-orientable two-dimensional manifold, i.e., it is a one-sided surface. M\"obius strip can be constructed from a square by gluing two of its sides together with a half-twist. The real projective plane can thus be constructed from a M\"obius strip by gluing opposite open edges of the strip together.

We define an equivalence relation $\sim$ on non-zero points of an affine space $\mathbbm{A}_K^d$ by setting
\begin{equation}\label{equiv-1}
(a_0,\ldots,a_{d-1})\sim (b_0,\ldots,b_{d-1})\,,
\end{equation}
if there is a non-zero element $\lambda\in K$ such that
\begin{equation}\label{equiv-2}
(a_0,\ldots,a_{d-1})=\lambda(b_0,\ldots,b_{d-1})\,.
\end{equation}

\begin{definition}[Projective space]
A $d$-dimensional projective space over the field $K$, denoted by $\mathbbm{P}_K^{d-1}$, is the set of equivalence classes of $\sim$ on $\mathbbm{A}_K^{d}$. That is
\begin{equation}
\mathbbm{P}_K^{d-1}=\frac{\mathbbm{A}_K^d}{\sim}\,.
\end{equation}
\end{definition}

\begin{definition}[Ray]
Equivalence classes of $\sim$ on $\mathbbm{A}_K^{d}$ are called rays (or fibers). So,
\[
\mathbbm{P}_K^{d-1}\cong\{\text{rays through the origin in}~K^d\}\,.
\]
\end{definition}

Given a point $p=(a_0,\ldots,a_{d-1})\in\mathbbm{A}_K^{d}$, its equivalence class $[p]\in\mathbbm{P}_K^{d-1}$ will be denoted by $[a_0:\cdots:a_{d-1}]$, and we will say that $[a_0:\cdots:a_{d-1}]$ are homogeneous coordinates of point $p$. Thus, for $\lambda\in K\setminus\{0\}$ we have
\begin{equation}
[a_0:\cdots:a_{d-1}]=[b_0:\cdots:b_{d-1}]~\Leftrightarrow~(a_0,\ldots,a_{d-1})=\lambda(b_0,\ldots,b_{d-1})\,.
\end{equation}

Whenever $a_0\neq0$, we may assume that $a_0=1$, and with this assumption the other coordinates are uniquely determined. Thus, we may identify the subset
\begin{equation}
\{[a_0:a_1\cdots:a_{d-1}]\mid a_0\neq0\}\subset\mathbbm{P}_K^{d-1}\,,
\end{equation}
with $\mathbbm{A}_K^{d-1}$ by letting $[a_0:a_1:\cdots:a_{d-1}]$ correspond to $(\frac{a_1}{a_0},\ldots,\frac{a_{d-1}}{a_0})$. The set
\begin{equation}
\{[a_0:a_1:\cdots:a_{d-1}]\mid a_0=0\}\subset\mathbbm{P}_K^{d-1}\,,
\end{equation}
corresponds to the points at infinity. This subset may in turn be identiﬁed with $\mathbbm{P}_K^{d-2}$ in an obvious way by forgetting the first coordinate $a_0$, which is zero. We obtain the following description of $\mathbbm{P}_K^{d-1}$:
\begin{equation}\label{Affine+Proj}
\mathbbm{P}_K^{d-1}=\mathbbm{A}_K^{d-1}\cup\mathbbm{P}_K^{d-2}\,,
\end{equation}
and we say that the projective space $\mathbbm{P}_K^{d-1}$ is obtained by adjoining to an affine space $\mathbbm{A}_K^{d-1}$ a projective space space $\mathbbm{P}_K^{d-2}$ of points at infinity.

\begin{corollary}
In projective geometry, affine space means the complement of a hyperplane at infinity in a projective space.
\end{corollary}

Regarding Eq. \eqref{Affine+Proj}, One can divide $\mathbbm{P}_K^{d-2}$ in a similar way, and repeating the process all the way down to $\mathbbm{P}_K^0$, that is,
\begin{equation}
\mathbbm{P}_K^{d-1}=\mathbbm{A}_K^{d-1}\cup\mathbbm{A}_K^{d-2}\cup\cdots\cup\mathbbm{A}_K^1\cup\mathbbm{P}_K^0\,,
\end{equation}
where $\mathbbm{P}_K^0$ is just only one point at infinity, i.e., $\mathbbm{P}_K^0=\{\infty\}$.



\subsection{Projective Hilbert space}

\begin{definition}[Ray]
The equivalence class of a vector $|\psi\rangle\in\mathcal{H}$ for the equivalence relation $\sim$
\begin{equation}
|\psi\rangle\sim|\varphi\rangle\qquad\text{iff}\quad\exists~\lambda\in\mathbbm{C}\setminus\{0\}\quad\text{s.t.}\quad|\psi\rangle=\lambda|\varphi\rangle\,,
\end{equation}
is called a ray. For $|\psi\rangle$, the associated ray is the set
\begin{equation}
[\psi]=\big\{|\varphi\rangle\in\mathcal{H}~\big|~|\varphi\rangle=\lambda|\psi\rangle~\forall~\lambda\in\mathbbm{C}\setminus\{0\}\big\}\,.
\end{equation}
\end{definition}

Therefore, it is more convenient to work on the projective Hilbert space by considering each ray as a point of it.

\begin{definition}[Projective Hilbert space]
Projective Hilbert space is the set of all rays of non-zero vectors in the Hilbert space which is induced by equivalence relation $\sim$. In other words, it is the quotient space 
\begin{equation}
\mathbbm{P}(\mathcal{H}):=\frac{\mathcal{H}\setminus\{0\}}{\sim}\,.
\end{equation}
\end{definition}

Points in the projective Hilbert space $\mathbbm{P}(\mathcal{H})$ are one-dimensional rays in the Hilbert space $\mathcal{H}$ or equivalently one-dimensional projectors:
\begin{equation}
|\psi\rangle\to P_{\psi}:=\frac{|\psi\rangle\langle\psi|}{\langle\psi|\psi\rangle}\,.
\end{equation}

In this dissertation, we shall consider only finite dimensional complex Hilbert spaces, i.e., we assume $\mathcal{H}=\mathbbm{C}^d$. Hence, normalized quantum pure states $|\psi\rangle\in\mathbbm{C}^d$ can be considered as points on a unit $(2d-1)$-dimensional sphere ($(2d-1)$-sphere, for short)
\begin{equation}
S^{2d-1}=\{|\psi\rangle\in\mathbbm{C}^d\mid\langle\psi|\psi\rangle=1\}\,.
\end{equation}

Now, two normalized quantum pure states $|\psi\rangle$ and $|\varphi\rangle$ as points on the unit $(2d-1)$-sphere $S^{2d-1}$ are equivalent iff $|\varphi\rangle={\rm{e}}^{\mathbf{i}\delta}|\psi\rangle$. Therefore, the
corresponding projective Hilbert space can be considered as a quotient space of the unit $(2d-1)$-sphere in $\mathbbm{C}^d$ under the action of ${\rm{U}}(1)$, i.e.,
\begin{equation}
\mathbbm{P}(\mathcal{H})=\mathbbm{C}\mathbbm{P}^{d-1}:=\frac{S^{2d-1}}{{\rm{U}}(1)}\cong\frac{S^{2d-1}}{S^1}\,.
\end{equation}
the last equality can be considered by the topology of the group ${\rm{U}}(1)$, that is ${\rm{U}}(1)\cong S^1$.

Concerning Eq. \eqref{Affine+Proj}, we also can write the following equality
\begin{equation}
\mathbbm{C}\mathbbm{P}^{d-1}=\mathbbm{C}^{d-1}\cup\mathbbm{C}\mathbbm{P}^{d-2}\,.
\end{equation}
Specifically, $\mathbbm{C}\mathbbm{P}^1=\mathbbm{C}\cup\{\infty\}$.

On the other hand, the transitive group of $\mathbbm{C}\mathbbm{P}^{d-1}$ is ${\rm{SU}}(d)$, that is, for every $p,q\in\mathbbm{C}\mathbbm{P}^{d-1}$ there exists a $g\in{\rm{SU}}(d)$ such that
\begin{align}\nonumber
&\varpi_g\colon{\rm{SU}}(d)\times\mathbbm{C}\mathbbm{P}^{d-1}\to\mathbbm{C}\mathbbm{P}^{d-1}\,, \\
&p\to\varpi_g(p)=q\,.
\end{align}
The isotropy group (or stabilizer) of $\mathbbm{C}\mathbbm{P}^{d-1}$ is the group ${\rm{U}}(d-1)$, hence
\begin{equation}
\mathbbm{C}\mathbbm{P}^{d-1}\cong\frac{{\rm{SU}}(d)}{{\rm{U}}(d-1)}\,.
\end{equation}

\subsubsection{Bloch sphere}
\begin{figure}[t]
\center{\includegraphics[width=4.5cm]{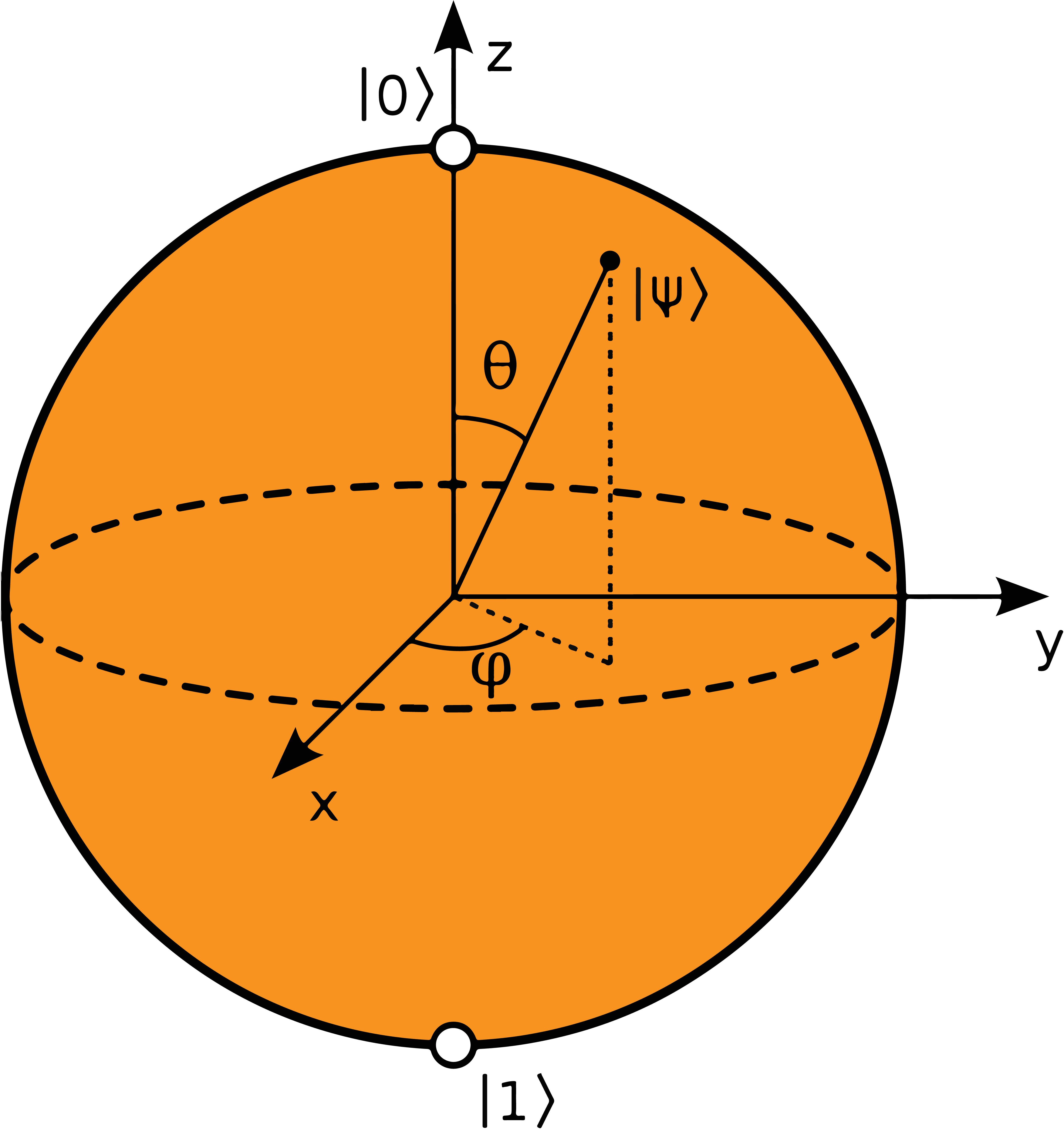}}
\caption[Bloch sphere]{\label{fig:3.3} Bloch sphere: the projective Hilbert space of a qubit $\mathbbm{P}(\mathbbm{C}^2)=\mathbbm{C}\mathbbm{P}^1$. Each point on this sphere represents a class of equivalent one-qubit pure states (see Eq. \eqref{qubit-polar}).}
\end{figure}
A well known example of the projective Hilbert space is $\mathbbm{C}\mathbbm{P}^1$ (also known as the complex projective line or Riemann sphere) which is the state space of one-qubit states.
\begin{equation}
\mathbbm{C}\mathbbm{P}^1=\frac{S^{3}}{{\rm{U}}(1)}\cong S^2\,.
\end{equation}
This two-dimensional sphere is known as the Bloch sphere in quantum mechanics (see Fig. \ref{fig:3.3}). Regarding Eq. \eqref{qubit-mixed}, the $\vec{r}$ is called the Bloch vector.  Therefore, any point on this sphere, i.e., $|\vec{r}|=1$, corresponds to a class of equivalent qubit pure states\footnote{Elements of the projective Hilbert space are not states but rather a class of equivalent states. However, since the members of the equivalent class represent the same quantum states (by the ${\rm{U}}(1)$ gauge invariance), we can use the term state for the elements of the projective Hilbert space. Indeed, we usually use a minimal form of a state from an Hilbert space as a representative of a point of the projective Hilbert space.}. If we consider a Bloch ball corresponding to $|\vec{r}|<1$, then any point in the Bloch ball represents a class of equivalent qubit mixed states and $|\vec{r}|$ is called the Bloch ball radius.

\subsection{Projective variety}
Here, our goal is to extend the definition of the variety in the affine space to the projective space. For this purpose, we will see that some care must be taken. For instance, in $\mathbbm{R}\mathbbm{P}^2$, we can try to construct the variety $\mathfrak{V}(x^2-y-z)$. This variety is the zero locus of the polynomial $x^2-y-z$. Moreover, in projective space, we should consider the equivalency of points (see Eqs. \eqref{equiv-1} and \eqref{equiv-2}). We can see that the point $p=[x:y:z]=[2:2:2]$ appears to be in this set since the components of $p$ satisfy the equation $x^2-y-z=0$. Although we know that the same point $p$ can be represented by the homogeneous coordinates $p=[4:4:4]$ but if we substitute these components into the polynomial $x^2-y-z$, we obtain $16-4-4=8\neq0$. Therefore, we get different results depending on which homogeneous coordinates we choose. To avoid this problem, we use homogeneous polynomials in the projective space.

\begin{definition}[Homogeneous polynomial]\label{def:Hom-Poly}
A homogeneous $d$-variate \mbox{degree-$n$} polynomial $f\in K[x_0,\ldots,x_{d-1}]_n$ is a finite linear combination of $d$-variate degree-$n$ monomials with coefficients in the field $K$. That is
\begin{equation}
f(x_0,\ldots,x_{d-1})=\sum_{\alpha}a_{\alpha}\mathbf{x}^{\alpha}\,, \qquad \text{with} \qquad \sum_{k\in\mathbbm{Z}_d}\alpha_k=n\,,
\end{equation}
where $\alpha=(\alpha_0,\ldots,\alpha_{d-1})$ and $a_{\alpha}\in K$.
\end{definition}

\begin{lemma}
Let $f\in K[x_0,\ldots,x_{d-1}]$ be a homogeneous polynomial. If $f$ vanishes on one of the homogeneous coordinates representing a point $p\in\mathbbm{P}_K^{d-1}$, then $f$ vanishes for all homogeneous coordinates of $p$.
\end{lemma}

Therefore, $\mathfrak{V}(f)=\{p\in\mathbbm{P}_K^{d-1}\mid f(p)=0\}$ for a homogeneous polynomial $f$, is a well-defined subset of the projective space $\mathbbm{P}_K^{d-1}$. So we can define the projective variety as follows.

\begin{definition}[Projective variety]
The zero locus of a set of homogeneous polynomials $F=\{f_1,\ldots,f_s\}\subset K[x_0,\ldots,x_{d-1}]$ is called a projective variety. That is
\begin{equation}
\mathfrak{V}(F)=\{[a_0:\cdots:a_{d-1}]\in\mathbbm{P}_K^{d-1}\mid f_i(a_0,\ldots,a_{d-1})=0~\forall~f_i\in F\}\,.
\end{equation}
\end{definition}


\begin{definition}[Projective equivalence]
We say $\mathfrak{V}\subseteq\mathbbm{P}_K^{d-1}$ and  $\mathfrak{W}\subseteq\mathbbm{P}_K^{d-1}$ are projectively equivalent if they are transformed into each other by a linear change of coordinates in $\mathbbm{P}^{d-1}$. It means that there exists a linear map $A\in{\rm{GL}}(d,K)$ which establishes an isomorphism
\begin{align}\nonumber
A:\mathfrak{V}&\to\mathfrak{W} \\
\vec{a}&\to A\vec{a}\,,
\end{align}
Where $\vec{a}=(a_0,\ldots,a_{d-1})$. The inverse of the linear map is given then by the inverse matrix $A^{-1}$.
\end{definition}

Note that in this situation $A$ and $\lambda A$ define the same transformation in the projective space $\mathbbm{P}_K^{d-1}$ for $\lambda\in K\setminus\{0\}$, and in fact the group which acts on the projective space is
\begin{equation}
{\rm{PGL}}(d,K)=\mathbbm{P}\big({\rm{GL}}(d,K)\big)=\frac{{\rm{GL}}(d,K)}{K\setminus\{0\}}\,.
\end{equation}

\begin{definition}[Closed subset]
A subset of a projective space $S\subset\mathbbm{P}_K^{d-1}$ is called closed subset if it consists of all points at which a finite number of homogeneous polynomials vanish.
\end{definition}

It is worth noting that the set of all homogeneous polynomials in $K[x_0,\ldots,x_{d-1}]$ that vanish at all points $p\in S$ forms an ideal $\mathbf{I}(S)$, called the ideal of the closed set $S$. Therefore, the closure of a subset of a projective space is the variety of the ideal of that subset, i.e.,
\begin{equation}
\bar{S}=\{[a_0:\cdots:a_{d-1}]\in\mathbbm{P}_K^{d-1}\mid f(a_0,\ldots,a_{d-1})=0~\forall~f\in\mathbf{I}(S)\}=\mathfrak{V}(\mathbf{I}(S))\,.
\end{equation}

\begin{remark}
A closed subset defined by one homogeneous equation $f=0$ is called a hypersurface, as in the affine case. The degree of the polynomial is the degree of the hypersurface. A hypersurface of degree two is called a quadric.
\end{remark}

\begin{definition}[Zariski topology on a projective space]
The Zariski topology on a projective space $\mathbbm{P}_K^{d-1}$ is the topology whose closed subsets are the projective varieties. The Zariski topology on a projective variety $\mathfrak{V}\subseteq\mathbbm{P}_K^{d-1}$ is the induced topology.
\end{definition}

\section{Morphisms}\label{sec.3.4}
In algebraic geometry, a morphism between algebraic varieties is made of local functions between the algebraic varieties that is given by a polynomial map.

\begin{definition}[Morphism]
A morphism of aﬃne spaces
\begin{equation}
F\colon\mathbbm{A}_K^{d_1}\to\mathbbm{A}_K^{d_2}\,,
\end{equation}
is given by a polynomial map
\begin{equation}
(a_0,\ldots,a_{d_1-1})\mapsto(f_0(\mathbf{a}),\ldots,f_{d_2-1}(\mathbf{a}))\,,
\end{equation}
where $\mathbf{a}=(a_0,\ldots,a_{d_1-1})$, $F=(f_0,\ldots,f_{d_2-1})$, and $f_i\in K[x_0,\ldots, x_{d_1-1}]$ for each $0\leq i \leq d_2-1$.
\end{definition}

A morphism $F\colon\mathfrak{V}\to\mathfrak{W}$ between aﬃne varieties $\mathfrak{V}\subseteq\mathbbm{A}_K^{d_1}$ and $\mathfrak{W}\subseteq\mathbbm{A}_K^{d_2}$ is given by a polynomial map $F\colon\mathbbm{A}_K^{d_1}\to\mathbbm{A}_K^{d_2}$ that restricts to $\mathfrak{V}$, such that $F(\mathfrak{V})\subseteq\mathfrak{W}$.

It is worth noting that an isomorphism is a morphism which has an inverse morphism.

\begin{remark}
Generally speaking, the image of a morphism need not be an aﬃne variety.
\end{remark}

\begin{remark}
A morphism of projective spaces is the same as the morphism of affine space where the local functions are homogeneous polynomials.
\end{remark}

\subsection{Veronese embedding}\label{subsec3.4.1}

\begin{definition}[Veronese embedding]
The degree $n$ Veronese embedding
\begin{equation}\label{Veronese-AG}
\mathcal{V}^{n}_{d-1}\colon\mathbbm{P}_K^{d-1}\hookrightarrow\mathbbm{P}_K^{m}\,,
\end{equation}
is an injective morphism defined by
\begin{equation}\label{Veronese-functions}
[a_0:\cdots:a_{d-1}]\to[f_0(\mathbf{a}):\cdots:f_m(\mathbf{a})]\,,
\end{equation}
where $\mathbf{a}=(a_0,\ldots,a_{d-1})$ and $\{f_i\}_{i=0}^m$ are all of the $d$-variate degree-$n$ monomials. That is, each $f_i$ has the form
\begin{equation}\label{homo-monomial}
\prod_{k\in\mathbbm{Z}_d}x_k^{\alpha_k}\,, \qquad \text{with} \qquad \sum_{k\in\mathbbm{Z}_d}\alpha_k=n\,.
\end{equation}
For $0\leq\alpha_k\leq n$ we have $\binom{n+d-1}{n}$ different forms. Therefore, $m=\binom{n+d-1}{n}-1$. The image of this morphism is called Veronese variety.
\end{definition}

\begin{definition}[Catalecticant matrix]
Let $[z_0:\cdots:z_m]$ denote the homogeneous coordinates of the local functions in Eq. \eqref{Veronese-functions}. The $i$-th catalecticant matrix of a point $p\in\mathbbm{P}^m$ is defined as
\begin{equation}\label{catalecticant}
\mathrm{C}_i(p)=\begin{pmatrix}
z_0 & z_1 & \cdots & z_i \\
z_1 & z_2 & \cdots & z_{i+1} \\
\vdots & \vdots & \ddots & \vdots \\
z_{m-i} & z_{m-i+1} & \cdots & z_{m}
\end{pmatrix},
\end{equation}
for $1\leq i\leq m-1$.
\end{definition}

\begin{remark}
The Veronese embedding can also be represented as follows
\begin{align}\nonumber
\mathcal{V}_{d-1}^n:\mathbbm{P}\big(K^d\big)&\hookrightarrow\mathbbm{P}\big({\rm{Sym}}^n K^d\big) \\
[v]&\to[v^n]\,,
\end{align}
where ${\rm{Sym}}^n K^d$, the $n$-th symmetric power of vector space $K^d$, is the quotient of $n$-fold product $(K^d)^{\times n}\colon K^d\times\cdots\times K^d$ by the permutation action of the symmetric group $\mathfrak{S}_n$, i.e.,
\begin{equation}
\frac{(K^d)^{\times n}}{\mathfrak{S}_n}\,.
\end{equation}
\end{remark}

\subsubsection{Rational normal curve}
In Eq. \eqref{Veronese-AG}, for $d-1=1$, i.e.,
\begin{equation}
\mathcal{V}_1^n\colon\mathbbm{P}_K^1\hookrightarrow\mathbbm{P}_K^n\,.
\end{equation}
the Veronese variety is known as the rational normal curve. For $n=1$ the Veronese map is the identity map on the projective line. For $n=2$ the Veronese variety is the standard parabola, $[a^2:ab:b^2]$, in affine coordinates $[1:z:z^2]$. For $n=3$ the Veronese variety is the twisted cubic, $[a^3:a^2b:ab^2:b^3]$, in affine coordinates $[1:z:z^2:z^3]$ (see Fig. \ref{fig:Twisted-Cube}).
\begin{figure}[t]
\center{\includegraphics[width=4.5cm]{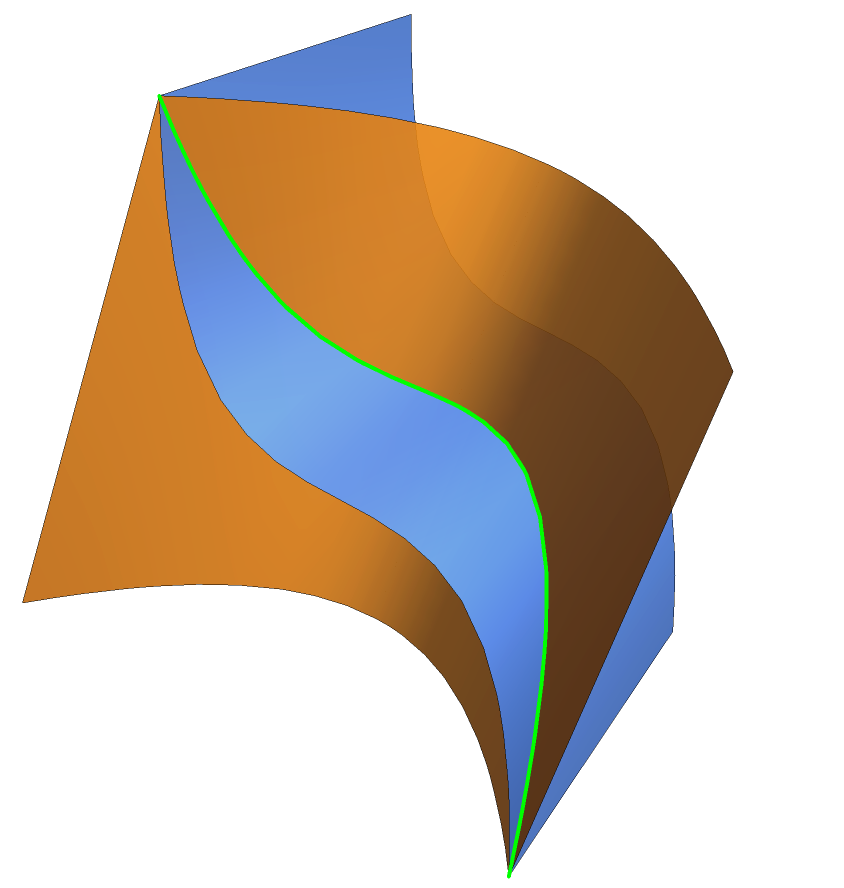}}
\caption[Twisted cube]{\label{fig:Twisted-Cube} 
Twisted cube: the Veronese variety of $\mathbbm{R}\mathbbm{P}^1\hookrightarrow\mathbbm{R}\mathbbm{P}^3$.}
\end{figure}

If we let $[z_0:z_1:z_2:z_3]$ denotes the homogeneous coordinates of $\mathbbm{P}_K^3$, then twisted cubic can be defined as the intersection of the following three quadrics which can be obtained by $2\times2$ minors of the catalecticant matrix in Eq. \eqref{catalecticant}
\begin{equation}\label{3-quadrics}
\begin{cases}
z_0z_2-z_1^2=0\,,\\
z_1z_3-z_2^2=0\,,\\
z_0z_3-z_1z_2=0\,.
\end{cases}
\end{equation}

\subsection{Segre embedding}\label{subsec.3.4.2}
The Zariski topology is strictly finer than the product topology\footnote{In topology, a product space is the Cartesian product of a family of topological spaces equipped with a topology called the product topology}. It means that the Zariski topology on $\mathbbm{P}_K^2$ is not identical with the product topology on $\mathbbm{P}_K^1\times\mathbbm{P}_K^1$. So, the Segre embedding is used to consider the Cartesian product of projective spaces as a projective subvariety of a bigger projective space and restrict the Zariski topology on the projective space to the subvariety.

\begin{definition}[Segre embedding]
The Segre embedding is an injective morphism deﬁned by
\begin{equation}
\Sigma_{(d_1-1,d_2-1)}^{2}\colon\mathbbm{P}_K^{d_1-1}\times\mathbbm{P}_K^{d_2-1}\hookrightarrow\mathbbm{P}_K^{d_1d_2-1}\,,
\end{equation}
which takes a pair of points $([a],[b])\in\mathbbm{P}^{d_1-1}\times\mathbbm{P}^{d_2-1}$ to their products
\begin{equation}
([a_{0}:\cdots:a_{d_1-1}],[b_{0}:\cdots:b_{d_2-1}])\rightarrow[a_{0}b_{0}:\cdots:a_{i}b_{j}:\cdots:a_{d_1-1}b_{d_2-1}]\,,
\end{equation}
where the notation refers to homogeneous coordinates and the $a_{i}b_{j}$ are taken in lexicographical order. The image of this morphism is called Segre variety.
\end{definition}

\begin{remark}
The Segre embedding can be considered as matrix multiplication\footnote{Notice that it is true for Segre embedding of two factors.}. That is
\begin{equation}
\Sigma_{(d_1-1,d_2-1)}^{2}\colon\mathbbm{P}\big(K^{d_1}\big)\times\mathbbm{P}\big(K^{d_2}\big)\hookrightarrow\mathbbm{P}\big(K^{d_1}\otimes K^{d_2}\big)\,,
\end{equation}
which is defined by
\begin{equation}
([a],[b])\rightarrow[a\otimes b]=[a_ib_j]\,.
\end{equation}
Indeed,
\begin{equation}
([a_{0}:\cdots:a_{d_1-1}],[b_{0}:\cdots:b_{d_2-1}])\rightarrow\left[
\begin{pmatrix}
a_0\\
\vdots\\
a_{d-1}
\end{pmatrix}
\begin{pmatrix}
b_0 & \cdots & b_{d-1}
\end{pmatrix}
\right].
\end{equation}
\end{remark}

\begin{example}
Consider the following Segre map
\begin{equation}\label{Segre-P1P1-P3}
\Sigma_{(1,1)}^{2}\colon\mathbbm{P}_K^1\times\mathbbm{P}_K^1\hookrightarrow\mathbbm{P}_K^3\,,
\end{equation}
which is defined by
\begin{equation}
[a_0:b_1]\times[b_0:b_1]\to[a_0b_0:a_0b_1:a_1b_0:a_1b_1]\,.
\end{equation}
Regarding $[z_0:z_1:z_2:z_3]$ as the homogeneous coordinates of $\mathbbm{P}_K^3$, then it is easy to check that the Segre variety in $\mathbbm{P}_K^3$ has equation
\begin{equation}
z_0z_3-z_1z_2=0\,.
\end{equation}
Thus, the Segre variety is a smooth quadric in $\mathbbm{P}_K^3$. If we take the real field $K=\mathbbm{R}$, then the Segre variety would be a hyperboloid ruled by straight lines (see Fig. \ref{fig:Segre}).
\begin{figure}[t]
\center{\includegraphics[width=4cm]{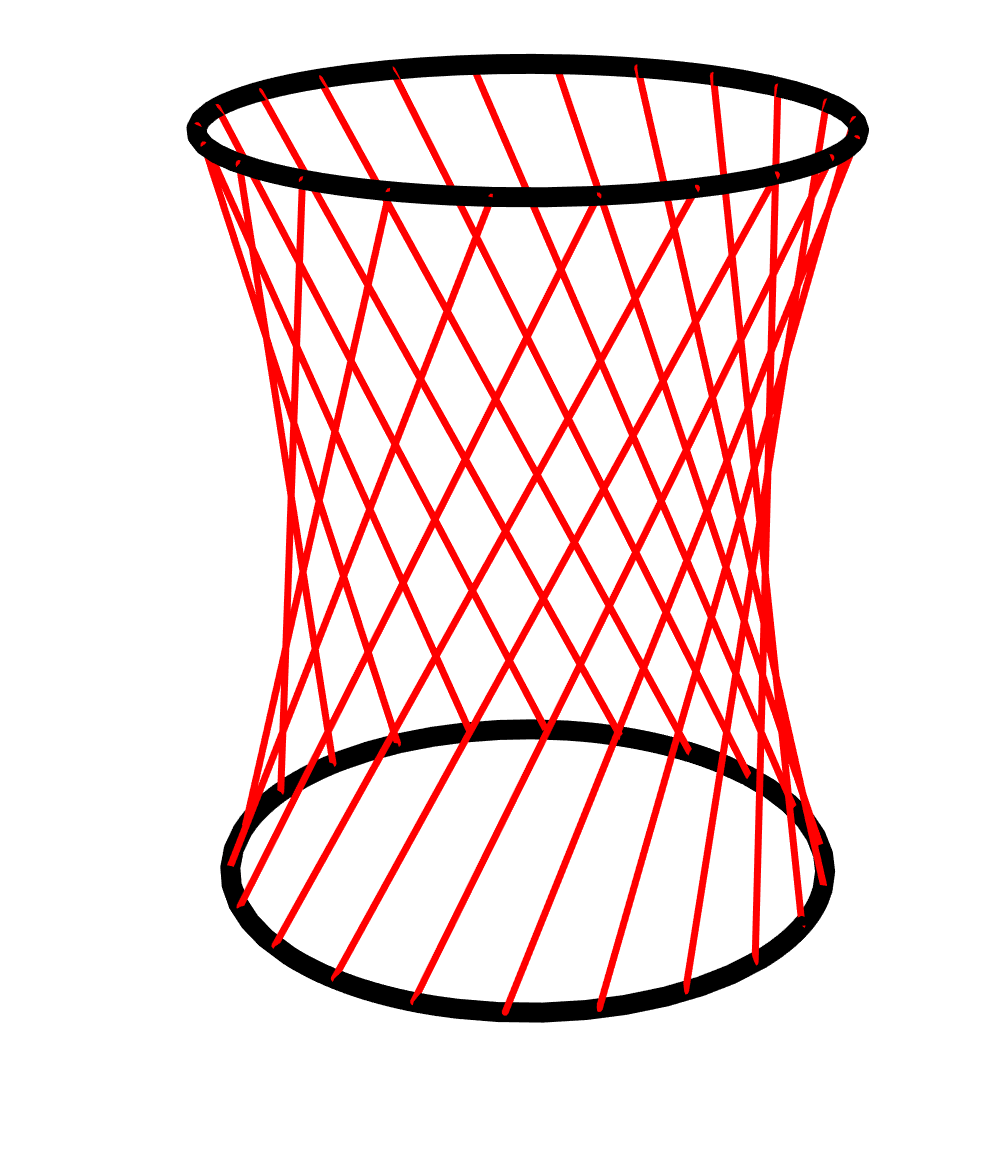}}
\caption[The real Segre variety]{\label{fig:Segre} 
The real Segre embedding $\mathbbm{R}\mathbbm{P}^1\times\mathbbm{R}\mathbbm{P}^1\hookrightarrow\mathbbm{R}\mathbbm{P}^3$: the Segre variety is a hyperboloid ruled by straight lines.}
\end{figure}
\end{example}

The Segre map in Eq. \eqref{Segre-P1P1-P3} can be also defined as follows
\begin{equation}
[a_0:a_1]\times[b_0:b_1]\to\begin{pmatrix}
a_0b_0 & a_0b_1\\
a_1b_0 & a_1b_1
\end{pmatrix}.
\end{equation}
Therefore, the Segre variety represents the set of $2\times2$ matrices of rank at most one and the ideal of the Segre variety is generated by the vanishing of the determinant of the generic matrix
\begin{equation}
\begin{pmatrix}
z_0 & z_1\\
z_2 & z_3
\end{pmatrix}.
\end{equation}

In the next section, we will present algebraic tools that are profitable for entanglement characterization. Particularly, we study flattening a tensor corresponding to a multipartite quantum state. We see that a given $\ell$-multilinear rank configuration determines a determinantal variety in the projective Hilbert space $\mathbbm{P}(\mathcal{H})$ which is a subvariety of $k$-secant varieties of the Segre variety.

It is worth remarking that for the sake of simplicity, we write $\mathbbm{C}\mathbbm{P}^d\equiv\mathbbm{P}^d$ in the rest of this thesis since we mainly work with the complex field $K=\mathbbm{C}$, unless otherwise noted.

\section{Algebraic geometry and SLOCC invariants}\label{sec.3.5}

\subsection{Tensors}\label{sec.Tensors}
Since rank and border rank of a tensor are invariant under the SLOCC they are profitable tools to study our central problems, that is, multipartite entanglement classification and the interconversion between different resources by the SLOCC and asymptotic SLOCC. Therefore, we provide these tools in this section.

A tensor is an algebraic object that can be represented by a multidimensional array of components that are functions of the coordinates of a space. We provide a general definition of a tensor based on multilinear maps in multilinear algebra.

\begin{definition}[Multilinear map]
Let $V_1,\ldots,V_n$, and $W$ be vector spaces. An $n$-linear map is a function
\begin{equation}\label{multilinear-map}
\phi\colon V_1\times\cdots\times V_n\to W\,,
\end{equation}
such that $(v_1,\ldots,v_i,\ldots,v_n)$ is a linear function on $v_i\in V_i$ if all of the variables but $v_i$ are held constant, for each $i\in\{1,\ldots,n\}$. Therefore, a multilinear map is a function of several variables that is linear separately in each variable. 
\end{definition}

\begin{definition}[Multilinear complexity]
Let us assume that there exists a sequence $(f_1^{(1)},\ldots,f_1^{(n)},w_1:\cdots:f_r^{(1)},\ldots,f_r^{(n)},w_r)$ with $f_k^{(j)}\in V_j^{\vee}$ and $w_k\in W$ such that the multilinear map in Eq. \eqref{multilinear-map} can be written as
\begin{equation}\label{multilinear-complexity}
\phi(v_1,\ldots,v_n)=\sum_{k=1}^{r}f_k^{(1)}(v_1)\cdots f_k^{(n)}(v_n)w_k \qquad \forall~v_j\in V_j\,,
\end{equation}
The minimum $r$ is called  multilinear complexity or the rank of the multilinear map $\phi$ in Eq. \eqref{multilinear-map}.
\end{definition}

The border rank of the multilinear map $\phi$ in Eq. \eqref{multilinear-map} is defined as the smallest $r$ such that $\phi$ can be written as a limit of a sequence of bilinear maps of rank $r$.

A multilinear map can be thought of as a tensor in $V_1^{\vee}\otimes\cdots\otimes V_n^{\vee}\otimes W$. Therefore, Eq. \eqref{multilinear-complexity} can be considered as decompositions of this tensor and the rank and border rank of the multilinear map coincides with the tensor rank and tensor border rank. In addition, it can be concluded that there is a natural one-to-one correspondence between the multilinear map $\phi$ in Eq. \eqref{multilinear-map} and the following linear map
\begin{equation}\label{tensor-map}
\mathcal{T}\colon V_1\otimes\cdots\otimes V_n\to W\,,
\end{equation}
where $V_1\otimes\cdots\otimes V_n$ denotes the tensor product of vector spaces $V_1,\ldots,V_n$. The relation between these two maps is given by $\phi(v_1,\ldots,v_n)=\mathcal{T}(v_1\otimes\cdots\otimes v_n)$.

\begin{remark}
Any pure multipartite quantum state $|\psi\rangle\in\otimes_{i=1}^n\mathbbm{C}^{d_i}$ corresponds to an order-$n$ tensor (or simply $n$-tensor). So one can treat a multipartite quantum state like a tensor.
\end{remark}

\begin{remark}
A tensor of the form $v_1\otimes\cdots\otimes v_n$ in $\otimes_{i=1}^n V_i$ is called simple tensor. Therefore, a nonzero simple tensor is equivalent to a fully separable quantum state.
\end{remark}

Let us consider the bilinear complexity of the matrix multiplication as an example. The matrix multiplication of two $2\times2$ matrices can be seen as a bilinear map
\begin{equation}\label{MaMu-2-2}
\mathrm{MaMu}\colon{\mathbbm{C}^4}^{\vee}\times{\mathbbm{C}^4}^{\vee}\to\mathbbm{C}^4\,,
\end{equation}
Now, the question is: what is the minimum number of scalar multiplications over the ground ﬁeld that are required to compute the map $\mathrm{MaMu}$ in Eq. \eqref{MaMu-2-2}? Actually, the answer of this question is equal to the rank of the bilinear map.

Let $A$ and $B$ be two $2\times2$ matrices
\begin{equation}
A=\begin{pmatrix}
a_{11} & a_{12} \\
a_{21} & a_{22}
\end{pmatrix}, \qquad B=\begin{pmatrix}
b_{11}& b_{12} \\
b_{21} & b_{22}
\end{pmatrix}.
\end{equation}
Using the definition of matrix multiplication we need eight scalar multiplications to compute their product. That is
\begin{equation}
A\times B=\begin{pmatrix}
a_{11}b_{11}+a_{12}b_{21} & a_{11}b_{12}+a_{12}b_{22} \\
a_{21}b_{11}+a_{22}b_{21} & a_{21}b_{12}+a_{22}b_{22}
\end{pmatrix}.
\end{equation}
However, in Ref. \cite{Strassen69}, Strassen presents an algorithm such that using the following seven scalar multiplications
\begin{align}\nonumber
&M_1=(a_{11}+a_{22})(b_{11}+b_{22})\,, \\ \nonumber
&M_2=(a_{21}+a_{22})b_{11}\,, \\ \nonumber
&M_3=a_{11}(b_{12}-b_{22})\,, \\ \nonumber
&M_4=a_{22}(b_{21}-b_{11})\,, \\ \nonumber
&M_5=(a_{11}+a_{12})b_{22}\,, \\ \nonumber
&M_6=(a_{21}-a_{11})(b_{11}+b_{12})\,,  \\ \label{Strassen-Algorithm}
&M_7=(a_{12}-a_{22})(b_{21}+b_{22})\,,
\end{align}
we can obtain $A\times B$ as
\begin{equation}
A\times B=\begin{pmatrix}
M_1+M_4-M_5+M_7 & M_3+M_5 \\
M_2+M_4 & M_1-M_2+M_3+M_6
\end{pmatrix}.
\end{equation}

While one needs $n^3$ scalar multiplications to compute the matrix multiplication of two $n\times n$ matrices based on its definition, one can do it by $7^m$ scalar multiplication for $n=2^m$, using Strassen’s algorithm iteratively. Therefore, this algorithm lowers the upper bound for the complexity of matrix multiplication from $n^3$ to $n^{\log_2 7}\simeq n^{2.81}$. A long-standing open problem in algebraic complexity theory is whether we can lower this upper bound to $n^2$.

The bilinear map in Eq. \eqref{MaMu-2-2} is equivalent to a three-tensor
\begin{equation}
\mathrm{MaMu}\in{\mathbbm{C}^4}\otimes{\mathbbm{C}^4}\otimes{\mathbbm{C}^4}\,.
\end{equation}
Therefore, considering the bilinear map as a tensor, we can write $\mathrm{MaMu}$ as a linear combination of simple tensors where each simple tensor represents an scalar multiplication. The naive algorithm for matrix multiplication implies that $\mathrm{MaMu}$ can be written as a linear combination of eight simple tensors. However, concerning Strassen's algorithm in Eq. \eqref{Strassen-Algorithm}, the tensor $\mathrm{MaMu}$ can be written as $\mathrm{MaMu}=\sum_{i=1}^{7}x_i\otimes y_i\otimes z_i$. It means that the upper bound of the tensor rank of $\mathrm{MaMu}$ is seven. Later, it was shown that the rank of $\mathrm{MaMu}$ is optimal \cite{Winogard71}. Furthermore, Landsberg proved that the border rank of $\mathrm{MaMu}$ is seven \cite{Landsberg06}.

The rank of a tensor $\mathcal{T}$ is defined as the minimum number of simple tensors that sum to $\mathcal{T}$ and it extends the notion of the rank of a matrix in algebra \cite{Bourbaki}, so it can be seen as a generalization of Schmidt rank. The following is a concrete definition of tensor rank.

\begin{definition}[Tensor rank]\label{def:rank}
Let $\mathcal{T}\in V_1\otimes\cdots\otimes V_n$ be a tensor where each $V_i$ is a vector space. The tensor rank of $\mathcal{T}$ is defined as follows
\begin{equation}\label{rank}
\rk(\mathcal{T})=\min\Big\{r~\big|~\mathcal{T}=\sum_{p=1}^{r}v_1^{(p)}\otimes\cdots\otimes v_n^{(p)}\,,~\text{for some}~v_i^{(p)}\in V_i\Big\}\,.
\end{equation}
\end{definition}
If $n=2$, the tensor rank of $\mathcal{T}$ is equal to the matrix rank of the linear map $\mathcal{T}\colon V_1^{\vee}\to V_2$.

The tensor border rank (border rank, for short) of a tensor $\mathcal{T}$ is defined as the smallest $r$ such that $\mathcal{T}$ is a limit of tensors of rank $r$, or equivalently the smallest $r$ such that $\mathcal{T}$ lies in the Zariski closure of the set of tensors of rank $r$ \cite{Landsberg}, so it can be seen as a counterpart of the generalized Schmidt rank. The following is a concrete definition of border rank of a tensor $\mathcal{T}$.

\begin{definition}[Border rank]\label{def:brank}
Let $\mathcal{T}\in V_1\otimes\cdots\otimes V_n$ be a tensor where each $V_i$ is a vector space. The border rank of $\mathcal{T}$ is the smallest $r$ such that $\mathcal{T}$ is a limit of tensors of rank $r$, i.e.,
\begin{equation}\label{brank}
\brk(\mathcal{T})=\min\Big\{r~\big|~\mathcal{T}=\lim_{\varepsilon\to 0}\mathcal{T}_{\varepsilon}\,,\text{s.t.}~\forall\varepsilon\neq0,~\rk(\mathcal{T}_{\varepsilon})=r\Big\}\,.
\end{equation}
\end{definition}
Clearly, $\brk(\mathcal{T})\leq\rk(\mathcal{T})$.

It is useful to introduce symmetric tensors and symmetric tensor rank (or Waring rank).

\begin{definition}[Symmetric tensor]
Let $\mathcal{T}\in V_1\otimes\cdots\otimes V_n$ be a tensor where each $V_i$ is a vector space. The tensor $\mathcal{T}$ is symmetric if it is invariant under the action of a symmetric group that permutes the tensor factors, that is, for all $\mathfrak{p}\in\mathfrak{S}_n$, $\mathfrak{p}(\mathcal{T})=\mathcal{T}$.
\end{definition}

Symmetric tensors correspond to symmetric multipartite quantum states that are invariant under any permutations of the parties. To estimate the symmetric tensor rank of a symmetric quantum state $|\psi\rangle$, denoted by $\srk(\psi)$, there is a correspondence between symmetric quantum states and homogeneous polynomials.

\begin{definition}[Waring rank]\label{def:Waring}
The Waring rank of a homogeneous $d$-variate degree-$n$ polynomial $f\in\mathbbm{C}[x_0,\ldots,x_{d-1}]_n$ is the minimum number of terms contained in $f$ when it is expressed as a combination of $n^{\text{th}}$ powers of linear forms, that is, the minimum number of $s$ in the following symmetric decomposition
\begin{equation}
f(x_0,\ldots,x_{d-1})=\sum_{i=1}^s(\beta_0^{(i)}x_0+\cdots+\beta_{d-1}^{(i)}x_{d-1})^n\,.
\end{equation}
\end{definition}

The Waring rank is a much-studied problem in algebraic geometry \cite{CGLM08,LT10,CCG12,Shitov18,Oeding20}. Since there is a unique correspondence between a symmetric quantum state $|\psi\rangle$ and a homogeneous $d$-variate degree-$n$ polynomial $f$, up to scaling the variables, the symmetric tensor rank is identical to the Waring rank. Although it is known that the symmetric tensor rank is not equal to the tensor rank \cite{Shitov18}, it can be considered as an upper bound for the tensor rank of a symmetric tensor $\mathcal{T}_{\sym}$ (see Refs. \cite{CCDJW10,LT10}), that is,
\begin{equation}\label{SymRank}
\rk(\mathcal{T}_{\sym})\leq\srk(\mathcal{T}_{\sym})\,.
\end{equation}

The Waring rank of a general monomial has been found in Ref. \cite{CCG12}.

\begin{theorem}[Ref. \cite{CCG12}]\label{theo:Warink-rk}
Let $\alpha=(0<\alpha_0\leq\cdots\leq\alpha_{d-1})\in\mathbbm{N}^{d+1}$. The Waring rank of the monomial $x_0^{\alpha_0}\cdots x_{d-1}^{\alpha_{d-1}}$ is equal to
\begin{equation}
\rk(x_0^{\alpha_0}\cdots x_{d-1}^{\alpha_{d-1}})=\prod_{i=1}^{d-1}(\alpha_i+1)\,.
\end{equation}
\end{theorem}

In addition, the conjecture in Ref. \cite{Oeding20} provides the symmetric border rank of a general monomial.

\begin{conjecture}[Ref. \cite{Oeding20}]\label{conj:Waring-brk}
Let $\alpha=(0<\alpha_0\leq\cdots\leq\alpha_{d-1})\in\mathbbm{N}^{d+1}$. The symmetric border rank of the monomial $x_0^{\alpha_0}\cdots x_{d-1}^{\alpha_{d-1}}$ is equal to
\begin{equation}
\brk(x_0^{\alpha_0}\cdots x_{d-1}^{\alpha_{d-1}})=\prod_{i=0}^{d-2}(\alpha_i+1)\,.
\end{equation}
\end{conjecture}

\subsection{Flattening}\label{sec.flattening}
Although it is customary to look at an $n$-partite quantum state
\begin{equation}\label{n-partite}
|\psi\rangle=\sum_{\alpha=1}^{n}\sum_{i_{\alpha}=0}^{d_{\alpha}-1}\mathbf{c}_{i_{1}\cdots i_{n}}|i_{1}\cdots i_{n}\rangle\,,
\end{equation}
as a vector, such a vector results from the vectorization of an $n$-tensor in the Hilbert space $\mathcal{H}_{n}=\otimes_{i=1}^{n}\mathbbm{C}^{d_{i}}$. In multilinear algebra, this vectorization is a kind of tensor reshaping. Here, we shall use a tensor reshaping known as tensor flattening (or matricization) \cite{Landsberg}. It consists in partitioning the $n$-fold tensor product space (here, $\mathcal{H}_{n}$) to two-fold tensor product spaces with higher dimensions. With respect to the partitioning, we define an ordered $\ell$-tuple $I=(i_{1},i_{2},\ldots,i_{\ell})$ where $1\leq\ell\leq{n-1}$ and $1\leq i_{1}<\cdots<i_{\ell}\leq{n}$ and an ordered $(n-\ell)$-tuple related to complementary partition $\bar{I}$ such that $I\cup\bar{I}=(1,2,\ldots,n)$. Therefore, $\mathcal{H}_{n}\simeq\mathcal{H}_{I}\otimes\mathcal{H}_{\bar{I}}$ where $\mathcal{H}_{I}=\otimes_{\alpha=i_{1}}^{i_{\ell}}\mathbbm{C}^{d_{\alpha}}$ and $\mathcal{H}_{\bar{I}}$ is the complementary Hilbert space. For any state $\psi$ with vector representation $|\psi\rangle\in\mathcal{H}_n$, the $\ell$-partition $I$ leads to a linear operator $\mathcal{M}_{I}[\psi]$ which maps 
the dual $\mathcal{H}_I^{\vee}$ of ${\cal H}_I$ to ${\cal H}_{\bar{I}}$,
\begin{equation}\label{flattening}
\mathcal{M}_I[\psi]\colon\mathcal{H}_I^{\vee}\to\mathcal{H}_{\bar{I}}\,.
\end{equation}
Using Dirac notation, the matricization of $|\psi\rangle$ reads
\begin{equation}
\mathcal{M}_{I}[\psi]=\left(\langle{e_1}|\psi\rangle,\ldots,\langle{e_{d_{I}}}|\psi\rangle\right)^{\rm{T}}\,,
\end{equation}
where $\{|e_j\rangle=|i_{1}\cdots{i_{\ell}}\rangle\}_{j=1}^{d_I=\Pi{d_{\alpha}}}$ is the computational basis of $\mathcal{H}_{I}$ and $\rm{T}$ denotes the matrix transposition.
Clearly, we shall consider all ordered $\ell$-tuples $I$ to avoid overlapping of entanglement families \cite{GM18}. Hence, for a given $|\psi\rangle$ we have as many matrix representations $\mathcal{M}_{I}[\psi]$ as the number of possible $\ell$-tuples $I$, which is ${\binom{n}{\ell}}$. In this way, we can define $\ell$-multilinear rank ($\ell$-multirank, for short) \cite{Landsberg} of $|\psi\rangle$ as a ${\binom{n}{\ell}}$-tuple of ranks of ${\binom{n}{\ell}}$ different flattening $\mathcal{M}_{I}[\psi]$. Obviously, the zero-multirank is just a number, namely 1, as well as the $n$-multirank.

It is worth noting that it is enough to check $\ell$-multiranks for partition $I$ with \mbox{$1\leq\ell\leq\lfloor\frac{n}{2}\rfloor$}, because for complementary partition $\bar{I}$ the matrices $\mathcal{M}_{\bar{I}}[\psi]$ are just the transposes of $\mathcal{M}_{I}[\psi]$ and transposition does not alter the rank of the matrix. Moreover, when $n$ is an even integer the half digits of $\binom{n}{\frac{n}{2}}$-tuple corresponding to $\frac{n}{2}$-multirank is iterative. Therefore, the number of independent flattening ranges from ${\binom{n}{1}}$ to $(1/2)^{n+1~{\rm{mod}}~2} {\binom{n}{\lfloor\frac{n}{2}\rfloor}}$.

\begin{remark}
Recall that a matrix has rank less than $r$ iff all of its $(r+1)\times(r+1)$ minors vanish. Similarly, a tensor has multilinear rank less than $(r_1,\ldots,r_m)$ iff all of its $(r_i+1)\times(r_i+1)$ minors of the $i$-flattening vanish for all $i$.
\end{remark}

\begin{remark}
Not only do the integers of the tuples ($\ell$-multiranks) tell us about the separability of the state but also the greater the integers are the more entanglement the parties of the state have. Each integer equals one means there is a separability between two parties
\end{remark}
\begin{figure}[t]
\center{\includegraphics[width=13cm]{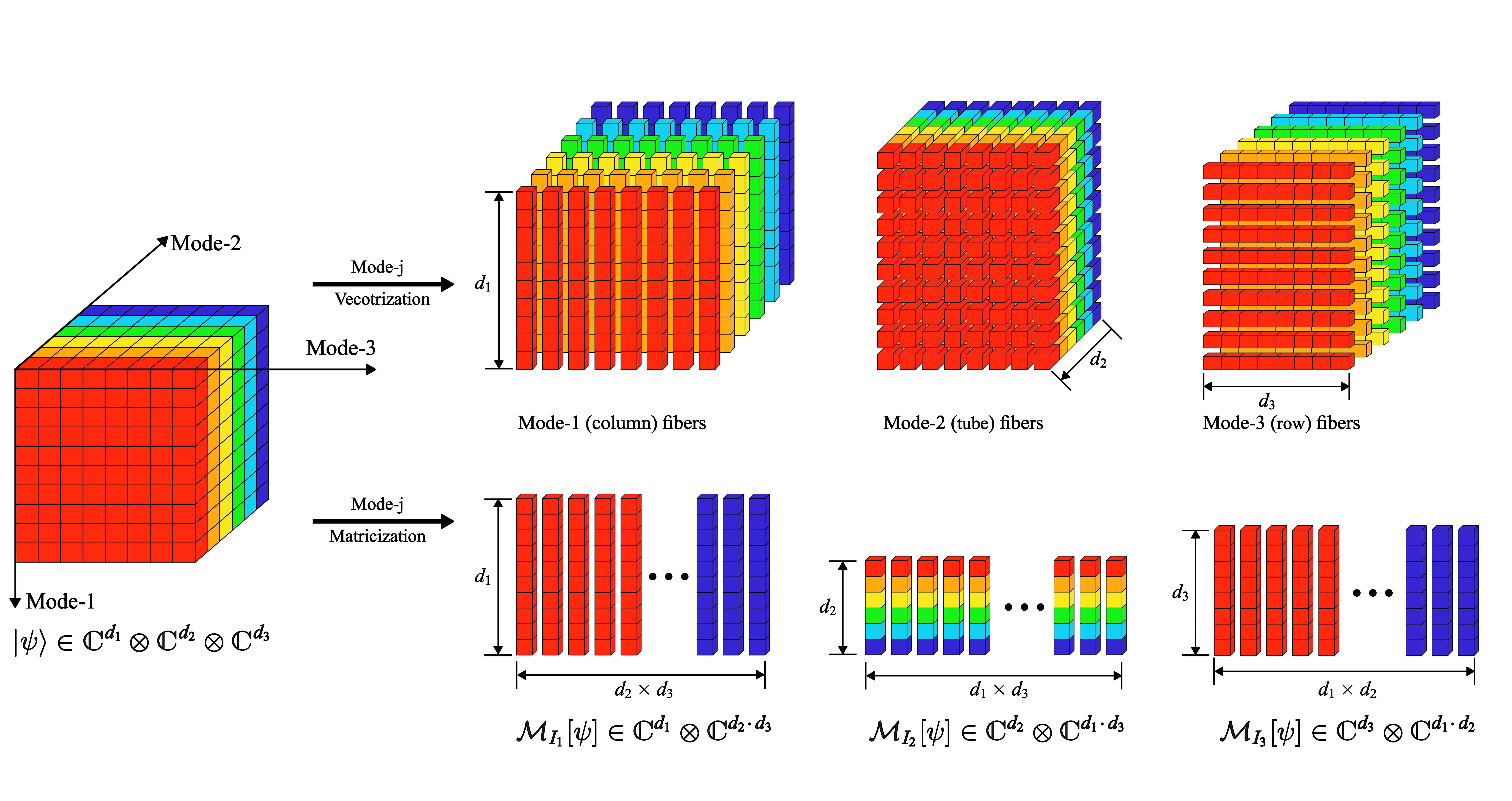}}
\caption[Flattening]{\label{fig:flattening} Flattening of $|\psi\rangle\in\mathbbm{C}^{d_1}\otimes\mathbbm{C}^{d_2}\otimes\mathbbm{C}^{d_3}$ to three different matrices.}
\end{figure}
\begin{example}
Let us consider $|\psi\rangle\in\mathbbm{C}^{d_1}\otimes\mathbbm{C}^{d_2}\otimes\mathbbm{C}^{d_3}$. Regarding Eq. \eqref{flattening}, we have three different matrices corresponding to partitions: $I_i=(i)$ for all $i\in\{1,2,3\}$. That is
\begin{align}\nonumber
&\mathcal{M}_{I_1}[\psi]\in\mathbbm{C}^{d_1}\otimes\mathbbm{C}^{d_2\cdot d_3}\,, \\ \nonumber
&\mathcal{M}_{I_2}[\psi]\in\mathbbm{C}^{d_2}\otimes\mathbbm{C}^{d_1\cdot d_3}\,, \\
&\mathcal{M}_{I_3}[\psi]\in\mathbbm{C}^{d_3}\otimes\mathbbm{C}^{d_1\cdot d_2}\,.
\end{align}
These three different flattening is illustrated in Fig. \ref{fig:flattening}. Therefore, one-multirank is a triple $(r_1,r_2,r_3)$ where $r_i=\mathrm{rank}(\mathcal{M}_{I_i})$.
\end{example}

\begin{remark}
The rank of $\mathcal{M}_{I}[\psi]$ is the same as the rank of the reduced density matrix obtained after tracing over the parties identified by the $(n-\ell)$-tuple $\bar{I}$, i.e.,
\begin{equation}
\varrho_{I}={\rm{Tr}}_{\bar{I}}\left(|\psi\rangle\langle\psi|\right)=\mathcal{M}_{I}[\psi]\mathcal{M}_{I}^{\dagger}[\psi]\,.
\end{equation}
\end{remark}

\begin{theorem}\label{l-multirank-SLOCC}
$\ell$-multirank is an SLOCC invariant.
\end{theorem}
{\it Proof.}
Regarding SLOCC equivalent states, i.e.,
\begin{equation}
|\tilde{\psi}\rangle = \left(\otimes_{i=1}^{n} A_{i}\right) |\psi\rangle\,,
\end{equation}
where $|\psi\rangle\in\mathcal{H}_{n}$ and $A_{i}\in{\rm{SL}}(d_{i},\mathbbm{C})$, we have following relation between their flattenings
\begin{equation}
\mathcal{M}_{I}[\tilde{\psi}]=\left(\otimes_{i\in I}A_{i}\right)\mathcal{M}_{I}[\psi]\left(\otimes_{i\in\bar{I}}A_{i}\right)^{\rm{T}}\,,
\end{equation}
that does not alter the matrix rank. Therefore, $\ell$-multirank is an SLOCC invariant.
\qed

\begin{remark}\label{remark-MR-GE}
A state is genuinely entangled iff all $\ell$-multiranks are greater than one.
\end{remark}

\subsection{Auxiliary varieties}

Since $\ell$-multiranks only depend on the quantum state, and not on the representation, and, furthermore, because statements about rank can be rephrased as statements about minors which are determinants, it follows that a given $\ell$-multirank configuration determines a determinantal variety in the projective Hilbert space $\mathbbm{P}\mathcal{H}$. Actually, the determinantal variety is a subset of all matrices with rank $r$ or less in $\mathbbm{P}\mathcal{H}$, that is just the common zero locus of the $(r+1)\times(r+1)$ minors. Pure multipartite states which have $\ell$-multiranks bounded by a given integer sequence make a subvariety of $\mathbbm{P}\mathcal{H}$. In particular, the Segre variety is an example of a determinantal variety; it is the zero locus of the $2\times2$ minors of the coefficient matrices in Eq. \eqref{n-partite}, i.e., common zero locus of the quadratic polynomials $\mathcal{M}_{ij}\mathcal{M}_{kl}-\mathcal{M}_{il}\mathcal{M}_{kj}$. Therefore, the projective variety of fully separable $n$-partite states has the structure of a Segre variety \cite{Miyake03, Heydari08} which is embedded in the ambient space as follows:
\begin{equation}\label{segre-SM}
\Sigma^{n}_{\textbf{d}-\textbf{1}}\colon\mathbbm{P}^{d_{1}-1}\times\mathbbm{P}^{d_{2}-1}\times\cdots\times\mathbbm{P}^{d_{n}-1}\hookrightarrow\mathbbm{P}^{D}\,.
\end{equation}
Here, $\textbf{d}-\textbf{1}=(d_{1}-1,\ldots,d_{n}-1)$, $D=\left(\Pi_{i=1}^{n}d_{i}\right)-1$, and $\times$ is the Cartesian product of sets.  One can readily check that $\Sigma^{n}_{\textbf{d}-\textbf{1}}$ is indeed the projective variety of fully separable states. Actually, if all partial traces are pure states, the corresponding ranks are all one. So we have that for all $\ell$-partitions the rank of 
$\mathcal{M}_I[\psi]$ is always one. Conversely, if all ranks are one, the state is fully separable.

It is worth noting that multipartite symmetric separable states with identical parties of dimension $d$ have the structure of Veronese variety.

\begin{definition}[Secant variety]
Let projective varieties $\mathfrak{X}$ and $\mathfrak{Y}$ be subvarieties of a projective variety. The joining of $\mathfrak{X}$ and $\mathfrak{Y}$ is given by the algebraic closure, for the Zariski topology, of the lines from one to the other,
\begin{equation}\label{joinvariety}
\mathfrak{J}(\mathfrak{X},\mathfrak{Y})=\overline{\bigcup_{x\in\mathfrak{X},y\in\mathfrak{Y},x\neq{y}}\mathbbm{P}^1_{xy}}\,,
\end{equation}
where $\mathbbm{P}^1_{xy}$ is the projective line that includes both $x$ and $y$. If $\mathfrak{X}=\mathfrak{Y}$, the joining is called the secant variety of $\mathcal{X}$, i.e., $\sigma(\mathfrak{X})=\mathfrak{J}(\mathfrak{X},\mathfrak{X})$.
\end{definition}

\begin{remark}
The determinantal varieties in Eq. \eqref{segre-SM} are subvarieties of secant varieties of the Segre variety.
\end{remark}

\begin{definition}[Tangent variety (see Ref. \cite{Zak})]
Let projective varieties $\mathfrak{X}$ and $\mathfrak{Y}$ be subvarieties of a projective variety and suppose now $\mathfrak{Y}\subset\mathfrak{X}$. Let $\mathfrak{T}^{\star}_{\mathfrak{X},\mathfrak{Y},y_{0}}$ denote the relative tangent star, which is the union of $\mathbbm{P}_{\star}^{1}=\lim_{x,y\to{y_{0}}}\mathbbm{P}^{1}_{xy}$ with $y_{0}\in\mathfrak{Y}$ and the $x$'s taken from $\mathfrak{X}$. The variety of relative tangent stars is defined as follows
\begin{equation}\label{tangentvariety}
\mathfrak{T}(\mathfrak{X},\mathfrak{Y})=\bigcup_{y\in\mathfrak{Y}}\mathfrak{T}^{\star}_{\mathfrak{X},\mathfrak{Y},y}\,.
\end{equation}
 If $\mathfrak{X}=\mathfrak{Y}$, we denote the tangential variety as $\tau(\mathfrak{X})=\mathfrak{T}(\mathfrak{X},\mathfrak{X})$. 
\end{definition} 
 
\begin{remark}
The iterated join of $k$ copies of $\mathcal{X}$ is called the $k$-secant variety of $\mathcal{X}$. Hence, the secant varieties that we have mentioned above are given by the algebraic closure of the joining of the Segre variety and the immediately previous secant variety:
\begin{equation}\label{recursivesecant}
\sigma_{k}(\Sigma)=\mathfrak{J}\left(\sigma_{k-1}(\Sigma),\Sigma\right)\,.
\end{equation}
\end{remark}

Notice that the first secant variety of the Segre variety coincides with the Segre variety itself, i.e., $\sigma_{1}(\Sigma)=\Sigma$\footnote{Actually, this statement is true for any variety.}. This means that a generic point of the $k$-secant is a combination of $k$ independent points of the Segre variety (the superposition of $k$ fully separable states). If $\sigma_k$ fills the ambient projective space we say that the generic tensor rank is $k$. A crucial element of the definitions is that the secants are closed. This means that in each $k$-secant ($k>1$) family there will be elements whose tensor rank will not be $k$ (it can be greater than $k$) but the border rank is $k$. Thus we will make the distinction between the proper secant and the tangent.

For instance, based on Strassen's algorithm discussed in Section \ref{sec.Tensors}, the tensor rank of $\mathrm{MaMu}$ is seven which means $\mathrm{MaMu}\in\sigma_7(\mathbbm{P}^3\times{\mathbbm{P}^3}\times{\mathbbm{P}^3})$ \cite{Strassen69}. Moreover, based on the information that the border rank of $\mathrm{MaMu}$ is also seven we can conclude that $\mathrm{MaMu}\notin\sigma_6(\mathbbm{P}^3\times{\mathbbm{P}^3}\times{\mathbbm{P}^3})$ \cite{Landsberg06}.

\begin{remark}
If the projective variety $\mathfrak{X}$ is non-degenerate, i.e., it is not contained in a linear subspace of $\mathbbm{P}(K^d)$, there is a natural sequence of inclusions given by
\begin{equation}
\mathfrak{X}\subset\sigma_2(\mathfrak{X})\subset\sigma_3(\mathfrak{X})\subset\cdots\subset\sigma_r(\mathfrak{X})=\mathbbm{P}(K^d)\,,
\end{equation}
where $r$ is the smallest integer such that the $r$-th secant variety fills the ambient space.
\end{remark}

We can also generalize the definition of tangent line to a curve by introducing its osculating planes \cite{Harris}. Hence, one can define varieties of different types of limiting curves inside the $k$-secant variety. To simplify the calculations, let $x_t$ be a smooth curve in $\Sigma$. Then, we can take higher order derivatives and calculate the higher dimensional tangential varieties as follows:
\begin{equation}\label{G-tangent}
\tau_{k}(\Sigma)=\overline{\{x_{0}+x'_{0}+\cdots+x^{(k-1)}_{0}|x_{t}\subset\Sigma~\text{is a smooth curve}\}}\,.
\end{equation}
Obviously $\tau_{k}(\Sigma)\subset\sigma_{k}(\Sigma)$ and $\mathfrak{T}(\tau_{k-1}(\Sigma),\Sigma)\subset\tau_{k}(\Sigma)$, the last inclusion is even an equality.

The expected dimension of $k$-secant variety of a projective variety $\mathfrak{X}\subset\mathbbm{P}^{d-1}$ arises just from the naive dimension count. That is, if the $k$-secant variety does not fill the ambient space, each point in an open set of $\sigma_k(\mathfrak{X})$ can be decomposed as a sum of $k$ points from the projective variety $\mathfrak{X}$. In the $k$-secant variety, there are $k$ points which leads to $k\times\dim\mathfrak{X}$ parameters. Concerning $k$ points, there are $k$ coefficients of the ground field. Moreover, one can divide all the $k$ coefficients by one of them which leads to $k-1$ independent parameters. Furthermore, we expect that since there is no deductive relation between these parameters, one can define the expected dimension of the $k$-secant variety of a projective variety $\mathfrak{X}$ as follows.
\begin{definition}[Expected dimension of $k$-secant variety]
For a projective variety $\mathfrak{X}\subseteq\mathbbm{P}^{d-1}$, with $\dim\,\mathfrak{X}=s$, the expected dimension of $\sigma_{k}(\mathfrak{X})$ is
\begin{equation}\label{exp-k-secant}
\mathrm{expdim}\,\sigma_k(\mathfrak{X})=\min\{ks+k-1,d-1\}\,.
\end{equation}
If $\mathrm{expdim}\,\sigma_k(\mathfrak{X})-\dim\sigma_k(\mathfrak{X})=a>0$, then $\sigma_{k}(\mathfrak{X})$ is called $k$-defective or simply defective and $a$ is the defect.
\end{definition}

It worth noting that the expected dimension is also the maximum dimension of the $k$-secant variety.

For qudits we take the projective variety $\mathfrak{X}$ to be the Segre variety and regarding Eq. \eqref{segre-SM} $\dim\Sigma^{n}_{\textbf{d}-\textbf{1}}=\sum_{i=1}^n(d_i-1)$.

\begin{conjecture}[Ref. \cite{AOP09}]
For a general multipartite quantum state in the Hilbert space $\mathcal{H}=\mathbbm{C}^{d_1}\otimes\cdots\otimes\mathbbm{C}^{d_n}$, the generic tensor rank, denoted by $\rk_{\text{gen}}$, is equal to the expected tensor rank. That is, the smallest $k$ in Eq. \eqref{exp-k-secant} such that $k$-secant variety fills the ambient space, i.e,
\begin{equation}
\rk_{\text{gen}}=\left\lceil\frac{\prod_{i=1}^n d_i}{\sum_{i=1}^n(d_i-1)+1}\right\rceil\,.
\end{equation}
\sloppy In this conjecture, there are exceptional cases: $\mathbbm{C}^{4\times{4}\times{3}}$, $\mathbbm{C}^{(2i+1)\times(2i+1)\times{3}}$, and $\mathbbm{C}^{(i+2)\times(i+2)\times{2}\times{2}}$, with $i\in\mathbbm{Z}^+$. In this exceptional cases, the generic tensor rank is equal to the expected plus one.
\end{conjecture}

So, based on this conjecture, it is also possible to classify entanglement in multipartite systems.


\begin{theorem}[ Alexander-Hirschowit (see Ref. \cite{AH95})]\label{theo:AH95}
The generic tensor rank of a symmetric tensor in ${\rm{Sym}}^n\mathbbm{C}^d$ is equal to the expected symmetric tensor rank which is
\begin{equation}\label{SymGenericRank}
\rk_{\text{gen}}=\left\lceil\frac{\binom{n+d-1}{n}}{d}\right\rceil\,,
\end{equation}
except for (i) $n=2$ where it is equal to $d$, and (ii) the pairs $(n,d)=(3,5)$, $(4,3)$, $(4,4)$, $(4,5)$ where the generic tensor rank is equal to the expected plus one \cite{AH95, BO08}.  Furthermore,
in these exceptional cases, all tensors of border rank at most $\brk=\rk_{\text{gen}}$ form a hypersurface in ${\rm{Sym}}^n\mathbbm{C}^d$ \cite{BFZ20}.
\end{theorem}

To obtain the dimension of the secants and tangents, one can utilize the following theorems \cite{FH79,Zak}.

\begin{theorem}[Fulton-Hansen (see Ref. \cite{FH79})]\label{theo:FH97}
Let $\mathfrak{X}$ be a projective algebraic variety of dimension $d$. Then, one of the following two properties holds
\begin{enumerate}
\item[(i)] $\dim \sigma_2(\mathfrak{X})=2d+1$ and $\dim \tau(\mathfrak{X})=2d$,
\item[(ii)] $\dim \sigma_2(\mathfrak{X})\leq2d$ and $\tau(\mathfrak{X})=\sigma_2(\mathfrak{X})$.
\end{enumerate}
\end{theorem}

\begin{theorem}[Zak (see Ref. \cite{Zak})]\label{theo:Zak}
Let $\mathfrak{X}\subset\mathbbm{P}^{D}$ be an irreducible nondegenerate (i.e., not contained in a hyperplane) $d_1$-dimensional projective variety. For an arbitrary nonempty irreducible $d_2$-dimensional variety $\mathfrak{Y}\subset\mathfrak{X}$ one of the following two properties holds
\begin{enumerate}
\item[(i)] ${\rm{dim}}~\mathfrak{J}(\mathfrak{X},\mathfrak{Y})=d_1+d_2+1 > {\rm{dim}}~\mathfrak{T}({X},\mathfrak{Y})=d_1+d_2$,
\item[(ii)] $\mathfrak{J}(\mathfrak{X},\mathfrak{Y})=\mathfrak{T}(\mathfrak{X},\mathfrak{Y})$.
\end{enumerate}
\end{theorem}

These two theorems provide us some information if we have the dimension of the $\sigma_2(\mathfrak{X})$ and the dimension of $\mathfrak{J}(\mathfrak{X},\mathfrak{Y})$, for the Fulton-Hansen and Zak theorems, respectively. Since the secant variety is a special case of join variety, we present the general form of the Terracini's lemma that tell us if varieties $\mathfrak{X}$ and $\mathfrak{Y}$ are of dimensions $d_1$ and $d_2$, respectively, then the expected dimension of join variety $\mathfrak{J}(\mathfrak{X},\mathfrak{Y})$ is $d_1+d_2+1$.

\begin{lemma}[Terracini's lemma]
Let $(x,y)\in(\mathfrak{X}\times\mathfrak{Y})_{\text{smooth}}$\footnote{If $\dim\mathfrak{T}_{p}\mathfrak{X}$ is locally constant near $p$, we say $p$ is a smooth point of $\mathfrak{X}$ and $\mathfrak{X}_{\text{smooth}}$ denotes the set of smooth points of $\mathfrak{X}$. So, $\mathfrak{X}_{\text{smooth}}$ is a complex manifold.} and $[z]=[x+y]\in\mathfrak{J}(\mathfrak{X},\mathfrak{Y})_{\text{smooth}}$. Then
\begin{equation}
\mathfrak{T}_{[z]}\mathfrak{J}(\mathfrak{X},\mathfrak{Y})=\mathfrak{T}_{[x]}\mathfrak{X}+\mathfrak{T}_{[y]}\mathfrak{Y}\,.
\end{equation}
\end{lemma}

Moreover, since the algebraic closure of the $\ell$-multirank is known to be the subspace variety \cite{Landsberg}, we have the following corollary.

\begin{corollary}\label{coro:MR-secant}
$\ell$-multiranks of a given tensor in the $k$-secant are at most $k$.
\end{corollary}

\begin{theorem}\label{k-secant-SLOCC}
The $k$-secant variety of Segre variety is invariant under the action of projective linear group and therefore is an SLOCC invariant.
\end{theorem}
{\it Proof.}
If the points of variety $\mathfrak{X}$ remains invariant under the action of a group $G$, then so is any of its auxiliary variety which is built from points of $\mathfrak{X}$.
\qed

That is why the Schmidt rank, which indeed is tensor rank, is a SLOCC invariant. On the other hand, since tangent lines can be seen as the limits of the secant lines, there exist asymptotic SLOCC equivalence between two different SLOCC classes and, hence, we can find exceptional states as defined in Ref. \cite{ST13}.

\subsection{Equations of higher secant varieties}

Regarding Eq. \eqref{segre-SM}, in order to distinguish the elements of higher secant varieties with the same $\ell$-multiranks, one can think about $\mathfrak{m}$ copies of the projective Hilbert space and utilize $\mathfrak{m}^{\rm{th}}$ Veronese embedding, i.e.,
\begin{equation}\label{m-veronese}
\mathcal{V}_{D}^{\mathfrak{m}}\colon\mathbbm{P}(\mathcal{H}_{n})\to\mathbbm{P}({\rm{Sym}}^{\mathfrak{m}}[\mathcal{H}_{n}])\,,
\end{equation}
where ${\rm{Sym}}^{\mathfrak{m}}[\mathcal{H}_{n}]$ is the $\mathfrak{m}^{\rm{th}}$ symmetric power of Hilbert space $\mathcal{H}_{n}$ (${\rm{Sym}}^{\mathfrak{m}}[\mathcal{H}_{n}]\sim{\rm{Sym}}[\mathcal{H}_{n}^{\otimes\mathfrak{m}}]$). According to this embedding, one can use minors of catalecticant matrices \cite{LO13}, to find the elements of higher secants.

Although, in principle, the minors of catalecticant matrices from Eq. (\ref{m-veronese}) provide us the invariant homogeneous polynomials\footnote{These provide some necessary but typically not all the polynomials that are needed.}, one can devise a more effective method. One of these, similar to the spirit of Ref. \cite{OS16}, could be based on projective invariants via an interpolation of representation theory \cite{Ottaviani13}. As we know, the vanishing locus of the minors of the catalecticant matrices are determinantal varieties and are invariant under the action of group $G={\rm{SL}}(d_{1},\mathbbm{C})\times\cdots\times{\rm{SL}}(d_{n},\mathbbm{C})$. Here, we should similarly provide homogeneous polynomials of degree $\mathfrak{m}$ which are invariant under the action of group $G$.
Given complex vector spaces $V_1\equiv\mathbbm{C}^{d_1},\ldots, V_n\equiv\mathbbm{C}^{d_n}$, the group $G$ acts over the tensor space ${\mathcal{H}}_{n}=\otimes_{i=1}^{n}V_i$ and, hence, on the polynomial ring,
\begin{equation}\label{polyring}
S=\sum_{\mathfrak{m}\ge 0}\mathrm{Sym}^{\mathfrak{m}}\left[{\mathcal{H}}_{n}\right]\,,
\end{equation}
where ${\mathcal{H}}_{n}^{\otimes\mathfrak{m}}\cong\left(V_{1}^{\otimes\mathfrak{m}}\right)\otimes\cdots\otimes\left(V_{n}^{\otimes\mathfrak{m}}\right)$. Since $G$ is a reductive group, every summand of degree $\mathfrak{m}$ of $S$ in Eq. (\ref{polyring}) decomposes as the sum of irreducible representations of $G$, which have the form $\otimes_{i=1}^{n}\mathfrak{S}_{\lambda_i}V_i$ for certain Young diagrams $\lambda_1,\ldots, \lambda_{n}$, each representation occurring with a multiplicity $m_{\lambda_{1}\cdots\lambda_{n}}$. When each $\lambda_i$ has a rectangular shape, with exactly $\dim{V_i}=d_{i}$ rows, all of the same length, we get that $\dim{\otimes_{i=1}^{n}\mathfrak{S}_{\lambda_i}V_i}=1$ and a generator of this space is known to be an invariant of degree $\mathfrak{m}$ and, indeed, all invariants occur in this way. In addition, these one-dimensional subspaces fill altogether the invariant subring $S^G$ of $S$, consisting of all invariant polynomials. It is known that such an invariant ring is finitely generated and in principle its generators and relations can be computed \cite{GoodmanWallach}. Note that the ideal of any $G$-invariant subvariety of the projective space $\mathbbm{P}(\mathcal{H}_{n})$, like the secant varieties, is generated by the generators of a finite number of summands of the form $\otimes_{i=1}^{n}\mathfrak{S}_{\lambda_i}V_i$. These subspaces are generally known as covariants, so an invariant is a covariant of dimension one, generated by a single $G$-invariant polynomial. A special case is given by codimension one $G$-invariant subvarieties of the projective space $\mathbbm{P}(\mathcal{H}_{n})$. Their ideal is principal and it is generated by a single invariant polynomial. Since the equations of any $k$-secant variety can be found among the $G$-covariants, which are invariant sets of polynomials, we give an explicit definition of a covariant and basic tools for constructing a complete set of covariants.

The $n$-partite state $|\psi\rangle$ in Eq. (\ref{n-partite}) can be interpreted as an $n$-linear form:
\begin{equation}\label{n-linear}
f({\bf{x}}^{1},\ldots,{\bf{x}}^{n})=\sum_{\alpha=1}^{n}\sum_{i_{\alpha}=0}^{d_{\alpha}-1}\mathbf{c}_{i_{1}\cdots i_{n}}{x}^{1}_{i_{1}}\cdots{x}^{n}_{i_{n}}\,.
\end{equation}
A covariant of $f$ is a multi-homogeneous $G$-invariant polynomial in the coefficients $\mathbf{c}_{i_{1}\cdots i_{n}}$ and the variables ${\bf{x}}^{\alpha}=\{{x}^{\alpha}_{i_{\alpha}}\}_{\alpha=1}^{n}$.
To construct covariants, we move on from Gour and Wallach \cite{GW13} who write all possible $\rm{SL}$ invariant polynomials for the action of $G$ over ${\mathcal{H}_n}$, following Schur-Weyl duality. Let $P_{d,\mathfrak{m}}$ denote the orthogonal projection of $\otimes^{\mathfrak{m}}\mathbbm{C}^d$ onto 
 $(\otimes^{\mathfrak{m}}\mathbbm{C}^d)^{{\rm{SL}}(d,\mathbbm{C})}$. Then, $P(v)=(P_{d_1,\mathfrak{m}}\otimes\cdots\otimes P_{d_n,\mathfrak{m}}(v^T))^T$, where $T$ stands for the intertwining map defined in Ref. \cite{GW13}, is the orthogonal projection from $\otimes^{\mathfrak{m}}{\mathcal{H}_n}$ to $(\otimes^{\mathfrak{m}}{\mathcal{H}_n})^G$. To compute $P_{d,\mathfrak{m}}$, first observe that it is zero if $\mathfrak{m}/d\notin{\mathbbm{Z}}$, while if $\mathfrak{m}=dr$ denote by $\chi_{d,r}$ the character of ${\mathfrak{S}}_\mathfrak{m}$ corresponding to the partition $\mathfrak{m}=r+\cdots+r$, and we get up to scalar multiples
\begin{equation}\label{inv-p}
P_{d,\mathfrak{m}}=\frac{d_{d,r}}{\mathfrak{m}!}\sum_{\pi\in{\mathfrak{S}}_{\mathfrak{m}}}\chi_{d,r}(\pi)\pi\,,
\end{equation}
where $d_{d,r}$ is the dimension of the irreducible representation corresponding to the partition $\mathfrak{m}=r+\cdots+r$ that can be calculated by the hook-length formula. This construction can be generalized to write all covariants of the above action, an invariant being a covariant of dimension $1$ as mentioned before. Every covariant of degree $\mathfrak{m}$ corresponds to $\otimes_{i=1}^{n}\mathfrak{S}_{\lambda_i}V_i$ for certain partitions $\lambda_i$ of $\mathfrak{m}$. Denoted by $\chi_{\lambda_i}$ the character of $\mathfrak{S}_{\mathfrak{m}}$ corresponding to the partition $\lambda_i$, we get again that up to scalar multiples,
\begin{equation}\label{cov-p}
P_{\lambda_i}=\frac{d_{\lambda_i}}{\mathfrak{m}!}\sum_{\pi\in{\mathfrak{S}}_{\mathfrak{m}}}\chi_{\lambda_i}(\pi)\pi\,,
\end{equation}
is the orthogonal projection from $\otimes^{\mathfrak{m}}V_{i}$ to the isotypical summand containing $\mathfrak{S}_{\lambda_i}V_{i}$, so the orthogonal projection from $\otimes^{\mathfrak{m}}{\mathcal{H}_n}$ to $\otimes_{i=1}^{n}\mathfrak{S}_{\lambda_i}V_i$ is $P(v)=(P_{\lambda_1}\otimes\cdots\otimes P_{\lambda_n}(v^T))^T$. The drawback of this construction is the difficulty to check in advance which linear combinations of the $P_{\lambda_i}$'s appear in a covariant of degree $\mathfrak{m}$, that is when $\otimes_{i=1}^{n}\mathfrak{S}_{\lambda_i}V_i$ comes from the subspace ${\mathrm{Sym}}^{\mathfrak{m}}[{\mathcal{H}_n}]\subset\otimes^{\mathfrak{m}}{\mathcal{H}_n}$,
this problem is known as plethysm. For example, the partition $4=2+1+1$ gives the projection in Eq. (\ref{cov-p}),
\begin{align}\nonumber
&v_1\otimes v_2\,\otimes \,v_3\otimes v_4 \mapsto \frac{1}{8}\Big(3\, v_1\otimes v_2\otimes v_3\otimes v_4-{\hspace{-2mm}\sum_{\pi\in(12)}}\hspace{-1mm}v_{\pi(1)}\otimes v_{\pi(2)}\otimes v_{\pi(3)}\otimes v_{\pi(4)}
\\ \nonumber
&\quad+{\hspace{-2mm}\sum_{\pi\in(1234)}}\!\!v_{\pi(1)}\otimes v_{\pi(2)}\otimes v_{\pi(3)}\otimes v_{\pi(4)}-{\hspace{-2mm}\sum_{\pi\in(12)(34)}}\!\!v_{\pi(1)}\otimes v_{\pi(2)}\otimes v_{\pi(3)}\otimes v_{\pi(4)}\Big),
\end{align}
where $(12)$ is the conjugacy class containing the six simple swaps and so on for the other conjugacy classes.

For the ``symmetric'' systems, there is also another well-known process in mathematics literature to construct the complete set of covariants. To interpolate physics and mathematics literatures, for a symmetric multiqubit system, the set of covariants is actually the set of joint covariants of binary forms and similarly for a symmetric multiqudit system, the set of covariants is the set of joint covariants of $d$-ary forms. A general method for constructing a complete set of covariants is known as transvectants, which are based on Cayley's omega process and are basic tools for this aim \cite{Olver}. Here, we give the procedure of creating transvectants for symmetric multiqudit systems [$d_{\alpha}=d$ for all $\alpha$ in Eq. (\ref{n-linear})]. Let functions $f_{1},\ldots,f_{d}$ be forms in variable ${\bf{x}}=(x_1,\ldots,x_d)$, and tensor product notation $f_1\otimes\cdots\otimes f_d$ denotes the polynomial $f_{1}({\bf{y}}_{1})\cdots f_{d}({\bf{y}}_{d})$ in $d$  independent variables (note that ${\bf{y}}_{\gamma}=(y_{\gamma,1},\ldots,y_{\gamma,d})$, $\gamma=1,\ldots,d$). The $d$-dimensional Cayley omega process is the $d^{\rm{th}}$-order partial differential operator:
\begin{equation}\label{cayley-omega}
\Omega_{\bf{x}}=\left|\begin{array}{ccc} \frac{\partial}{\partial{y_{1,1}}} & \cdots & \frac{\partial}{\partial{y_{d,1}}} \\ \vdots & \ddots & \vdots \\ \frac{\partial}{\partial{y_{1,d}}} & \cdots & \frac{\partial}{\partial{y_{d,d}}} \end{array}\right|\,.
\end{equation}
The $r^{\rm{th}}$ transvectant of functions $f_{1},\ldots,f_{d}$ is
\begin{equation}\label{transvectant}
\left(f_{1},\ldots,f_{d}\right)^{(r)}={\rm{tr}}~\Omega^{r}_{\bf{x}}(f_{1}\otimes\cdots\otimes f_{d})\,,
\end{equation}
where ${\rm{tr}}$ sets all variables equal, i.e., ${\bf{y}}_{1}=\cdots={\bf{y}}_{d}={\bf{x}}$. For instance, the first and second transvectants are known as the Jacobian determinant and polarized form of Hessian. Now, if functions $f_{1},\ldots,f_{d}$ are $n$-tuple forms in $n$ independent $d$-ary variables ${\bf{x}}^{1},\ldots,{\bf{x}}^{n}$, one can define a multiple transvectant for any $\vec{\jmath}=(j_{1},\ldots,j_{n})\in\mathbbm{N}^{n}$ as follows:
\begin{equation}\label{n-transvectant}
\left(f_{1},\ldots,f_{d}\right)^{(\vec{\jmath})}={\rm{tr}}\prod_{i=1}^{n}\Omega^{j_{i}}_{{\bf{x}}^{i}}(f_{1}\otimes\cdots\otimes f_{d})\,.
\end{equation}
By building iterative tansvectants in the multigraded setting and starting with the covariant of degree 1, i.e., Eq. (\ref{n-linear}), one can provide a complete system of covariants for multiqudit systems. For instance, in Ref. \cite{BLT03} the complete set of covariants has been found for four-qubit systems with this method.

\chapter{Fine-Structure Classification of Multiqubit Entanglement}\label{chap4}

\epigraph{``I like crossing the imaginary boundaries people set up between different fields.''}{\textit{Maryam Mirzakhani}}

In this chapter, which is based on the Ref. \cite{GMO20}, we present a fine-structure entanglement classification under stochastic local operation and classical communication (SLOCC) for multiqubit pure states. To this end, we employ specific algebraic-geometry tools that are SLOCC invariants, namely $k$-secant varieties, to show that for $n$-qubit systems there are $\lceil\frac{2^{n}}{n+1}\rceil$ entanglement families. By using another invariant, $\ell$-multilinear ranks, each family can be further split into a finite number of subfamilies. Not only does this method facilitate the classification of multipartite entanglement, but it also turns out to be operationally meaningful as it quantifies entanglement as a resource.

\section{Classification algorithm}\label{sec.4.1}
Let us consider an $n$-qubit state:
\begin{equation}\label{qubitstate}
|\psi\rangle = \sum_{i\in\{0,1\}^{n}}\mathfrak{c}_{i}|i\rangle\,.
\end{equation}
The space of states $|\psi\rangle$ that are fully separable has the structure of a Segre variety \cite{Miyake03, Heydari08} which is embedded in the ambient space as follows:
\begin{equation}\label{segre}
\Sigma^{n}_{\textbf{1}}:~\mathbbm{P}^{1}\times\mathbbm{P}^{1}\times\cdots\times\mathbbm{P}^{1}\hookrightarrow\mathbbm{P}^{2^{n}-1}\,,
\end{equation}
where $\textbf{1}=(1,\ldots,1)$ and $\times$ is the Cartesian product of sets.
A $k$-secant of the Segre variety joins its $k$ points, each of which represents a possibly distinct separable state. Thus, the joining of points corresponds to an entangled state being a superposition of separable states. The closure of the union of $k$-secants of the Segre variety $\Sigma^{n}_{\textbf{1}}$ gives rise to the $k$-secant variety $\sigma_k(\Sigma^{n}_{\textbf{1}})$. This is as much as the set of entangled states arising from the superposition of $k$ separable states.
Since $k$-secant varieties are SLOCC invariants (see \cref{k-secant-SLOCC}), SLOCC classes congregate naturally into entanglement families. Therefore, the dimension of the highest $k$-secant, which fills the projective Hilbert space of $n$ qubits, can indicate the number of entanglement families in a coarse classification by border rank.
The $k$-secant varieties in $\mathbbm{P}(\otimes^n\mathbbm{C}^2)$, have the expected dimension
\begin{equation}
{{\rm{dim}}}~\sigma_{k}(\Sigma^{n}_{\textbf{1}})= \min\{k(n+1)-1, 2^{n}-1\}\,,
\end{equation}
for every $k$ and $n$, except $\sigma_{3}(\Sigma^{4}_{\textbf{1}})$ which has dimension $13$ \cite{CGG11}. Consequently, the $k$-secant fills the ambient space, when
\begin{equation}
k=\left\lceil\frac{2^{n}}{n+1}\right\rceil\,.
\end{equation}
This $k$ indicates the number of entanglement families and remains finite (although growing exponentially) with the number of qubits.

The proper $k$-secant (the states that belongs to $k$-secant but not to ($k-1$)-secant), i.e., the set $\sigma_{k}(\Sigma^{n}_{\textbf{1}})\setminus\sigma_{k-1}(\Sigma^{n}_{\textbf{1}})$, is the union of the $k$-secant hyperplanes $\mathcal{S}_{k}\subset\sigma_{k}(\Sigma^{n}_{\textbf{1}})$ represented by
\vspace{-3mm}
\begin{equation}\label{G-secant}
\mathcal{S}_{k}=\sum_{i=1}^{k}\lambda_{i}p_{i}\,,
\end{equation}
with $\{\lambda_{i}\}_{i=1}^{k}\neq{0}$ and each $p_{i}$ is an independent point in $\Sigma^{n}_{\textbf{1}}$.

It is worth saying that each secant, with regards to its dimension, could have tangents as its closure (see \cref{theo:Zak}) which discriminate subfamilies with the same $\ell$-multiranks and provide us exceptional states \cite{ST13}. Let us now consider the limits of secants to obtain the tangents. Let $(i_{1},i_{2},\ldots,i_{k})$ be a rearrangement of points indices in Eq. (\ref{G-secant}). The first limit type is when one point tends to another one, i.e., $p_{i_{2}}\to{p_{i_{1}}}$, and let us call the result $p'_{i_{1}}$. The second limit type can be considered as the closure of the first limit type so the third point is  approaching $p_{i_{1}}+\eta p'_{i_{1}}$. The third limit type can be considered as the closure of the second limit type so two points tend to $p_{i_{1}}$ and $p_{i_{2}}$ (if the join of $p_{i_{1}}$ and $p_{i_{2}}$ is still in $\Sigma^{n}_{\textbf{1}}$) \cite{BL14}. As we can always redefine Eq. (\ref{G-secant}) to have the desired form and new coefficients rather than $\lambda_{j}$, we can formulate these limits as
\vspace{-2mm}
\begin{align}\label{G-tangent-g1}
T^{(1)}_{k}=&\lim_{\epsilon\to{0}}\frac{\lambda_{i_{2}}}{\epsilon}\big(p_{i_{2}}(\epsilon)-p_{i_{1}}\big)+\sum_{j=i_{3}}^{i_{k}}\lambda_{j}p_{j}\,,
\\ \label{G-tangent-g2}
T^{(2)}_{k}=&\mu_{1}p'_{i_{1}}+\lim_{\eta\to{0}}\frac{\mu_{2}}{\eta^{2}}\big(p_{i_{3}}(\eta)-(p_{i_{1}}+\eta\, p'_{i_{1}})\big)+\sum_{j=i_{4}}^{i_{k}}\lambda_{j}p_{j}\,,
\\ \label{G-tangent-g3}
T^{(3)}_{k}=&\lim_{\epsilon\to{0}}\frac{\nu_{1}}{\epsilon}\big(p_{i_{3}}(\epsilon)-p_{i_{1}}\big)+\lim_{\epsilon\to{0}}\frac{\nu_{2}}{\epsilon}\big(p_{i_{4}}(\epsilon)-p_{i_{2}}\big)+\sum_{j=i_{4}}^{i_{k}}\lambda_{j}p_{j}\,.
\end{align}
These processes can be generalized if we consider all extra limit types which may occur by adding the next points \cite{LT10}. This will provide us with higher tangential varieties, although these are not all of the types of limits that one can consider.

On the other hand, using ordered $\ell$-tuples $I=(i_{1},i_{2},\ldots,i_{\ell})$, where $1\leq\ell\leq\lfloor\frac{n}{2}\rfloor$ (see Section \ref{sec.flattening}), and $(n-\ell)$-tuples $\bar{I}$ such that $I\cup\bar{I}=(1,2,\ldots,n)$, we can flatten the $n$-fold tensor product Hilbert space $\mathcal{H}_{n}=\otimes_{i=1}^{n}\mathbbm{C}^{2}$ to two-fold tensor product Hilbert spaces with higher dimensions, i.e.,
\begin{equation}
\mathcal{H}_{n}\simeq\mathcal{H}_{I}\otimes\mathcal{H}_{\bar{I}}\,,
\qquad
\mathcal{H}_{I}=\mathbb{C}^{2^{\ell}}\,,
\quad
\mathcal{H}_{\bar{I}}=\mathbb{C}^{2^{n-\ell}}\,.
\end{equation}
Using Dirac notation, the flattening of $|\psi\rangle$ reads
\begin{equation}
\mathcal{M}_{I}[\psi]=\left(\langle e_0|\psi\rangle,\langle e_2|\psi\rangle, \ldots,\langle e_{2^{\ell}-1}|\psi\rangle\right)^{\rm{T}}\,,
\end{equation}
where $|e_j\rangle\equiv|j\rangle$, with $j\in\{0,1\}^\ell$, are the canonical basis of $\mathcal{H}_{I}$. For the $n$-qubit systems, the order of such matrices can be from $2 \times 2^{n-1}$ to $2^{\lfloor\frac{n}{2}\rfloor} \times 2^{\lceil\frac{n}{2}\rceil}$ and the number of these matrices ranges from ${\binom{n}{1}}$ to $(1/2)^{n+1~{\rm{mod}}~2} {\binom{n}{\lfloor\frac{n}{2}\rfloor}}$. The rank of these matrices give us $\ell$-multiranks. Since $\ell$-multiranks are also SLOCC invariants (see \cref{l-multirank-SLOCC}), the SLOCC classes in each family can be grouped into subfamilies.

Therefore, we use $k$-secant varieties and $\ell$-multiranks as the SLOCC invariants to group orbits (classes) into finite number of families and subfamilies. In addition, one can split $k$-secant families, according to \cref{theo:Zak}, by identifying their closure as $k$-tangent. Hence, the classification algorithm can be summarized as follows:
\begin{enumerate}
\item[\bf(i)] find families by identifying $\Sigma^{n}_{\textbf{1}}$, $\sigma_{2}(\Sigma^{n}_{\textbf{1}}), \ldots, \sigma_{k}(\Sigma^{n}_{\textbf{1}})$, 

\item[\bf(ii)] split families to secants and tangents by identifying $\tau_{2}(\Sigma^{n}_{\textbf{1}}), \ldots, \tau_{k}(\Sigma^{n}_{\textbf{1}})$,

\item[\bf(iii)] find subfamilies by identifying $\ell$-multiranks.
\end{enumerate}

\section{Case study}\label{sec.4.2}

\subsection{Two-qubit entanglement}\label{subsec.4.2.1}
The classification of two-qubit states is fairly trivial, nonetheless it can be instructive for working out the developed concepts. For the Segre surface $\Sigma^{2}_{\textbf{1}}$, we shall use homogeneous coordinates associated with the induced basis $\left\{|00\rangle,|01\rangle,|10\rangle,|11\rangle\right\}$. That is to say, a point $p\in\mathbbm{P}^3$ is written in homogenous coordinates $\left[\mathfrak{c}_0:\mathfrak{c}_1:\mathfrak{c}_2:\mathfrak{c}_3\right]$ whenever $p$ is the projective class of the equivalent two-qubit state of Eq. (\ref{qubitstate}). Then, the Segre surface $\Sigma^{2}_{\textbf{1}}$ is the projective variety with points in an open set given by affine coordinates $[1:a:b:ab]$, where $a$ and $b$ are complex parameters. This expression must be properly understood, in that the limits of $a$ and/or $b$ going to  infinity, must be included.

Let us now move on to the proper two-secant variety, i.e., the set $\sigma_2(\Sigma^2_{\mathbf{1}})\setminus\Sigma^2_{\mathbf{1}}$, which is the union of the secant planes $\mathcal{S}_2$ represented by Eq. \eqref{G-secant}. Hence, the proper two-secant variety is given by $\sigma_{2}=[\lambda_{1}+\lambda_{2}:\lambda a_{1}+\lambda_{2} a_{2}:\lambda_{1} b_{1}+\lambda_{2} b_2:\lambda_{1} a_{1}b_{1}+\lambda_{2} a_{2}b_{2}]$. It is easy to see that
\begin{align}\nonumber
&|\Phi^{\pm}\rangle=[1:0:0:\pm 1]\in\sigma_{2}(\Sigma^{2}_{\textbf{1}})\,, \\
&|\Psi^{\pm}\rangle=[0:1:\pm 1:0]\in\sigma_{2}(\Sigma^{2}_{\textbf{1}})
\end{align}
that are well-known Bell states.

To create the closure of the two-secant, let consider $p_{2}(\epsilon)=[1:a_{1}+\epsilon:b_{1}+\epsilon:(a_{1}+\epsilon)(b_{1}+\epsilon)]$. Using Eq. (\ref{G-tangent-g1}) we have the special situation that all points on the tangent lines $T^{(1)}_{2}$ lie also on two-secants. Since
\begin{equation}
T^{(1)}_{2}=\lim_{\epsilon\to{0}}\frac{\lambda}{\epsilon}\big(p_2(\epsilon)-p_1\big)=[0:\lambda:\lambda:\lambda(a_1+b_1)],
\end{equation}
and
\begin{equation}
T^{(1)}_{2}=[1:a:b:ab]-[1:c:d:cd],
\end{equation}
with $a=a_1+\frac{\lambda}{2}$, $b=b_1+\frac{\lambda}{2}$, $c=a_1-\frac{\lambda}{2}$, and $d=b_1-\frac{\lambda}{2}$. It means that all elements of $\mathbbm{P}^{3}$ are elements of $\sigma_{2}(\Sigma^{2}_{\textbf{1}})$.

One can thus conclude that all entangled states of two qubits are linear combinations of two independent separable states, which is the same result obtainable by the Schmidt decomposition. Here, the two entanglement families coincide with the two SLOCC classes, namely, separable and entangled (see Table \ref{table:2qubit}).

\begin{table}[t]
\centering
\caption[Fine-structure classification of two-qubit entanglement]{\label{table:2qubit} Fine-structure classification of two-qubit entanglement.}
\begin{tabularx}{\linewidth}{XX}
\hline\hline
& \\ [-2ex]
$\Sigma^{3}_{\textbf{1}}$ & $\sigma_{2}$ \\ [0.5ex]
\hline
& \\ [-2ex]
$|{\rm{Sep}}\rangle$ & $|{\rm{Bell}}\rangle$ \\ [0.5ex]
\hline\hline
\end{tabularx}
\end{table}

Already from this example we can draw a general conclusion. That is, for $n\geq 2$ we have
\begin{equation}\label{general2}
\mathfrak{p}\{|{\rm{Bell}}\rangle|q_2
\rangle^{\otimes{(n-2)}}\}\in\sigma_{2}(\Sigma^{n}_{\textbf{1}})\,,
\end{equation}
where $\mathfrak{p}\{\cdot\}$ denotes all possible permutations and $|q_2\rangle$ is a general one-qubit state.

\subsection{Three-qubit entanglement}\label{subsec.4.2.2}
For three qubits the Segre three-fold $\Sigma^{3}_{\textbf{1}}\subset\mathbbm{P}^{7}$ consists of general points given by afﬁne coordinates $[1:a:b:ab:c:ac:bc:abc]$ with the possibility of $a$ and/or $b$ and/or $c$ going to infinity.

Moving on to the proper two-secant variety, i.e., $\sigma_2(\Sigma^3_{\mathbf{1}})\setminus\Sigma^3_{\mathbf{1}}$, we have generic elements as $[\lambda_{1}+\lambda_{2}:\lambda_{1} a_{1}+\lambda_{2} a_{2}:\lambda_{1} b_{1}+\lambda_{2} b_{2}:\lambda_{1} a_{1}b_{1}+\lambda_{2} a_{2}b_{2}: \lambda_{1} c_{1}+\lambda_{2} c_{2}:\lambda_{1} a_{1}c_{1}+\lambda_{2} a_{2}c_{2}:\lambda_{1} b_{1}c_{1}+\lambda_{2} b_{2}c_{2}:\lambda_{1} a_{1}b_{1}c_{1}+\lambda_{2} a_{2}b_{2}c_{2}]$. One can check that
\begin{equation}
|{\rm{GHZ}}_{3}\rangle=[1:0:0:0:0:0:0:1]\,,
\end{equation}
is an elelemnt of $\sigma_{2}(\Sigma^{3}_{\textbf{1}})$.

We also need to consider situations in which one or more parameters in the coordinate of proper two-secant tend to infinity. As an example, let us take $a_{1}=b_{1}=\sqrt{\lambda_{2}}\to\infty$ with $c_{1}=c_{2}$, which gives the biseparable state
\begin{equation}
|{\rm{B}}_{A-BC}\rangle=[1:a:b:c:d:ad:bd:cd]\,,
\end{equation}
is an element of $\sigma_{2}(\Sigma^{3}_{\textbf{1}})$.

Hence, the state $|{\rm{GHZ}}_{3}\rangle$ with one-multirank equal to $(222)$ and all three biseparable states $|{\rm{B}}_{i}\rangle_{i=1}^{3}$ with the same form as Eq. (\ref{general2}) and one-multiranks equal to $(122)$, up to a permutation, are elements of $\sigma_{2}(\Sigma^{3}_{\textbf{1}})$.

However, the tangent points defined in Eq. \eqref{G-tangent-g1} cannot be expressed as elements of $\sigma_{2}(\Sigma^{3}_{\textbf{1}})$, which spans all $\mathbbm{P}^{7}$ only if the tangential variety is included as its closure. If we consider the tangent to $p_{1}=[1:0:0:0:0:0:0:0]$ (equivalent to all points on $\Sigma^{3}_{\textbf{1}}$ by an SLOCC), we have $T^{(1)}_{2}=[1:\lambda:\lambda:0:\lambda:0:0:0]$.
For instance,
\begin{equation}
|{\rm{W}}_{3}\rangle=\lim_{\lambda\to\infty}T^{(1)}_{2}=[0:1:1:0:1:0:0:0]\,,
\end{equation}
with one-multirank equal to $(222)$, is an element of $\tau_{2}(\Sigma^{3}_{\textbf{1}})$.

We saw that one-multirank equal to $(222)$ can be discriminated by secant and/or tangent classification. From now on, we use a prime for the states in tangent to discriminate secant and tangent families where they have same $\ell$-multiranks.

In summary, this classification provides us two secant families (three secant / tangent families), and six subfamilies (Table \ref{table:3qubit}, see also Ref. \cite[Example 14.4.5]{GKZ}) that coincide with the six SLOCC classes of Ref. \cite{DVC00}.

\begin{table}[t]
\centering
\caption[Fine-structure classification of three-qubit entanglement]{\label{table:3qubit} Fine-structure classification of three-qubit entanglement. Each column corresponds to a family ($\tau_2$ is the closure of $\sigma_2$). Within a column, each row corresponds to a subfamily.}
\begin{tabularx}{\linewidth}{XXX}
\hline\hline
& & \\ [-2ex]
$\Sigma^{3}_{\textbf{1}}$ & $\sigma_{2}$ & $\tau_{2}$ \\ [0.5ex]
\hline
& & \\ [-2ex]
$|{\rm{Sep}}\rangle$ & $|{\rm{GHZ}}_{3}\rangle$ & $|{\rm{W}}_{3}\rangle$ \\ [0.5ex]
& $|{\rm{B}}_{i}\rangle_{i=1}^{3}$ &  \\ [0.5ex]
\hline\hline
\end{tabularx}
\end{table}
\pagebreak
Also from this example we can extrapolate general results. That is, for $n\geq{r}\geq{3}$, we have
\begin{align}\label{general3}\nonumber
\hspace{8mm}&|{\rm{GHZ}}_{n}\rangle=|0\rangle^{\otimes{n}}+|1\rangle^{\otimes{n}} &\in \sigma_{2}(\Sigma^{n}_{\textbf{1}})\,,\hspace{8mm} \\ \nonumber
&\mathfrak{p}\{|{\rm{GHZ}}_{r}\rangle|q_2\rangle^{\otimes{(n-r)}}\} &\in \sigma_{2}(\Sigma^{n}_{\textbf{1}})\,,\hspace{8mm} \\ \nonumber
&|{\rm{W}}_{n}\rangle=|{\rm{D}}_{n}^{1}\rangle &\in \tau_{2}(\Sigma^{n}_{\textbf{1}})\,,\hspace{8mm} \\
&\mathfrak{p}\{|{\rm{W}}_{r}\rangle|q_2\rangle^{\otimes{(n-r)}}\} &\in \tau_{2}(\Sigma^{n}_{\textbf{1}})\,,\hspace{8mm}
\end{align}
where
\begin{equation}\label{Dicke-n-qubit}
|{\rm{D}}_{n}^{l}\rangle={\binom{n}{l}}^{-(1/2)}\sum_{i}\mathfrak{p}_{i}\{|0\rangle^{\otimes{(n-l)}}\otimes|1\rangle^{\otimes{l}}\}\,,
\end{equation}
are the so-called Dicke states (with $l$ excitations).

\subsection{Four-qubit entanglement}\label{subsec.4.2.3}
Due to \cref{remark-MR-GE} and \cref{coro:MR-secant} in Chapter \ref{chap3} and classification of two- and three-qubit states, we have
\begin{enumerate}
\item All triseparable states $|{\rm{T}}_{i}\rangle_{i=1}^{6}$ from Eq. (\ref{general2}) are elements of $\sigma_{2}(\Sigma^{4}_{\textbf{1}})$.
\item All biseparable states $|{\rm{B}}_{i}^{{\rm{GHZ}}_{3}}\rangle_{i=1}^{4}$ and $|{\rm{B}}_{i}^{{\rm{W}}_{3}}\rangle_{i=1}^{4}$ from Eq. (\ref{general3}) are, respectively, elements of $\sigma_{2}(\Sigma^{4}_{\textbf{1}})$ and $\tau_{2}(\Sigma^{4}_{\textbf{1}})$.
\item The states $|{\rm{GHZ}}_{4}\rangle$ and $|{\rm{W}}_{4}\rangle$ are elements of $\sigma_{2}(\Sigma^{4}_{\textbf{1}})$ and $\tau_{2}(\Sigma^{4}_{\textbf{1}})$, respectively.
\end{enumerate}

The rest of the subfamilies of four-qubit states can be identified by considering the elements of three- and four-secants and their closures.

The proper three-secant, i.e., the set $\sigma_{3}(\Sigma^{4}_{\textbf{1}})\setminus\sigma_{2}(\Sigma^{4}_{\textbf{1}})$, is the union of the secant hyperplanes $\mathcal{S}_{3}$ represented by Eq. (\ref{G-secant}). For instance,
\begin{equation}
|{\rm{M}}_4\rangle=|0000\rangle+|1111\rangle+\mathfrak{p}\{|0011\rangle\}\,,
\end{equation}
which comes from joining $|{\rm{GHZ}}_{4}\rangle$ and an independent element of $\Sigma^{4}_{\textbf{1}}$ is and element of $\sigma_{3}(\Sigma^{4}_{\textbf{1}})$ with two-multirank equal to a permutation of (233).

Using Eq. \eqref{Dicke-n-qubit} one can see that four-qubit Dicke state with two excitations
\begin{equation}\label{D-4-2}
|{\rm{D}}_{4}^{2}\rangle=\frac{1}{\sqrt{6}}(|0011\rangle+|0101\rangle+|0110\rangle+|1001\rangle+|1010\rangle+|1100\rangle)\,,
\end{equation}
is an element of $\sigma_{3}(\Sigma^{4}_{\textbf{1}})$ with two-multirank equal to $(333)$. This is because we can relate the above-mentioned symmetric state to the monomial $x^2y^2$and we can decompose this monomial as follows:
\begin{equation}
x^2y^2=\frac{1}{18\omega} \big(\omega^4(x+y)^4+(x+\omega^2y)^4+(\omega^2x+y)^4\big)\,,
\end{equation}
where $\omega$ is the nonreal cube root of unity. So, using the Dirac notation, we can rewrite the state $|{\rm{D}}_{4}^{2}\rangle$ in Eq. \eqref{D-4-2} based on the above decomposition as follows:
\begin{equation}
|{\rm{D}}_4^2\rangle=\frac{1}{3\omega}\big(\omega^4(|0\rangle+|1\rangle)^{\otimes 4}+(|0\rangle+\omega^2|1\rangle)^{\otimes 4}+(\omega^2|0\rangle+|1\rangle)^{\otimes 4}\big)\,.
\end{equation}

To construct the closure of $\sigma_{3}$, we consider different limit types as in Eqs. (\ref{G-tangent-g1})-(\ref{G-tangent-g3}) at $p_{1}=[1:0:\cdots:0]$, equivalent to all points on $\Sigma^{4}_{\textbf{1}}$ by an SLOCC. Then,
\begin{enumerate}
\item Regarding the first limit type, i.e., Eq. (\ref{G-tangent-g1}), the following states
\begin{equation}
|{\rm{W}}_{4}\rangle+|1111\rangle\,,
\end{equation}
and
\begin{equation}
|{\rm{W}}_{4}\rangle+\mathfrak{p}\{|0011\rangle\}\,,
\end{equation}
are elements of $\tau_{3}(\Sigma^{4}_{\textbf{1}})$ with two-multirank equal to $(333)$ and a permutation of $(233)$, respectively.

\item Regarding the second limit type, i.e., Eq. (\ref{G-tangent-g2}), the following state
\begin{equation}
|{\rm{W}}_{4}\rangle+|{\rm{D}}_{4}^2\rangle\,,
\end{equation}
is an element of $\tau_{3}(\Sigma^{4}_{\textbf{1}})$ with two-multirank equal to $(333)$.

\item For the third limit type (Eq. (\ref{G-tangent-g3})), one can take $p_{1}=[0:1:0:\cdots:0]$ as the second point, where $\lambda_{1}p_{1}+\lambda_{2}p_{2}\in\Sigma^{4}_{\textbf{1}}$ and hence
\begin{equation}
|{\rm{W}}_{4}\rangle+\alpha|0011\rangle+\beta|0101\rangle+\gamma|1001\rangle\,,
\end{equation}
is an element in $\tau_{3}(\Sigma^{4}_{\textbf{1}})$ with two-multirank equal to $(333)$.
\end{enumerate}
We denote the union of these points as the tangential variety $\tau_{3}(\Sigma^{4}_{\textbf{1}})$.

An important observation is that, all elements in the three-secant variety are genuinely entangled. This can be useful for characterizing genuine multilevel entanglement when we look at four qubits as two ququarts \cite{KRBHG18}.

The proper four-secant, i.e., the set $\sigma_{4}(\Sigma^{4}_{\textbf{1}})\backslash\sigma_{3}(\Sigma^{4}_{\textbf{1}})$, is the union of the secant hyperplanes $\mathcal{S}_{4}$ represented by Eq. (\ref{G-secant}). For instance, all biseparable states
\begin{equation}
|{\rm{BB}}_{i}\rangle_{i=1}^{3}=|{\rm{Bell}}\rangle_{j}|{\rm{Bell}}\rangle_{k}\,,
\end{equation}
which are tensor products of two Bell states of different parts, that is $j\in\{12,13,14\}$ and k is the complementary such that $j\cup k=1234$, is an element in $\sigma_{4}(\Sigma^{4}_{\textbf{1}})$ with two-multirank equal to $(144)$, up to a permutation.

In the proper three-secant we have other elements with different two-multiranks. For instance,
\begin{align}\nonumber
&|{\rm{Cl}}_{4}^{(1)}\rangle=\frac{1}{2}(|0000\rangle+|0011\rangle+|1100\rangle-|1111\rangle)\,, \\ \nonumber
& |{\rm{Cl}}_{4}^{(2)}\rangle=\frac{1}{2}(|0000\rangle+|0101\rangle+|1010\rangle-|1111\rangle)\,, \\
& |{\rm{Cl}}_{4}^{(3)}\rangle=\frac{1}{2}(|0000\rangle+|0110\rangle+|1001\rangle-|1111\rangle)\,,
\end{align}
which are known as cluster states \cite{BR01}, and are elements in $\sigma_{4}(\Sigma^{4}_{\textbf{1}})$ with two-multirank equal to $(244)$, $(424)$ and $(442)$, respectively. Also, the following states
\begin{align}\nonumber
&|{\rm{Cl}}_{4}^{(1)}\rangle+|0101\rangle\,, \\ \nonumber
&|{\rm{Cl}}_{4}^{(2)}\rangle+|0110\rangle\,, \\
&|{\rm{Cl}}_{4}^{(3)}\rangle+|0011\rangle\,,
\end{align}
are elements in $\sigma_{4}(\Sigma^{4}_{\textbf{1}})$ with two-multirank equal to $(344)$, $(434)$ and $(443)$, respectively.

Since the highest tensor rank for a four-qubit state is 4 \cite{Brylinski02}, we do not need to construct the four-tangent. Therefore, any general state of four-qubit system with two-multirank equal to $(444)$ can be considered as an element of four-secant family.

To have an exhaustive classification, we have written each subfamily of three- and four-secant families in terms of their two-multiranks in Table \ref{table:4qubit}. Also, we have used a prime for the states in tangent to discriminate secant and tangent families where they have same one- and two-multiranks.

Briefly, this classification provide us four secant families (six secant/tangent families), and $35$ subfamilies (Table \ref{table:4qubit}). The petal-like classification of SLOCC orbits is presented in Fig. \ref{fig:Flower4}.

\begin{table}[t]
\centering
\caption[Fine-structure classification of four-qubit entanglement]{\label{table:4qubit} Fine-structure classification of four-qubit entanglement.}
\begin{tabularx}{\linewidth}{XXXXXX}
\hline\hline
& & & & & \\ [-2ex]
$\Sigma^{4}_{\textbf{1}}$ & $\sigma_{2}$ & $\tau_{2}$ & $\sigma_{3}$ & $\tau_{3}$ & $\sigma_{4}$\\ [0.5ex]
\hline
& & & & & \\ [-2ex]
$|{\rm{Sep}}\rangle$ & $|{\rm{GHZ}}_{4}\rangle$ & $|{\rm{W}}_{4}\rangle$ & $|(333)\rangle$ & $|(333)'\rangle$ & $|(444)\rangle$\\ [0.5ex]
& $|{\rm{B}}_{i}^{{\rm{GHZ}}_{3}}\rangle_{i=1}^{4}$ & $|{\rm{B}}_{i}^{{\rm{W}}_{3}}\rangle_{i=1}^{4}$ & $|(332)\rangle$ & $|(332)'\rangle$
 & $|(443)\rangle$\\ [0.5ex]
& $|{\rm{T}_{i}}\rangle_{i=1}^{6}$ & & $|(323)\rangle$ & $|(323)'\rangle$
 & $|(434)\rangle$\\ [0.5ex]
& & & $|(233)\rangle$ & $|(233)'\rangle$
 & $|(344)\rangle$\\ [0.5ex]
& & & & & $|(442)\rangle$\\ [0.5ex]
& & & & & $|(424)\rangle$\\ [0.5ex]
& & & & & $|(244)\rangle$\\ [0.5ex]
& & & & & $|{\rm{BB}}_{i}\rangle_{i=1}^{3}$\\ [0.5ex]
\hline\hline
\end{tabularx}
\end{table}

\begin{figure}[H]
\center{\includegraphics[width=10cm]{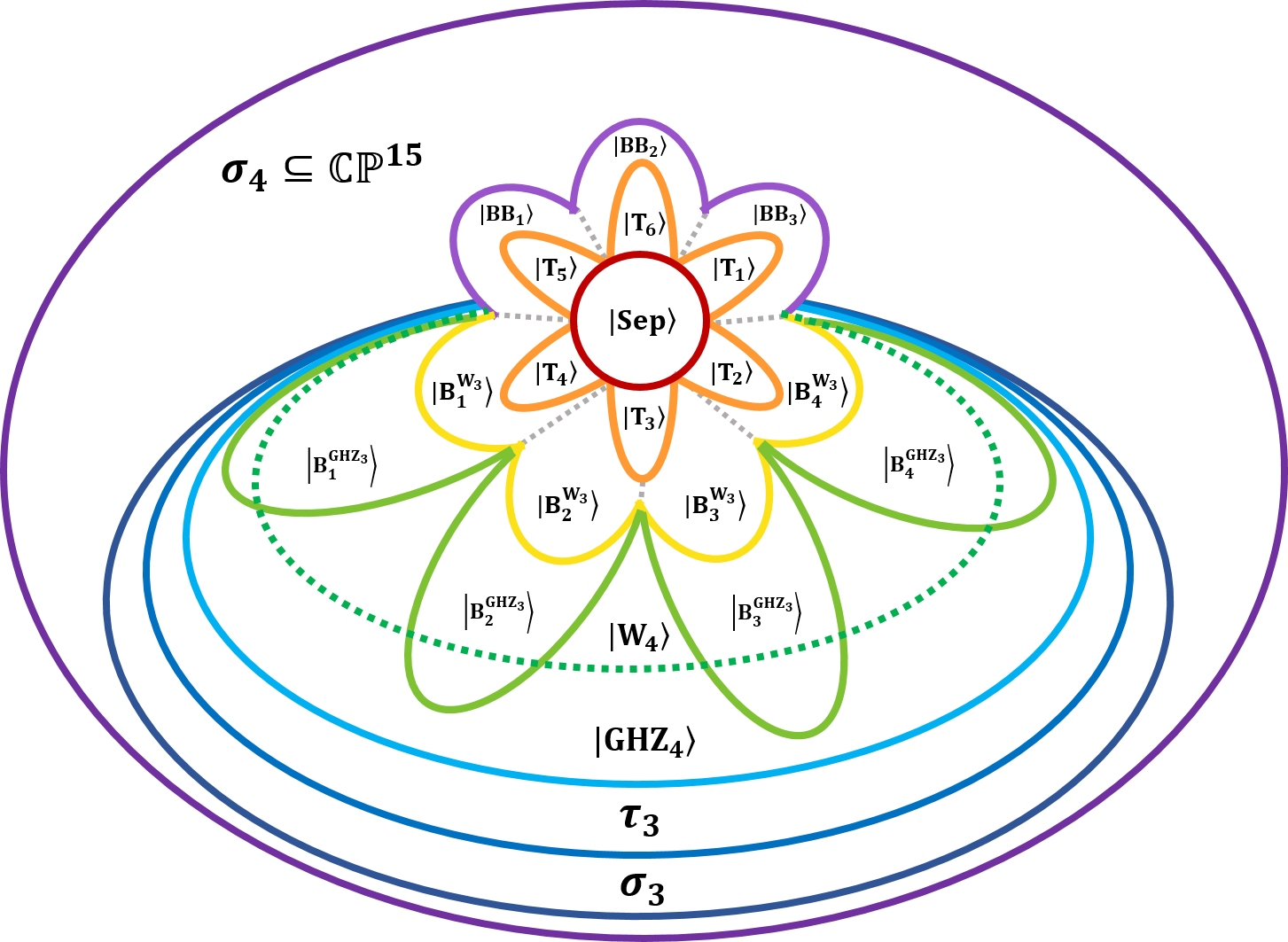}}
\caption[Petal-like classification of four-qubit entanglement]{\label{fig:Flower4} Petal-like classification of SLOCC orbits of four-qubit states. Dashed gray lines in the core show that each $|{\rm{BB}}_{i}\rangle$ encompasses two triseparable subfamilies, while each $|{\rm{B}}_{i}^{{\rm{W}}_{3}}\rangle$ encompasses three triseparable subfamilies. The convex hull of $|{\rm{W}}_{4}\rangle$ (dashed green curve) indicates that this family does not encompass biseparable states $|{\rm{B}}_{i}^{{\rm{GHZ}}_{3}}\rangle$, while both encompass the yellow, orange, and red subsets. From the outer classes, one can go to the inner ones by noninvertible SLOCC (from $\sigma_k$ to $\tau_k$ also in an approximate way), thus generating the entanglement hierarchy. (See Fig. \ref{fig:Hasse4} in the following subsection for more details.)}
\end{figure}

\subsubsection{Much ado about two-multiranks for four-qubit systems}\label{subsubsec.4.3.3.1}
Carlini and Kleppe have classified all possible one-multiranks for any number of qudits \cite{CK11}. The case of two-multiranks is more subtle. The partial result of two-multiranks of four-qubit states which is related to the Fig. \ref{fig:Flower4} can be seen in Hasse diagram in Fig. \ref{fig:Hasse4}. A partial classification was given classically in Ref. \cite{Segre}, where the case $(442)$ and its permutations were forgotten. The full classification is achieved by the following

\begin{figure}[t]
\center{\includegraphics[width=9cm]{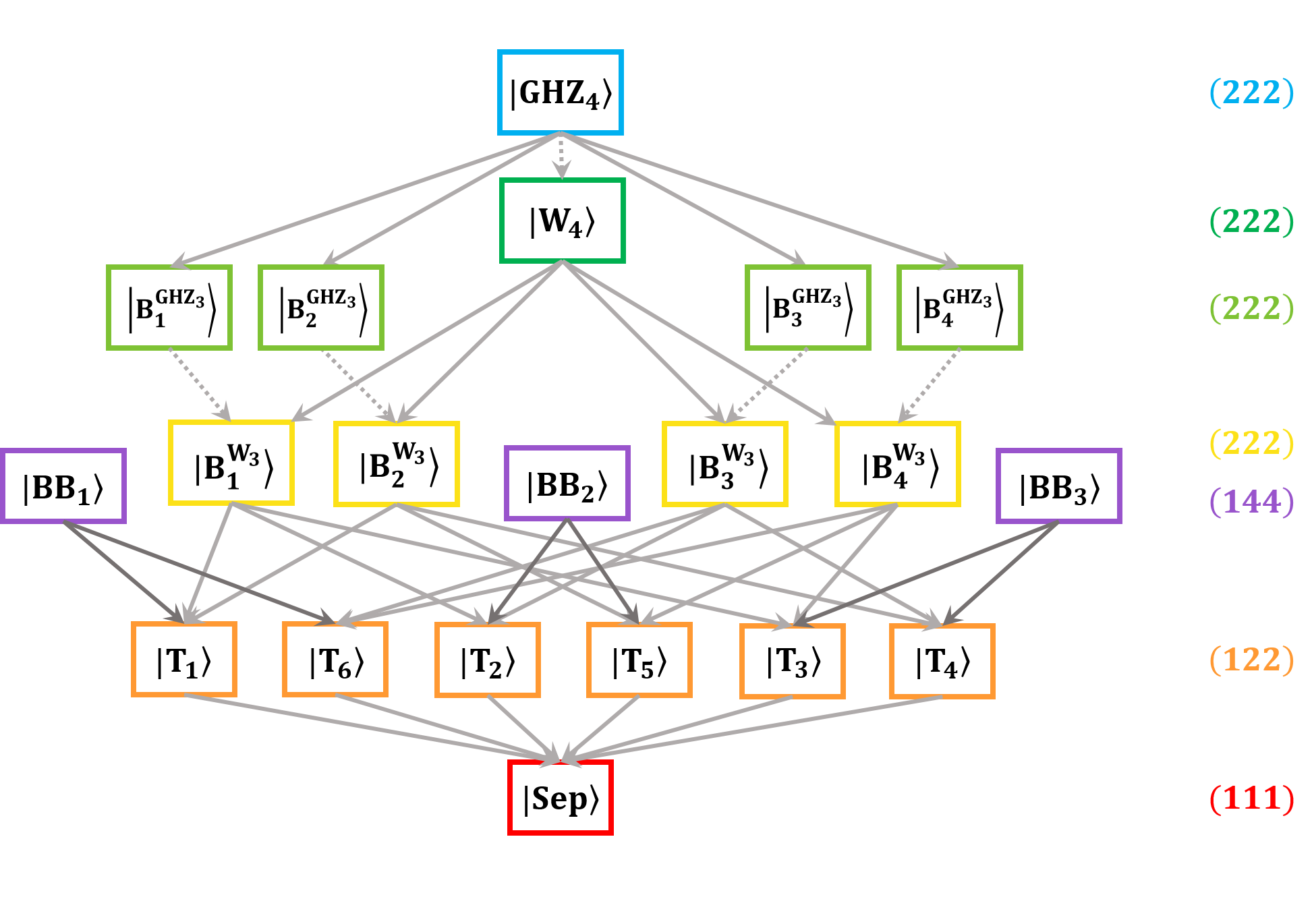}}
\caption[Hasse diagram SLOCC classification of four-qubit]{\label{fig:Hasse4} Hasse diagram of the central SLOCC classification of four-qubit states and their corresponding two-multiranks. The arrows denote noninvertible SLOCC transformations. When the arrow is dashed, the transformation is also approximated.}
\end{figure}

\begin{theorem}\label{theo:two-multirank}
(i) For any four-qubit system, the maximum among the three two-multiranks is attained at least twice. \\
(ii) The constraint in (i) is the only constraint for triples of two-multiranks of four-qubit systems, with the only exception of the triple $(133)$, which cannot be achieved.
\end{theorem}

{\it Proof.}
If the minimum of the three two-multiranks is $\ge{3}$, the result follows from the fact that the three $4\times{4}$ determinants of the three flattenings sum to zero, as proved a century ago by Segre \cite{Segre}. Then, we assume that the minimum is $\le{2}$, attained by ${\mathcal{M}}_{xy}$ and we have three distinct cases as follows up to SLOCC [referring to Eq. (\ref{n-linear})]; here, multi-homogeneous coordinates for the four-qubit system are $x_iy_jz_kt_l$ for $i,j,k,l=\{0,1\}$).
\begin{enumerate}
\item[(1)] Secant:
$$f=x_0y_0(\sum a_{ij}z_it_j)+x_1y_1(\sum b_{ij}z_it_j)\,.$$
Here, the two-flattenings are $4\times 4$ matrices with the block form
$${\mathcal{M}}_{xz}=\left(\begin{array}{c|c}A&0 \\ \vspace{-3.5mm} \\ \hline \vspace{-3.5mm} \\ 0&B\end{array}\right)\,, \qquad {\mathcal{M}}_{xt}=\left(\begin{array}{c|c}A^{T}&0 \\ \vspace{-3.5mm} \\ \hline \vspace{-3.5mm} \\ 0&B^{T}\end{array}\right)\,,$$
which have the same rank. If this rank is one, then $A=0$ or $B=0$ and $f$ is a decomposable tensor.
\item[(2)] Tangent:
$$f=x_0y_0(\sum a_{ij}z_it_j)+(x_0y_1+x_1y_0)(\sum b_{ij}z_it_j)\,.$$
The two-flattenings have the block form
$${\mathcal{M}}_{xz}=\left(\begin{array}{c|c}A&B \\ \vspace{-3.5mm} \\ \hline \vspace{-3.5mm} \\ B&0\end{array}\right)\,, \qquad {\mathcal{M}}_{xt}=\left(\begin{array}{c|c}A^{T}&B^{T} \\ \vspace{-3.5mm} \\ \hline \vspace{-3.5mm} \\ B^{T}&0\end{array}\right)\,,$$
which again have the same rank. If this rank is one then $B=0$ and $f$ is a decomposable tensor.
\item[(3)] Isotropic:
$$f=x_0y_0(\sum a_{ij}z_it_j)+x_0y_1(\sum b_{ij}z_it_j)\,.$$
Here ${\mathcal{M}}_{xy}$ has rank $1$ iff $a$ and $b$ are proportional.
The two-flattenings have the block form
$${\mathcal{M}}_{xz}=\left(\begin{array}{c|c}A&B \\ \vspace{-3.5mm} \\ \hline \vspace{-3.5mm} \\ 0&0\end{array}\right)\,, \qquad {\mathcal{M}}_{xt}=\left(\begin{array}{c|c}A^{T}&B^{T} \\ \vspace{-3.5mm} \\ \hline \vspace{-3.5mm} \\ 0&0\end{array}\right)\,,$$
which have both rank $\le{2}$. If they have both rank one, then $A$ and $B$ are proportional, moreover $\mathrm{rk}(A)=\mathrm{rk}(B)=1$. This concludes the proof of (i). (ii) follows by exhibiting a representative for each case, as in Table \ref{table:4qubit}. The nonexistence of case $(133)$ follows since when one two-multirank is $1$, then we may assume $f=(\sum a_{ij}x_iy_j)(\sum b_{ij}z_it_j)$
and depending on the pair $(\mathrm{rk}(A),\mathrm{rk}(B))=(1,1), (1,2), (2,2)$
we have, correspondingly, the triples $(111)$, $(122)$, $(144)$, so $(133)$ is not achieved.
\end{enumerate}
\qed

\begin{remark}\label{remark1}
States living in the higher secant and/or tangent can produce all states in the lower secants and/or tangents by means of degenerations, that is performing some limits.
\end{remark}

As for what concerns the possibility of producing states in the lower secants and/or tangents from states in the higher secant and/or tangent by degeneration (Remark \ref{remark1}), from Fig. \ref{fig:Hasse4}, it follows that we can asymptotically produce $|{\rm{W}}_4\rangle$ from $|{\rm{GHZ}}_4\rangle$ with a noninvertible SLOCC transformation, i.e., we cannot produce $|{\rm{GHZ}}_4\rangle$ from $|{\rm{W}}_4\rangle$. As a matter of fact, employing the singular (for $\epsilon\to{0}$) SLOCC transformation \begin{equation}
A_{\epsilon}=\epsilon^{-1/4}\begin{pmatrix}
\sqrt[4]{-1} & 1 \\
\epsilon & 0 \\
\end{pmatrix},
\end{equation}
we get
\begin{equation}
\lim_{\epsilon\to0}A_{\epsilon}^{\otimes4}|\rm{GHZ}_4\rangle=|\rm{W}_4\rangle\,.
\end{equation}
Let $|\rm{X}_{4}\rangle=\mathfrak{d}_{1}(|0001\rangle+|0010\rangle+|0100\rangle+|1000\rangle)+\mathfrak{d}_{4}|1111\rangle=\mathfrak{d}_{1}|\rm{W}_{4}\rangle+\mathfrak{d}_{4}|1111\rangle$. It is a symmetric state in $\tau_{3}(\Sigma^{4}_{\textbf{1}})$. It is obvious that if $\mathfrak{d}_4$ tends to zero we can approximately produce $|\rm{W}_4\rangle$ from $|\rm{X}_4\rangle$. As a matter of fact, employing the singular (for $\epsilon\to{0}$) SLOCC transformation 
\begin{equation}
B_{\epsilon}=\epsilon^{-\frac{1}{4}}\begin{pmatrix}
\sqrt[4]{-1} & (-1)^{7/12} 2^{2/3} \\
\epsilon & 0 \\
\end{pmatrix},
\end{equation}
we can get
\begin{equation}
\lim_{\epsilon\to0}B_{\epsilon}^{\otimes4}|\rm{X}_4\rangle=|\rm{W}_4\rangle\,.
\end{equation}
As another example, employing the singular (for $\epsilon\to{0}$) SLOCC transformation
\begin{equation}
C_{\epsilon}=\epsilon^{-\frac{1}{4}}\begin{pmatrix}
\sqrt[4]{-1} & \pm\sqrt{\frac{1}{2} \left(-\sqrt{3}-\mathbf{i}\right)} \\
\epsilon & 0 \\
\end{pmatrix},
\end{equation}
we can asymptotically produce $|\rm{W}_{4}\rangle$ from $|\rm{M}_4\rangle=\alpha|0000\rangle+\beta|0011\rangle+\gamma|1111\rangle$ belonging to $\sigma_{3}(\Sigma^{4}_{\textbf{1}})$, i.e.,
\begin{equation}
\lim_{\epsilon\to0}C_{\epsilon}^{\otimes4}|\rm{M}_4\rangle=|\rm{W}_4\rangle\,.
\end{equation}
It is also obvious that we can approximately produce $|\rm{GHZ}_4\rangle$ from $|\rm{M}_4\rangle$ by letting $\beta$ go to zero.

\subsubsection{($n\geq4$)-qubit entanglement}\label{subsub4.3.3.2}
We can draw the following conclusions for $n\geq{4}$:
\begin{align}\label{general4}\nonumber
&|{\rm{M}}_{n}^{r}\rangle:=|{\rm{GHZ}}_{n}\rangle+\mathfrak{p}\{|0\rangle^{\otimes{r}}|1\rangle^{\otimes{(n-r)}}\} &\in \sigma_{3}(\Sigma^{n}_{\textbf{1}})\,, \\ \nonumber
&|0\rangle_{i}|{\rm{GHZ}}_{n-1}\rangle+|1\rangle_{i}~\mathfrak{p}\{|1\rangle^{\otimes{s}}|0\rangle^{\otimes{(n-s-1)}}\} &\in \sigma_{3}(\Sigma^{n}_{\textbf{1}})\,, \\ \nonumber
&|{\rm{N}}_{n}^{t}\rangle:=|{\rm{W}}_{n}\rangle+\mathfrak{p}\{|1\rangle^{\otimes{t}}|0\rangle^{\otimes{(n-t)}}\} &\in \tau_{3}(\Sigma^{n}_{\textbf{1}})\,,\\ \nonumber
&|0\rangle_{i}|{\rm{W}}_{n-1}\rangle+|1\rangle_{i}~\mathfrak{p}\{|1\rangle^{\otimes{(t-1)}}|0\rangle^{\otimes{(n-t)}}\} &\in \tau_{3}(\Sigma^{n}_{\textbf{1}})\,, \\ \nonumber
&|{\rm{G}}_{n}^{r}\rangle:=\mathfrak{p}\{\alpha|0\rangle^{\otimes{n}}+\beta|0\rangle^{\otimes{r}}|1\rangle^{\otimes{(n-r)}} & \\
&\qquad\qquad\quad+\gamma|1\rangle^{\otimes{r}}|0\rangle^{\otimes{(n-r)}}+\delta|1\rangle^{\otimes{n}}\} &\in \sigma_{4}(\Sigma^{n}_{\textbf{1}})\,,
\end{align}
where $2\leq{r}\leq{n-2}$, $1\leq{s}\leq n-2$, $2\leq{t}\leq n$, and $i\in\{1,\ldots,n\}$. It is worth noting that the state $|{\rm{G}}_{n}^{r}\rangle$ is a generalization of bipartite state $\alpha|00\rangle+\beta|01\rangle+\gamma|10\rangle+\delta|11\rangle$ and its minor is $2|\alpha\delta-\beta\gamma|$, which coincides with the definition of concurrence \cite{Wootters98}. Therefore, if $\alpha\delta\neq\beta\gamma$, the state $|{\rm{G}}_{n}^{r}\rangle$ is genuinely entangled, otherwise it is biseparable (a tensor product of two $r$- and ($n-r$)-partite entangled states).

\begin{proposition}
For $n\geq{4}$ qubits, there is no symmetric entangled state in the proper locus of the $k$-secant variety of the Segre variety, with $k>\lceil\frac{n+1}{2}\rceil$.
\end{proposition}

The superposition of $n$-qubit Dicke states with all possible excitations  
\begin{equation}\label{symmetric}
|\psi_{n}^{\rm{Sym}}\rangle=\sum_{l=0}^{n}\mathfrak{d}_{l}|{\rm{D}}_{n}^{l}\rangle\,,
\end{equation}
is the most general symmetric entangled state. The symmetric $n$-qubit separable states have the structure of the Veronese variety ($\mathcal{V}^{n}_{1}$) and its $k$-secant varieties are SLOCC families \cite{BH01, ST13, SBSE17}. The highest $k$-secant variety fills the ambient space for $k=\lceil\frac{n+1}{2}\rceil$. Comparing with the highest $k$-secant variety in the Segre embedding ($k=\lceil\frac{2^{n}}{n+1}\rceil$), it proves the proposition. Moreover, we will show below that each Dicke state with $1\leq{l}\leq\lfloor\frac{n}{2}\rfloor$ (the same for the spin-flipped version, i.e., $|{\rm{D}}_{n}^{n-l}\rangle$) is in a $k$-secant family of Veronese embedding, and hence, Segre embedding for $2\leq{k}\leq\lfloor\frac{n}{2}\rfloor+1$, respectively. Thus, this method can be useful to classify entanglement of symmetric states and the corresponding number of families grows slower than Ref. \cite{BKMGLS09}.

Consider the following $n$-qubit separable state:
\begin{equation}
|{\rm{S}}_n(\varepsilon)\rangle=(|0\rangle+\varepsilon|1\rangle)^{\otimes{n}}=\sum_{l=0}^{n}\varepsilon^{l}|{\rm{D}}_{n}^{l}\rangle\,.
\end{equation}
Thanks to the definition of tangent star and Eqs. (\ref{recursivesecant}) and (\ref{G-tangent}) in \cref{chap3}, we can write
\begin{equation}\label{tangent-dicke}
\lim_{\varepsilon\to{0}}\frac{1}{\varepsilon^{m+1}}\Big(|{\rm{S}}_n(\varepsilon)\rangle-\sum_{i=0}^{m}\varepsilon^{i}|{\rm{D}}_{n}^{i}\rangle\Big)=|{\rm{D}}_{n}^{m+1}\rangle\in\tau_{m+2}(\Sigma^{n}_{\textbf{1}})\,,
\end{equation}
where $0\leq{m}\leq\lfloor\frac{n}{2}\rfloor-1$. Furthermore, $\lfloor\frac{n}{2}\rfloor$-multiranks of the Dicke states with $1\leq l\leq\lfloor\frac{n}{2}\rfloor$ (and similarly $|{\rm{D}}_{n}^{n-l}\rangle$) are $l+1=k$ ($\ell$-multiranks with $\ell<\lfloor\frac{n}{2}\rfloor$ have the same value or maximum rank). We guess that this is a general behavior which holds true for symmetric multiqudit systems as well. In a similar way, one can check that the states $|{\rm{N}}_{n}^{r}\rangle$ are on the limiting lines of the states $|{\rm{M}}_{n}^{r}\rangle$ in Eq. (\ref{general4}), and therefore, are exceptional states.

Consider now $|\psi_{4}^{\rm{Sym}}\rangle$ from Eq. (\ref{symmetric}) which belongs to $\tau_{3}(\Sigma^{4}_{\textbf{1}})$. It can asymptotically produce lower tangent elements, like $|{\rm{W}}_{4}\rangle$.

\subsection{Five-qubit entanglement}\label{subsec.4.2.4}
For five-qubit states, due to \cref{remark-MR-GE}, \cref{coro:MR-secant}, and classification of two-, three-, and four-qubit states, we have
\begin{enumerate}
\item All quadriseparable states $|{\rm{Q}}_{i}\rangle_{i=1}^{10}$ from Eq. (\ref{general2}) are elements of $\sigma_{2}(\Sigma^{5}_{\textbf{1}})$.
\item All triseparable states $|{\rm{T}}_{i}^{{\rm{GHZ}}_{3}}\rangle_{i=1}^{10}$ and $|{\rm{T}}_{i}^{{\rm{W}}_{3}}\rangle_{i=1}^{10}$ from Eq. (\ref{general3}) are, respectively, elements of $\sigma_{2}(\Sigma^{5}_{\textbf{1}})$ and $\tau_{2}(\Sigma^{5}_{\textbf{1}})$.
\item All biseparable states $|{\rm{B}}_{i}^{{\rm{GHZ}}_{4}}\rangle_{i=1}^{5}$ and $|{\rm{B}}_{i}^{{\rm{W}}_{4}}\rangle_{i=1}^{5}$ from Eq. (\ref{general3}) are, respectively, elements of $\sigma_{2}(\Sigma^{5}_{\textbf{1}})$ and $\tau_{2}(\Sigma^{5}_{\textbf{1}})$.
\end{enumerate}
Considering Eq. (\ref{general3}), we can also find that states $|{\rm{GHZ}}_{5}\rangle$ and $|{\rm{W}}_{5}\rangle$ are elements of $\sigma_{2}(\Sigma^{5}_{\textbf{1}})$ and $\tau_{2}(\Sigma^{5}_{\textbf{1}})$, respectively. These results are given in Table \ref{table:5qubit1}.

In a similar way to Eq. (\ref{general4}), all biseparable states of the form $|\sigma_{3}(\Sigma^{4}_{\textbf{1}})\rangle|q_2\rangle$ and $|\tau_{3}(\Sigma^{4}_{\textbf{1}})\rangle|q_2\rangle$ are elements of $\sigma_{3}(\Sigma^{5}_{\textbf{1}})$ and $\tau_{3}(\Sigma^{5}_{\textbf{1}})$, respectively. Note that the number of distinct subfamilies that these biseparable states create in each $\sigma_{3}(\Sigma^{5}_{\textbf{1}})$ and $\tau_{3}(\Sigma^{5}_{\textbf{1}})$, according to the permutations of the one-qubit state, is, respectively, four times the number of subfamiles in $\sigma_{3}(\Sigma^{4}_{\textbf{1}})$ and $\tau_{3}(\Sigma^{4}_{\textbf{1}})$, i.e., $16$ subfamilies. Other elements of three-secant can be written in a similar way to Eq. (\ref{general4}) with a two-multirank including at least one 3 and no 4 (see \cref{coro:MR-secant}). We denote these elements as $|(3\cdots)\rangle\in\sigma_{3}(\Sigma^{5}_{\textbf{1}})$ and $|(3\cdots)'\rangle\in\tau_{3}(\Sigma^{5}_{\textbf{1}})$ (see Table \ref{table:5qubit2}).

The remaining families of five-qubit states have different two-multiranks, including at least one 4.

\begin{table*}[th]
\centering
\caption[Fine-structure classification of five-qubit entanglement (Segre \& two-secant)]{\label{table:5qubit1} Fine-structure classification of five-qubit entanglement (Segre \& two-secant).}
\begin{tabularx}{\linewidth}{XXX}
\hline\hline
& & \\ [-2ex]
$\Sigma^{5}_{\textbf{1}}$ & $\sigma_{2}$ & $\tau_{2}$ \\ [0.5ex]
\hline
& & \\ [-2ex]
$|{\rm{Sep}}\rangle$ & $|{\rm{GHZ}}_{5}\rangle$ & $|{\rm{W}}_{5}\rangle$ \\ [0.5ex]
& $|{\rm{B}}_{i}^{{\rm{GHZ}}_{4}}\rangle_{i=1}^{5}$ & $|{\rm{B}}_{i}^{{\rm{W}}_{4}}\rangle_{i=1}^{5}$ \\ [0.5ex]
& $|{\rm{T}}_{i}^{{\rm{GHZ}}_{3}}\rangle_{i=1}^{10}$ & $|{\rm{T}}_{i}^{{\rm{W}}_{3}}\rangle_{i=1}^{10}$ \\ [0.5ex]
& $|{\rm{Q}}_{i}\rangle_{i=1}^{10}$ & \\ [0.5ex]
\hline\hline
\end{tabularx}
\end{table*}
\begin{table*}[hbt!]
\centering
\caption[Fine-structure classification of five-qubit entanglement (three- \& four-secant)]{\label{table:5qubit2} Fine-structure classification of five-qubit entanglement (three- \& four-secant).}
\begin{tabularx}{\linewidth}{XXXX}
\hline\hline
& & & \\ [-2ex]
$\sigma_{3}$ & $\tau_{3}$ & $\sigma_{4}$ & $\tau_{4}$\\ [0.5ex]
\hline
& & & \\ [-2ex]
$|(3333333333)\rangle$ & $|(3333333333)'\rangle$ & $|(4444444444)_{4}\rangle$ & $|(4444444444)'_{4}\rangle$ \\ [0.5ex]
$\vdots$ & $\vdots$ & $\vdots$ & $\vdots$ \\ [0.5ex]
$|(3\cdots)\rangle$ & $|(3\cdots)'\rangle$ & $|(4\cdots)_{4}\rangle$ & $|(4\cdots)'_{4}\rangle$ \\ [0.5ex]
$\mathfrak{p}_{i}\{|\sigma_{3}(\Sigma^{4}_{\textbf{1}})\rangle|q_2\rangle\}_{i=1}^{16}$ & & $\mathfrak{p}_{i}\{|\sigma_{4}(\Sigma^{4}_{\textbf{1}})\rangle|q_2\rangle\}_{i=1}^{40}$ & \\ [0.5ex]
\hline\hline
\end{tabularx}
\end{table*}
\begin{table*}[hbt!]
\centering
\caption[Fine-structure classification of five-qubit entanglement (five \& six-secant)]{\label{table:5qubit3} Fine-structure classification of five-qubit entanglement (five \& six-secant).}
\begin{tabularx}{\linewidth}{XXXX}
\hline\hline
& & & \\ [-2ex]
$\sigma_{5}$ & $\tau_{5}$ & $\sigma_{6}$ & $\tau_{6}$ \\ [0.5ex]
\hline
& & & \\ [-2ex]
$|(4444444444)_{5}\rangle$ & $|(4444444444)'_{5}\rangle$ & $|(4444444444)_{6}\rangle$ & $|(4444444444)'_{6}\rangle$ \\ [0.5ex]
$\vdots$ & $\vdots$ & $\vdots$ & $\vdots$ \\ [0.5ex]
$|(4\cdots)_{5}\rangle$ & $|(4\cdots)'_{5}\rangle$ & $|(4\cdots)_{6}\rangle$ & $|(4\cdots)'_{6}\rangle$ \\ [0.5ex]
\hline\hline
\end{tabularx}
\end{table*}

Considering classification of four-qubit as the core structure of five-qubit classification, all biseparable state of the form $|\sigma_{4}(\Sigma^{4}_{\textbf{1}})\rangle|q_2\rangle$ are elements of $\sigma_{4}(\Sigma^{5}_{\textbf{1}})$ ($40$ subfamilies). Here, we have a new type of biseparable state in our five-qubit classification, i.e., $\mathfrak{p}\{|{\rm{Bell}}\rangle|{\rm{GHZ}}_{3}\rangle\}$, which creates 10 subfamilies in $\sigma_{4}(\Sigma^{5}_{\textbf{1}})$ (see Table \ref{table:5qubit2}). Note that one can generate genuine entangled states from them which would be of the form $|{\rm{G}}_{5}^{2}\rangle$ ($\sim|{\rm{G}}_{5}^{3}\rangle$) in Eq. (\ref{general4}). On the limiting lines of these states, one can find the biseparable states $\mathfrak{p}\{|{\rm{Bell}}\rangle|{\rm{W}}_{3}\rangle\}$ and the genuine entangled versions as the elements of $\tau_{4}(\Sigma^{5}_{\textbf{1}})$.
As another example, using reasoning similar to Eq. (\ref{tangent-dicke}), we can draw the following results for $n\geq 5$:
\begin{align}\label{general5}\nonumber
|{\rm{W}}_{n}\rangle+|1\rangle^{\otimes{n}}+\mathfrak{p}\{|0\rangle^{\otimes{r}}|1\rangle^{\otimes{(n-r)}}\} &\in \tau_{4}(\Sigma^{n}_{\textbf{1}})\,, \\
|{\rm{D}}_{n}^{2}\rangle+\mathfrak{p}\{|1\rangle^{\otimes{s}}|0\rangle^{\otimes{(n-s)}}\} &\in \tau_{4}(\Sigma^{n}_{\textbf{1}})\,,
\end{align}
where $2\leq{r}\leq{n-2}$ and $3\leq{s}\leq n-1$.

It is worth noting that since in the five-qubit states ($\otimes^5\mathbbm{C}^{2}$), we just have flattenings of sizes $2\times{16}$ and $4\times{8}$ with maximum ranks of 2 and 4, respectively, they do not provide nontrivial equations to find the elements of five-secant. Hence, with the method of \cref{chap3}, one can find, as in Ref. \cite{OS16}, homogeneous polynomials of degrees $6$ and $16$ where the rank of the Jacobian of these two equations gives the desired information (if the point is not singular for the five-secant then it cannot stay in the four-secant, i.e., it is an element of the proper five-secant family).

To have an exhaustive classification, we denote the other elements of four-, five-, and six-secants as $|(4\cdots)_{i}\rangle\in\sigma_{i}(\Sigma^{5}_{\textbf{1}})$ and $|(4\cdots)'_{i}\rangle\in\tau_{i}(\Sigma^{5}_{\textbf{1}})$ where $i\in\{4,5,6\}$ (see Tables \ref{table:5qubit2} and \ref{table:5qubit3}). It is worth noting that in the classification of five-qubit states, all the elements in five- and six-secant families are genuinely entangled.

\subsection{$n$-qubit Dicke states}\label{subsec.4.2.5}
Regarding Theorem \ref{theo:Warink-rk} and Conjecture \ref{conj:Waring-brk}, we have the following result for the $n$-qubit Dicke states $|{\rm{D}}_{n}^l\rangle$ (with $l$ excitations). If $1\leq{l}<\lfloor\frac{n}{2}\rfloor$, the tensor rank and border rank of $|{\rm{D}}_n^l\rangle$ are equal to $n-l+1$ and $l+1$, respectively. For $l=\lfloor\frac{n}{2}\rfloor$, we have two situations; (1) if $n=\text{even}$, the tensor rank and border rank are both equal to $\frac{n}{2}+1$, and (2) if $n=\text{odd}$, the tensor rank and border rank are equal to $\lceil\frac{n}{2}\rceil+1$ and $\lfloor\frac{n}{2}\rfloor+1$, respectively. Hence, the relation between tensor rank and border rank of $n$-qubit Dicke states is as follows:
\begin{equation}
\rk\big(|{\rm{D}}_{n}^l\rangle\big)+\brk\big(|{\rm{D}}_{n}^l\rangle\big)=n+2\,.
\end{equation}
Based on this fact, we draw the following result
\begin{equation}
|{\rm{D}}_{n}^{\lfloor\frac{n}{2}\rfloor}\rangle\in
\begin{cases}
  \sigma_{\frac{n}{2}+1}(\Sigma^{n}_{\textbf{1}}) & \text{if $n=$ even}\,, \\
  \tau_{\lfloor\frac{n}{2}\rfloor+1}(\Sigma^{n}_{\textbf{1}})  & \text{if $n=$ odd}\,.
\end{cases}
\end{equation}
Therefore, for an even number of qubits, regarding the rank and border rank information the Dicke state $|{\rm{D}}_{n}^{\frac{n}{2}}\rangle$ is in the proper $(\frac{n}{2}+1)$-secant family while based on the higher derivative information it is in the osculating hyperplane that we take it in the tangent family. Geometrically, it means that this special state is in the intersection of the proper $(\frac{n}{2}+1)$-secant family and the $(\frac{n}{2}+1)$-tangent family.

\chapter{Fine-Structure Classification of Tripartite Entanglement}\label{chap5}

\epigraph{``All happy families are alike, but every unhappy family is unhappy in its own way''}{\textit{Lew Tolstoy}}

In this chapter of the dissertation, which is based on the Ref. \cite{GM21}, we characterize entanglement of tripartite $\mathbbm{C}^d\otimes\mathbbm{C}^d\otimes\mathbbm{C}^d$ systems. To this aim, we employ algebraic-geometric tools that are invariants under Stochastic Local Operation and Classical Communication (SLOCC), namely $k$-secant varieties and one-multilinear ranks. Indeed, by means of them, we present a classification of pure tripartite states in terms of a finite number of families and subfamilies. At the core of it stands out a fine-structure grouping of three-qutrit entanglement.

\section{Classification method}\label{sec.5.1}
At the core of Ref. \cite{GMO20} was the identification of determinantal and secant varieties of the Segre variety as those that classify multiqubit entanglement. Here, we shall extend this approach to classify three-qudit pure states
\begin{equation}\label{3quditstate}
|\psi\rangle = \sum_{i\in\{0,\ldots,d-1\}^{3}}\mathfrak{c}_{i}|i\rangle\,.
\end{equation}
To this end, we shall be examining maps $\mathcal{M}$ that are produced by tensor flattening \cite{Landsberg} from the quantum states in Eq. (\ref{3quditstate}). Consider the tensor Hilbert space 
$\mathcal{H}=\mathcal{H}_1\otimes\mathcal{H}_2\otimes\mathcal{H}_3$, where $\mathcal{H}_i\simeq \mathbbm{C}^d$. We shall define $\ell$-partitions as ordered $\ell$-tuples $I=\left(i_1,\ldots,i_{\ell}\right)$, where $1\leq\ell\leq{2}$, and $1\leq i_1<\cdots<i_{\ell}\leq{3}$. Given an $\ell$-partition $I$, we define the complementary partition $\bar{I}$ as the $(3-\ell)$-partition such that $I\cup\bar{I}=\{1,2,3\}$. Therefore, $\mathcal{H}\simeq\mathcal{H}_I\otimes\mathcal{H}_{\bar{I}}$, where $\mathcal{H}_I=\otimes^{\ell}\mathbbm{C}^{d}$ and $\mathcal{H}_{\bar{I}}$ is the Hilbert space with complementary indices. For any state $\psi$ with vector representation $|\psi\rangle\in\mathcal{H}$, the $\ell$-partition $I$ leads to a linear operator $\mathcal{M}_{I}[\psi]$ (see Eq. \eqref{flattening}). Given a state $\psi$ and a number $1\leq\ell\leq{2}$, we call the sequence of ranks $r_{I}[\psi]={\rm{rank}}\left(\mathcal{M}_I[\psi]\right)$ for all $\ell$-partitions $I$, the $\ell$-multilinear rank (hereafter $\ell$-multirank) of the state $\psi$. Although there are six partitions, with three complementary pairs $(1) \leftrightarrow (23)$, $(2) \leftrightarrow (13$), $(3) \leftrightarrow (12)$, it is enough to check $\ell$-multiranks for partition $I$ with $\ell=1$. Note that for the complementary partition $\bar{I}$ the matrices $\mathcal{M}_{\bar{I}}[\psi]$ are just the transpose of $\mathcal{M}_{I}[\psi]$ and transposition does not change the rank of the matrix.

An important observation is that $\ell$-multirank is an invariant under SLOCC (see \cref{l-multirank-SLOCC}).

Since $\ell$-multiranks only depend on the state vector and, furthermore, because statements about rank can be rephrased as statements about minors which are determinants, it follows that a given $\ell$-multirank configuration determines a determinantal variety in the projective Hilbert space and pure multipartite states which have $\ell$-multiranks bounded by a given integer sequence make a subvariety of $\mathbbm{P}(\mathcal{H}_{n})$. Indeed, these determinantal varieties are subvarieties of secant varieties of the projective variety of fully separable states. For a tripartite quantum state, the space of fully separable states is the Segre variety \cite{Miyake03, Heydari08}: with embedding
\begin{equation}\label{Segre-3qudit}
\Sigma^{3}_{\textbf{d}-\textbf{1}}:~\mathbbm{P}^{d-1}\times\mathbbm{P}^{d-1}\times\mathbbm{P}^{d-1}\hookrightarrow\mathbbm{P}^{d^{3}-1}\,,
\end{equation}
where $\textbf{d}-\textbf{1}=(d-1,d-1,d-1)$ and $\times$ is the Cartesian product of sets. One can easily check that $\Sigma$ is the projective variety of fully separable states. Indeed, if all partial traces give pure states, the corresponding ranks are all one. Conversely, if all $\ell$-multiranks are one, the state is fully separable. 

Since $k$-secant varieties of the Segre variety are invariant under the action of the projective linear group therefore they are SLOCC invariants (see \cref{k-secant-SLOCC}). This means that SLOCC classes can be grouped naturally into entanglement families. For this reason, the dimension of the highest $k$-secant variety, that fills the projective Hilbert space of three qudits, can indicate the number of entanglement families.
The higher $k$-secant variety fills the ambient space $\mathbbm{P}(\mathbbm{C}^{d}\otimes\mathbbm{C}^{d}\otimes\mathbbm{C}^{d})$ when
\begin{equation}\label{TripartiteGenericRank}
k=\left\lceil\frac{d^3}{3d-2}\right\rceil\,,
\end{equation}
except for $d=3$ where the generic rank is five \cite{Strassen,Lickteig}. This $k$ indicates the number of entanglement families which remains finite with the dimension of parties.

Since $\sigma_{k-1}\subset\sigma_k$ we need to distinguish the elements of each $k$-secant family by defining the proper secant. If it exists, the proper $k$-secant [the states that belongs to $k$-secant but not to $(k-1)$-secant], i.e., the set $\sigma_{k}(\Sigma^{3}_{\textbf{d}-\textbf{1}})\setminus\sigma_{k-1}(\Sigma^{3}_{\textbf{d}-\textbf{1}})$, is the union of the $k$-secant hyperplanes $\mathcal{S}_{k}\subset\sigma_{k}(\Sigma^{3}_{\textbf{d}-\textbf{1}})$ represented by
\begin{equation}\label{S-plane}
\mathcal{S}_{k}=\sum_{i=1}^{k}\lambda_{i}p_{i}\,,
\end{equation}
with $\{\lambda_{i}\}_{i=1}^{k}\neq{0}$ and each $p_{i}$ is an independent point in Segre variety.

It is worth noting that in addition to the standard flattenings, as the standard tensor contraction shown in Eq. (\ref{flattening}), for tripartite systems $\mathbbm{C}^{2m+1}\otimes\mathbbm{C}^{2m+1}\otimes\mathbbm{C}^{2m+1}$, we have another flattening map (it often called a Koszul flattening, or more generally one of the Young flattenings) as follows:
\begin{equation}\label{flattening2}
\Lambda^{m}\mathcal{H}_1\otimes\mathcal{H}^{\vee}_2\to\Lambda^{m+1}\mathcal{H}_1\otimes\mathcal{H}_3\,,
\end{equation}
where $\Lambda^{m}$ denotes the $m^{\text{th}}$ exterior power. Hence, the size $(k+1)\binom{2m}{m}$ minors of Eq. \eqref{flattening2} provides equations for $k$-secant varieties up to $k=(2m+1)^2/(m+1)$ (see Refs. \cite{Landsberg, LO13}).

Therefore, similar to the spirit of Ref. \cite{GMO20}, we use $k$-secant varieties and one-multiranks as the SLOCC invariants to bunch entanglement orbits (classes) of tripartite 
$\otimes^3\mathbbm{C}^d$ systems into a finite number of families and subfamilies. Hence, the classification algorithm can be summarized as:
\begin{itemize}
\item[{\bf (i)}] find families by identifying $k$-secant varieties $\Sigma^{3}_{\textbf{d}-\textbf{1}}, \sigma_{2}(\Sigma^{3}_{\textbf{d}-\textbf{1}}), \ldots, \sigma_{k}(\Sigma^{3}_{\textbf{d}-\textbf{1}})$;
\item[{\bf (ii)}] split families to secants and tangents by identifying $\tau_{2}(\Sigma^{3}_{\textbf{d}-\textbf{1}}), \ldots, \tau_{k}(\Sigma^{3}_{\textbf{d}-\textbf{1}})$;
\item[{\bf (iii)}] find subfamilies by identifying one-multiranks.
\end{itemize}

\section{Fine-structure classification of two-qutrit entanglement}\label{sec.5.2}

Although two-qutrit states do not belong to tripartite systems we provide a full entanglement classification for two-qutrit states which can be used as the core for the entanglement classification of three-qutrit states.

For two-qutrit states, the Segre four-fold $\Sigma^{2}_{\textbf{2}}\subset\mathbbm{P}^{8}$, i.e., the set of fully separable states of two qutrits, consists of general points given by affine coordinates $p=[1:a:b:c:ac:bc:d:ad:bd]$ where $a$, $b$, $c$, and $d$ are complex parameters and one or more parameters can tend to infinity.

Moving on to the proper two-secant variety, i.e., the union of the secant planes $\mathcal{S}_{2}=\lambda_1p_1+\lambda_2p_2$, we have generic elements given by the following coordinates:
\begin{align}\nonumber
[& \lambda_1+\lambda_2:\lambda_1a_1+\lambda_2a_2:\lambda_1b_1+\lambda_2b_2:\lambda_1c_1+\lambda_2c_2: \\ \nonumber
& \lambda_1a_1c_1+\lambda_2a_2c_2:\lambda_1b_1c_1+\lambda_2b_2c_2:\lambda_1d_1+\lambda_2d_2: \\
& \lambda_1a_1d_1+\lambda_2a_2d_2:\lambda_1b_1d_1+\lambda_2b_2d_2]\,.
\end{align}
It is easy to see that $[1:0:0:0:1:0:0:0:0]$ is an elements of $\sigma_2(\Sigma^{2}_{\textbf{2}})$. Actually, the following general state:
\begin{equation}\label{GHZ2-1}
|{\rm{GHZ}}_{2}^{(1)}\rangle=|\alpha\alpha\rangle+|\beta\beta\rangle\,,
\end{equation}
where $\alpha\neq\beta\in\{0,1,2\}$, can represent all elements of proper two-secant family with one-multiranks equal to $(22)$.

Obviously, one can rewrite the secant planes as $\mathcal{S}_{2}=p_1+\mu(p_2-p_1)$ where 
$\lambda_1=1-\mu$ and $\lambda_2=\mu$. Now, we consider the situation where second point tends to the first one, i.e., $p_2\to{p_1}$, by taking $p_2(\epsilon)=[1:a_1+\epsilon:b_1+\epsilon:c_1+\epsilon:(a_1+\epsilon)(c_1+\epsilon):(b_1+\epsilon)(c_1+\epsilon):d_1+\epsilon:(a_1+\epsilon)(d_1+\epsilon):(b_1+\epsilon)(d_1+\epsilon)]$. This will give us the coordinates of the elements in the proper two-tangent variety. However, for two-qutrit states, we have the special situation that all points on the tangent, i.e.,
\begin{align}\nonumber
p'=&\lim_{\epsilon\rightarrow 0}\Big(p_1+\frac{\mu}{\epsilon}\big(p_2(\epsilon)-p_1\big)\Big)=[1:a_1+\mu:b_1+\mu: \\ \nonumber
& c_1+\mu:a_1c_1+\mu(a_1+c_1):b_1c_1+\mu(b_1+c_1): \\
& d_1+\mu:a_1d_1+\mu(a_1+d_1):b_1d_1+\mu(b_1+d_1)]\,,
\end{align}
lie also on the proper two-secant since
\begin{align}\nonumber
p'=[& 1:a_1+\mu:b_1+\mu:c_1+\mu:(a_1+\mu)(c_1+\mu): \\ \nonumber
& (b_1+\mu)(c_1+\mu):d_1+\mu:(a_1+\mu)(d_1+\mu): \\
& (b_1+\mu)(d_1+\mu)]-\mu^2[0:0:0:0:1:1:0:1:1]\,,
\end{align}
which explicitly comes from joining of two independent points in the Segre variety, i.e., superposition of two fully separable states. It means that the proper two-tangent is equal to the proper two-secant.

The proper three-secant, i.e., the set $\sigma_3(\Sigma^{2}_{\textbf{2}})/\sigma_2(\Sigma^{2}_{\textbf{2}})$,  is the union of the secant hyperplanes $\mathcal{S}_3$ represented by Eq. (\ref{S-plane}). Indeed, joining of three independent points in the Segre variety gives rise to elements of three-secant family. For instance,
\begin{equation}\label{GHZ2-2}
|{\rm{GHZ}}_{2}^{(2)}\rangle=|00\rangle+|11\rangle+|22\rangle\,,
\end{equation}
is an element of $\sigma_3(\Sigma^{2}_{\textbf{2}})$ with one-multirank equals to $(33)$. In a similar way to two-secant variety, one can see that the proper three-tangent is equal to the proper three-secant.

\begin{table}[t]
\centering
\caption[Fine-structure classification of two-qutrit entanglement]{\label{table:5.2} Fine-structure classification of two-qutrit entanglement.}
\vspace{8pt}
\begin{tabularx}{\linewidth}{XXX}
\hline\hline
& & \\ [-2ex]
$\Sigma^{2}_{\textbf{2}}$ & $\sigma_{2}$ & $\sigma_{3}$ \\ [0.5ex]
\hline
& & \\ [-2ex]
$|\rm{Sep}\rangle$ & $|{\rm{GHZ}}_{2}^{(1)}\rangle$ & $|{\rm{GHZ}}_{2}^{(2)}\rangle$ \\ [0.5ex]
\hline\hline
\end{tabularx}
\end{table}

Briefly, this classification provide us three secant families that coincide with the three SLOCC classes, namely, separable and two inequivalently entangled states that come from superposition of two and three fully separable states (Table \ref{table:5.2}).

Already from this classification we can draw a general conclusion. That is, for $n\geq{2}$ qutrits we have
\begin{align}\label{general-qutrit}
&\mathfrak{p}\{|{\rm{GHZ}}_{2}^{(1)}\rangle|q_3\rangle^{\otimes{(n-2)}}\} \in \sigma_2(\Sigma^{n}_{\textbf{2}})\,, \\
&\mathfrak{p}\{|{\rm{GHZ}}_{2}^{(2)}\rangle|q_3\rangle^{\otimes{(n-2)}}\} \in \sigma_3(\Sigma^{n}_{\textbf{2}})\,,
\end{align}
where $\mathfrak{p}\{\cdot\}$ denotes all possible permutations of subsystems and $|q_3\rangle$ is a general one-qutrit state.

\section{Fine-structure classification of three-qutrit entanglement}\label{sec.5.3}
For the Segre variety $\Sigma^{3}_{\textbf{2}}\subset\mathbbm{P}^{26}$, we shall use homogeneous coordinates associated with the induced basis $\left\{|000\rangle,|001\rangle,\ldots,|222\rangle\right\}$. That is to say, a point $p\in\mathbbm{P}^{26}$ is written in homogeneous coordinates $\left[\mathfrak{c}_0:\mathfrak{c}_1:\cdots:\mathfrak{c}_{26}\right]$ whenever $p$ is the projective class of the three-qutrit state of Eq. (\ref{3quditstate}). Then, the Segre surface $\Sigma^{3}_{\textbf{2}}$ is the projective variety with points given by affine coordinates $[1:a:b:c:ac:bc:d:ad:bd:e:ae:be:ce:ace:bce:de:ade:
bde:f:af:bf:cf:acf:bcf:df:adf:bdf]$, where $a$, $b$, $c$, $d$, $e$, and $f$ are complex parameters. This expression must be properly understood, in that the limits of $a$ and/or $b$ and/or $c$ and/or $d$ and/or $e$ and/or $f$ going to infinity, must be included. For instance, also points of the form $[0:1:0:0:c:0:0:d:0:0:e:0:0:ce:0:0:de:0:0:f:0:0:cf:0:0:df:0]$, which corresponds to $a\rightarrow\infty$, are part of $\Sigma^{3}_{\textbf{2}}$.

Thanks to Ref. \cite{CK11}, all one-multiranks can be found for states of any number of qudits. For three-qutrit states we have
\begin{equation}
r_i\leq\prod_{j\neq{i}}r_j\quad\forall~i,j\in\{1,2,3\}\,,
\end{equation}
where $0\leq{r_i}\leq{3}$ stands for the rank of the corresponding flattening. Therefore, all the 
one-multiranks of three-qutrit states are: (111) which indicates a fully separable states; (122) and (133) and their permutations, which indicate biseparable states; (222), all permutations of (223), all permutations of (233), and (333), which indicate genuinely entangled states.

Standard flattenings are not enough to construct higher secant families in $\mathbbm{P}^{26}$. So based on Eq. \eqref{flattening2} we have the following flattening:
\begin{equation}\label{flat2}
\mathcal{F}:~\mathcal{H}_1\otimes\mathcal{H}^{\vee}_2\to\Lambda^2\mathcal{H}_1\otimes\mathcal{H}_3\,,
\end{equation}
that can be considered as the composition of
$$
\mathcal{H}_1\otimes\mathcal{H}^{\vee}_2\xrightarrow[]{{\rm{Id}}_{\mathcal{H}_1}\otimes\mathcal{M}_2}\mathcal{H}_1\otimes\mathcal{H}_1\otimes\mathcal{H}_3\,,
$$
and
$$
\mathcal{H}_1\otimes\mathcal{H}_1\otimes\mathcal{H}_3\xrightarrow[]{P_\wedge\otimes{\rm{Id}}_{\mathcal{H}_3}}\Lambda^2\mathcal{H}_1\otimes\mathcal{H}_3\,,
$$
where $\mathcal{M}_2:\mathcal{H}^{\vee}_{2}\to\mathcal{H}_1\otimes\mathcal{H}_3$ is the standard flattening and $P_\wedge:\mathcal{H}_1\otimes\mathcal{H}_1\to\Lambda^2\mathcal{H}_1$ is the projection onto the skew-symmetric component \cite{Ottaviani09}. Based on the map in Eq. \eqref{flat2}, we have the following $9\times{9}$ matrix (known as the Ottaviani-Strassen matrix) for the general three-qutrit state of Eq. (\ref{3quditstate}),
\begin{equation}\label{strassen}
\mathcal{F}=
\begin{pmatrix}
 0 & 0 & 0 & \mathfrak{c}_{0} & \mathfrak{c}_{1} & \mathfrak{c}_{2} & -\mathfrak{c}_{9} & -\mathfrak{c}_{10} & -\mathfrak{c}_{11} \\
 0 & 0 & 0 & \mathfrak{c}_{3} & \mathfrak{c}_{4} & \mathfrak{c}_{5} & -\mathfrak{c}_{12} & -\mathfrak{c}_{13} & -\mathfrak{c}_{14} \\
 0 & 0 & 0 & \mathfrak{c}_{6} & \mathfrak{c}_{7} & \mathfrak{c}_{8} & -\mathfrak{c}_{15} & -\mathfrak{c}_{16} & -\mathfrak{c}_{17} \\
-\mathfrak{c}_{0} & -\mathfrak{c}_{1} & -\mathfrak{c}_{2}  & 0 & 0 & 0 & \mathfrak{c}_{18} & \mathfrak{c}_{19} & \mathfrak{c}_{20} \\
-\mathfrak{c}_{3} & -\mathfrak{c}_{4} & -\mathfrak{c}_{5} & 0 & 0 & 0 & \mathfrak{c}_{21} & \mathfrak{c}_{22} & \mathfrak{c}_{23} \\
-\mathfrak{c}_{6} & -\mathfrak{c}_{7} & -\mathfrak{c}_{8} & 0 & 0 & 0 & \mathfrak{c}_{24} & \mathfrak{c}_{25} & \mathfrak{c}_{26} \\
\mathfrak{c}_{9} & \mathfrak{c}_{10} & \mathfrak{c}_{11} & -\mathfrak{c}_{18} & -\mathfrak{c}_{19} & -\mathfrak{c}_{20} & 0 & 0 & 0 \\
\mathfrak{c}_{12} & \mathfrak{c}_{13} & \mathfrak{c}_{14} & -\mathfrak{c}_{21} & -\mathfrak{c}_{22} & -\mathfrak{c}_{23} & 0 & 0 & 0 \\
\mathfrak{c}_{15} & \mathfrak{c}_{16} & \mathfrak{c}_{17}  & -\mathfrak{c}_{24} & -\mathfrak{c}_{25} & -\mathfrak{c}_{26} & 0 & 0 & 0
\end{pmatrix}\,.
\end{equation}
Actually, the determinant of matrix $\mathcal{F}$, which is an ${\rm{SL}}(3,\mathbbm{C})^{\times{3}}$-invariant of degree nine, indicates the four-secant hyperplane. It means that if $\mathcal{F}$ is full rank for a given state, i.e., rank of the matrix $\mathcal{F}$ is nine, that state is in five-secant family. Indeed the quantity
\begin{equation}
k=\left\lceil\frac{\text{rank}\mathcal{F}}{2}\right\rceil\,,
\end{equation}
indicates the $k$-secant family to which the state belongs.

Let us now move on to the proper two-secant variety, i.e., the set $\sigma_{2}(\Sigma^{3}_{\textbf{2}})\setminus\Sigma^{3}_{\textbf{2}}$, which is the union of the secant planes $\mathcal{S}_{2}$ represented by Eq. (\ref{S-plane}). The generic element in the proper two-secant comes from joining two independent points (superposition of two fully separable states), i.e., $\lambda_1 p_1+\lambda_2 p_2$ with $\lambda_1,\lambda_2\neq{0}$. For instance, it is easy to see that
\begin{equation}\label{GHZ3-1}
|{\rm{GHZ}}_3^{(1)}\rangle=|\alpha\alpha\alpha\rangle+|\beta\beta\beta\rangle\,,
\end{equation}
where $\alpha\neq\beta\in\{0,1,2\}$ is an element of $\sigma_{2}(\Sigma^{3}_{\textbf{2}})$ with one-multirank equal to $(222)$. 

Also, the classification of two-qutrit states provides us the following biseparable states with one-multirank equal to $(122)$, up to a permutation, as other elements of 
$\sigma_{2}(\Sigma^{3}_{\textbf{2}})$:
\begin{equation}
|{\rm{B}}_{i}^{(1)}\rangle_{i=1}^{3}=\mathfrak{p}\{|{\rm{GHZ}}_{2}^{(1)}\rangle|q_3\rangle\}\,,
\end{equation}
where $\mathfrak{p}\{\cdot\}$ denotes all possible permutations of subsystems, $|q_3\rangle$ is a generic one-qutrit state, and similarly to Eq. \eqref{GHZ3-1} $|{\rm{GHZ}}_{2}^{(1)}\rangle=|\alpha\alpha\rangle+|\beta\beta\rangle$. Note that this is the situation in which one or more parameters on the proper two-secant variety tend to infinity.

Now, considering the tangent to the point $p_1=[1:0:\cdots:0]$ (equivalent to all points on $\Sigma^{3}_{\textbf{2}}$ by an SLOCC), we have the affine coordinate $[1:\mu:\mu:\mu:0:0:\mu:0:0:\mu:0:0:0:0:0:0:0:0:\mu:0:0:0:0:0:0:0:0]$. Letting $\mu\to\infty$, we have the state $|00\upsilon \rangle+|0\upsilon0\rangle+|\upsilon00\rangle$ with $|\upsilon\rangle=|1\rangle+|2\rangle$ which is obviously a three-qutrit $\rm{W}$-type state. Bearing in mind this result, we can derive the following state as an element of $\tau_{2}(\Sigma^{3}_{\textbf{2}})$ with one-multirank equal to $(222)$:
\begin{equation}\label{W3}
|{\rm{W}}_3\rangle=|{\rm{D}}_{3}^{\mathfrak{p}(2,1,0)}\rangle=\sum_{i}\mathfrak{p}_{i}\{|\alpha\alpha\beta\rangle\}\,,
\end{equation}
where $\alpha\neq\beta\in\{0,1,2\}$ and
\begin{equation}\label{Dicke-3-qudit}
|{\rm{D}}_{3}^{\jmath}\rangle=\sqrt{\frac{\prod_{i}{j_i!}}{3!}}\sum_{\pi\in\mathfrak{S}_{3}}\pi\{|0\rangle^{\otimes{j_1}}\otimes\cdots\otimes|d-1\rangle^{\otimes{j_d}}\}\,,
\end{equation}
are the so-called $3$-qudit Dicke states (with excitations shown as $\jmath=(j_1,\ldots,j_d)$ where $j_1+\cdots+j_d=3$). Also, we can explicitly see that $|{\rm{W}}_3\rangle$ can asymptotically be obtained from $|{\rm{GHZ}}_3^{(1)}\rangle$ as follows:
\begin{equation}
|{\rm{W}}_3\rangle=\lim_{\varepsilon\to{0}}\frac{1}{\varepsilon}\left((|\alpha\rangle+\varepsilon|\beta\rangle)^{\otimes{3}}-|\alpha\alpha\alpha\rangle\right)\,.
\end{equation}

The proper three-secant, i.e., the set $\sigma_{3}(\Sigma^{3}_{\textbf{2}})\setminus\sigma_{2}(\Sigma^{3}_{\textbf{2}})$, is the union of the secant hyperplanes $\mathcal{S}_{3}$ represented by Eq. (\ref{S-plane}). So, joining three independent points in the Segre variety (superposition of three fully separable states) that satisfies Eq. (\ref{S-plane}), gives rise to elements of proper three-secant family. For instance,
\begin{equation}\label{GHZ3-2}
|{\rm{GHZ}}_{3}^{(2)}\rangle=|000\rangle+|111\rangle+|222\rangle\,, 
\end{equation}
is an element of $\sigma_{3}(\Sigma^{3}_{\textbf{2}})$ with one-multirank equal to $(333)$. In the proper three-secant we have other elements with different one-multiranks. For instance,
\begin{equation}\label{(332)}
|{\rm{GHZ}}_3^{(1)}\rangle+\mathfrak{p}\{|\alpha\gamma\gamma\rangle\}\,,
\end{equation}
and
\begin{equation}\label{(322)}
|{\rm{GHZ}}_3^{(1)}\rangle+\mathfrak{p}\{|\alpha\beta\gamma\rangle\}\,,
\end{equation}
where $\alpha\neq\beta\neq\gamma\in\{0,1,2\}$ are all elements of $\sigma_{3}(\Sigma^{3}_{\textbf{2}})$ with one-multirank equal to $(233)$ and $(223)$, up to a permutation, respectively. The states in Eqs. \eqref{(332)} and \eqref{(322)} are the joining of a $|{\rm{GHZ}}_3^{(1)}\rangle$ state and an independent point in the Segre variety. One can write these elements of proper three-secant in terms of joining biseparable states $|{\rm{B}}^{(1)}\rangle$ and an independent point of Segre variety as $|\alpha\rangle(|\alpha\alpha\rangle+|\beta\gamma\rangle)+|\beta\beta\beta\rangle$ and $|\alpha\rangle(|\alpha\alpha\rangle+|\gamma\gamma\rangle)+|\beta\beta\beta\rangle$, respectively.

From the classification of two-qutrit states, we have biseparable states with one-multirank equal to $(133)$, up to a permutation, as other elements of 
$\sigma_{3}(\Sigma^{3}_{\textbf{2}})$:
\begin{equation}
|{\rm{B}}_{i}^{(2)}\rangle_{i=1}^{3}=\mathfrak{p}\{|{\rm{GHZ}}_{2}^{(2)}\rangle|q_3\rangle\}\,.
\end{equation}

To construct the closure of the three-secant variety, i.e., the three-tangent, one can use different limit types at $p_{1}=[1:0:\cdots:0]$.  For instance, we can consider the first limit type which is the addition of Eq. \eqref{W3} with an extra point from the Segre variety (see Ref. \cite{GMO20}). Then, we get
\begin{equation}\label{X3}
|{\rm{X}}_3\rangle=|{\rm{W}}_3\rangle+|\gamma\gamma\gamma\rangle\,,
\end{equation}
where $\alpha\neq\beta\neq\gamma\in\{0,1,2\}$ as an element of $\tau_{3}(\Sigma^{3}_{\textbf{2}})$ with one-multirank equal to $(333)$. 
Indeed, based on the inclusion $\tau_3\subset\sigma_3$, we can conclude that $|{\rm{X}}_3\rangle$ can asymptotically be produced by $|{\rm{GHZ}}_{3}^{(2)}\rangle$. This can be shown by considering the following points:
\begin{equation}
p(\varepsilon)=\frac{1}{\varepsilon}\left((|0\rangle+\varepsilon|1\rangle+\varepsilon|2\rangle)^{\otimes{3}}+\varepsilon|222\rangle-|000\rangle\right)\,.
\end{equation}
For all $\varepsilon\neq{0}$ they correspond to ${\rm{GHZ}}^{(2)}$-type states and indicate a smooth curve in $\sigma_{3}(\Sigma^{3}_{\textbf{2}})$. When $\varepsilon\to{0}$ we have
\begin{equation}
\lim_{\varepsilon\to{0}}p(\varepsilon)=|00\upsilon \rangle+|0\upsilon0\rangle+|\upsilon00\rangle+|222\rangle\,,
\end{equation}
that is equivalent to the state in Eq. \eqref{X3}.

In a similar way, we can also derive as limit process the following states from Eq. (\ref{(332)}), in order to get other elements of $\tau_{3}(\Sigma^{3}_{\textbf{2}})$ with one-multiranks equal to a permutation of $(233)$
\begin{equation}
|{\rm{W}}_3\rangle+\mathfrak{p}\{|\alpha\gamma\gamma\rangle\}\,,
\end{equation}
where $\alpha\neq\beta\neq\gamma\in\{0,1,2\}$. Additionally, the states
\begin{equation}
|\alpha\rangle_i|{\rm{GHZ}}_{2}^{(2)}\rangle_{jk}+|\beta\rangle_i|q_3\rangle^{\otimes2}_{jk}\,,
\end{equation}
where $\{i,j,k\}=\{1,2,3\}$, $\langle{\rm{GHZ}}_{2}^{(2)}|q_3\rangle^{\otimes2}=0$, and $\langle\alpha|\beta\rangle=0$, have tensor rank and border rank equal to three and four, respectively. So they can be as well considered as elements of $\tau_{3}(\Sigma^{3}_{\textbf{2}})$ with one-multiranks equal to a permutation of $(233)$.

Note that in the three-tangent family, we do not have any element with one-multirank equals to $(223)$ and its permutations. In fact, if the one-multirank of a given sate is equal to $(223)$, then the state lives in a smaller tensor product space, here is $\mathbbm{C}^2\otimes\mathbbm{C}^2\otimes\mathbbm{C}^3$, and its border rank is bounded by three, and it is a balanced case \cite{AOP09,BL13}. Let us consider two cases which have one-multiranks equal to a permutations of $(223)$: (1) Concerning Eq. (\ref{(322)}), one can consider the states $|{\rm{W}}_3\rangle+|\alpha\beta\gamma\rangle$. It is obvious that we can write these states as $|\alpha\alpha\beta\rangle+|\alpha\beta\rangle(|\alpha\rangle+|\gamma\rangle)+|\beta\alpha\alpha\rangle$ which clearly have tensor rank and border rank equal to three. (2) With a better choice of basis one can also consider the sates $|{\rm{W}}_3\rangle+|\beta\beta\gamma\rangle$. These states can be easily written as 
$|\alpha\alpha\rangle(-|\alpha\rangle+|\beta\rangle)+(|\alpha\rangle+|\beta\rangle)(|\alpha\rangle+|\beta\rangle)|\alpha\rangle+|\beta\beta\rangle(-|\alpha\rangle+|\gamma\rangle)$ which clearly have tensor rank and border rank equal to three.

The proper four-secant, i.e., the set $\sigma_{4}(\Sigma^{3}_{\textbf{2}})\backslash\sigma_{3}(\Sigma^{3}_{\textbf{2}})$, is the union of the secant hyperplanes $\mathcal{S}_{4}$ represented by Eq. (\ref{S-plane}). For instance, the following states which explicitly come from joining of four independent points in the Segre variety are elements of $\sigma_{4}(\Sigma^{3}_{\textbf{2}})$ with one-multirank equal to $(333)$
\begin{align}\nonumber
& |000\rangle+|011\rangle+|122\rangle+|221\rangle\,, \\
& |000\rangle+|111\rangle+|122\rangle+|221\rangle\,,
\end{align}
which can be considered as adding two different types of biseparable states $|{\rm{B}}^{(1)}\rangle$, or adding two different types of $|{\rm{GHZ}}_3^{(1)}\rangle$ states, or adding a biseparable state $|{\rm{B}}^{(1)}\rangle$ and a $|{\rm{GHZ}}_3^{(1)}\rangle$ state. Other examples of the proper four-secant family with one-multirank equals to $(333)$ can be considered as joining an independent point to the state in Eq. (\ref{GHZ3-2}) as follows:
\begin{equation}\label{GHZ-G}
|{\rm{GHZ}}_{3}^{(2)}\rangle+\mathfrak{p}\{|012\rangle\}\,,
\end{equation}
and
\begin{equation}\label{G3}
|{\rm{G}}_3\rangle=|{\rm{GHZ}}_{3}^{(2)}\rangle+|\omega_1\omega_1\omega_1\rangle\,,
\end{equation}
where $|\omega_1\rangle=|0\rangle+|1\rangle+|2\rangle$.

Using Eq. (\ref{Dicke-3-qudit}) one can see that the higher symmetric entangled state
\begin{equation}\label{D(111)}
|{\rm{D}}_{3}^{(1,1,1)}\rangle=|012\rangle+|021\rangle+|102\rangle+|120\rangle+|201\rangle+|210\rangle\,,
\end{equation}
is also an element of $\sigma_{4}(\Sigma^{3}_{\textbf{2}})$ with one-multirank equal to $(333)$. This is because we can relate the above-mentioned symmetric state to the monomial $xyz$ (actually all symmetric states can be related to some homogeneous polynomials since the variables in polynomials are invariant under permutation and each variable can be related to a basis) and we can decompose this monomial as follows
\begin{equation}\label{xyz}
xyz=\frac{1}{24}\big((x+y+z)^3+(-x-y+z)^3+(-x+y-z)^3+(x-y-z)^3\big)\,.
\end{equation}
So, using the Dirac notation, we can rewrite the state $|{\rm{D}}_{3}^{(1,1,1)}\rangle$ in Eq. (\ref{D(111)}) based on the above decomposition as follows:
\begin{equation}
|{\rm{D}}_{3}^{(1,1,1)}\rangle=\frac{1}{4}\big(|\omega_1\rangle^{\otimes 3}+|\omega_2\rangle^{\otimes 3}+|\omega_3\rangle^{\otimes 3}+|\omega_4\rangle^{\otimes 3}\big)\,,
\end{equation}
where $|\omega_2\rangle=-|0\rangle-|1\rangle+|2\rangle$, $|\omega_3\rangle=-|0\rangle+|1\rangle-|2\rangle$, and $|\omega_4\rangle=|0\rangle-|1\rangle-|2\rangle$. So the tensor rank and border rank of this state are at most 4. Moreover, using the Eq. \eqref{strassen}, one can see that the rank and border rank of this state are at least 4. Hence, both the tensor rank and the border rank of $|{\rm{D}}_{3}^{(1,1,1)}\rangle$ are four.

In the four-secant family, we do not have any element with one-multirank equals to $(233)$ and its permutations. Indeed, if one-multirank of a given sate is equal to $(233)$ then the state lives in $\mathbbm{C}^2\otimes\mathbbm{C}^3\otimes\mathbbm{C}^3$, and its border rank is bounded by three, but tensor rank can be three or four \cite{BL13}.

Concerning the closure of the four-secant variety, i.e., the four-tangent, we use the results of Ref. \cite{Bernardi}. The following state which has tensor rank and border rank equal to five and four, respectively,
\begin{equation}
|010\rangle+|100\rangle+|112\rangle+|201\rangle+|222\rangle\,,
\end{equation}
is an element of $\tau_{4}(\Sigma^{3}_{\textbf{2}})$ with one-multirank equal to $(333)$.

\begin{table*}[t]
\centering
\caption[Fine-structure classification of three-qutrit entanglement]{\label{table:5.1} Fine-structure classification of three-qutrit entanglement. Each column corresponds to a family ($\tau_k$ is the closure of $\sigma_k$ family that is split based on tensor rank). Within a column, each row corresponds to a subfamily. A subscript $k$ is used to indicate members appearing in different $k$-secant families while having the same one-multirank. A prime symbol is used for states in the $k$-tangent variety that appear (with the same one-multirank) in the boundary of the $k$-secant variety.}
\begin{tabularx}{\linewidth}{XXXXXXXX}
\hline\hline
& & & & & & & \\ [-2ex]
$\Sigma^{3}_{\textbf{2}}$ & $\sigma_{2}$ & $\tau_{2}$ & $\sigma_{3}$ & $\tau_{3}$ & $\sigma_{4}$ & $\tau_{4}$ & $\sigma_{5}$\\ [0.5ex]
\hline
& & & & & & & \\ [-2ex]
$|{\rm{Sep}}\rangle$ & $|{\rm{GHZ}}_3^{(1)}\rangle$ & $|{\rm{W}}_3\rangle$ & $|{\rm{GHZ}}_3^{(2)}\rangle$ & $|(333)'_{3}\rangle$ & $|(333)_{4}\rangle$ & $|(333)'_{4}\rangle$ & $|(333)_{5}\rangle$ \\ [0.5ex]
& $|{\rm{B}}_{i}^{(1)}\rangle_{i=1}^{3}$ & & $|(332)\rangle$ & $|(332)'\rangle$ & & &  \\ [0.5ex]
& & & $|(323)\rangle$ & $|(323)'\rangle$ & & & \\ [0.5ex]
& & & $|(233)\rangle$ & $|(233)'\rangle$ & & &  \\ [0.5ex]
& & & $|(322)\rangle$ & & & &  \\ [0.5ex]
& & & $|(232)\rangle$ & & & &  \\ [0.5ex]
& & & $|(223)\rangle$ & & & &  \\ [0.5ex]
& & & $|{\rm{B}}_{i}^{(2)}\rangle_{i=1}^{3}$ & & & &  \\ [0.5ex]
\hline\hline
\end{tabularx}
\end{table*}

Although any general state of three-qutrit system that has a non-zero determinant of matrix $\mathcal{F}$ in Eq. (\ref{strassen}) can be considered as an element of proper five-secant family, the following state which explicitly comes from joining of five independent points in the Segre variety and obeys Eq. (\ref{S-plane}), is an element of $\sigma_{5}(\Sigma^{3}_{\textbf{2}})$ with one-multirak equal to $(333)$
\begin{equation}
|{\rm{G}}_{3}\rangle+t~(|1\rangle+|2\rangle)\otimes(|0\rangle+|2\rangle)\otimes(|0\rangle+|1\rangle)\,,
\end{equation}
where $t\in\mathbbm{C}\setminus\{0,1\}$. The determinant of the matrix $\mathcal{F}$ for this state is $2t(1-t)$. Note that for $t=1$ the border rank is four and the tensor rank is also four, so the state belongs to the four-secant family in this case.

Since the highest tensor rank for a three-qutrit state is five \cite{BH13}, we do not need to construct the Zariski closure of the five-secant family.

It is worth noting that in the classification of three-qutrit states, all the elements in the proper four- and five-secant families are genuinely entangled.

To have an exhaustive classification, we have written each subfamily of three-, four-, and five-secant families in terms of their one-multiranks in Table \ref{table:5.1}. Also, we have used a prime for the states in tangent to discriminate secant and tangent families where they have same one-multiranks. In addition, we have put a subscript $k$ to indicate members appearing in different $k$-secant families with the same one-multirank.

\begin{figure}[th]
\center{\includegraphics[width=10cm]{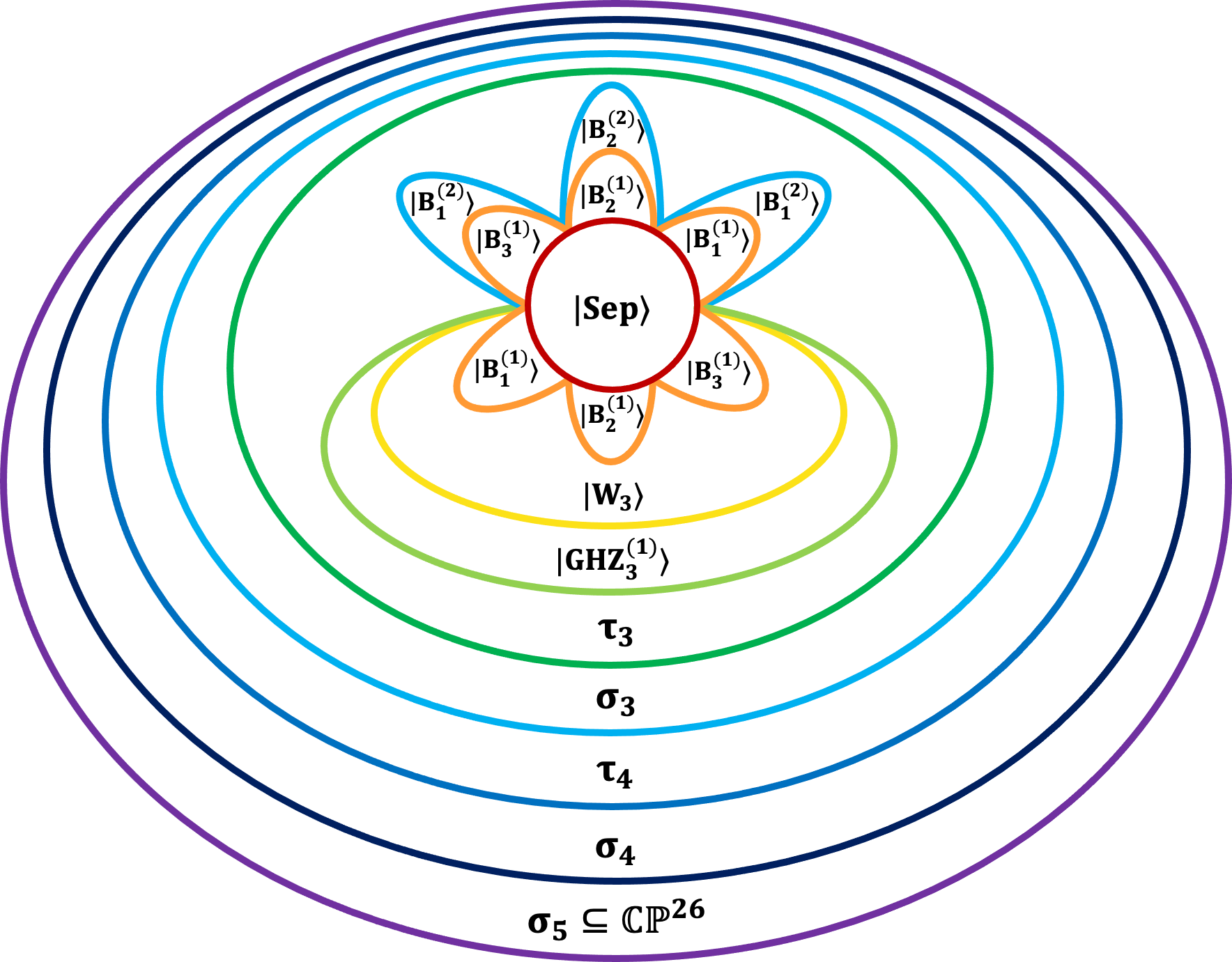}}
\caption[Petal-like classification of three-qutrit entanglement]{\label{fig:5.1} Petal-like classification of SLOCC orbits of three-qutrit states. By noninvertible SLOCC one can go from the outer classes to the inner ones  (from $\sigma_k$ to $\tau_k$ also in an approximate way), thus generating the entanglement hierarchy. Note that states $|{\rm{B}}_i^{(1)}\rangle$ appear with a double petal because to emphasize that they can be obtained starting from either $|{\rm{W}}_3\rangle$ states or $|{\rm{B}}_i^{(2)}\rangle$ states. In contrast, $|{\rm{B}}_i^{(2)}\rangle$ states cannot be obtained from $|{\rm{W}}_3\rangle$ states.}
\end{figure}

In summary, this classification provides us five secant families (eight secant / tangent families), and 23 subfamilies (Table \ref{table:5.1}). These classes are pictorially represented in Fig. \ref{fig:5.1}. {Obviously, a finer classification can be obtained by utilizing an extra SLOCC invariant (see \Cref{subsec.5.3.1}).}

\subsection{Finer classification of three-qutrit entanglement}\label{subsec.5.3.1}

Since the Schmidt measure can be defined as the logarithm of the tensor rank of a quantum state, one can conclude that tensor rank is itself an SLOCC invariant. Therefore, we can employ it to improve the classification algorithm by eventually splitting subfamilies into sub-subfamilies with the same tensor rank. Although determining the tensor rank of a given quantum state is NP hard \cite{Haastad90}, it 
could also results a useful tool for studying SLOCC interconversions among specific quantum states.

In Ref. \cite{BHO14}, a classification of three-qutrit entanglement is presented in five families according to the description of fundamental invariants provided in Refs. \cite{Nurmiev1,Nurmiev2}. It is also determined which fundamental invariants of ${\rm{SL}}(3,\mathbbm{C})^{\times{3}}$ vanish on tensors for each possible tensor rank. Here, we utilize tensor rank as an extra SLOCC invariant to present a finer classification of  three-qutrit entanglement with respect to the classification presented in Table \ref{table:5.1}, such that it contains the information of Ref. \cite{BHO14}.

To this end, consider the following state:
\begin{equation}\label{Y3}
|{\rm{Y}}_3\rangle=|002\rangle+|020\rangle+|200\rangle+|011 \rangle+|101\rangle+|110\rangle\,,
\end{equation}
and the following points:
\begin{equation}
q(\varepsilon)=\frac{1}{\varepsilon^2}\big((|0\rangle+\frac{\varepsilon}{\sqrt{2}}|1\rangle+\varepsilon^{2}|2\rangle)^{\otimes{3}}+(|0\rangle-\frac{\varepsilon}{\sqrt{2}}|1\rangle)^{\otimes{3}}-2|000\rangle\big)\,,
\end{equation}
that for all $\varepsilon\neq{0}$ correspond to ${\rm{GHZ}}^{(2)}$-type states. When $\varepsilon\to{0}$ we have
\begin{equation}\label{limY3}
\lim_{\varepsilon\to{0}}q(\varepsilon)=|002\rangle+|020\rangle+|200\rangle+|011 \rangle+|101\rangle+|110\rangle\,,
\end{equation}
that is equivalent to the state in Eq. \eqref{Y3}. So $|{\rm{Y}}_3\rangle$ can be considered as another element of $\tau_{3}(\Sigma^{3}_{\textbf{2}})$ with one-multirank equal to $(333)$. Moreover, it can asymptotically be obtained from $|{\rm{GHZ}}_{3}^{(2)}\rangle$. It is worth noting that the states in Eqs. \eqref{X3} and \eqref{Y3} are not equivalent since the tensor rank of the former is four, while of the later is five
In fact, we can rewrite Eq. \eqref{Y3} as follows:
\begin{align}\nonumber
|{\rm{Y}}_3\rangle=& \frac{1}{3}\left[(2|0\rangle+|2\rangle)^{\otimes{3}}-2(|0\rangle+|2\rangle)^{\otimes{3}}+|222\rangle\right] \\
&+\frac{1}{2\sqrt{3}i}\left[(2\xi+1)|0\rangle-|1\rangle)^{\otimes{3}}-((2\xi^2+1)|0\rangle-|1\rangle)^{\otimes{3}}\right],
\end{align}
with $\xi=\exp(2\pi i / 3)$. Hence, using the tensor rank as the third SLOCC invariant, we can split the subfamily $|(333)'_3\rangle\in\tau_{3}(\Sigma^{3}_{\textbf{2}})$ into two sub-subfamilies with tensor ranks equal to four and five, respectively (see Fig. \ref{fig:5.2}).
\begin{figure}[t]
\center{\includegraphics[width=6cm]{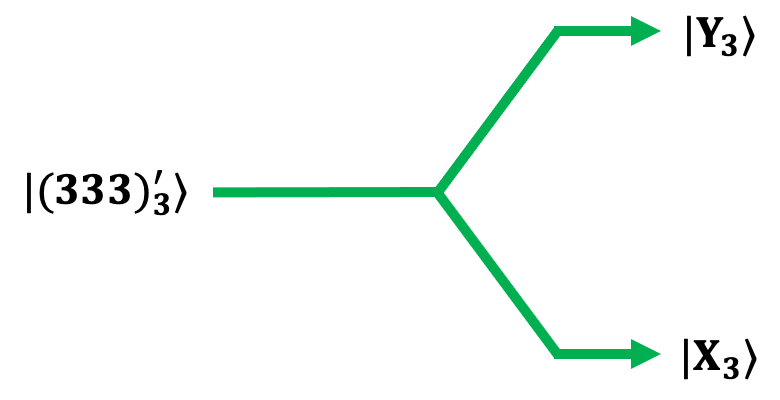}}
\caption[Sub-subfamilies of $|(333)'_3\rangle$]{\label{fig:5.2} Pictorial representation of the fact that using tensor rank as the third SLOCC invariant, the subfamily $|(333)'_3\rangle$ of Table \ref{table:5.1} can be split into two sub-subfamilies $|{\rm{X}}_3\rangle$ and $|{\rm{Y}}_3\rangle$ with tensor ranks equal to four and five, respectively.}
\end{figure}

\section{Generalization}\label{sec.5.4}
We generalize here some of the results found in the previous section to tripartite $\mathbbm{C}^d\otimes\mathbbm{C}^d\otimes\mathbbm{C}^d$ systems as well to $n$-qudit systems.

As one can see, going beyond the qubit setting, there are several types of $\rm{GHZ}$ states (see for instance, Eqs. (\ref{GHZ3-1}) and (\ref{GHZ3-2})). This is because we have different types of excitations rather than qubit systems. So we can draw the following conclusions for $d,n\geq{3}$:
\begin{equation}\label{GHZn-z}
|{\rm{GHZ}}_{n}^{(\zeta)}\rangle=|\alpha_1\rangle^{\otimes{n}}+\cdots+|\alpha_{\zeta+1}\rangle^{\otimes{n}}\in\sigma_{\zeta+1}(\Sigma^{n}_{\textbf{d}-\textbf{1}})\,,
\end{equation}
where $1\leq\zeta\leq{d-1}$ and $\alpha_i\neq\alpha_j\in\{0,1,\ldots,d-1\}$. Then, based on Eq. (\ref{GHZn-z}), we can create $(n-m+1)$-separable states as follows:
\begin{equation}
\mathfrak{p}\{|{\rm{GHZ}}_{m}^{(\zeta)}\rangle|q_d\rangle^{\otimes{n-m}}\}\in\sigma_{\zeta+1}(\Sigma^{n}_{\textbf{d}-\textbf{1}})\,,
\end{equation}
where $2\leq{m}\leq{n-1}$ and $|q_d\rangle$ is a general one-qudit state.

From Eq. (\ref{W3}), we can draw the following conclusion for $d,n\geq{3}$:
\begin{equation}\label{Wn}
|{\rm{W}}_n\rangle=|{\rm{D}}_{n}^{\mathfrak{p}(n-1,1,0,\cdots,0)}\rangle=\sum_{i}\mathfrak{p}_{i}\{|\alpha\rangle^{\otimes{n-1}}\otimes|\beta\rangle\}\in\tau_{2}(\Sigma^{n}_{\textbf{d}-\textbf{1}})\,,
\end{equation}
where $\alpha\neq\beta\in\{0,1,\ldots,d-1\}$ and
\begin{equation}\label{Dicke-n-qudit}
|{\rm{D}}_{n}^{\jmath}\rangle=\sqrt{\frac{\prod_{i}{j_i!}}{n!}}\sum_{\pi\in\mathfrak{S}_{n}}\pi\{|0\rangle^{\otimes{j_1}}\otimes\cdots\otimes|d-1\rangle^{\otimes{j_d}}\}\,,
\end{equation}
are the so-called $n$-qudit Dicke states, with excitations shown as $\jmath=(j_1,\ldots,j_d)$ where $j_1+\cdots+j_d=n$.

Furthermore, from Eq. (\ref{X3}) we can conclude, for $d,n\geq{3}$:
\begin{equation}\label{Xn}
|{\rm{X}}_n\rangle=|{\rm{W}}_n\rangle+|\gamma\gamma\gamma\rangle\in\tau_{3}(\Sigma^{n}_{\textbf{d}-\textbf{1}})\,,
\end{equation}
where $\gamma\in\{0,1,\ldots,d-1\}$ is different from $\alpha$ and $\beta$ in Eq. (\ref{Wn}).

For $d$-qudit states we have the following results, which respectively comes from Eqs. (\ref{GHZ-G}) and (\ref{G3}),
\begin{align}
&|{\rm{GHZ}}_{d}^{(d-1)}\rangle+\mathfrak{p}\{|01\cdots(d-1)\rangle\}&\in~\sigma_{d+1}(\Sigma^{d}_{\textbf{d}-\textbf{1}})\,, \\
&|{\rm{G}}_{d}\rangle=|{\rm{GHZ}}_{d}^{(d-1)}\rangle+|\Omega\rangle^{\otimes{d}}&\in~\sigma_{d+1}(\Sigma^{d}_{\textbf{d}-\textbf{1}})\,,
\end{align}
where $|\Omega\rangle=|0\rangle+\cdots+|d-1\rangle$.

Let us now discuss Dicke states. Since they correspond to monomials, up to scaling the variables, they are symmetric, i.e., they are invariant under any permutation of the parties. Thus, their symmetric tensor rank can be computed as the Waring rank of the corresponding monomials. We might expect that for monomials the symmetric rank and the tensor rank agree, especially since the smallest known counterexample to Comon’s conjecture (rank and symmetric rank of symmetric tensors are equal \cite{CGLM08}) by Shitov  \cite{Shitov18}  is of size $800\times800\times800$. It is widely expected that Comon's conjecture should be true for tensors of small size.

\begin{proposition}
For $d\geq{3}$, there is no symmetric entangled state in the proper locus of the $k$-secant variety of the Segre variety of $\mathbbm{P}^d\times\mathbbm{P}^d\times\mathbbm{P}^d$ with $k>\lceil\frac{\binom{d+2}{3}}{d}\rceil$, except in the case $d=5$, where the bound becomes $k>8$.
\end{proposition}
{\it Proof.}
The proof of the proposition follows directly from Theorem \ref{theo:AH95}, together with a comparison of Eq. (\ref{SymGenericRank}) for $n = 3$ and Eq. (\ref{TripartiteGenericRank}).
\qed

\begin{proposition}
For $n\geq{4}$ qutrits, Dicke states are not in the highest secant variety of $\mathbbm{P}({\rm{Sym}}^n\mathbbm{C}^3)$.
\end{proposition}
{\it Proof.}
Based on Theorem \ref{theo:Warink-rk} and Conjecture \ref{conj:Waring-brk}, for an $n$-qutrit Dicke state, the maximum border rank achieved when $\jmath=(\lceil\frac{n}{3}\rceil,\lfloor\frac{n}{3}\rfloor,n-\lceil\frac{n}{3}\rceil-\lfloor\frac{n}{3}\rfloor)$ in Eq. (\ref{Dicke-n-qudit}). So,
\begin{equation}
|{\rm{D}}_{n}^{(\lceil\frac{n}{3}\rceil,\lfloor\frac{n}{3}\rfloor,n-\lceil\frac{n}{3}\rceil-\lfloor\frac{n}{3}\rfloor)}\rangle\in\begin{cases}
  \sigma_{(\lfloor\frac{n}{3}\rfloor+1)(n-\lceil\frac{n}{3}\rceil-\lfloor\frac{n}{3}\rfloor+1)}(\Sigma^{n}_{\textbf{2}}) & \text{if}~n=3i~(i\in\mathbbm{N})\,, \\
  \tau_{(\lfloor\frac{n}{3}\rfloor+1)(n-\lceil\frac{n}{3}\rceil-\lfloor\frac{n}{3}\rfloor+1)}(\Sigma^{n}_{\textbf{2}})  & \text{otherwise}\,.
\end{cases}
\end{equation}
On the other hand, the generic symmetric rank of a tensor in ${\rm{Sym}}^n\mathbbm{C}^3$ is equal to $\left\lceil\frac{(n+1)(n+2)}{6}\right\rceil\,,$ except for $n=4$ where it is six. Hence, in contrast to multiqubit Dicke states, multiqutrit Dicke states are not in the highest secant variety in $\mathbbm{P}({\rm{Sym}}^n\mathbbm{C}^3)$.
\qed

Moreover, since there is no symmetric entangled state in the higher secant family of $3$-qutrit systems, it turns out that for $n\geq{3}$ qutrits, there is no symmetric entangled state in the higher secant variety.

\chapter{Persistent Tensors}\label{chap6}

\epigraph{``Beauty will save the world.''}{\textit{Fyodor Dostoevsky}}

Within this chapter, which is based on the Ref. \cite{GL22}, we construct a lower bound of the tensor rank for a new class of tensors, which we call \emph{persistent tensors}. We present three specific families of persistent tensors, of which the lower bound is tight. We show that there is a chain of degenerations between these three families of minimal-rank persistent tensors that can be used to study the entanglement transformation between them. In addition, we show that these three families of persistent tensors are indeed different generalizations of multiqubit $\rm{W}$ state within multiqudit systems and are geometrically in the orbit closure of multiqudit $\rm{GHZ}$ states. Consequently, we show that one can obtain every one of the generalizations of the $\rm{W}$ state from a multiqudit $\rm{GHZ}$ state via asymptotic Stochastic Local Operations and Classical Communication (SLOCC) with rate one. Finally, we extend the obtained lower bound of the tensor rank to direct sums with persistent summands and to even more general combinations of tensors, which we call \emph{block pyramidal tensors}. As a result, we show that the tensor rank is multiplicative under the Kronecker and tensor products of minimal-rank persistent tensors with the $\rm{GHZ}$ tensor.

\section{Preliminaries}\label{sec.6.2}

\subsection{Multipartite quantum states as multipartite tensors}\label{subsec.6.2.1}
A state of a multipartite quantum system can be considered as a multipartite tensor in the tensor product of Hilbert spaces of each individual subsystem. Let $\mathcal{H}_n^{\mathbf{d}}=\otimes_{i=1}^{n}\mathbbm{C}^{d_i}$ be the Hilbert space representing the state space of an $n$-partite quantum system where $\mathbf{d}=(d_1,\ldots,d_n)$ indicates the dimensions of the Hilbert spaces of each individual subsystem.

Here, we are dealing with two different notions of product for tensors. Suppose that we have two tensors $\mathcal{T}_1 \in\mathcal{H}_{n_1}^{\mathbf{d}}=\otimes_{i=1}^{n_1}\mathbbm{C}^{d_i}$ and $\mathcal{T}_2\in \mathcal{H}_{n_2}^{\mathbf{d'}}=\otimes_{i=1}^{n_2}\mathbbm{C}^{d'_i}$ corresponding to two multipartite quantum systems, the first with $n_1$ parties and the second with $n_2$ parties, respectively. Assume $n_1\leq n_2$ (without any loss of generality). The first product is the tensor product that corresponds to an ($n_1+n_2$)-partite system, and we denote it by $\mathcal{T}_1\otimes\mathcal{T}_2\in\mathcal{H}_{n_1}^{\mathbf{d}}\otimes\mathcal{H}_{n_2}^{\mathbf{d'}}$. The second product is the Kronecker product, which corresponds to an $n_2$-partite system, and we denote it by $\mathcal{T}_1\boxtimes\mathcal{T}_2 \in \mathcal{H}_{n_1}^{\mathbf{d}}\boxtimes\mathcal{H}_{n_2}^{\mathbf{d'}}$. In fact, $\mathcal{H}_{n_1}^{\mathbf{d}}\otimes\mathcal{H}_{n_2}^{\mathbf{d'}}=(\otimes_{i=1}^{n_1}\mathbbm{C}^{d_i})\otimes(\otimes_{j=1}^{n_2}\mathbbm{C}^{d'_j})$ and $\mathcal{H}_{n_1}^{\mathbf{d}}\boxtimes\mathcal{H}_{n_2}^{\mathbf{d'}}=(\otimes_{i=1}^{n_1}\mathbbm{C}^{d_i+d'_i})\otimes(\otimes_{j=n_1+1}^{n_2}\mathbbm{C}^{d'_j})$.

We also need the notion of direct sum of tensors. Let $\mathcal{T}_1\in\mathcal{H}_{n}^{\mathbf{d}}=\otimes_{i=1}^{n}\mathbbm{C}^{d_i}$ and $\mathcal{T}_2\in\mathcal{H}_{n}^{\mathbf{d'}}=\otimes_{i=1}^{n}\mathbbm{C}^{d'_i}$ be two tensors with the same number of factors. By considering the spaces $\mathbbm{C}^{d_i}$ and $\mathbbm{C}^{d'_i}$ as two summands of $(\mathbbm{C}^{d_i}\oplus \mathbbm{C}^{d'_i})\cong\mathbbm{C}^{d_i+d'_i}$, we can embed $\mathcal{T}_1$ and $\mathcal{T}_2$ into a larger Hilbert space $\mathcal{H}_{n}^{\mathbf{d}+\mathbf{d'}}\cong\otimes_{i=1}^{n}(\mathbbm{C}^{d_i}\oplus\mathbbm{C}^{d'_i})$. The direct sum $\mathcal{T}_1\oplus\mathcal{T}_2 \in \mathcal{H}_{n}^{\mathbf{d}+\mathbf{d'}}$ is the sum of the two tensors embedded in this way.

In this chapter, we take $\{|j\rangle\mid j\in\mathbbm{Z}_d\}$ as the canonical basis of $\mathbbm{C}^d$. We do not distinguish multipartite quantum states from the tensors that represent them. We denote specific tensors by calligraphic capital letters. It should be noted that we do not consider the normalization of the quantum states, since all properties that we work with can be defined for tensors in general and are invariant under scaling.

The state of a composite system is always expressible as a superposition of tensor products of the states of individual subsystems. A quantum state is called fully separable (or unentangled) if it can be written as a tensor product of individual subsystem states, i.e., $|\psi\rangle=|\varphi_1\rangle\otimes\cdots\otimes|\varphi_n\rangle$. Therefore, it is desirable to characterize the entanglement in a composite system. The tensor rank is a good tool for this purpose.

For example, consider the $n$-qubit $\rm{W}$ and $\rm{GHZ}$ states. An $n$-qubit $\rm{W}$ state, i.e.,
\begin{equation}\label{W}
\mathcal{W}_n=\sum_{\mathfrak{p}\in\mathfrak{S}_n}\mathfrak{p}\big\{|0\rangle^{\otimes(n-1)}|1\rangle\big\}=\sum_{i=0}^{n-1}|0\rangle^{\otimes(n-i-1)}|1\rangle|0\rangle^{\otimes{i}}\,,
\end{equation}
where $\mathfrak{p}$ denotes non-redundant elements of the symmetric group $\mathfrak{S}_n$, corresponds to a symmetric tensor in $\otimes^{n}\mathbbm{C}^2$ and its tensor rank and border rank are known to be $\rk(\mathcal{W}_n)=n\,$\footnote{A complete proof seems to be missing in the literature. In Ref. \cite[Theorem 3]{CCDJW10} the tensor rank of multiqubit Dicke states is presented. The proof is based on induction and the base case is cited to Ref. \cite{DVC00}. While the techniques from~\cite{DVC00} can be used to infer that the tensor rank of $n$-qubit $\W$ state is $n$, it does not seem clear at first sight.} and $\brk(\mathcal{W}_n)=2$, respectively. We will prove this fact that the tensor rank of a $n$-qubit $\rm{W}$ state is indeed $n$. A generalized $n$-qudit $\rm{GHZ}$ state, i.e.,
\begin{equation}\label{GHZ}
\mathcal{G}(d,n)=\sum_{j=0}^{d-1}|j\rangle^{\otimes{n}}\,,
\end{equation}
corresponds to a symmetric tensor in $\otimes^{n}\mathbbm{C}^d$ and its tensor rank and border rank are $\rk(\mathcal{G}(d,n))=\brk(\mathcal{G}(d,n))=d$.

It is known that a multiqudit $\rm{GHZ}$-equivalent state can be transformed into a quantum state $|\psi\rangle$ iff $d \geq \rk(|\psi\rangle)$ \cite{CDS08}. 
Therefore, the tensor rank of a quantum state can be characterized as follows
\begin{equation}\label{rank-SLOCC}
\rk(\mathcal{T})=\min\Big\{d~\big|~\mathcal{G}(d,n)\xrightarrow[]{\text{SLOCC}}\mathcal{T}\Big\}\,.
\end{equation}
Similarly, the border rank of a quantum state is the smallest $d$ such that a multiqudit $\rm{GHZ}$-equivalent state degenerates into it.
\begin{equation}\label{brank-degeneration}
\brk(\mathcal{T})=\min\Big\{d~\big|~\mathcal{G}(d,n)\xdashrightarrow[]{\text{SLOCC}}\mathcal{T}\Big\}\,.
\end{equation}

Now, assume that $|{\rm{S}}\rangle$ and $|{\rm{T}}\rangle$ are two quantum states in Hilbert spaces $\mathcal{H}$ and $\mathcal{H}'$ whose tensor ranks are $\rk(\mathcal{S})$ and $\rk(\mathcal{T})$, respectively. Then $\mathcal{T}\boxtimes\mathcal{S}\in\mathcal{H}\boxtimes\mathcal{H}'$, $\mathcal{T}\otimes\mathcal{S}\in\mathcal{H}\otimes\mathcal{H}'$, and we have the following inequalities
\begin{equation}\label{rank-inequalities}
\rk(\mathcal{S}\boxtimes\mathcal{T})\leq\rk(\mathcal{S}\otimes\mathcal{T})\leq\rk(\mathcal{S})\rk(\mathcal{T})\,.
\end{equation}
These operations (the Kronecker product and the tensor product) can be applied to the study of the (asymptotic) SLOCC interconversion between multipartite entangled states \cite{CDS08,CCDJW10,GMO20,YCGD10,CDS10,YGD14,VC15,VC17}. It is known that the tensor rank is not multiplicative under the Kronecker product. This is the reason why multicopy entanglement transformation by SLOCC is quite challenging. In Refs. \cite{CDS08,CCDJW10}, the tensor rank of two copies of the three-qubit $\rm{W}$ state is shown to be $\rk(\mathcal{W}_3\boxtimes\mathcal{W}_3)=7$, and in general the tensor rank of two copies of the $n$-qubit $\rm{W}$ state is shown to be $\rk(\mathcal{W}_n\boxtimes\mathcal{W}_n)=3n-2$. The tensor rank has also been shown to not be multiplicative under the tensor product \cite{CJZ18}. In Ref. \cite{CF18} it is shown that $\rk(\mathcal{W}_3\otimes\mathcal{W}_3)=8$ but the tensor rank of the tensor product of two $n$-qubit $\rm{W}$ states is still unknown.

In the following, we present the definitions of the SLOCC and asymptotic SLOCC transformations that are, respectively, known as restriction and degeneration in algebraic geometry and algebraic complexity theory. Although we have already given a definition for SLOCC transformation in \cref{chap2}, the following definition is more general.

\begin{definition}[SLOCC transformation]\label{def:SLOCC-transformation}
Let $|\psi\rangle\in U_1\otimes\cdots\otimes U_n$ and $|\varphi\rangle\in V_1\otimes\cdots\otimes V_n$ be two $n$-partite quantum states, where $U_i$ and $V_i$ are the Hilbert spaces of individual subsystems. We say that $|\psi\rangle$ can be transformed into $|\varphi\rangle$ via SLOCC (denoted by $|\psi\rangle\xrightarrow[]{\text{SLOCC}}|\varphi\rangle$) if there exist linear maps $A_i\colon U_i\to V_i$ such that
\begin{equation}\label{SLOCC}
(\otimes_{i=1}^{n}A_i)|\psi\rangle=|\varphi\rangle\,.
\end{equation}
\end{definition}

A generalization of the concept of SLOCC conversion is that of asymptotic SLOCC conversion. Here, instead of an exact transformation according to Eq. \eqref{SLOCC}, we consider asymptotic transformations between quantum states by local operations.

\begin{definition}[Degeneration]\label{def:degeneration}
Let $|\psi\rangle\in U_1\otimes\cdots\otimes U_n$ and $|\varphi\rangle\in V_1\otimes\cdots\otimes V_n$ be two $n$-partite quantum states, where $U_i$ and $V_i$ are the Hilbert spaces of individual subsystems. We say that $|\psi\rangle$ degenerates into $|\varphi\rangle$ with error degree $e$ via SLOCC (denoted by $|\psi\rangle\xdashrightarrow[]{\text{SLOCC}}|\varphi\rangle$) if there exist linear maps $A_i(\varepsilon)\colon U_i\to V_i$ depending polynomially on $\varepsilon$ such that
\begin{equation}\label{degenration}
(\otimes_{i=1}^{n}A_i(\varepsilon))|\psi\rangle=\varepsilon^d|\varphi\rangle+\sum_{l=1}^e\varepsilon^{d+l}|\tilde{\varphi}_l\rangle\,,
\end{equation}
for some state $|\tilde{\varphi}_l\rangle$ and $d\in\mathbbm{N}$ which is called the approximation degree.
\end{definition}
Indeed, if the quantum state $|\psi\rangle$ degenerates into the quantum state $|\varphi\rangle$, then $|\varphi\rangle$ can be approximated to an arbitrary precision by restrictions of $|\psi\rangle$, i.e.,
\begin{equation}\label{asymptoticSLOCC}
\lim_{\varepsilon\to 0}\frac{1}{\varepsilon^d}(\otimes_{i=1}^{n}A_i(\varepsilon))|\psi\rangle=|\varphi\rangle\,.
\end{equation}

In a similar spirit to the LOCC-based entanglement dilution \cite{Nielsen-Chuang}, we can use a quantity that indicates the minimum number of copies of a source quantum state $|\psi\rangle$ that can be used to obtain a single copy of the target quantum state $|\varphi\rangle$ by SLOCC transformation, in an asymptotic setting. This quantity is the rate of asymptotic SLOCC transformation from $|\psi\rangle$ into $|\varphi\rangle$ and is defined as follows
\begin{equation}\label{rate}
\omega(\psi,\varphi)=\lim_{n\to\infty}\frac{1}{n}\inf\left\{m\in\mathbbm{N}\,\big\vert~|\psi\rangle^{\boxtimes m}\xrightarrow[]{\text{SLOCC}}|\varphi\rangle^{\boxtimes n}\right\}\,.
\end{equation}

\subsection{Concise tensors}\label{subsec.6.2.2}
Informally, a tensor is concise if it cannot be written as a tensor in a smaller ambient space.
For example, a tensor $\mathcal{T}\in\mathbbm{C}^{a}\otimes\mathbbm{C}^{b}\otimes\mathbbm{C}^{c}$ is concise if its multilinear rank is $(a,b,c)$, which means that the tensor $\mathcal{T}$ uses all dimensions of the local spaces. In the following, we define concise tensors for multipartite systems.

\begin{definition}[Concise tensor]\label{def:concise}
A tensor $\mathcal{T}\in V_1\otimes\cdots\otimes V_n$ is called \emph{concise in the first factor}, or \emph{$1$-concise}, if $\mathcal{T}\notin V'_1\otimes V_2\otimes\cdots\otimes V_n$ with $V'_1\subsetneq V_1$. Conciseness in other factors is defined analogously. A tensor is called \emph{concise} if it is $i$-concise for all $i\in\{1,\ldots,n\}$.
\end{definition}

\begin{remark}
In quantum information theory, concise tensors correspond to quantum states with maximal one-to-group marginal entanglement (bipartite entanglement between each single party and the remaining parties). Conciseness can be alternatively characterized as all single-party reduced density matrices being full rank.
\end{remark}

The following lemma gives several equivalent characterizations of $1$-concise tensors.

\begin{lemma}\label{lem:concise-tfae}
Let $\mathcal{T}\in V_1\otimes\cdots\otimes V_n$ be a tensor and $\dim V_i=d_i$.
The following statements are equivalent:
\begin{enumerate}
    \item $\mathcal{T}$ is $1$-concise;
    \item For every non-zero covector $\langle f| \in V_1^{\vee}$ the contraction $\langle f| \mathcal{T}$ is non-zero;
    \item For every basis $\{|e_j\rangle\mid j\in\mathbbm{Z}_{d_1}\}$ of $V_1$ the decomposition $\mathcal{T} = \sum_{j = 0}^{d_1-1} |e_j\rangle\otimes\mathcal{T}_j$ has all $\mathcal{T}_j$ non-zero.
\end{enumerate}
\end{lemma}
\begin{proof}
$(1) \Leftrightarrow (2)$: 
Note that $\langle f| \mathcal{T} = 0$ iff $\mathcal{T}$ is in $(\ker\langle f|)\otimes V_2\otimes\cdots\otimes V_n$.
If $\mathcal{T}$ is $1$-concise, then $\langle f|\mathcal{T}=0$ iff $\ker\langle f|=V_1$, that is, $\langle f|=0$.
Conversely, assume that $\mathcal{T}$ is not $1$-concise, that is, $\mathcal{T} \in V_1'\otimes V_2 \otimes\cdots\otimes V_n$ with $V_1'\subsetneq V_1$.
There exists a non-zero covector $\langle f|$ vanishing on $V_1'$, and for this covector we have $\langle f|\mathcal{T}=0$.

$(2) \Rightarrow (3)$: Let $\{\langle f_j|\mid j\in\mathbbm{Z}_{d_1}\}$ be the dual basis. $\mathcal{T}_j = \langle f_{j}| \mathcal{T}$ is non-zero.

$(3) \Rightarrow (2)$: For every non-zero $\langle f|$ there exists a basis $\{|e_j\rangle\mid j\in\mathbbm{Z}_{d_1}\}$ such that $\langle f|e_{j}\rangle=0$ for $j>0$ and $\langle f|e_0\rangle=1$. We have $\langle f|\mathcal{T}=\mathcal{T}_0\neq0$.
\end{proof}

\begin{corollary}\label{cor:concise-matrix}
A tensor $\mathcal{T}\in V_1\otimes V_2$ is $1$-concise iff $\rk(\mathcal{T})=\dim V_1$.
\end{corollary}

We will need the following property of rank decompositions for $i$-concise tensors.
\begin{lemma}\label{lem:concise-rank}
Let $\mathcal{T}\in V_1\otimes\cdots\otimes V_n$ be an $i$-concise tensor. If
\begin{equation}\label{rank-decomposition}
\mathcal{T}= \sum_{p = 1}^r v_1^{(p)}\otimes\cdots\otimes v_n^{(p)}\,,
\end{equation}
is a tensor rank decomposition of $\mathcal{T}$, then the vectors $\{v_i^{(1)},\ldots,v_i^{(r)}\}$ span $V_i$.
\end{lemma}
{\it Proof.}
We prove the statement for $1$-concise tensors. Let the set of vectors $\{v_1^{(1)},\ldots, v_1^{(r)}\}$ span a vector space $U$. Note that all summands of rank decomposition are contained in $U\otimes V_2\otimes\cdots\otimes V_n$. It follows that $\mathcal{T}$ also lies in this space. Since $\mathcal{T}$ is a $1$-concise tensor, we then have $U=V_1$.

The proof for $i$-concise tensors is analogous.
\qed

\section{Persistent Tensors}\label{sec.6.3}

The substitution method is a method to obtain lower bounds for the tensor rank by zeroing the summands in a rank decomposition by applying appropriate projection maps to the tensor factors (see \cite[Ch.17]{BCS97} or \cite[Appx.~B]{AFT11}). A lower bound is obtained by keeping track of the number of summands zeroed. In this section, we introduce a class of tensors for which we can prove the tensor rank lower bounds by repeated application of the substitution method.
We call these tensors persistent.
For persistent tensors in $\otimes^n\mathbbm{C}^d$ we get a lower bound of $(n-1)(d-1)+1$.

\begin{definition}[Persistent tensor]\label{def:persistent-tensor}
We define \emph{persistent tensors} inductively.
\begin{itemize}
\item[(i)] A tensor $\mathcal{P}\in V_1\otimes V_2$ is persistent if it is $1$-concise.
\item[(ii)] A tensor $\mathcal{P}\in V_1\otimes\cdots\otimes V_n$ with $n > 2$ is persistent if it is $1$-concise and there exists a subspace $S \subsetneq V_1^{\vee}$ such that the contraction $\langle f|\mathcal{P}\in V_2\otimes\cdots\otimes V_n$ is persistent whenever $\langle f|\notin S$.
\end{itemize}
\end{definition}

The following lemma gives different characterizations of the class of persistent tensors which are useful for checking persistence.
\begin{lemma}\label{lem:persistence-equivalent}
Let $\mathcal{P}\in V_1\otimes\cdots\otimes V_n$ ($n>2$) be a persistent tensor with $\dim V_i=d_i$. The following statements are equivalent:
\begin{enumerate}
\item $\mathcal{P}$ is persistent.
\item $\mathcal{P}$ is $1$-concise and there exists a non-zero vector $|e\rangle\in V_1$ such that the following implication holds:
\[
\langle f|e\rangle\neq0 \Rightarrow \text{$\langle f|\mathcal{P}$ is persistent.}
\]
\item $\mathcal{P}$ is $1$-concise and there exists a non-zero vector $|e\rangle \in V_1$ such that the following implication holds:
\[
\langle f|e\rangle=1 \Rightarrow \text{$\langle f|\mathcal{P}$ is persistent.}
\]
\item For every basis $\{|e_j\rangle\mid j\in\mathbbm{Z}_{d_1}\}$ of $V_1$ the decomposition
\begin{equation}\label{eq:basis-decomposition-persistent}
\mathcal{P}=\sum_{j=0}^{d_1-1}|e_j\rangle\otimes\mathcal{P}_j\,,
\end{equation}
has all $\mathcal{P}_j$ non-zero and at least one of them is persistent.
\end{enumerate}
\end{lemma}
\begin{proof}
$(1) \Rightarrow (2)$:
Let $S \subsetneq V_1^{\vee}$ be the subspace in the definition of persistence. Choose any non-zero $|e\rangle \in S^{\perp}$.

$(2) \Rightarrow (1)$:
Take $S=|e\rangle^{\perp}$.

$(2) \Rightarrow (3)$: Trivial

$(3) \Rightarrow (2)$: If $\langle f|e\rangle=a\neq 0$, then $\langle f'|e\rangle=1$ for $\langle f'|=\frac{1}{a}\langle f|$. Hence $\frac{1}{a}\langle f|\mathcal{P}$ is persistent, and $\langle f|\mathcal{P}$ is persistent because persistence is scaling-invariant.

$(1) \Rightarrow (4)$:
Because $\mathcal{P}$ is $1$-concise, all $\mathcal{P}_i$ are non-zero. Let $\{\langle f_j|\mid j\in\mathbbm{Z}_{d_1}\}$ be the dual basis to $\{|e_j\rangle\mid j\in\mathbbm{Z}_{d_1}\}$. At least one $\langle f_j|$ does not lie in the subspace $S\subsetneq V_1^{\vee}$ from the definition of persistence. It follows that $\mathcal{P}_j=\langle f_j|\mathcal{P}$ is persistent.

$(4) \Rightarrow (3)$:
Let $\{|e_j\rangle\mid j\in\mathbbm{Z}_{d_1}\}$ be a basis such that the decomposition~\eqref{eq:basis-decomposition-persistent} has the minimum possible number of persistent tensors $\mathcal{P}_j$.
Assume without loss of generality that $\mathcal{P}_0$ is persistent.

For every $\alpha_1,\ldots,\alpha_{d_1-1}$ we can rewrite the decomposition~\eqref{eq:basis-decomposition-persistent} to get
\begin{equation}
\mathcal{P}=|e_0\rangle\otimes\big(\mathcal{P}_0+\sum_{j=1}^{d_1-1} \alpha_j\mathcal{P}_j\big)+\sum_{j=1}^{d_1-1}(|e_j\rangle-\alpha_j|e_0\rangle)\otimes\mathcal{P}_j.
\end{equation}
This is a decomposition corresponding to a different basis $|e_0'\rangle=\frac{1}{\alpha_0}|e_0\rangle, |e_j'\rangle=|e_j\rangle-\frac{\alpha_j}{\alpha_0}|e_0\rangle$.
Since the number of persistent slices in this decomposition cannot be less than that in the original, the tensor $\mathcal{P}_0+\sum_{j=1}^{d_1-1}\alpha_j\mathcal{P}_j$ is persistent.

Let $\langle f|\in V_1^{\vee}$ be a covector such that $\langle f|e_0\rangle=1$.
Note that $\langle f|\mathcal{P}=\mathcal{P}_0+\sum_{j=1}^{d_1-1}\langle f|e_j\rangle\mathcal{P}_j$ is persistent by the previous discussion.
We have proven $(3)$ with $|e\rangle=|e_0\rangle$.
\end{proof}

In the following, we present some examples of persistent and non-persistent tensors.

\begin{itemize}
\item[(i).] Non-persistent tensors:
\begin{itemize}
\item[1.] The diagonal tensor $\mathcal{G}(d,n)$ (correspondingly, $n$-qudit $\rm{GHZ}$ state) is not a persistent tensor for $n>2$ and $d\geq{2}$. This can be understood using \Cref{lem:persistence-equivalent}(4). We have
\begin{equation}
\mathcal{G}(d,n)=\sum_{j=0}^{d - 1}|j\rangle\otimes\mathcal{T}_j \quad \text{and} \quad \mathcal{T}_j=|j\rangle^{\otimes(n-1)}\,,
\end{equation}
where all $\mathcal{T}_j$ are not $1$-concise and therefore not persistent.
\item[2.] The Dicke state $\mathcal{D}_4^2=|0011\rangle+|0101\rangle+|0110\rangle+|1001\rangle+|1010\rangle+|1100\rangle$ is not a persistent tensor. This can be seen from the decomposition~\eqref{eq:basis-decomposition-persistent} corresponding to the basis $|\pm\rangle=|0\rangle\pm|1\rangle$ and the fact that $\mathcal{W}_3\pm\overline{\mathcal{W}}_3\equiv\mathcal{G}(2,3)$, where $\overline{\mathcal{W}}_3=|011\rangle+|101\rangle+|110\rangle$.
\item[3.] All unnormalized multiqubit Dicke states
\begin{equation}\label{Dicke-qubit}
\mathcal{D}_n^l=\sum_{\mathfrak{p}\in\mathfrak{S}_n}\mathfrak{p}\big\{|0\rangle^{\otimes(n-l)}\otimes|1\rangle^{\otimes l}\big\}\,,
\end{equation}
with $l$ excitations are not persistent tensors except when $l=1$. This can be understood from the previous example.
\end{itemize}
\item[(ii).] Persistent tensors:
\begin{itemize}
\item[1.] $\mathcal{W}_n$ is a persistent tensor because for every $\langle f|$ such that $\langle f|0\rangle=1$ we have
\begin{equation}
\langle f|\mathcal{W}_n=\mathcal{W}_{n-1}+\langle f|1\rangle|0\rangle^{\otimes(n-1)}\,,
\end{equation}
which is equivalent to $\mathcal{W}_{n-1}$. Repeating this construction, we arrive at the base case $\mathcal{W}_2=|01\rangle+|10\rangle$ which is a persistent tensor. Indeed, the $n$-qubit $\rm{W}$ state is the only symmetric persistent tensor in multiqubit systems.
\item[2.] An example of a nonsymmetric persistent tensor is the four-qubit state $\mathcal{T} = \alpha^2|0011\rangle+\beta^2|0101\rangle+(\alpha \pm \beta)^2|0110\rangle+|1001\rangle+|1010\rangle+|1100\rangle$. For every $\langle f|$ such that $\langle f|1\rangle = 1$ the contraction $\langle f|\mathcal{T}$ is equivalent to $\mathcal{W}_3$. This can be checked by computing the tangle~\cite{LLHL06}.
\item[3.] As another example, 3-qutrit $\mathcal{Y}_3=|002\rangle+|020\rangle+|200\rangle+|011\rangle+|101\rangle+|110\rangle$ is a persistent tensor. In the following, we will show that the $n$-qutrit $\Y$ state given by
\begin{equation}\label{Y}
\mathcal{Y}_n=\sum_{\mathfrak{p}\in\mathfrak{S}_n}\mathfrak{p}\{|0\rangle^{\otimes(n-2)}(|02\rangle+|11\rangle)\}\,,
\end{equation}
which corresponds to a symmetric tensor in $\otimes^n\mathbbm{C}^3$, is a persistent tensor.
\end{itemize}
\end{itemize}

\begin{theorem}\label{theo:PTLB}
If $\mathcal{P}\in V_1\otimes\cdots\otimes V_n$ is a persistent tensor and $\dim{V_k}=d_k$, then
\begin{equation}\label{PTLB}
\rk(\mathcal{P})\geq\sum_{k=1}^{n-1}(d_k-1) + 1\,.
\end{equation}
Moreover, in every rank decomposition
\begin{equation}\label{P-rank-decomposition}
\mathcal{P}=\sum_{p=1}^r u_1^{(p)}\otimes\cdots\otimes u_n^{(p)}\,,
\end{equation}
one can permute the summands in such a way that the rearranged decomposition
\begin{equation}\label{P-rank-decomposition-2}
\mathcal{P}=\sum_{p=1}^r v_1^{(p)}\otimes\cdots\otimes v_n^{(p)}\,,
\end{equation}
has the following property:
for every $j<n$ the vectors $\{v_j^{(D_j+1)},\ldots, v_j^{(D_j + d_j)}\}$ form a basis of $V_j$, where $D_j = \sum_{k = 1}^{j - 1} (d_k - 1)$ (for $j = 1$ we take $D_1 = 0$).
\end{theorem}
{\it Proof.}
We prove the statement by induction.

Base case: $n = 2$. If $n = 2$, then $\mathcal{P}$ is persistent iff it is $1$-concise. By \Cref{cor:concise-matrix} we have $\rk(\mathcal{P}) = d_1$, so the required lower bound is maintained.
Moreover, in any decomposition
\begin{equation}
\mathcal{P}=\sum_{p=1}^r u_1^{(p)}\otimes u_2^{(p)}\,,
\end{equation}
Regarding \Cref{lem:concise-rank}, $\Span\{u_1^{(1)},\ldots,u_1^{(r)}\}=V_1$. Therefore, we can permute the summands to get a decomposition
\begin{equation}
\mathcal{P}=\sum_{p=1}^r v_1^{(p)}\otimes v_2^{(p)}\,,
\end{equation}
where $\{v_1^{(1)},\ldots,v_1^{(d_1)}\}$ form a basis of $V_1$.

Consider now the case $n > 2$. Let $\mathcal{P}$ be a persistent tensor in $V_1 \otimes \cdots \otimes V_n$ and let Eq. \eqref{P-rank-decomposition} be a rank decomposition of $\mathcal{P}$. We rearrange the summands of this decomposition in several steps.

First, since $\mathcal{P}$ is $1$-concise, based on \Cref{lem:concise-rank}, $V_1=\Span\{u_1^{(1)}, \ldots, u_1^{(r)}\}$. So, we can permute the summands to obtain a rearranged decomposition in Eq. \eqref{P-rank-decomposition-2} such that $\{v_1^{(1)},\ldots,v_1^{(d_1)}\}$ form a basis of $V_1$. We choose the order of this basis in such a way that in the decomposition
\begin{equation}
\mathcal{P} = \sum_{k = 1}^{d_1} v_1^{(k)} \otimes \mathcal{P}_k\,,
\end{equation}
the tensor $\mathcal{P}_{d_1}$ is persistent.

Let $V_1'=\Span\{v_1^{(1)}, \ldots, v_1^{(d_1 - 1)}\}$.
As a second step, we separate the summands with $v_1^{(p)} \notin V_1'$ from those with $v_1^{(p)} \in V_1'$.
We rearrange the summands with indices from $d_1 + 1$ to $r$ to get a second rearranged decomposition
\begin{equation}\label{P-rank-decomposition-3}
\mathcal{P}=\sum_{p=1}^r w_1^{(p)}\otimes\cdots\otimes w_n^{(p)}\,,
\end{equation}
such that $w_1^{(p)}\notin V_1'$ if $d_1\leq p\leq s$ and $w_1^{(p)}\in V_1'$ if $p > s$ for an appropriate $s \leq r$.

Let $\langle f|$ be a covector such that $\langle f|v_1^{(k)}\rangle=0$ if $k<d_1$ and $\langle f|v_1^{(d_1)}\rangle=1$.
We have
\begin{equation}
\mathcal{P}_{d_1} =\langle f|\mathcal{P}=\sum_{p=d_1}^s \langle f|w_1^{(p)}\rangle w_2^{(p)}\otimes\cdots\otimes w_n^{(p)}\,,
\end{equation}
which is a rank decomposition for $\mathcal{P}_{d_1}$.

By the induction hypothesis, the number of summands in this decomposition, $s-d_1+1$, is at least $\sum_{k=2}^{n-1}(d_k-1)+1$, from which we obtain $r \geq s \geq \sum_{k = 1}^{n - 1} (d_k - 1) + 1$ as required.
Moreover, we can rearrange the summands to get the following
\begin{equation}
\mathcal{P}_{d_1}=\langle f|\mathcal{P} =\sum_{p = d_1}^s\langle f|y_1^{(p)}\rangle y_2^{(p)}\otimes\cdots\otimes y_n^{(p)},
\end{equation}
with $\{y_j^{(D_j + 1)},\ldots, y_j^{(D_j + d_j)}\}$ being a basis of $V_j$.
Applying the same permutation to the summands with indices from $d_1$ to $s$ in Eq. \eqref{P-rank-decomposition}, we get
\begin{equation}
\mathcal{P}=\sum_{p=1}^{r} y_1^{(p)}\otimes\cdots\otimes y_n^{(p)}\,.
\end{equation}
Note that~~~$\{y_1^{(1)},\ldots, y_1^{(d_1 - 1)}, y_1^{(d_1)}\}$~~is still a basis of~~ $V_1$ since~~ $y_1^{(d_1)}\notin V_1'$ ~~and \\ $V_1'=\Span\{y_1^{(1)},\ldots, y_1^{(d_1-1)}\}$.
\qed

Due to the following lemma, which is the essence of the substitution method (see also Ref. \cite[Appx.~B]{AFT11}) we have the following alternative proof of \Cref{theo:PTLB}.

\begin{lemma}\label{lem:substitution-method}
Let $\mathcal{T}\in V_1\otimes\cdots\otimes V_n$ be an $i$-concise tensor. For every subspace $V'_i\subsetneq V_i$ there exists a projection $\pi_i\colon V_i\to V'_i$ such that 
\begin{equation}\label{substitution-method}
\rk(\mathcal{T})-\rk(\pi_i\mathcal{T})\geq\dim{V_i}-\dim{V'_i}\,,
\end{equation}
where $\pi_i\mathcal{T}$ denotes the application of $\pi_i$ to the $i$-th factor of the tensor $\mathcal{T}$, i.e., $(\mathbb{1}^{\otimes(i-1)}\otimes\pi_i\otimes\mathbb{1}^{\otimes(n-i)})\mathcal{T}$.
\end{lemma}
{\it Proof.}
Suppose Eq. \eqref{rank-decomposition} is a tensor rank decomposition of $\mathcal{T}$, i.e., $\rk(\mathcal{T})=r$. By \Cref{lem:concise-rank} the vectors $\{v_i^{(1)},\ldots,v_i^{(r)}\}$ span $V_i$. Thus, there exists a subset $S_i\subset\{v_i^{(1)},\ldots,v_i^{(r)}\}$ consisting of $c_i=\dim{V_i}-\dim{V'_i}$ vectors such that $W_i=\Span\{S_i\}$ is complementary to $V'_i$. Consider the projection $\pi_i$ onto $V'_i$ along $W_i$. Applying it to the $i$-th factor of each summand of the tensor rank decomposition in Eq. \eqref{rank-decomposition} we obtain a decomposition of $\pi_i\mathcal{T}$ with at most $r - c_i$ summands, because the summands containing vectors from $S_i$ are sent to $0$. It follows that $\rk(\pi_i\mathcal{T})\leq\rk(\mathcal{T})-(\dim{V_i}-\dim{V'_i})$, and we obtain the required statement by rearranging the terms.
\qed

Therefore, the proof of \Cref{theo:PTLB} can be given as follows:

{\it Proof.}
We prove the statement by induction.

If $n = 2$, then $\mathcal{P}$ is persistent iff it is $1$-concise. By \Cref{cor:concise-matrix} we have $\rk(\mathcal{P}) = d_1$, so the required lower bound is maintained.

Consider now the case $n > 2$. Since $\mathcal{P}$ is a persistent tensor, by \Cref{lem:persistence-equivalent} there exists a vector $|e\rangle\in V_1$ such that for every covector $\langle f|$ in the dual space of $V_1$, $\langle f|\mathcal{P}$ is a persistent tensor whenever $\langle f|e\rangle\neq 0$. Let $V'_1=\Span\{|e\rangle\}$ be a $1$-dimensional subspace of $V_1$. Apply \Cref{lem:substitution-method} to find the projection $\pi_1\colon V_1\to V'_1$ such that $\rk(\mathcal{P})-\rk(\pi_1\mathcal{P})\geq d_1-1$. Since $V'_1$ is a $1$-dimensional subspace, $\pi_1\mathcal{P}=|e\rangle\otimes\mathcal{P}'$ for some $\mathcal{P}'\in V_2\otimes\cdots\otimes V_n$. It follows that $\rk(\mathcal{P}')=\rk(\pi_1\mathcal{P})$ and thus $\rk(\mathcal{P})\geq d_1-1+\rk(\mathcal{P}')$. Note that $\mathcal{P}'=\langle f|\mathcal{P}$ where $\langle f|\in V_1^{\vee}$ is the composition of $\pi_1$ with the linear map $V'_1\to\mathbbm{C}$ that sends $|e\rangle$ to $1$. So, we have $\langle f|e\rangle=1$ and $\mathcal{P}'$ is a persistent tensor.
By the induction hypothesis, we have $\rk(\mathcal{P}')\geq\sum_{k=2}^{n-1}(d_k-1)+1$, and therefore $\rk(\mathcal{P})\geq d_1-1+\rk(\mathcal{P}')\geq\sum_{k=1}^{n-1}(d_k-1)+1$.
\qed

\begin{corollary}\label{cor:W-rank}
The tensor rank of the $n$-qubit $\rm{W}$ state is $n$.
\end{corollary}
{\it Proof.}
The upper bound is obvious from the definition of $\mathcal{W}_n$, which has $n$ summands. According to \Cref{theo:PTLB}, the lower bound of the tensor rank of $\mathcal{W}_n$ is $n$ as it is a persistent tensor.
\qed

\section{Multiqudit generalization of $\rm{W}$ state}\label{sec.6.4}

We now introduce several families of multipartite tensors in $\otimes^n\mathbbm{C}^d$ (corresponding to $n$-qudit states) which can be thought as different generalizations of multiqubit $\rm{W}$ state within multiqudit systems. In the tripartite case, these tensors have been studied before in the context of matrix multiplication complexity in connection with the Coppersmith-Winograd algorithm~\cite{CW90}.

The first family we present we call $n$-qudit $\LL$ states
\begin{equation}\label{L-state}
\mathcal{L}(d,n)=\sum_{j_1+\cdots+j_n=d-1}|j_1\cdots j_n\rangle\,.
\end{equation}
These states are a special case of the weight states considered by Christandl et al. in Ref. \cite{CGFW21}.
For $n=d=3$ this tensor appeared as the $\Y$ state in Ref. \cite{GM21}.
In algebraic complexity theory, a tensor equivalent to $\mathcal{L}(d, 3)$ appeared as the structure tensor of truncated polynomial multiplication or as a ``lower triangular'' version of the cyclic group tensor~\cite{AVW18}.

The second family we introduce is the family of $n$-qudit $\M$ states, which generalizes the Coppersmith-Winograd tensors used in matrix multiplication algorithms~\cite{CW90} to the multipartite case.
We give two versions of these tensors, which are SLOCC equivalent.
The first version is
\begin{align}\nonumber
&\mathcal{M}(d,n)=\sum_{\mathfrak{p}\in\mathfrak{S}_n}\mathfrak{p}\Big\{\sum_{j=0}^{\lfloor\frac{d-1}{2}\rfloor}|0\rangle^{\otimes(n-2)}|j\rangle|d-j-1\rangle\Big\} \\ \label{M-state}
&=\sum_{i=0}^{n-1}|0\rangle^{\otimes(n-i-1)}|d-1\rangle|0\rangle^{\otimes i}+\sum_{i+k+l=n-2}\sum_{j=1}^{d-2}|0\rangle^{\otimes i}|j\rangle|0\rangle^{\otimes k}|d-j-1\rangle|0\rangle^{\otimes l}\,.
\end{align}
The second version is (for $d\geq3$)
\begin{align}\nonumber
\mathcal{M}'(d,n)&=\sum_{\mathfrak{p}\in\mathfrak{S}_n}\mathfrak{p}\Big\{|0\rangle^{\otimes(n-1)}|d-1\rangle+\sum_{j=1}^{d-2}|0\rangle^{\otimes(n-2)}|jj\rangle\Big\} \\ \label{M-state-2}
&=\sum_{i=0}^{n-1}|0\rangle^{\otimes(n-i-1)}|d-1\rangle|0\rangle^{\otimes i}+\sum_{i+k+l=n-2}\sum_{j=1}^{d-2}|0\rangle^{\otimes i}|j\rangle|0\rangle^{\otimes k}|j\rangle|0\rangle^{\otimes l}\,.
\end{align}
For $d=3$, the two versions are equal. For $d\ge4$, the two versions are SLOCC equivalent, because $\mathcal{M}$ is transformed into $\mathcal{M}'$ by applying to each factor the following change of basis
\begin{equation}
\begin{cases}
|j\rangle\mapsto\frac{1}{\sqrt{2}}(|j\rangle+\mathbf{i}|d-j-1\rangle) \\
|d-j-1\rangle\mapsto\frac{1}{\sqrt{2}}(|j\rangle-\mathbf{i}|d-j-1\rangle)
\end{cases} \text{for} \quad 1\leq j\leq\left\lfloor\frac{d-2}{2}\right\rfloor\,,
\end{equation}
where $\mathbf{i}=\sqrt{-1}$. In fact, it is a direct consequence of the Schmidt decomposition that $|j\rangle|d-j-1\rangle+|d-j-1\rangle|j\rangle$ is equivalent to $|j\rangle|j\rangle+|d-j-1\rangle|d-j-1\rangle$ \cite{Schmidt}.

The last family we present is the family of $n$-qudit $\N$ states defined as
\begin{equation}\label{N-state}
\mathcal{N}(d,n)=|0\rangle^{\otimes(n-1)}|d-1\rangle+\sum_{i=0}^{n-2}\sum_{j=1}^{d-1}|0\rangle^{\otimes(n-i-2)}|j\rangle|0\rangle^{\otimes i}|d-j-1\rangle\,.
\end{equation}
by applying the map $|j\rangle\mapsto|d-j-1\rangle$ to the last tensor factor, we get an equivalent tensor as follows
\begin{equation}\label{N-state-2}
\mathcal{N}'(d,n)=|0\rangle^{\otimes n}+\sum_{i=0}^{n-2}\sum_{j=1}^{d-1}|0\rangle^{\otimes(n-i-2)}|j\rangle|0\rangle^{\otimes i}|j\rangle\,.
\end{equation}
This form of $n$-qudit $\N$ state generalizes the tripartite tensor considered by Copersmith-Winograd as the asymmetric version of its construction~\cite{CW90}.

All three families of tensors generalize the $n$-qubit $\rm{W}$ state in the sense that $\mathcal{W}_n=\mathcal{L}(2, n)=\mathcal{M}(2, n)=\mathcal{N}(2, n)$. Moreover, it is easy to see that $\mathcal{W}_n^{\boxtimes 2}$ is equivalent to $\mathcal{M}(4, n)$ under the identification $|00\rangle \mapsto |0\rangle, |01\rangle \mapsto |1\rangle, |10\rangle \mapsto |2\rangle, |11\rangle \mapsto |3\rangle$.

All these families consist of persistent tensors, which follows from a more general statement below.

\begin{theorem}
Let $\mathcal{T}\in\otimes^n\mathbbm{C}^d$ be a tensor of the form
\begin{equation}
\mathcal{T}=\sum_{j_1+\cdots+j_n< d} t_{j_1\cdots j_n} |j_1\cdots j_n\rangle\,.
\end{equation}
If the coefficients before $|0\rangle^{\otimes(n-i-2)}|j\rangle|0\rangle^{\otimes i}|d-j-1\rangle$ are non-zero for all $i\leq n-2$ and $j\leq d-1$ then $\mathcal{T}$ is a persistent tensor.
\end{theorem}
\begin{proof}
We prove the statement by induction on $n$.

If $n=2$, then 
\begin{equation}
\mathcal{T} =\sum_{j=0}^{d-1} t_{j,d-j-1}|j\rangle\otimes\Big(|d-j-1\rangle+\sum_{k = 0}^{d-j-2} \frac{t_{jk}}{t_{j,d-j-1}} |k\rangle\Big)\,,
\end{equation}
has matrix rank $d$ and therefore is $1$-concise.

For $n>2$, note that for $\langle f|=\sum_{j=0}^{d-1} f_j\langle j|$ we have
\begin{equation}
\langle f|\mathcal{T}={\hspace{-2mm}\sum_{j_2+\cdots+j_n< d}} {\hspace{-2mm}s_{j_2\cdots j_n}|j_2 \cdots j_n\rangle} \quad \text{where} \quad s_{j_2\cdots j_n}={\hspace{-2mm}\sum_{j_1=0}^{d-(j_2+\cdots+ j_n)-1}} {\hspace{-4mm} f_{j_1} t_{j_1\cdots j_n}}\,,
\end{equation}

If $\langle f|\neq 0$ and $j$ is the minimum index such that $f_j \neq 0$, then $s_{0\cdots0,d-j-1}= f_j t_{j0\cdots0,d-j-1}\neq0$, so $\langle f|\mathcal{T} \neq 0$. By \Cref{lem:concise-tfae} $\mathcal{T}$ is $1$-concise. Additionally, if $\langle f|0\rangle= 1$ and $j_2+\cdots+j_n=d-1$, then $s_{j_2\cdots j_n}=t_{0j_2\cdots j_n}$. In particular, the coefficients before $|0\rangle^{\otimes (n-i-3)} |j\rangle |0\rangle^{\otimes i} |d-j-1\rangle$ in $\langle f|\mathcal{T}$ are non-zero, so by the induction hypothesis $\langle f|\mathcal{T}$ is persistent.
Therefore, $\mathcal{T}$ is persistent by \Cref{lem:persistence-equivalent}(3) with $|e\rangle=|0\rangle$.
\end{proof}

\begin{corollary}
The tensors $\mathcal{L}(d,n)$, $\mathcal{M}(d,n)$ and $\mathcal{N}(d,n)$ are persistent.
\end{corollary}

The persistence allows us to use the lower bound of~\Cref{theo:PTLB} to find the ranks of $\mathcal{L}(d, n)$, $\mathcal{M}(d, n)$ and $\mathcal{N}(d, n)$.

\begin{theorem}
The tensor rank of the $n$-qudit $\LL$ state is $\rk(\mathcal{L}(d,n))=(n-1)(d-1)+1$.
\end{theorem}
\begin{proof}
The lower bound follows from~\Cref{theo:PTLB}.
For the upper bound, we give an explicit decomposition.

Let $r=(n-1)(d-1)+1$ and let $\zeta =\exp(\frac{2\pi\mathbf{i}}{r})$ be a primitive $r$-th root of unity.
Using the property of roots of unity
\begin{equation}\label{roots-unity-2}
\sum_{p = 0}^{r-1} \zeta^{p q}=
\begin{cases}
r & \text{if $r\,|\,q$} \\
0 & \text{otherwise,}
\end{cases}
\end{equation}
we see that
\begin{align}\nonumber
\mathcal{L}(d,n)&=\sum_{j_1+\cdots+j_n=d-1}|j_1\cdots j_n\rangle
\\ \nonumber
&=\frac{1}{r}\sum_{s=0}^{n(d-1)}\sum_{j_1+\cdots+j_n=s} \Big(\sum_{p=0}^{r-1}\zeta^{p(s-d+1)}\Big)|j_1\cdots j_n\rangle  \\ \nonumber
&=\frac{1}{r}\sum_{j_1+\cdots+j_n=s}\Big(\sum_{p=0}^{r-1}\zeta^{p(\sum_{i=1}^n j_i-d+1)}\Big)|j_1\cdots j_n\rangle
\\
&=\frac{1}{r}\sum_{p=0}^{r-1}\zeta^{p(-d+1)}\Big(\sum_{j=0}^{d-1}\zeta^{p j}|j\rangle\Big)^{\otimes n}\,.
\end{align}
\end{proof}

\begin{corollary}\label{cor:tensor-rank-Y}
The tensor rank of the $n$-qutrit $\Y$ state ($\mathcal{Y}_n = \mathcal{L}(3, n)$) is $2n-1$.
\end{corollary}

\begin{theorem}\label{lem:Tensor-rank-M}
The tensor rank of the $n$-qudit $\M$ state is $\rk(\mathcal{M}(d,n))=(n-1)(d-1)+1$.
\end{theorem}
\begin{proof}
For the lower bound, we again use~\Cref{theo:PTLB}.

We prove the upper bound for the SLOCC equivalent tensor $\mathcal{M}'(d, n)$.
Note that $\mathcal{M}'(d,n)$ is the sum of $d - 2$ tensors of the form $\sum_{\mathfrak{p} \in \mathfrak{S}_n} \mathfrak{p}\lbrace|0\rangle^{\otimes(n-2)}|jj\rangle\rbrace$, which are equivalent to the Dicke state $\mathcal{D}^2_n$, and the tensor $\sum_{\mathfrak{p}\in\mathfrak{S}_n} \mathfrak{p}\lbrace|0\rangle^{\otimes(n-1)}|d-1\rangle\rbrace$, which is equivalent to $\mathcal{W}_n$.
The rank of $\mathcal{D}^2_n$ is $n-1$ \cite{CDS10} and the rank of $\mathcal{W}_n$ is $n$. We obtain the required upper bound by summing these ranks.
\end{proof}

\begin{corollary}\label{cor:tensor-rank-W2}
$\rk(\mathcal{W}_n\boxtimes\mathcal{W}_n)=\rk(\mathcal{M}(4,n))=3n-2$ which already has been obtained in Ref. \cite{CCDJW10}.
\end{corollary}

\begin{theorem}
The tensor rank of the $n$-qudit $\N$ state is $\rk(\mathcal{N}(d,n))=(n-1)(d-1)+1$.
\end{theorem}
\begin{proof}
The lower bound again follows from~\Cref{theo:PTLB}, and the upper bound is obvious from the definition of $\mathcal{N}$, which has $(n-1)(d-1)+1$ summands.
\end{proof}

Thus, the three families of persistent tensors we introduced have rank $(n-1)(d-1)+1$, which is the minimum possible rank for a persistent tensor in $\otimes^n\mathbbm{C}^d$.
We can show that their border rank also has the minimum possible value $d$.

\begin{theorem}\label{theo:LMN-degenerations}
For $n\geq3$ we have a chain of degenerations
\begin{equation}\label{LMN-degenerations}
\mathcal{L}(d,n)\xdashrightarrow[]{\text{SLOCC}}\mathcal{M}(d,n)\xdashrightarrow[]{\text{SLOCC}}\mathcal{N}(d,n)\,.
\end{equation}
\end{theorem}
\begin{proof}
To degenerate from $\mathcal{L}(d, n)$ to $\mathcal{M}(d, n)$, apply the family of linear maps $A(\varepsilon)=\mathrm{diag}(\varepsilon^{-2},\varepsilon^{n-2},\ldots,\varepsilon^{n-2}, \varepsilon^{2(n-1)})$ to each tensor factor and let $\varepsilon\to0$.

To degenerate from $\mathcal{M}(d, n)$ to $\mathcal{N}(d, n)$, apply $A(\varepsilon)=\mathrm{diag}(1,\varepsilon,\ldots,\varepsilon,1)$ to the first $n-1$ factors and $A(\varepsilon)^{-1}$ to the last factor, and let $\varepsilon\to 0$.
\end{proof}

\begin{theorem}\label{theo:LMN-brank}
$\brk(\mathcal{L}(d,n)) = \brk(\mathcal{M}(d,n)) = \brk(\mathcal{N}(d,n)) = d$.
\end{theorem}
\begin{proof}
The lower bound follows from the $1$-conciseness of the tensors.

We give an explicit approximation for $\mathcal{L}(d, n)$.
Let $\xi=\exp(\frac{2\pi\mathbf{i}}{d})$ be the primitive $d$-th root of unity.
Using the property of roots of unity as follows
\begin{equation}\label{roots-unity}
\sum_{p = 0}^{d-1} \xi^{p q}=
\begin{cases} d & \text{if $d\,|\,q$} \\
0 & \text{otherwise,}
\end{cases}
\end{equation}
we can give an approximation of the $n$-qudit $\LL$ state as follows
\begin{equation}\label{border-rank-L}
\mathcal{L}(d,n)=\lim_{\varepsilon\to 0} \frac{1}{d\,\varepsilon^{d-1}}\sum_{p = 0}^{d-1} \xi^p \Big(\sum_{j=0}^{d-1} \varepsilon^j\xi^{p j}|j\rangle\Big)^{\otimes n}\,.
\end{equation}
The upper bound for $\mathcal{M}(d,n)$ and $\mathcal{N}(d,n)$ can be transferred from $\mathcal{L}(d,n)$ using degenerations from~\Cref{theo:LMN-degenerations}.
\end{proof}

Alternatively, we can give approximate decompositions for $\mathcal{M}'(d, n)$ and $\mathcal{N}'(d, n)$ as follows
\begin{align}\nonumber
\mathcal{M}'(d,n)=\lim_{\varepsilon\to 0}\frac{1}{\varepsilon^2}\Big(&\sum_{j=1}^{d-2}(|0\rangle+\varepsilon|j\rangle+\frac{\varepsilon^2}{d-2}|d-1\rangle)^{\otimes n}
\\ \label{M-state-BR}
&-\frac{1}{\varepsilon}\big(|0\rangle+\varepsilon^2\sum_{j=1}^{d-2}|j\rangle\big)^{\otimes n}-(d-2-\frac{1}{\varepsilon})|0\rangle^{\otimes n}\Big)\,,
\end{align}
\begin{equation}\label{N-state-BR}
\mathcal{N}'(d,n)=\lim_{\varepsilon\to 0}\frac{1}{\varepsilon}\Big(\sum_{j=1}^{d-1}(|0\rangle+\varepsilon|j\rangle)^{\otimes (n-1)} \otimes |j\rangle + |0\rangle^{\otimes (n-1)}\otimes (\varepsilon|0\rangle - \sum_{j = 1}^{d - 1}|j\rangle)\Big)\,.
\end{equation}

An immediate result of \Cref{theo:LMN-brank} is that the multiqudit $\LL$, $\M$, and $\N$ states are geometrically in the orbit closure of the multiqudit $\rm{GHZ}$ states, similarly to how the $\rm{W}$ state lie in the orbit closure of the $\rm{GHZ}$ state. Again, we see that our families of minimal-rank persistent tensors can be considered as generalizations of multiqubit $\rm{W}$ state within multiqudit systems.

\section{Multiqudit entanglement transformation}\label{sec.6.5}
Here, we study the SLOCC interconversion between the $n$-qudit $\rm{GHZ}$ state and each generalization of the $\rm{W}$ state, that is, the $n$-qudit $\LL$, $\M$, and $\N$ states. Concerning the chain of degenerations between $\mathcal{L}(d,n)$, $\mathcal{M}(d,n)$, and $\mathcal{N}(d,n)$ in Eq. \eqref{LMN-degenerations}, we are able to study the asymptotic SLOCC transformation between them.

The following proposition relates the Schmidt rank to the asymptotic SLOCC transformation \cite{VC15}.
\begin{proposition}[Ref. \cite{VC15}]\label{prop:rate-SchmidtRank}
Let $|\psi\rangle\in V_1\otimes\cdots\otimes V_n$ be an $n$-partite quantum state and let $\rk_S(\psi)$ denotes the Schmidt rank of $|\psi\rangle$ for any bipartite cut (bipartition) $S|\overline{S}$ where $S\subseteq[n]$ and $\overline{S}=[n]\setminus S$. For any bipartitions, we have
\begin{equation}\label{rate-SchmidtRank}
\omega(\psi,\varphi)\geq\max_{S\subseteq[n]}\frac{\log\rk_S(\varphi)}{\log\rk_S(\psi)}\,,
\end{equation}
where $\omega(\psi,\phi)$ is the rate of asymptotic SLOCC transformation from $|\psi\rangle$  into $|\phi\rangle$ (see Eq. \eqref{rate}).
\end{proposition}

The following theorem relates degeneration to the asymptotic SLOCC transformation (see Ref. \cite{VC15} for a proof).

\begin{theorem}[Ref. \cite{VC15}]\label{theo:rate-degenration}
Let $|\psi\rangle$ and $|\varphi\rangle$ be two n-partite quantum states. If $|\psi\rangle$ degenerates into $|\varphi\rangle$ via SLOCC, then $\omega(\psi,\varphi)\leq 1$.
\end{theorem}

\begin{theorem}\label{theo:SLOCC-LMN}
An $n$-qudit $\LL$ state can be transformed into an $n$-qudit $\M$ state by asymptotic SLOCC with rate one. An $n$-qudit $\M$ state can be transformed into an $n$-qudit $\N$ state by asymptotic SLOCC with rate one. Formally,
\begin{align}\label{rate-LM}
&\omega(\mathcal{L}(d,n),\mathcal{M}(d,n))=1\,,
\\ \label{rate-MN}
&\omega(\mathcal{M}(d,n),\mathcal{N}(d,n))=1\,.
\end{align}
\end{theorem}
{\it Proof.}
The Schmidt rank of the $n$-qudit $\LL$, $\M$, and $\N$ states across any bipartition is $d$. Therefore, with respect to \Cref{prop:rate-SchmidtRank}, the rates the aforementioned asymptotic SLOCC transformations is greater than one. Regarding \Cref{theo:LMN-degenerations}, the upper bounds of the both rates are one. This concludes the proof.
\qed

\begin{corollary}
Based on \Cref{theo:SLOCC-LMN}, we can conclude
\begin{equation}\label{rate-LN}
\omega(\mathcal{L}(d,n),\mathcal{N}(d,n))=1\,.
\end{equation}
\end{corollary}

\begin{theorem}\label{theo:SLOCC-GHZ-L}
An $n$-qudit $\rm{GHZ}$ state can be transformed into an $n$-qudit $\LL$ state by asymptotic SLOCC with rate one, i.e.,
\begin{equation}\label{rate-GHZ-L}
\omega(\mathcal{G}(d,n),\mathcal{L}(d,n))=1\,.
\end{equation}
\end{theorem}
{\it Proof.}
The Schmidt rank of the $n$-qudit $\rm{GHZ}$ states across any bipartition is $d$. Therefore, with respect to \Cref{prop:rate-SchmidtRank}, the rate of the aforementioned asymptotic SLOCC transformation is greater than one. Regarding \Cref{theo:LMN-degenerations} and \Cref{theo:LMN-brank}, we have the following chain of degenerations
\begin{equation}\label{LMN-degenerations}
\mathcal{G}(d,n)\xdashrightarrow[]{\text{SLOCC}}\mathcal{L}(d,n)\xdashrightarrow[]{\text{SLOCC}}\mathcal{M}(d,n)\xdashrightarrow[]{\text{SLOCC}}\mathcal{N}(d,n)\,.
\end{equation}
Therefore, the upper bound of the rate is one. This concludes the proof.
\qed

\begin{corollary}
Concerning \Cref{theo:SLOCC-LMN} and \Cref{theo:SLOCC-GHZ-L}, one can conclude the following results
\begin{align}\label{rate-GHZ-M}
&\omega(\mathcal{G}(d,n),\mathcal{M}(d,n))=1\,,
\\ \label{rate-GHZ-N}
&\omega(\mathcal{G}(d,n),\mathcal{N}(d,n))=1\,.
\end{align}
\end{corollary}

\section{Direct sums of persistent tensors}\label{sec.6.6}
The lower bound obtained in~\Cref{theo:PTLB} can be extended to direct sums with persistent summands and to even more general combinations of tensors, which we call \emph{block pyramidal tensors}. Our lower bound techniques for direct sums and pyramidal tensors generalize some of the constructions of Alder and Strassen in Ref. \cite{AS81} to multipartite tensors. Recent work of Buczy\'{n}ski et al. in Ref. \cite{BPR20} uses similar ideas to prove the rank additivity for some tripartite tensors using the substitution method.

\begin{theorem}\label{theo:PTLB-DirectSum}
Let $\mathcal{T}\in U_1\otimes\cdots\otimes U_n$ be an arbitrary tensor of rank $\rk(\mathcal{T})$. If $\mathcal{P}\in V_1\otimes\cdots\otimes V_n$ is a persistent tensor and $\dim V_k=d_k$, then
\begin{equation}\label{PTLB-DirectSum}
\rk(\mathcal{T}\oplus\mathcal{P})\geq\rk(\mathcal{T})+\sum_{k=1}^{n-1}(d_k-1)+1\,.
\end{equation}
\end{theorem}
\begin{proof}
Consider a tensor rank decomposition
\begin{equation}\label{eq:decomp-DirectSum}
\mathcal{T} \oplus \mathcal{P} =\sum_{p=1}^{r}w_1^{(p)}\otimes\cdots\otimes w_n^{(p)}\,,
\end{equation}
where $w_j^{(p)} = u_j^{(p)} + v_j^{(p)}$ with $u_j^{(p)} \in U_j$ and $v_j^{(p)} \in V_j$.

Let $\pi_{V_j}\colon U_j\oplus V_j\to V_j$ be the canonical projection onto $V_j$ and let $\pi_V=\pi_{V_1}\otimes\cdots \otimes\pi_{V_n}$
Applying $\pi_V$ to the both sides of the decomposition, we get
\begin{equation}
\mathcal{P} =\sum_{p=1}^{r} v_1^{(p)}\otimes\cdots\otimes v_n^{(p)}\,.
\end{equation}

By \Cref{theo:PTLB} we can assume that for $j < n$, $\{v_j^{(D_j+1)},\ldots,v_j^{(D_j+d_j)}\}$ forms a basis of $V_j$, where $D_j = \sum_{k = 1}^{j - 1} (d_k - 1)$.

Let $\pi_{U_n} \colon U_n\oplus V_n \to U_n$ be the canonical projection onto $U_n$. Define the new projections $\Pi_j \colon U_j \oplus V_j \to U_j$ for $j<n$ as
\begin{align}
\begin{cases}
\Pi_j v_j^{(D_j+k)}=-u_j^{(D_j+k)}\,, & \text{$1\leq k\leq d_j$}\,, \\
\Pi_j u=u\,, & \text{$u\in U_j$}\,,
\end{cases}
\end{align}
so that we have $\Pi_j w_j^{(D_j+1)}=\cdots=\Pi_j w_j^{(D_j+d_j)}=0$.
Let $\Pi=\Pi_1\otimes\cdots\otimes\Pi_{n-1}\otimes \pi_{U_n}$.

Note that $\Pi$ sends the first $s=\sum_{k=1}^{n-1}(d_k-1)+1$ summands of the decomposition in Eq.~\eqref{eq:decomp-DirectSum} to zero. Furthermore, $\Pi(\mathcal{T}\oplus\mathcal{P}) =\mathcal{T}$. Therefore, applying $\Pi$ to both sides of the decomposition in Eq.~\eqref{eq:decomp-DirectSum}, one gets a rank decomposition for $\mathcal{T}$ with $r-s$ summands, so $r\geq\rk(\mathcal{T})+s$.
\end{proof}

As a corollary, we find that the rank is additive for direct sums with persistent tensors of minimum rank. Similar rank additivity statements are known for tripartite tensors, the rank of which can be determined using the substitution method~\cite{BPR20,LM17}.

\begin{corollary}
Let $V_1,\ldots,V_n$ be vector spaces with $\dim V_k=d_k$. If $\mathcal{P} \in V_1\otimes\cdots\otimes V_n$ is a persistent tensor of the minimum possible rank $\rk(\mathcal{P})=\sum_{k=1}^{n-1}(d_k-1)+1$, then for any arbitrary $n$-partite tensor $\mathcal{T}$ we have
\begin{equation}
\rk(\mathcal{T}\oplus\mathcal{P})=\rk(\mathcal{T})+\rk(\mathcal{P})\,.
\end{equation}
\end{corollary}

Furthermore, the tensor rank is multiplicative under the Kronecker and tensor products of the persistent tensor of the minimum rank with $\rm{GHZ}$ tensors.

\begin{lemma}\label{lem:GHZ-P}
Let $V_1,\ldots,V_n$ be vector spaces with $\dim V_k = d_k$. If $\mathcal{P} \in V_1\otimes\cdots\otimes V_n$ is a persistent tensor of the minimum possible rank $\rk(\mathcal{P}) = \sum_{k=1}^{n-1}(d_k-1)+1$, then
\begin{equation}
\rk(\mathcal{G}(d, n)\boxtimes\mathcal{P})=\rk(\mathcal{G}(d, n)\otimes\mathcal{P})=d\cdot\rk(\mathcal{P})\,.
\end{equation}
\end{lemma}
\begin{proof}
Regarding Eq. \eqref{rank-inequalities}, we have
\[
\rk(\mathcal{G}(d, n) \boxtimes \mathcal{P}) \leq \rk(\mathcal{G}(d, n) \otimes \mathcal{P}) \leq d \cdot \rk(\mathcal{P})\,.
\]

To get a lower bound, note that the tensors $\mathcal{G}(d, n) \boxtimes \mathcal{P}$ and $\mathcal{P}^{\oplus d}$ are (isometrically) equivalent.
More specifically, $\mathcal{P}^{\oplus d}$ is transformed into $\mathcal{G}(d,n)\boxtimes\mathcal{P}$ by applying to each tensor factor the isomorphism $V_j^{\oplus d} \xrightarrow{\sim} \mathbbm{C}^d \otimes V_j$ sending $(v_0,\ldots,v_{d-1})$ to $\sum_{k=0}^{d-1} e_k \otimes v_k$.

We then apply the previous corollary repeatedly to the direct sum $\mathcal{P}^{\oplus d}$ to obtain the lower bound of the rank $d\cdot\rk(\mathcal{P})$.
\end{proof}

\begin{corollary}\label{cor:Tensor-product-GHZ-W-M-L}
From \Cref{lem:GHZ-P} we have the following corollaries:
\begin{enumerate}
    \item $\rk\big(\mathcal{G}(d,n)\boxtimes\mathcal{W}_n\big)=\rk\big(\mathcal{G}(d,n)\otimes\mathcal{W}_n\big)=nd$.
    \item $\rk\big(\mathcal{G}(d,n)\boxtimes\mathcal{W}_n^{\boxtimes 2}\big)=\rk\big(\mathcal{G}(d,n)\otimes\mathcal{W}_n^{\boxtimes 2}\big) = (3n - 2)d$.
    \item $\rk\big(\mathcal{G}(d_1,n)\boxtimes\mathcal{L}(d_2,n)\big)=\rk\big(\mathcal{G}(d_1,n)\otimes\mathcal{L}(d_2,n)\big)$
    \\
    $=d_1\big((n-1)d_2-n+2\big)$.
    \item $\rk\big(\mathcal{G}(d_1,n)\boxtimes\mathcal{M}(d_2,n)\big)=\rk\big(\mathcal{G}(d_1,n)\otimes\mathcal{M}(d_2,n)\big)$
    \\
    $=d_1\big((n-1)d_2-n+2\big)$.
    \item $\rk\big(\mathcal{G}(d_1,n)\boxtimes\mathcal{N}(d_2,n)\big)=\rk\big(\mathcal{G}(d_1,n)\otimes\mathcal{N}(d_2,n)\big)$
    \\
    $=d_1\big((n-1)d_2-n+2\big)$.
\end{enumerate}
\end{corollary}

In particular, for $\mathcal{G}(d,n) \boxtimes \mathcal{W}_n$ we answer an open question posed in Ref. \cite{CF18}.

The same lower bound method can be applied not only to direct sums but also to a more general class of block tensors.

\begin{definition}[Block pyramidal tensor]
A tensor $\mathcal{Q} \in (U_1 \oplus V_1) \otimes\cdots\otimes(U_n \oplus V_n)$ is a \emph{block pyramidal tensor} if $\mathcal{Q}\in U_1\otimes\cdots\otimes U_n\oplus(U_1\oplus V_1)\otimes\cdots\otimes(U_{n-1}\oplus V_{n-1})\otimes V_n$. Denote by $\pi_{U_k}$ and $\pi_{V_k}$, the canonical projections associated with the summands of $U_k \oplus V_k$. The tensor $(\pi_{U_1} \otimes \dots \otimes \pi_{U_n}) \mathcal{Q}$ is called the \emph{head block} of $\mathcal{Q}$, and $(\pi_{V_1} \otimes \dots \otimes \pi_{V_n}) \mathcal{Q}$ is called the \emph{step block} of $\mathcal{Q}$.
\end{definition}

\begin{theorem}\label{theo:PTLB-pyramidal}
Let $\mathcal{Q} \in (U_1 \oplus V_1) \otimes\cdots\otimes(U_n \oplus V_n)$ be a block pyramidal tensor with the head block $\mathcal{T}\in U_1\otimes\cdots\otimes U_n$ and the step block $\mathcal{P}\in V_1\otimes\cdots\otimes V_n$.
If $\mathcal{P}$ is a persistent tensor and $\dim V_k=d_k$, then
\begin{equation}\label{PTLB-DirectSum}
\rk(\mathcal{Q})\geq\rk(\mathcal{T})+\sum_{k=1}^{n-1}(d_k-1)+1\,.
\end{equation}
\end{theorem}
\begin{proof}
The proof is the same as for~\Cref{theo:PTLB-DirectSum}, with $\mathcal{T} \oplus \mathcal{P}$ replaced by $\mathcal{Q}$.
\end{proof}

\chapter{Summary and Outlook}

\epigraph{``pauca sed matura''}{\textit{Carl Friedrich Gau{\ss}}}

\section{Conclusions}

This dissertation links algebraic geometry to entanglement theory to answer two central problems, that is, classification of multipartite entanglement and SLOCC interconversion between different resources.

In \cref{chap2}, mathematical formulation and the fundamental concepts of quantum information theory are introduced. Mainly, we have introduced the description of entangled states and the mathematical tools and concepts needed for entanglement classification in bipartite and multipartite quantum systems.

\cref{chap3} is devoted to introduce the algebro-geometric tools we have used for addressing the above mentioned problems in entanglement theory. Namely, we have introduced $k$-secant varieties of Segre embedding, $\ell$-multilinear ranks, and tensor rank and border rank that are SLOCC invariants.

In \cref{chap4}, we focused on the central problem of quantum information theory which concerns the classification of multipartite entanglement. To this end, we presented a fine-structure entanglement classification that can be interpreted as a Mendeleev table, where the structure of an element can be used as a core structure of another. As a matter of fact, for $n$-qubit classification we are fixing the elements in $k$-secant families [see Eqs. (\ref{general2})-(\ref{general4})], and, indeed, one can always use $n$-qubit classification as a partial classification of ($n+1$)-qubit case. Then, we just need to find the elements of new $k$-secants for the classification of ($n+1$)-qubit states. As we have already illustrated in our examples, new $k$-secants' elements can be identified by joining points of previous $k$-secant families, and considering all tangential varieties. More interesting is that joining randomly chosen elements from both $\sigma_{i}$ and $\sigma_{j}$ would land in $\sigma_{i+j}\setminus\sigma_{i+j-1}$, with probability one \cite{Landsberg}. Therefore, one can always create a general element in a desired secant family. In addition, all the genuine entangled states in higher secants and tangents can be, respectively, considered as the generalizations of ${\rm{GHZ}}$ and ${\rm{W}}$ states in two-secant and two-tangent [one can also see a footprint of ${\rm{GHZ}}$ and ${\rm{W}}$ states in the higher secants and tangents from Eq. (\ref{general4})].

To clearly show the potential of our approach, we have elaborated the classification for $n=5$ qubits in \cref{subsec.4.2.4}. We believe the method can be extended to find a classification of multipartite entanglement for higher dimensional systems as we have already provided a conjecture for the classification of symmetric multiqudit states.

Within \cref{chap5} we follow the problem of multipartite entanglement classification in higher dimensional systems. Indeed, using algebraic-geometric tools, we studied entanglement characterization of three-qudit $\mathbbm{C}^d\otimes\mathbbm{C}^d\otimes\mathbbm{C}^d$ systems. Specifically, we used secant varieties and one-multiranks that are SLOCC invariants, to present entanglement classification of three-qudit entanglement as a generalization to our previous work in \cref{chap4}. As a prominent example we have provided a fine-structure classification for three-qutrit pure states. This  can be considered as the core  classification of tripartite $\mathbbm{C}^d\otimes\mathbbm{C}^d\otimes\mathbbm{C}^d$ states as well as ($n\geq{4}$)-qutrit states. Indeed, with this method, one can always use $n$-qudit classification as a partial classification of $(n + 1)$-qudit systems. Outside this core classification, the results for larger systems ($d>3$ and/or $n>3$) have been derived by also relying on conjectures of tensor theory.

Not only is our classification operationally meaningful as it quantifies entanglement as a resource but also this classification can be seen in terms of the order of entanglement strength from Segre variety that contains no entanglement, to the higher secant family. Indeed, the tools we have used for entanglement characterization, i.e., tensor rank and border rank, can be seen as the generalized Schmidt rank and its counterpart. More precisely, the Schmidt measure that quantify entanglement of a multipartite state $|\psi\rangle$ can be defined as the logarithm of the rank of the tensor $\psi$. On the other hand, generic tensor rank can be considered as a discrete measure of entanglement. Based on this fact, one can conclude that symmetric states are much less entangled than general states. Although we have shown this fact for multiqubit systems in \cref{chap4}, and for three-qudit and multiqutrit systems in Sec. \ref{sec.5.4},  this is a general fact since generic symmetric tensor rank has a polynomial growth while generic tensor rank has an exponential growth.

In \cref{chap6}, we focused on a central problem in quantum information theory which concerns the interconversion between different resources by SLOCC and asymptotic SLOCC. The tensor rank, known as the generalized Schmidt rank, plays an important role in the study of the classification and transformation of multipartite entanglement. Furthermore, the tensor border rank is a powerful tool for studying degeneration and asymptotic SLOCC transformation.

In this chapter, we have introduced a new class of tensors in $\otimes_{i=1}^n\mathbbm{C}^{d_i}$ that we call persistent tensors and have constructed a lower bound of the tensor rank for this class. We also have introduced several families of persistent tensors in $\otimes^n\mathbbm{C}^d$ where the lower bound is tight. Moreover, we have studied the asymptotic SLOCC transformation of these families of minimal-rank persistent tensors to each other. Showing that the border rank of these families of minimal-rank persistent tensors is $d$, we have concluded that geometrically they are in orbit closure of the $n$-qudit $\rm{GHZ}$ states, and we can consider them as generalizations of the $n$-qubit $\rm{W}$ state within multiqudit systems. Consequently, we have shown that these generalizations of the $\rm{W}$ state can be approximated with arbitrary precision by states in the $\rm{GHZ}$ orbit via asymptotic SLOCC. Actually, we have shown that the rate of asymptotic SLOCC transformation from an $n$-qudit $\rm{GHZ}$ state into each generalization of the $\rm{W}$ state is one. Furthermore, we have proved that the achieved lower bound can be extended to direct sums with persistent summands and to even more general combinations of tensors, which we call block pyramidal tensors. we show that the tensor rank is multiplicative under the Kronecker product and the tensor product of minimal-rank persistent tensors with the GHZ tensor.

\section{Potential applications and open problems}

Here, not only do we provide potential applications and open problems that originates from our research works in Refs. \cite{GMO20,GM21,GL22}, but we also outline developments in slightly different directions.

\begin{enumerate}

\item[A.] We emphasize the operational meaning of the proposed classification in \cref{chap4} and \cref{chap5} as it somehow measures the amount of entanglement in multipartite systems, where a well-established entanglement monotone is still lacking. Furthermore, the tools we proposed for entanglement characterization can also be useful as states complexity measures, since they share analogies with the tree size method presented in Refs. \cite{LCWRS14, CLS15}. Indeed, the notion of tree size can be understood as the length of the shortest bracket representation of a state, which in turn is the tensor rank. Additionally, they offer a perspective for evaluating the computational complexity of quantum algorithms, by analyzing how the classes change while running them (see also Refs. \cite{HJN16, JH18}).

\item[B.] Still along the applicative side, since in a system with a growing number of particles, most of the states cannot be realistically prepared and will thus never occur neither in natural nor in engineered quantum systems \cite{WGE17}, our coarse-grain classification in \cref{chap4} and \cref{chap5} could provide a tool to singling out states that we do effectively need (e.g., a representative of each family and/or subfamily). For instance, ${\rm{W}}$ states that are living in a lower secant, although useful for many processes like the realization of quantum memories \cite{LST04}, are known to be more robust but not very entangled. Hence, for other tasks, like quantum teleportation, the usage of ${\rm{GHZ}}$ states that are more entangled has been suggested \cite{ZCZYBP09}, i.e., move up from the tangent to the proper secant of the lower secant family. Indeed ${\rm{GHZ}}$ states provide some degree of precision in frequency measurements \cite{BIWH96}, but in Ref. \cite{HMPEPC97} this is increased (even in the presence of decoherence), using a state lying in higher secant. Hence, it seems that higher secant families offer better estimation accuracy in quantum metrology (see also Refs. \cite{HLKSWWPS12, Toth12}). Also, our results about the cluster state $|{\rm{Cl}}_{4}\rangle$, supports the idea that states living in higher secants are more suitable as a resource for measurement-based quantum computation \cite{BBDRV09}. Actually, going to higher secants makes states more entangled and at the same time also more robust (at least with respect to losses) because even losing one qubit there would always be some residual entanglement left.

\item[C.]
Based on our classification in \cref{chap4} and \cref{chap5}, one can construct new entanglement witnesses to be used for detecting entanglement in multipartite mixed states (where state tomography is not efficient). Already, in Ref. \cite{BSV12} it has been shown that one can find, following a geometric approach, device-independent entanglement witnesses that allow us to discriminate between various types of multiqubit entanglement. We believe that this could also pave the way to extend this classification to mixed states, and to study the entanglement depth \cite{LRBPBG15, L-etal-prx-18} of each class.

\item[D.]
Along the potential applications mentioned in the conclusion of \cref{chap4} that can also be considered for the higher dimensional systems, it is captivating that this kind of classification can also be considered as a reference to study SLOCC and asymptotic SLOCC interconversions among different resources based on tensor rank \cite{CDS08,CCDJW10,YCGD10} and border rank \cite{YGD14,VC15}, respectively.

\item[E.] 
It would be also desirable to extend the proposed classification method in \cref{chap4} and \cref{chap5}, to mixed states. This goal will be pursued starting from possible connections with the Schmidt number vector classification of Refs. \cite{HD13, HPD13}. Indeed, the idea is utilizing the generalization of Schmidt rank for pure states to Schmidt number for mixed states \cite{TH00}. So, the Schmidt rank vector is nothing but the multirank we have used in our method and the Schmidt number vector is a tuple of digits obtained from a particular ensemble decomposition of a given mixed state \cite{HD13}.

\item[F.] Concerning persistent tensors as a new class of tensors with a lower bound of the tensor rank, it would be interesting not only to study this class of tensors more deeply from the complexity theory point of view, but also to study their properties concerning their application in quantum technology. We believe that any application of multiqubit $\rm{W}$ state can be studied for its generalizations within multiqudit systems. Indeed, qudit provides several advantages over qubit. For instance, using qudits as building blocks of a quantum circuit can provide a reduction in circuit complexity, since they provide a larger Hilbert space and hence a larger capacity to store and process information \cite{WHSK20}. In addition, several benefits of qudits have been proposed, including better noise resistance, higher information coding capacity, stronger nonlocality, and enhanced security, have been proposed \cite{CMM12,CBKG02,DWS03,SS11,ZZLZCLZ17}.

\item[G.] 
Although we have checked many examples that the Kronecker product of two persistent tensors is a persistent tensor, we leave the proof as an open problem.

\begin{conjecture}\label{conj:P1P2}
If $\mathcal{P}_1\in U_1\otimes\cdots\otimes U_n$ and $\mathcal{P}_2\in V_1\otimes\cdots\otimes V_n$ are two persistent tensors, where $U_1,\ldots,U_n$ and $V_1,\ldots,V_n$ are vector spaces with $\dim U_k=d_k$ and $\dim V_k=d'_k$, respectively, then their Kronecker product is also a persistent tensor. Therefore,
\begin{equation*}
\rk(\mathcal{P}_1\boxtimes\mathcal{P}_2)\geq\sum_{k=1}^{n}(d_k+d'_k-1)+1\,.
\end{equation*}
\end{conjecture}

\item[H.] Let consider $n$-partite quantum state $|\psi\rangle$ in the Hilbert space $\mathcal{H}_n=\otimes_{i=1}^n\mathbbm{C}^{d_i}$. A proper subspace $\mathtt{C}\subset\mathcal{H}_n$ is called a Completely Entangled Subspace (CES) if it contains no fully separable state. A set of pairwise orthogonal fully product vectors $\{|\psi_l\rangle\equiv\otimes_{s=1}^{n}|\varphi_s^l\rangle_{A_s}\}_{l=1}^{u}$ that spans a proper subspace of $\mathcal{H}_n$ (i.e., $u<\rm{dim}\,\mathcal{H}_n$) is called Unextendible Product Basis (UPB) if its orthogoanl complement subspace is a CES \cite{BDMSST99}. We can provie a nonorthogonal UPB (nUPB) for $n$-qubit systems by using Veronese map such that the CES is of maximum dimension. We also can generalize this method to multipartite systems. We are working on this problem to find a minimum-dimension nUPB or UPB which by construction would give rise to a genuinely entangled subspace (a subspace that only contains genuinely entangled states) as the orthocomplementary subspace \cite{GM-progress}.

\end{enumerate}

\newpage
\addcontentsline{toc}{chapter}{Bibliography}
\bibliographystyle{alpha}

\begin{thebibliography}{100}
\bibitem{example} {
For example, for three qutrits where the partition is 1|23,
let us consider two states, one from the family span$\{\Psi_1,\Psi_2\}$ and another from the family span$\{\Psi_0,\Psi_0,\Psi_2\}$ as follow:
\begin{align*}
&|000\rangle+|011\rangle+|022\rangle+|101\rangle+|112\rangle= \\
&|0\rangle(|00\rangle+|11\rangle+|22\rangle)+|1\rangle(|01\rangle+|12\rangle)\,,
\end{align*}
and
\begin{align*}
&|000\rangle+|011\rangle+|022\rangle+|101\rangle+|202\rangle= \\
&|0\rangle(|00\rangle+|11\rangle+|22\rangle)+|1\rangle|01\rangle+|2\rangle|02\rangle\,.
\end{align*}
Now, let us consider the partition 12|3 for the above states. We will have
\begin{align*}
&|000\rangle+|011\rangle+|022\rangle+|101\rangle+|202\rangle= \\
&|00\rangle|0\rangle+(|01\rangle+|10\rangle)|1\rangle+(|02\rangle+|11\rangle)|2\rangle\,,
\end{align*}
and
\begin{align*}
&|000\rangle+|011\rangle+|022\rangle+|101\rangle+|202\rangle= \\
&|00\rangle|0\rangle+(|01\rangle+|10\rangle)|1\rangle+(|02\rangle+|20\rangle)|2\rangle\,.
\end{align*}
As we see by changing the partition and algorithm two states from different families go to the family span$\{\Psi_0,\Psi_1,\Psi_1\}$ where $\Psi_0$ denotes a fully separable state while $\Psi_1$ and $\Psi_2$ denote two inequivalent fully inseparable states.
}
\end{thebibliography}


\appendix
\chapter*{Appendix \\ $\quad$ \\ List of all publications and preprints}
\addcontentsline{toc}{chapter}{List of all publications}

\begin{enumerate}
\item[{[1]}]
\textbf{M. Gharahi}, and V. Lysikov, 
\textit{Persistent Tensors and Multiqudit Entanglement Transformation},
\href{https://doi.org/10.22331/q-2024-01-31-1238}{Quantum \textbf{8}, 1238 (2024)}.
\item[{[2]}]
\textbf{M. Gharahi} and S. Mancini,
\textit{Algebraic-geometric characterization of tripartite entanglement}, 
\href{https://doi.org/10.1103/PhysRevA.104.042402}{Phys. Rev. A \textbf{104}, 042402 (2021)}.
\item[{[3]}]
\textbf{M. Gharahi}, S. Mancini, and G. Ottaviani,
\textit{Fine-structure classification of multiqubit entanglement by algebraic geometry}, 
\href{https://doi.org/10.1103/PhysRevResearch.2.043003}{Phys. Rev. Research \textbf{2}, 043003 (2020)}.
\item[{[4]}]
\textbf{M. Gharahi} and S. Mancini,
\textit{Comment on ``Inductive entanglement classification of four qubits under stochastic local operation and classical communication''}, 
\href{https://doi.org/10.1103/PhysRevA.98.066301}{Phys. Rev. A \textbf{98}, 066301 (2018)}.
\item[{[5]}]
L. Anticoli and \textbf{M. Gharahi},
\textit{Modeling tripartite entanglement in quantum protocols using evolving entangled hypergraphs}, 
\href{https://doi.org/10.1142/S0219749918500557}{Int. J. Quant. Inf. \textbf{16}, 1850055 (2018)}.
\item[{[6]}]
\textbf{M. Gharahi} and S. J. Akhtarshenas,
\textit{Entangled graphs: a classification of four-qubit entanglement}, 
\href{https://doi.org/10.1140/epjd/e2016-60729-1}{Eur. Phys. J. D \textbf{70}, 54 (2016)}.

\end{enumerate}


\newpage
\begin{titlepage}
	\begin{flushright}
	\mbox{}
    \vfill
	{\Large ...and I did it in the projective Hilbert space.\par}
	\end{flushright}
\end{titlepage}
	

\end{document}